\documentclass[10pt]{report}%
\pdfoutput=1
\input{Qcircuit}
\usepackage{epsfig}
\usepackage{amsfonts}
\usepackage{amsmath}
\usepackage{amsthm}
\usepackage{amssymb}
\usepackage{graphicx}
\usepackage{verbatim}
\usepackage{USCthesis2004}%
\providecommand{\U}[1]{\protect\rule{.1in}{.1in}}

\renewcommand{\bibname}{BIBLIOGRAPHY}

\newtheorem{definition}{Definition}[section]
\newtheorem{theorem}{Theorem}[section]

\newtheorem{lemma}{Lemma}[section]
\newtheorem{corollary}{Corollary}[section]
\newtheorem{conjecture}{Conjecture}[section]
\theoremstyle{remark}
\newtheorem{example}{Example}[section]
\newtheorem{remark}{Remark}[section]
\renewcommand{\thechapter}{\arabic{chapter}}

\begin{document}

\title{Quantum Coding with Entanglement}
\author{Mark McMahon Wilde}
\date{August 2008}
\maketitle

\begin{preface}
\pagebreak   
\mbox{}
\vspace{0.5in}
\addcontentsline{toc}{chapter}{Dedication}
\begin{center}
{\Large \bf Dedication}
\end{center}
\vspace{2in}
\begin{center}
{\it \large To my parents\\
Gregory and Sharon Wilde\\
And to my grandparents\\
Joseph and Rose McMahon\\
Norbert Jay and Mary Wilde}
\end{center}

\pagebreak   
\mbox{}
\vspace{0.5in}
\addcontentsline{toc}{chapter}{Acknowledgements}
\begin{center}
{\Large \bf Acknowledgements}
\end{center}
The completion of this Ph.D. has been a long and difficult road
with many wonderful memories along the way. I can truly say that
there is no way I could have completed it without the help and encouragement
from many people.

I first thank my advisor Todd Brun. His help and deep knowledge have been indispensable throughout
and I am grateful that he took me on as his student. Our many research meetings
have been valuable and I am thankful that Todd gave
his time for me. I cannot imagine what a better advisor would be like.
I can recall visiting USC in
March 2004 for a Ph.D. recruitment trip. I saw a posting on the wall
that said there would be a seminar on quantum computing and that Todd Brun was
the host. It was from this point that I became fascinated by quantum information,
had a desire to contribute original research to the field, and thought it would be great if Todd Brun
were my advisor. I look back now and am pleased that all of these dreams have
come to fruition.

I next thank Igor Devetak. He was the first to teach me quantum information theory
and I am grateful for the many times that he met with me. I met with Igor about one hour
per week in the academic semesters of 2005-2006 and these formative
meetings helped shape and refine my knowledge of quantum information theory. I am also
grateful that Igor invited me to join Todd and Igor's research group meetings.

I thank Jonathan Dowling for accepting me into his research group at LSU during
the summers of 2005-2006. I learned so much about quantum optics
by conversing with him and the other members of his lab. The research environment that Jon creates
by having so many seminars and bringing in visitors is a great one to be involved in.

I am grateful to Daniel Lidar for his collaboration and for teaching his course on quantum error
correction at USC. After this class, I can say that my research really took off because I
was able to solidify my knowledge of quantum error correction.

I am in debt to Bart Kosko for working with me during the initial period of my Ph.D.
He gave valuable advice for how to be a good teacher and researcher and
taught me how to have the discipline required of a good researcher. When I was a teaching assistant,
I simply tried to teach like he does and I think I won the Best Teaching Assistant Award
for this reason.
 
I thank Diane
Demetras for being a great grad student ``mom'' away from home.
I owe thanks to Antonio Ortega for initially giving me advice. I also thank the fifth floor
staff---Milly Montenegro, Mayumi Thrasher, and Gerrielyn Ramos---for providing a
pleasant working environment and for filing all those technical reports that I asked to file.
I hope our softball team wins! Shock 'em, Volts!

I thank my collaborators Hari Krovi and Austin Lund for helping me along the way. Austin
was invaluable during the first summer at LSU when he answered my
questions. Hari was a big help in all of our collaborations and especially when we were
first trying to figure out what Todd, Igor, and Min-Hsiu had done with
entanglement-assisted quantum codes. I also thank everyone else in the
group, Min-Hsiu Hsieh, Zhicheng Luo, Ognyan Oreshkov, Shesha Raghunathan, Bilal Shaw, Martin Varbanov, for being a good group of guys
to work with. I learned a lot from our debates and conversations.

I finally thank my family for their support. I especially thank my dad ``Jack'' for the
many long Ph.D. conversations that we have had on the phone and for being excited about my research.
I thank all of the wonderful friends I have made while in the City of Angels,
especially Ali, Ben, Mary L., and Mary B., and I thank
Fr. Lawrence for our spiritual conversations.
I lastly give a ``shout out'' to Santa Monica power yoga studios for providing a great environment
for destressing and to Pete Carroll for having an awesome Trojan football team! Fight on!

\begin{singlespace}
\tableofcontents    
\pagebreak
\addcontentsline{toc}{chapter}{List of Tables}
\listoftables     
\pagebreak
\addcontentsline{toc}{chapter}{List of Figures}
\listoffigures    
\end{singlespace}
\pagebreak
\begin{abstract}
Quantum error-correcting codes will be the ultimate enabler of a future quantum
computing or quantum communication device. This theory forms the cornerstone
of practical quantum information theory. We provide several contributions to the
theory of quantum error correction---mainly to the theory of ``entanglement-assisted''
quantum error correction where
the sender and receiver share entanglement in the form of entangled bits (\textit{ebits})
before quantum communication begins. Our first contribution is an algorithm for
encoding and decoding an entanglement-assisted quantum block code. We then give several formulas
that determine the optimal number of ebits for an entanglement-assisted code. The major contribution of this thesis
is the development of the theory of entanglement-assisted quantum convolutional coding. A convolutional code is
one that has memory and acts on an incoming stream of qubits. We explicitly show how to encode
and decode a stream of information qubits with the help of ancilla qubits and ebits.
Our entanglement-assisted convolutional codes include those with a Calderbank-Shor-Steane structure and those with a more general
structure. We then formulate convolutional protocols that correct errors in noisy entanglement.
Our final contribution is a unification of the theory of quantum error correction---these
unified convolutional codes exploit all of the known resources
for quantum redundancy.
\end{abstract}
\addcontentsline{toc}{chapter}{Abstract}
\end{preface}

\chapter{Introduction}

\begin{saying}
If computers that you build are quantum,\\
Then spies of all factions will want 'em.\\
Our codes will all fail,\\
And they'll read our email,\\
Till we've crypto that's quantum, and daunt 'em.\\
---Jennifer and Peter Shor
\end{saying}

\PARstart{Q}{uantum} computation and communication have enormous potential to
revolutionize the way we compute and communicate. From communicating by
teleportation of quantum information \cite{PhysRevLett.70.1895}\ to breaking
RSA\ encryption \cite{shor94,shor97}, quantum technological breakthroughs will
change society in a way that is difficult to predict.

One might say that the field of quantum computation began when Richard Feynman
and others suggested that a computer operating according to the principles of
quantum mechanics would have advantages over a computer operating according to
the laws of classical physics \cite{ijtp1982feynman,benioff,manin}. They
suggested that it might be able to simulate quantum-mechanical processes more
efficiently than a classical computer could. Later work then showed that it
was possible to achieve this simulation speedup \cite{science1996lloyd}. Other
major advances in the theory of quantum computing followed---examples are
Shor's algorithm for breaking RSA\ encryption\ \cite{shor94,shor97}\ and
Grover's algorithm for searching a database \cite{grover96,grover97}.

The field of modern quantum communication began with the discovery of quantum
key distribution \cite{bb84}. Major discoveries for quantum communication
followed---examples are the ability to send a quantum bit using two classical
bits and a bit of entanglement (quantum teleportation
\cite{PhysRevLett.70.1895}) and the ability to send two classical bits by
sending a quantum bit and consuming a bit of entanglement (quantum superdense
coding \cite{PhysRevLett.69.2881}). Quantum information theorists then began
to think more deeply about the ways in which we could combine the resources of
classical communication, quantum communication, and entanglement to formulate
new communication protocols.

In spite of these spectacular advances for the theories of quantum computation
and communication, a dark shadow lingered over them. A few authors disputed
that reliable quantum computation or communication would be possible because
small quantum errors would accumulate as the computation proceeds or channel
errors would destroy the quantum bits as we communicate them
\cite{landauer1982}. Rolf Landauer even urged his colleagues to include the
following disclaimer in their papers on quantum
computation:\ \textquotedblleft This scheme, like all other schemes for
quantum computation, relies on speculative technology, does not in its current
form take into account all possible sources of noise, unreliability and
manufacturing error, and probably will not work.\textquotedblright\ Various
other obstacles such as the no-cloning theorem \cite{nat1982} and measurement
destroying a quantum state seemed to pose an insurmountable barrier to a
protocol for quantum error correction.

Despite the aforementioned obstacles, Shor demonstrated the first quantum
error-correcting code that reduces the negative effects of decoherence on a
quantum bit \cite{PhysRevA.52.R2493}. Shor's code overcame all of the above
difficulties and established the basic principles for constructing a general
theory of quantum error correction
\cite{thesis97gottesman,PhysRevLett.78.405,ieee1998calderbank}. Shor's code
exploits many of the signature principles of quantum
mechanics:\ superposition, entanglement, unitary evolution, and measurement.
Mermin proclaims it a \textquotedblleft miracle\textquotedblright\ that
quantum error correction is even possible \cite{book2007mermin}.

\section{Quantum Error Correction}

Quantum error correction theory
\cite{PhysRevA.52.R2493,PhysRevA.54.1098,PhysRevLett.77.793,thesis97gottesman,ieee1998calderbank}%
\ now plays a prominent role in the practical realization and engineering of
quantum computing and communication devices. The first quantum
error-correcting codes
\cite{PhysRevA.52.R2493,PhysRevA.54.1098,PhysRevLett.77.793,ieee1998calderbank}%
\ are strikingly similar to classical block codes \cite{book1983code} in their
operation and performance. Quantum error-correcting codes restore a noisy,
decohered quantum state to a pure quantum state. Any future quantum
information processing device will operate faithfully only if it employs an
error correction scheme. This scheme can be an active scheme
\cite{thesis97gottesman},\ a passive scheme
\cite{PhysRevLett.79.3306,mpl1997zanardi,PhysRevLett.81.2594},\ or a
combination of both techniques
\cite{kribs:180501,qic2006kribs,poulin:230504,isit2007brun,hsieh:062313}.

The field of quantum error correction has rapidly expanded in recent years
because there are so many ways in which one can correct quantum error. We
briefly introduce three notable developments in the theory---these three are
by no means exhaustive. The developments include the stabilizer formalism, the
entanglement-assisted stabilizer formalism, and the convolutional stabilizer formalism.

\subsection{Stabilizer Quantum Coding}

\label{sec:intro-stab}Gottesman formalized the theory of quantum block coding
by establishing the stabilizer formalism \cite{thesis97gottesman}. A
\textit{stabilizer} quantum error-correcting code appends ancilla qubits to
information qubits that we want to protect. A unitary encoding circuit rotates
the Hilbert space of the information qubits into a subspace of a larger
Hilbert space. This highly entangled, encoded state corrects for local noisy
errors. A quantum error-correcting code makes quantum computation and quantum
communication practical by providing a way for a sender and receiver to
simulate a noiseless qubit channel given a noisy qubit channel that has a
particular error model.

The stabilizer theory of quantum error correction allows one to import some
classical binary or quaternary codes for use as a quantum code
\cite{ieee1998calderbank}. The Calderbank-Shor-Steane
(CSS)\ construction is the name for the method
for importing classical binary codes \cite{book2000mikeandike}. This idea of importing codes is useful because quantum code
designers can utilize classical block codes with high performance to construct
quantum codes with high performance. The only \textquotedblleft
catch\textquotedblright\ when importing is that the classical code must
satisfy the dual-containing or self-orthogonality constraint. Researchers have
found many examples of classical codes satisfying this constraint
\cite{ieee1998calderbank}, but most classical codes do not.

\subsection{Entanglement-Assisted Stabilizer Quantum Coding}

\label{sec:intro-EA}Bowen was the first to extend the stabilizer formalism by
providing an example of a code that exploits entanglement shared between a
sender and a receiver \cite{PhysRevA.66.052313}. The underlying assumption of
Bowen's code is that the sender and receiver share a set of noiseless ebits
(entangled qubits)\ before quantum communication begins. Many quantum
protocols such as teleportation \cite{PhysRevLett.70.1895}\ and superdense
coding \cite{PhysRevLett.69.2881}\ are \textquotedblleft
entanglement-assisted\textquotedblright\ protocols because they assume that
noiseless ebits are available.

Brun, Devetak, and Hsieh extended the standard stabilizer theory of quantum
error correction by developing the entanglement-assisted stabilizer formalism
\cite{science2006brun,arx2006brun}. They included entanglement as a resource
that a sender and receiver can exploit for a quantum error-correcting code.
They provided a \textquotedblleft direct-coding\textquotedblright%
\ construction in which a sender and receiver can use ancilla qubits and ebits
in a quantum code. An ebit is a nonlocal bipartite Bell state%
\[
\left\vert \Phi^{+}\right\rangle =\left(  \left\vert 00\right\rangle
+\left\vert 11\right\rangle \right)  /\sqrt{2}.
\]
Gottesman later showed that their construction is optimal \cite{unpub2007got}%
---it gives the minimum number of ebits required for the entanglement-assisted
quantum code.

The major benefit of the entanglement-assisted stabilizer formalism\ is that
we can construct an entanglement-assisted quantum code from two arbitrary
classical binary block codes or from an arbitrary classical quaternary block
code. The rate and error-correcting properties of the classical codes
translate to the resulting quantum codes. The entanglement-assisted stabilizer
formalism may be able to reduce the problem of finding high-performance
quantum codes approaching the quantum capacity
\cite{PhysRevA.55.1613,capacity2002shor,ieee2005dev,PhysRevLett.83.3081,ieee2002bennett}%
\ to the problem of finding good classical linear codes approaching the
classical capacity \cite{book1991cover}. The entanglement-assisted stabilizer
formalism thus is a significant and powerful extension of the stabilizer formalism.

\subsection{Convolutional Stabilizer Quantum Coding}

\label{sec:intro-conv}Another extension of the theory of quantum error
correction protects a potentially-infinite stream of quantum information
against the corruption induced by a noisy quantum communication channel
\cite{PhysRevLett.91.177902,arxiv2004olliv,isit2006grassl,ieee2006grassl,ieee2007grassl,isit2005forney,ieee2007forney,cwit2007aly,arx2007aly,arx2007wildeCED,arx2007wildeEAQCC,arx2008wildeUQCC}%
. Quantum convolutional codes have numerous benefits. The periodic structure
of the encoding and decoding circuits for a quantum convolutional code ensures
a low complexity for encoding and decoding while also providing higher
performance than a block code with equivalent encoding complexity
\cite{ieee2007forney}. The encoding and decoding circuits have the property
that the sender Alice and the receiver Bob can respectively send and receive
qubits in an \textquotedblleft online\textquotedblright\ fashion. Alice can
encode an arbitrary number of information qubits without worrying beforehand
how many she may want to send over the quantum communication channel.

We believe that quantum convolutional coding theory is one step along the way
to finding quantum error-correcting codes that approach the capacity of a
noisy quantum communication channel for sending quantum information
\cite{PhysRevA.55.1613,capacity2002shor,ieee2005dev,qcap2008first,qcap2008second,qcap2008third,qcap2008fourth}%
. Poulin et al. have recently incorporated some of this well-developed theory
of quantum convolutional coding into a theory of quantum serial-turbo coding
\cite{arx2007poulin}\ with the goal of designing quantum codes that come close
to achieving capacity.

The development of quantum convolutional codes has been brief but successful.
Chau was the first to construct some quantum convolutional codes
\cite{PhysRevA.58.905,PhysRevA.60.1966}, though some authors
\cite{ieee2007forney}\ argue that his construction does not give a true
quantum convolutional code. Several authors have established a working theory
of quantum convolutional coding based on the stabilizer formalism and
classical self-orthogonal codes over the finite field $GF\left(  4\right)  $
\cite{PhysRevLett.91.177902,arxiv2004olliv,isit2005forney,ieee2007forney}.
Others have also provided a practical way for realizing \textquotedblleft
online\textquotedblright\ encoding and decoding circuits for quantum
convolutional codes
\cite{PhysRevLett.91.177902,arxiv2004olliv,ieee2006grassl,isit2006grassl}.

Forney et al. have determined a method for importing an arbitrary classical
self-orthogonal quaternary code for use as a quantum convolutional code
\cite{isit2005forney,ieee2007forney}.\ The technique is similar to that for
importing a classical block code as a quantum block code
\cite{ieee1998calderbank}. Forney et al. showed how to produce quantum
convolutional codes from classical convolutional codes by extending the ideas
from the CSS\ construction to the convolutional setting \cite{ieee2007forney};
but again, these imported classical convolutional codes have to satisfy the
restrictive dual-containing constraint in order to form valid quantum codes.
The dual-containing constraint is actually quite a bit more restrictive in the
convolutional case because each generator for the quantum convolutional code
must commute not only with the other generators, but it must commute also with
any arbitrary shift of itself and any arbitrary shift of the other generators.
Forney et al. performed specialized searches to determine classical quaternary
codes that satisfy the restrictive self-orthogonality constraint
\cite{ieee2007forney}.

\section{Connections to Quantum Shannon Theory}

Quantum error correction theory gives practical ways to build codes for
protection of quantum information against decoherence. A quantum code
typically encodes some number $k$ of information qubits into a larger number
$n$ of physical qubits. We say that the rate of quantum communication for this
code is $k/n$.

One may ask the question: what is the maximum rate of quantum communication
for a given noisy quantum communication channel? This question is the
fundamental question of quantum Shannon theory and the answer to it is a
quantity known as the capacity of the quantum channel for quantum
communication, or more simply, the quantum capacity
\cite{PhysRevA.55.1613,capacity2002shor,ieee2005dev,arx2005dev}. Quantum
Shannon theory has many other questions. What is the rate of classical
communication for a noisy quantum communication channel?\ What is the rate if
the sender and receiver have access to entanglement?\ What is the
rate of simultaneous classical and quantum communication? The list of questions goes on and on.

The goal of quantum Shannon theory is to quantify the amount of quantum
communication, classical communication, and entanglement required for various
information processing tasks
\cite{PhysRevA.55.1613,capacity2002shor,ieee2005dev,arx2005dev}. Quantum
teleportation and superdense coding provided the
initial impetus for quantum Shannon theory because these protocols demonstrate
that we can combine entanglement, noiseless quantum communication, and
noiseless classical communication to transmit quantum or classical
information. In practice, the above resources are not noiseless because
quantum systems decohere by interacting with their surrounding environment.
Quantum Shannon theory is a collection of capacity theorems that determine the
ultimate limits for noisy quantum communication channels. Essentially all
quantum protocols have been unified as special cases of a handful of abstract
protocols \cite{arx2005dev}.

The techniques in quantum Shannon theory determine the asymptotic limits for
communication, but these techniques do not produce practical ways of realizing
these limits. This same practical problem exists with classical Shannon theory
because random coding techniques determine the limit of a noisy classical
communication channel \cite{bell1948shannon}.

\subsection{Communication Protocols}

The fundamental theorem of quantum Shannon theory is the celebrated
Lloyd-Shor-Devetak quantum capacity theorem. It determines the capacity of a
given noisy quantum communication channel for reliable quantum communication.
Suppose that a noisy quantum
channel $\mathcal{N}$ connects a sender Alice $A$ to a receiver Bob $B$. Let
$\left[  q\rightarrow q\right]  $ denote one qubit of noiseless quantum
communication.
We can state the coding theorem as a resource inequality:%
\begin{equation}
\left\langle \mathcal{N}\right\rangle \geq Q\left[  q\rightarrow q\right]  , \label{eq:q-coding-res}
\end{equation}
The above resource inequality states that $n$ uses of the noisy quantum
channel $\mathcal{N}$ are sufficient to communicate $nQ$ noiseless qubits in
the limit of a large number $n$ of channel uses. The rate $Q$ is equal to a
quantity known as the coherent information:
$I\left(  A\rangle B\right)  $ \cite{book2000mikeandike}. The entropic quantity is maximized with
respect to any resource state $\left\vert \psi\right\rangle ^{ABE}$ associated
to the noisy channel $\mathcal{N}$ and shared between Alice, Bob, and the
environment $E$. The goal of the stabilizer codes and convolutional stabilizer
codes mentioned respectively in Sections~\ref{sec:intro-stab} and
\ref{sec:intro-conv}\ is for their rates to approach the maximum capacity
given above.

Another example of an important capacity theorem from quantum Shannon theory
results from the \textquotedblleft father\textquotedblright\ protocol
\cite{arx2005dev}. The father capacity theorem determines the optimal
trade-off between the rate $E$ of ebits (entangled qubits in the state
$\left\vert \Phi^{+}\right\rangle ^{AB}\equiv(\left\vert 00\right\rangle
^{AB}+\left\vert 11\right\rangle ^{AB})/\sqrt{2}$) and the rate $Q$ of qubits
in entanglement-assisted quantum communication. This protocol is the
\textquotedblleft father\textquotedblright\ protocol because it generates many
of the protocols in the family tree of quantum information theory
\cite{arx2005dev}. The nickname \textquotedblleft father\textquotedblright\ is
a useful shorthand for classifying the protocol---there exists a mother,
grandmother, and grandfather protocol as well. Let $\left[  qq\right]  $ denote one ebit of entanglement.
The following resource inequality is a statement of the capability of the
father protocol:%
\begin{equation}
\left\langle \mathcal{N}\right\rangle +E\left[  qq\right]  \geq Q\left[
q\rightarrow q\right]  . \label{eq:father}%
\end{equation}
The above resource inequality states that $n$ uses of the noisy quantum
channel $\mathcal{N}$ and $nE$ noiseless ebits are sufficient to communicate
$nQ$ noiseless qubits in the limit of large $n$. The rates $E$ and $Q$ are
respectively equal to $\frac{1}{2}I\left(  A;E\right)  $ and $\frac{1}%
{2}I\left(  A;B\right)  $. The entropic quantities are maximized with respect
to any resource state $\left\vert \psi\right\rangle ^{ABE}$ associated to the
noisy channel $\mathcal{N}$ and shared between Alice, Bob, and the environment
$E$ (the rate $E$ is different from the environment $E$). The father capacity
theorem gives the optimal limits on the resources, but it does not provide a
useful quantum coding technique for approaching the above limits. The goal of
the entanglement-assisted stabilizer formalism mentioned in
Section~\ref{sec:intro-EA}\ is to develop quantum codes that approach the
maximal rates given in the father capacity theorem. We spend a significant
portion of this thesis (Chapters~\ref{chp:EAQCC-CSS}, \ref{chp:general-case},
and \ref{chp:free-ent}) developing entanglement-assisted convolutional
stabilizer codes. Future research might be able to incorporate these convolutional codes into
a framework for designing codes that come close to achieving the maximal rates.

Another important capacity theorem determines the ability of a noisy quantum
channel to send \textquotedblleft classical-quantum\textquotedblright\ states
\cite{cmp2005dev}. Let $\left[  c\rightarrow c\right]  $ denote one
bit of noiseless classical communication. The result of the classical-quantum
capacity theorem is also a resource inequality:%
\begin{equation}
\left\langle \mathcal{N}\right\rangle \geq Q\left[  q\rightarrow q\right]
+R\left[  c\rightarrow c\right]  . \label{eq:CQ}%
\end{equation}
The resource inequality states that $n$ uses of the noisy quantum channel
$\mathcal{N}$ are sufficient to communicate $nQ$ noiseless qubits and $nR$
noiseless classical bits in the limit of large $n$. The rates $Q$ and $R$ are
respectively equal to $I\left(  X;B\right)  $ and $I\left(  A\rangle
BX\right)  $. The entropic quantities are with respect to any state resulting
from sending the $A^{\prime}$ system of the following classical-quantum state%
\begin{equation}
\sum_{x}p_{x}\left\vert x\right\rangle \left\langle x\right\vert ^{X}%
\otimes\left\vert \phi_{x}\right\rangle \left\langle \phi_{x}\right\vert
^{AA^{^{\prime}}}%
\end{equation}
through the quantum channel $\mathcal{N}^{A^{\prime}\rightarrow B}$. $X$ and
$A$ are systems that the sender keeps. This theorem proves that we can devise
clever classical-quantum codes that perform better than time-sharing a noisy
quantum channel $\mathcal{N}$ between purely quantum codes and purely
classical codes.

The \textquotedblleft grandfather\textquotedblright\ capacity theorem
determines the optimal triple trade-off between qubits, ebits, and classical
bits for simultaneous transmission of classical and quantum information using
an entanglement-assisted noisy quantum channel $\mathcal{N}$
\cite{prep2008dev}. The grandfather resource inequality is as follows:%
\begin{equation}
\left\langle \mathcal{N}\right\rangle +E\left[  qq\right]  \geq Q\left[
q\rightarrow q\right]  +R\left[  c\rightarrow c\right]  . \label{eq:GF}%
\end{equation}
The above resource inequality is again an asymptotic statement and its meaning
is similar to that in (\ref{eq:q-coding-res}),
(\ref{eq:father}), and (\ref{eq:CQ}). The optimal rates
in the above resource inequality coincide with the father inequality
(\ref{eq:father}) when $R=0$, with the classical-quantum inequality
(\ref{eq:CQ})\ when $E=0$, and with the quantum capacity
(\ref{eq:q-coding-res}) when both $R=0$ and
$E=0$. The optimal strategy for the grandfather protocol is not time-sharing
the channel between father codes and entanglement-assisted classical codes. It
remains to be proven whether this optimal strategy outperforms time-sharing
\cite{prep2008dev}. We develop grandfather convolutional codes in
Chapter~\ref{chp:unified}. Again, future research might be able to incorporate
these grandfather convolutional codes into a framework for designing codes that come close to achieving the maximal limits given
in the grandfather capacity theorem.

\subsection{Distillation Protocols}

The aforementioned protocols employ a noisy quantum communication channel in
their operation. We now ask the question:\ given noisy entanglement, is it
possible to produce noiseless entanglement using classical communication?
Indeed, it is possible, and this procedure is entanglement distillation
\cite{PhysRevLett.76.722,PhysRevA.54.3824}. The following resource inequality
summarizes the resources consumed and produced in an entanglement distillation
protocol:%
\[
\left\langle \rho^{AB}\right\rangle +R\left[  c\rightarrow c\right]  \geq
Q\left[  qq\right]  .
\]
The statement of the resource inequality is as follows: given a large number
$n$ of noisy bipartite states $\rho^{AB}$ shared between two parties, it is
possible to generate $nQ$ noiseless ebits by consuming $nR$ classical bits in
the limit of a large number $n$ of copies of the noisy state $\rho^{AB}$. The
optimal rates are $R=I\left(  A;E\right)$ and\ $Q=I\left(  A\rangle
B\right)$ where the entropic quantities are with respect to a state
$\left\vert \psi\right\rangle ^{ABE}$ that purifies the state $\rho^{AB}$. In
Chapter~\ref{chp:ced}, we develop convolutional entanglement distillation
protocols whose goal is to approach the optimal rates given in the above
resource inequality.

\section{Organization of this Thesis}

This thesis represents a significant contribution to quantum error correction
theory. In it, we give several methods for designing quantum codes. These
methods should be useful in the future when coherent quantum encoding devices
are available. The quantum code designer will have a plethora of techniques to
exploit for protecting quantum information and the techniques in this thesis
should be part of the arsenal.

A quick glance over the thesis reveals that much of the mathematics involves
only linear algebra over matrices of binary numbers or binary polynomials. It
is a testament to the strength of the theory that it reduces to this simple
form. We show in detail especially in Chapters~\ref{chp:stab} and
\ref{chp:conv} how to reduce the theory to have this simple form.
This thesis has contributions from the following publications
\cite{arx2007wildeEAQCC,arx2008wildeOEA,prep2007wildeEAOQEC,prep2008wildeQCCFE,arx2008wildeUQCC,arx2007wildeCED,pra2007wildeEA}.

We structure the thesis as follows. The next chapter introduces the stabilizer
formalism for quantum error correction. We show in this chapter how to
represent quantum codes with binary matrices.
Chapter~\ref{chp:ent-assisted-review} introduces the entanglement-assisted
stabilizer formalism and shows how to enhance the stabilizer formalism by
assuming that the sender and receiver share entanglement. This chapter gives
two contributions:\ a method for determining the encoding circuit for a
quantum block code and a method to determine the minimal amount of
entanglement that an entanglement-assisted code requires.
Chapter~\ref{chp:conv} reviews the convolutional stabilizer formalism. In
particular, it details the operation of a quantum convolutional
code and shows how to simplify the mathematics with matrices whose entries
are binary polynomials. Chapters~\ref{chp:EAQCC-CSS}, \ref{chp:general-case},
and \ref{chp:free-ent} form the heart of this thesis. They show how to
construct entanglement-assisted quantum codes that have a convolutional
structure. The techniques developed represent a significant extension of the
entanglement-assisted stabilizer formalism. Chapter~\ref{chp:unified}%
\ presents a unifying structure for quantum error correction under which most
current schemes for quantum error correction fall. We present a method for
distilling entanglement in a convolutional manner in Chapter \ref{chp:ced}.
The major benefit of all of the above convolutional methods is that we can
import arbitrary high-performance classical codes to construct
high-performance quantum codes.

\chapter{Stabilizer Quantum Codes}

\label{chp:stab}\begin{saying}
``Fight entanglement with entanglement,'' you say,\\
'Cause entanglement keeps Eve away,\\
The hell that she raises,\\
She damps and dephases,\\
She's no match---like UCLA!
\end{saying}The stabilizer formalism is a mathematical framework for quantum
error correction \cite{PhysRevA.54.1862,thesis97gottesman}. This framework has
many similarities with classical coding theory, and it is even possible to
import a classical code for use in quantum error correction by employing the
CSS\ construction
\cite{PhysRevA.54.1098,PhysRevLett.77.793,book2000mikeandike}. We review the
stabilizer theory for quantum block codes.

The stabilizer formalism is versatile because it appears not only in quantum
error correction but also in other subjects of quantum information theory such
as stabilizer or graph states \cite{graph2005}, cluster-state quantum
computation \cite{cluster2001}, and entanglement theory \cite{arx2004fattal}.
The stabilizer formalism simplifies the theory of quantum error correction by
reducing the mathematics to linear algebra over binary matrices.

\section{Review of the Stabilizer Formalism}

\label{sec:standard-stabilizer}The stabilizer formalism exploits elements of
the Pauli group $\Pi$. The set $\Pi=\left\{  I,X,Y,Z\right\}  $ consists of
the Pauli operators:%
\[
I\equiv%
\begin{bmatrix}
1 & 0\\
0 & 1
\end{bmatrix}
,\ X\equiv%
\begin{bmatrix}
0 & 1\\
1 & 0
\end{bmatrix}
,\ Y\equiv%
\begin{bmatrix}
0 & -i\\
i & 0
\end{bmatrix}
,\ Z\equiv%
\begin{bmatrix}
1 & 0\\
0 & -1
\end{bmatrix}
.
\]
The above operators act on a single qubit---a vector in a two-dimensional
Hilbert space---and are the most important in formulating a quantum
error-correcting code. Two crucial properties of these matrices are useful:
each matrix in $\Pi$ has eigenvalues equal to $+1$ or $-1$, and any two
matrices in $\Pi$ either commute or anticommute.

In general, a quantum error-correcting code uses $n$ physical qubits to
protect a smaller set of information qubits against decoherence or quantum
noise. An $n$-qubit quantum error-correcting code employs elements of the
Pauli group $\Pi^{n}$. The set $\Pi^{n}$ consists of $n$-fold tensor products
of Pauli operators:%
\begin{equation}
\Pi^{n}=\left\{  e^{i\phi}A_{1}\otimes\cdots\otimes A_{n}:\forall j\in\left\{
1,\ldots,n\right\}  ,\ \ A_{j}\in\Pi,\ \ \phi\in\left\{  0,\pi/2,\pi
,3\pi/2\right\}  \right\}  .
\end{equation}
Matrices in $\Pi^{n}$ act on a $2^{n}$-dimensional complex vector, or
equivalently, an $n$-qubit quantum register. We omit tensor product symbols in
what follows so that $A_{1}\cdots A_{n}\equiv A_{1}\otimes\cdots\otimes A_{n}%
$. The above two crucial properties for the single-qubit Pauli group $\Pi$
still hold for the Pauli group $\Pi^{n}$ (up to an irrelevant phase for the
eigenvalue property). The $n$-fold Pauli group $\Pi^{n}$\ plays an important
role in both the encoding circuit and the error-correction procedure of a
quantum stabilizer code over $n$ qubits.

We can phrase the theory of quantum error correction in purely mathematical
terms using elements of $\Pi^{n}$. Consider a matrix $g_{1}\in\Pi^{n}$ that is
not equal to $\pm I$. Matrix $g_{1}$ then has two eigenspaces each of
dimension $2^{n-1}$. We can identify one eigenspace with the eigenvalue $+1$
and the other eigenspace with eigenvalue $-1$. Consider a matrix $g_{2}\in
\Pi^{n}$ different from both $g_{1}$ and the identity that commutes with $g_{1}$.
Matrix $g_{2}$ also has two eigenspaces each of size $2^{n-1}$ and identified
similarly by its eigenvalues $\pm1$. Then $g_{1}$ and $g_{2}$ have
simultaneous eigenspaces because they commute. These matrices together have
four different eigenspaces, each of size $2^{n-2}$ and identified by the
eigenvalues $\pm1,\pm1$ of $g_{1}$ and $g_{2}$ respectively. We can continue
this process of adding more commuting and independent matrices to a set
$\mathcal{S}$. The matrices in $\mathcal{S}$ are independent in the sense that
no matrix in $\mathcal{S}$ is a product of two or more other matrices in
$\mathcal{S}$. Adding more matrices from $\Pi^{n}$\ to $\mathcal{S}$ continues
to divide the eigenspaces of matrices in $\mathcal{S}$. In general, suppose
$\mathcal{S}$ consists of $n-k$ independent and commuting matrices $g_{1}$,
\ldots, $g_{n-k}\in\Pi^{n}$. These $n-k$ matrices then have $2^{n-k}$
different eigenspaces each of size $2^{k}$ and identified by the eigenvalues
$\pm1$, \ldots, $\pm1$\ of $g_{1}$, \ldots, $g_{n-k}$ respectively. Consider
that the Hilbert space of $k$ qubits has size $2^{k}$. A dimension count
immediately suggests that we can encode $k$ qubits into one of the eigenspaces
of $\mathcal{S}$. We typically encode these $k$ qubits into the simultaneous
$+1$-eigenspace of $g_{1}$, \ldots, $g_{n-k}$. This eigenspace is the
\textit{codespace}. The operators $g_{1},\ldots,g_{n-k}$ function in the same
way as a parity check matrix does for a classical linear block code. An
$\left[  n,k\right]  $ quantum error-correcting code encodes $k$ information
qubits into the simultaneous $+1$-eigenspace of $n-k$ matrices $g_{1}$,
\ldots, $g_{n-k}\in\Pi^{n}$. Its stabilizer $\mathcal{S}$\ is an abelian
subgroup of the $n$-fold Pauli group $\Pi^{n}$: $\mathcal{S}\subset\Pi^{n}$.
Note that it is possible to multiply the generators in the stabilizer together
and obtain an equivalent representation of the stabilizer in much the same way
that we can add vectors in a basis together and obtain an equivalent basis.%

\begin{figure}
[ptb]
\begin{center}
\includegraphics[
natheight=3.386600in,
natwidth=8.973300in,
height=1.9614in,
width=5.1742in
]%
{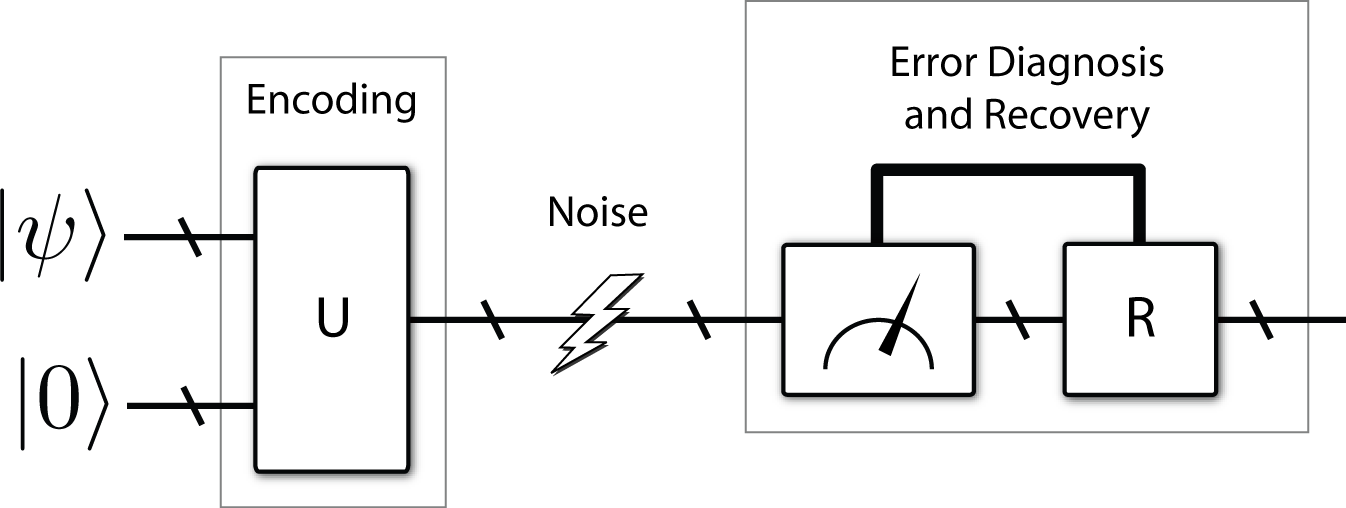}%
\caption{The operation of a stabilizer code. Thin lines denote quantum
information and thick lines denote classical information. Slanted bars denote
multiple qubits. A sender encodes a multi-qubit state $\left\vert
\psi\right\rangle $ with the help of some ancilla qubits $\left\vert
0\right\rangle $. She sends the encoded state over a noisy quantum channel.
The receiver performs multi-qubit measurements to extract information about
the errors. He finally performs a recovery operation $R$\ to reverse the
channel error.}%
\label{fig:stabilizer-code}%
\end{center}
\end{figure}

The operation of an $\left[  n,k\right]  $ quantum error-correcting code
consists of four steps. Figure~\ref{fig:stabilizer-code}\ highlights these steps.

\begin{enumerate}
\item A unitary operation $U$ encodes both $k$ information qubits in a general
state $\left\vert \psi\right\rangle $\ and $n-k$ ancilla qubits in the state
$\left\vert 0\right\rangle ^{\otimes n-k}$\ into the simultaneous
$+1$-eigenspace of the matrices $g_{1}$, \ldots, $g_{n-k}$.

\item The sender transmits the $n$ encoded qubits by using the noisy quantum
communication channel $n$ times.

\item The receiver performs quantum measurements of the $n-k$ matrices $g_{1}%
$, \ldots, $g_{n-k}$ in the stabilizer $\mathcal{S}$. These measurements learn
only about errors that may occur and do not disturb the encoded quantum
information. Each measurement gives a bit result equal to $+1$ or $-1$, and
the result of all the measurements is to project the $n$-qubit quantum
register into one of the $2^{n-k}$ different eigenspaces of $g_{1}$, \ldots,
$g_{n-k}$. Suppose that no error occurs. Then the measurements project the $n$
qubits into the simultaneous $+1$-eigenspace and return a bit vector
consisting of $n-k$ ones. This \textquotedblleft projection of the error
concept\textquotedblright\ is one of the fundamental notions in quantum error
correction theory. It suffices to correct a discrete error set with support in
the Pauli group $\Pi^{n}$ \cite{PhysRevA.52.R2493}. Now suppose that a quantum
error in an error set $\mathcal{E}\subset\Pi^{n}$ occurs. The error takes the
encoded quantum state out of the codespace and into one of the other
$2^{n-k}-1$ orthogonal eigenspaces. The measurements can detect that an error
has occurred because the result of the measurements is a bit vector differing
from the all ones vector. The receiver can identify uniquely which
error in $\mathcal{E}$\ has occurred if the set $\mathcal{E}$ satisfies the
following quantum error correction conditions:%
\begin{equation}
\forall E_{a},E_{b}\in\mathcal{E\ \ \ \ \ }\exists\ g_{i}\in\mathcal{S}%
:\left\{  g_{i},E_{a}^{\dag}E_{b}\right\}  =0\text{ or }E_{a}^{\dag}E_{b}%
\in\mathcal{S}.
\end{equation}
The first condition corresponds to the active error-correcting capability of the code, and the second condition
corresponds to its passive error-correcting capability.

\item If the receiver can identify which error occurs, he can then apply a
unitary operation $R$\ that is the inverse of the error. He finally performs a
decoding unitary that decodes the $k$ information qubits.
\end{enumerate}

\section{Clifford Encoding Unitary}

\label{sec:cliff-encoding}We comment briefly on the encoding operation $U$.
The encoding operation $U$ is a special type of unitary matrix called a
Clifford operation. A Clifford operation $U$\ is one that preserves elements
of the Pauli group under conjugation:\ $A\in\Pi^{n}\Rightarrow UAU^{\dag}%
\in\Pi^{n}$. The CNOT\ gate, the Hadamard gate $H$, and the phase gate $P$
suffice to implement any unitary matrix in the Clifford group
\cite{thesis97gottesman}. A quantum code with the CSS\ structure needs only
the CNOT\ and Hadamard gates for encoding and decoding. The matrix for the
CNOT\ gate acting on two qubits is%
\begin{equation}
\text{CNOT}=%
\begin{bmatrix}
1 & 0 & 0 & 0\\
0 & 1 & 0 & 0\\
0 & 0 & 0 & 1\\
0 & 0 & 1 & 0
\end{bmatrix}
,
\end{equation}
the matrix for the Hadamard gate $H$ acting on a single qubit is%
\begin{equation}
H=\frac{1}{\sqrt{2}}%
\begin{bmatrix}
1 & 1\\
1 & -1
\end{bmatrix}
,
\end{equation}
and the matrix for the phase gate $P$ acting on a single qubit is%
\begin{equation}
P=%
\begin{bmatrix}
1 & 0\\
0 & i
\end{bmatrix}
.
\end{equation}

The standard basis for $\Pi^{1}$ is $X$ and $Z$ because any element of
$\Pi^{1}$ is a product of elements in this generating set up to an irrelevant
phase. The standard basis for elements of the two-qubit Pauli group $\Pi^{2}$
is as follows%
\begin{equation}%
\begin{array}
[c]{cc}%
Z & I
\end{array}
,\ \
\begin{array}
[c]{cc}%
I & Z
\end{array}
,\ \
\begin{array}
[c]{cc}%
X & I
\end{array}
,\ \
\begin{array}
[c]{cc}%
I & X
\end{array}
,
\end{equation}
for the same reasons.

The
Hadamard gate $H$ transforms the standard basis of $\Pi^{1}$ under conjugation
as follows:%
\begin{equation}
Z\rightarrow X,\ \ \ \ X\rightarrow Z
\end{equation}
and the phase gate $P$ transforms the standard basis as follows:%
\begin{equation}
Z\rightarrow Z,\ \ \ \ X\rightarrow Y.
\end{equation}
For the CNOT\ gate, the first qubit is the \textquotedblleft
control\textquotedblright\ qubit and the second qubit is the \textquotedblleft
target\textquotedblright\ qubit. The CNOT\ gate transforms the standard basis
of $\Pi^{2}$ under conjugation as follows%
\begin{equation}%
\begin{array}
[c]{cc}%
Z & I
\end{array}
\rightarrow%
\begin{array}
[c]{cc}%
Z & I
\end{array}
,\ \ \
\begin{array}
[c]{cc}%
I & Z
\end{array}
\rightarrow%
\begin{array}
[c]{cc}%
Z & Z
\end{array}
,\ \ \
\begin{array}
[c]{cc}%
X & I
\end{array}
\rightarrow%
\begin{array}
[c]{cc}%
X & X
\end{array}
,\ \ \
\begin{array}
[c]{cc}%
I & X
\end{array}
\rightarrow%
\begin{array}
[c]{cc}%
I & X
\end{array}
,
\end{equation}
where the first qubit is the control and the second qubit is the target.

Section~\ref{sec:encoding-alg-ent-assist} of
Chapter~\ref{chp:ent-assisted-review}\ details an algorithm that determines a
Clifford encoding circuit for any stabilizer code or any entanglement-assisted
stabilizer code (we review entanglement-assisted codes in the next chapter).

\section{The Logical Operators and the Information Qubits}

Another aspect of the theory of quantum error correction is later useful for
our purposes in quantum convolutional coding. This aspect concerns the
information qubits and the operators that change them. Consider that the
initial unencoded state of a quantum error-correcting code is a simultaneous
+1-eigenstate of the matrices $Z_{k+1},\ldots,Z_{n}$ where $Z_{i}$ has a $Z$
matrix operating on qubit $i$ and the identity $I$ on all other qubits.
Therefore, the matrices $Z_{k+1},\ldots,Z_{n}$ constitute a stabilizer for the
unencoded state. The initial unencoded logical operators for the
$k$\ information qubits are $Z_{1},X_{1},\ldots,Z_{k},X_{k}$. The encoding
operation $U$ rotates the unencoded stabilizer matrices $Z_{k+1},\ldots,Z_{n}$
and the unencoded logical operators $Z_{1},X_{1},\ldots,Z_{k},X_{k}$ to the
encoded stabilizer operators $\bar{Z}_{k+1}$, \ldots, $\bar{Z}_{n}$ and the
encoded logical operators $\bar{Z}_{1},\bar{X}_{1},\ldots,\bar{Z}_{k},\bar
{X}_{k}$ respectively. The encoded matrices $\bar{Z}_{k+1},\ldots,\bar{Z}_{n}$
are respectively equivalent to the matrices $g_{1}$, \ldots, $g_{n-k}$ in the
discussion of the previous section.

The encoded operators obey the same commutation relations as their unencoded
counterparts. We would violate the uncertainty principle if this invariance
did not hold. Therefore, each of the encoded logical operators commutes with
elements of the stabilizer $\mathcal{S}$. Let $A$\ denote an arbitrary logical
operator from the above set and let $\bar{Z}_{i}$ denote an arbitrary element
of the stabilizer $\mathcal{S}$. The operator $A\bar{Z}_{i}$ (or equivalently
$\bar{Z}_{i}A$) is an equivalent logical operator because $A\bar{Z}_{i}$ and
$A$ have the same effect on an encoded state $\left\vert \bar{\psi
}\right\rangle $:%
\begin{equation}
\bar{Z}_{i}A\left\vert \bar{\psi}\right\rangle =A\bar{Z}_{i}\left\vert
\bar{\psi}\right\rangle =A\left\vert \bar{\psi}\right\rangle .
\end{equation}
We make extensive use of the above fact in our work.

The logical operators also provide a useful way to characterize the
information qubits. Gottesman showed that the logical operators for the
information qubits provide a straightforward way to characterize the
information qubits as they progress through a quantum circuit
\cite{thesis97gottesman}. As an example of this technique, he develops quantum
teleportation in the stabilizer formalism. The logical operators at the
beginning of the protocol are $X_{1}$ and $Z_{1}$ and become $X_{3}$ and
$Z_{3}$ at the end of the protocol. This transformation implies that the
quantum information in qubit one teleports to qubit three because the logical
operators act on only qubit three at the end of the protocol. We use the same
idea throughout this thesis to determine if our decoding circuits have truly
decoded the information qubits.

\section{The Pauli-to-Binary Isomorphism}

A simple but useful mapping\ exists between elements of $\Pi$ and the binary
vector space $\left(  \mathbb{Z}_{2}\right)  ^{2}$. This mapping gives a
simplification of quantum error correction theory. It represents quantum codes
with binary vectors and binary operations rather than with Pauli operators and
matrix operations respectively. We first give the mapping for the one-qubit
case. Suppose $\left[  A\right]  $ is a set of equivalence classes of an
operator $A$ that have the same phase:%
\begin{equation}
\left[  A\right]  =\left\{  \beta A\ |\ \beta\in\mathbb{C},\ \left\vert
\beta\right\vert =1\right\}  . \label{eq:equivalence-class}%
\end{equation}
Let $\left[  \Pi\right]  $ be the set of phase-free Pauli operators where
$\left[  \Pi\right]  =\left\{  \left[  A\right]  \ |\ A\in\Pi\right\}  $.
Define the map $N:\left(  \mathbb{Z}_{2}\right)  ^{2}\rightarrow\Pi$ as%
\begin{equation}%
\begin{tabular}
[c]{l|llll}%
$\Pi$ & $I$ & $X$ & $Y$ & $Z$\\\hline
$\left(  \mathbb{Z}_{2}\right)  ^{2}$ & $00$ & $01$ & $11$ & $10$%
\end{tabular}
\ \ \ \ . \label{eq:z2-pauli}%
\end{equation}
Suppose $u,v\in\left(  \mathbb{Z}_{2}\right)  ^{2}$. Let us employ the
shorthand $u=\left[  z|x\right]  $ and $v=\left[  z^{\prime}|x^{\prime
}\right]  $\ where $z$, $x$, $z^{\prime}$, $x^{\prime}\in\mathbb{Z}_{2}$. For
example, suppose $u=\left[  0|1\right]  $. Then $N\left(  u\right)  =X$. The
map $N$\ induces an isomorphism $\left[  N\right]  :\left(  \mathbb{Z}%
_{2}\right)  ^{2}\rightarrow\left[  \Pi\right]  $ because addition of vectors
in $\left(  \mathbb{Z}_{2}\right)  ^{2}$ is equivalent to multiplication of
Paulis up to a global phase:%
\begin{equation}
\left[  N\left(  u+v\right)  \right]  =\left[  N\left(  u\right)  \right]
\left[  N\left(  v\right)  \right]  . \label{eq:isomorphism}%
\end{equation}
We name this isomorphism the \textit{P2B isomorphism} (for Pauli-to-binary
isomorphism).

Let $\odot$\ denote the \textit{symplectic product} between two elements
$u,v\in\left(  \mathbb{Z}_{2}\right)  ^{2}$:%
\begin{equation}
u\odot v\equiv zx^{\prime}-xz^{\prime}.
\end{equation}
The symplectic product $\odot$ gives the commutation relations of elements of
$\Pi$:%
\[
N\left(  u\right)  N\left(  v\right)  =\left(  -1\right)  ^{\left(  u\odot
v\right)  }N\left(  v\right)  N\left(  u\right)  .
\]
The symplectic product and the P2B\ isomorphism $\left[  N\right]  $ thus give
a useful way to phrase Pauli relations in terms of binary algebra.

The extension of the above definitions and the P2B\ isomorphism $\left[
N\right]  $ to multiple qubits is straightforward. Let $\mathbf{A}%
=A_{1}\otimes\cdots\otimes A_{n}$ denote an arbitrary element of $\Pi^{n}$. We
can similarly define the phase-free $n$-qubit Pauli group $\left[  \Pi
^{n}\right]  =\left\{  \left[  \mathbf{A}\right]  \ |\ \mathbf{A}\in\Pi
^{n}\right\}  $\ where%
\begin{equation}
\left[  \mathbf{A}\right]  =\left\{  \beta\mathbf{A}\ |\ \beta\in
\mathbb{C},\ \left\vert \beta\right\vert =1\right\}  .
\end{equation}
The group operation $\ast$ for the above equivalence class is as follows:%
\begin{equation}
\left[  \mathbf{A}\right]  \ast\left[  \mathbf{B}\right]  \equiv\left[
A_{1}\right]  \ast\left[  B_{1}\right]  \otimes\cdots\otimes\left[
A_{n}\right]  \ast\left[  B_{n}\right]  =\left[  A_{1}B_{1}\right]
\otimes\cdots\otimes\left[  A_{n}B_{n}\right]  =\left[  \mathbf{AB}\right]  .
\end{equation}
The equivalence class $\left[  \Pi^{n}\right]  $ forms a commutative group
under operation $\ast$. Consider the $2n$-dimensional vector space%
\begin{equation}
\left(  \mathbb{Z}_{2}\right)  ^{2n}=\left\{  \left(  \mathbf{z,x}\right)
:\mathbf{z},\mathbf{x}\in\left(  \mathbb{Z}_{2}\right)  ^{n}\right\}  .
\end{equation}
It forms the commutative group $(\left(  \mathbb{Z}_{2}\right)  ^{2n},+)$ with
operation $+$ defined as binary vector addition. We employ the notation
$\mathbf{u}=\left[  \mathbf{z}|\mathbf{x}\right]  $ and $\mathbf{v}=\left[
\mathbf{z}^{\prime}|\mathbf{x}^{\prime}\right]  $ to represent any two vectors
$\mathbf{u,v}\in\left(  \mathbb{Z}_{2}\right)  ^{2n}$ respectively. Each
vector $\mathbf{z}$ and $\mathbf{x}$ has elements $\left(  z_{1},\ldots
,z_{n}\right)  $ and $\left(  x_{1},\ldots,x_{n}\right)  $ respectively with
similar representations for $\mathbf{z}^{\prime}$ and $\mathbf{x}^{\prime}$.
The \textit{symplectic product} $\odot$ of $\mathbf{u}$ and $\mathbf{v}$ is%
\begin{equation}
\mathbf{u}\odot\mathbf{v\equiv}\sum_{i=1}^{n}z_{i}x_{i}^{\prime}-x_{i}%
z_{i}^{\prime}=\sum_{i=1}^{n}u_{i}\odot v_{i},
\end{equation}
where $u_{i}=\left[  z_{i}|x_{i}\right]  $ and $v_{i}=\left[  z_{i}^{\prime
}|x_{i}^{\prime}\right]  $. Let us define a map $\mathbf{N}:\left(
\mathbb{Z}_{2}\right)  ^{2n}\rightarrow\Pi^{n}$ as follows:%
\begin{equation}
\mathbf{N}\left(  \mathbf{u}\right)  \equiv N\left(  u_{1}\right)
\otimes\cdots\otimes N\left(  u_{n}\right)  .\label{eq:map_symp_hw}%
\end{equation}
Let%
\begin{equation}
\mathbf{X}\left(  \mathbf{x}\right)  \equiv X^{x_{1}}\otimes\cdots\otimes
X^{x_{n}},\ \ \ \ \ \ \mathbf{Z}\left(  \mathbf{z}\right)  \equiv Z^{z_{1}%
}\otimes\cdots\otimes Z^{z_{n}},
\end{equation}
so that $\mathbf{N}\left(  \mathbf{u}\right)  $ and $\mathbf{Z}\left(
\mathbf{z}\right)  \mathbf{X}\left(  \mathbf{x}\right)  $ belong to the same
equivalence class:%
\begin{equation}
\left[  \mathbf{N}\left(  \mathbf{u}\right)  \right]  =\left[  \mathbf{Z}%
\left(  \mathbf{z}\right)  \mathbf{X}\left(  \mathbf{x}\right)  \right]  .
\end{equation}
The map $\left[  \mathbf{N}\right]  :\left(  \mathbb{Z}_{2}\right)
^{2n}\rightarrow\left[  \Pi^{n}\right]  $ is an isomorphism for the same
reason given in (\ref{eq:isomorphism}):
\begin{equation}
\left[  \mathbf{N}\left(  \mathbf{u+v}\right)  \right]  =\left[
\mathbf{N}\left(  \mathbf{u}\right)  \right]  \left[  \mathbf{N}\left(
\mathbf{v}\right)  \right]  ,
\end{equation}
where $\mathbf{u,v}\in\left(  \mathbb{Z}_{2}\right)  ^{2n}$. The symplectic
product captures the commutation relations of any operators $\mathbf{N}\left(
\mathbf{u}\right)  $ and $\mathbf{N}\left(  \mathbf{v}\right)  $:%
\begin{equation}
\mathbf{N\left(  \mathbf{u}\right)  N}\left(  \mathbf{v}\right)  =\left(
-1\right)  ^{\left(  \mathbf{u}\odot\mathbf{v}\right)  }\mathbf{N}\left(
\mathbf{v}\right)  \mathbf{N}\left(  \mathbf{u}\right)  .
\end{equation}
The above P2B\ isomorphism and symplectic algebra are useful in making the
relation between classical linear error correction and quantum error
correction more explicit.

\section{Example}

\label{ex:5-qubit-code}An example of a stabilizer code is the five qubit
$\left[  5,1\right]  $ stabilizer code
\cite{PhysRevLett.77.198,PhysRevA.54.3824}. It encodes $k=1$ logical qubit
into $n=5$ physical qubits and protects against an arbitrary single-qubit
error. Its stabilizer consists of $n-k=4$ Pauli operators:%
\begin{equation}%
\begin{array}
[c]{ccccccc}%
g_{1} & = & X & Z & Z & X & I\\
g_{2} & = & I & X & Z & Z & X\\
g_{3} & = & X & I & X & Z & Z\\
g_{4} & = & Z & X & I & X & Z
\end{array}
\label{eq:five-qubit-code}%
\end{equation}
The above operators commute. Therefore the codespace is the simultaneous
+1-eigenspace of the above operators. Suppose a single-qubit error occurs on
the encoded quantum register. A single-qubit error is in the set $\left\{
X_{i},Y_{i},Z_{i}\right\}  $ where $A_{i}$ denotes a Pauli error on qubit $i$.
It is straightforward to verify that any arbitrary single-qubit error has a
unique syndrome. The receiver corrects any single-qubit error by identifying
the syndrome and applying a corrective operation.

\section{Closing Remarks}

The mathematics developed in this chapter form the basis for the work in later
chapters. In particular, the binary representation of quantum codes is
important. The role of the \textquotedblleft twisted\textquotedblright%
\ symplectic product plays an important role for the entanglement-assisted
codes of the next chapter because it helps in determining the commutation
relations of an arbitrary (not necessarily commuting) set of Pauli generators.

\chapter{Entanglement-Assisted Stabilizer Quantum Codes}

\label{chp:ent-assisted-review}\begin{saying}
Classical to quantum made dinero,\\
But the twisted strange product was zero,\\
Little did we know,\\
What Entanglement could show,\\
And again become quantum's good hero.
\end{saying}The entanglement-assisted stabilizer formalism is a significant
extension of the standard stabilizer formalism that incorporates shared
entanglement as a resource for protecting quantum information
\cite{arx2006brun,science2006brun}. The advantage of entanglement-assisted
stabilizer codes is that the sender can exploit the error-correcting
properties of an arbitrary set of Pauli operators. The sender's Pauli
operators do not necessarily have to form an abelian subgroup of $\Pi^{n}$.
The sender can make clever use of her shared ebits so that the global
stabilizer is abelian and thus forms a valid quantum error-correcting code.
Figure~\ref{fig:entanglement-assisted-code}\ demonstrates the operation of a
generic entanglement-assisted stabilizer code.

Several references provide a review of this technique and generalizations of
the basic theory to block entanglement distillation \cite{luo:010303},
continuous-variable codes \cite{pra2007wildeEA}, and entanglement-assisted
operator codes for discrete-variable \cite{isit2007brun,hsieh:062313}\ and
continuous-variable systems \cite{prep2007wildeEAOQEC}.
Chapters~\ref{chp:EAQCC-CSS}, \ref{chp:general-case}, \ref{chp:free-ent},
\ref{chp:unified}, and \ref{chp:ced} of this thesis extend the
entanglement-assisted stabilizer formalism to many different \textquotedblleft
convolutional\textquotedblright\ coding scenarios. We explain what we mean by
\textquotedblleft convolutional\textquotedblright\ in the next chapter.

For now, we provide a review of the basics of entanglement-assisted coding.
This chapter also gives two original contributions:\ a method to determine the
encoding and decoding circuit for an entanglement-assisted code and several formulas that
determine the number of ebits that an arbitrary entanglement-assisted block code requires.

\section{Review of the Entanglement-Assisted Stabilizer Formalism}

\label{sec:ent-ass-gram}The fundamental unit of bipartite entanglement is the
\textit{ebit}. We must first understand the commutativity properties of the
ebit before developing the entanglement-assisted formalism. We express the
state $\left\vert \Phi^{+}\right\rangle $\ of an ebit shared between a sender
Alice and a receiver Bob as follows:%
\begin{equation}
\left\vert \Phi^{+}\right\rangle \equiv\frac{\left\vert 00\right\rangle
^{AB}+\left\vert 11\right\rangle ^{AB}}{\sqrt{2}}.
\end{equation}
The two operators that stabilize this ebit state are $X^{A}X^{B}$ and
$Z^{A}Z^{B}$. These two operators commute,%
\[
\left[  X^{A}X^{B},Z^{A}Z^{B}\right]  =0,
\]
but the local operators (operating only on either party) anticommute,%
\[
\left\{  X^{A},Z^{A}\right\}  =\left\{  X^{B},Z^{B}\right\}  =0.
\]

The above commutation relations hint at a way that we can resolve
anticommutativity in a set of generators. Suppose that we have two generators
that anticommute. We can use an ebit of entanglement to resolve the
anticommutativity in the two generators. We explain this idea in more detail in what follows.

We review the general construction of an entanglement-assisted code. An $\left[
n,k;c\right]  $\ entanglement-assisted code employs $c$\ ebits and $a$ ancilla
qubits to encode $k$ information qubits. Suppose that there is a nonabelian
subgroup $\mathcal{S}\subset\Pi^{n}$ of size $2c+a$. Application of the
fundamental theorem of symplectic geometry\footnote{We loosely refer to this
theorem as the fundamental theorem of symplectic geometry because of its
importance in symplectic geometry and in quantum coding theory.}
\cite{book2001symp}\ (Lemma~1 in \cite{science2006brun})\ states that there
exists a minimal set\ of independent generators $\left\{  \bar{Z}_{1}%
,\ldots,\bar{Z}_{a+c},\bar{X}_{a+1},\ldots,\bar{X}_{a+c}\right\}  $ for
$\mathcal{S}$ with the following commutation relations:%
\begin{align}
\forall i,j\ \ \ \ \left[  \bar{Z}_{i},\bar{Z}_{j}\right]   &
=0,\ \ \ \ \forall i,j\ \ \ \ \left[  \bar{X}_{i},\bar{X}_{j}\right]
=0,\label{eq:comm-relations}\\
\forall i\neq j\ \ \ \ \left[  \bar{X}_{i},\bar{Z}_{j}\right]   &
=0,\ \ \ \ \forall i\ \ \ \ \left\{  \bar{X}_{i},\bar{Z}_{i}\right\}
=0.\nonumber
\end{align}
The decomposition of $\mathcal{S}$ into the above minimal generating set
determines that the code requires $a$ ancilla qubits and $c$ ebits.
The parameters $a$ and $c$ generally depend on the set of generators in $\mathcal{S}$.

There exists a symplectic Gram-Schmidt orthogonalization procedure that gives
the decomposition
\cite{arx2006brun,science2006brun,arx2007wildeCED,prep2007shaw}. Specifically,
the algorithm performs row operations (multiplication of the Pauli generators)
that do not affect the code's error-correcting properties and thus gives a set
of generators that form an equivalent code. The decomposition also minimizes
the number of ebits required for the code and we prove this optimality in
Section~\ref{sec:opt-ebit}.

We present a simple stabilizer version of the symplectic Gram-Schmidt
algorithm. Suppose we have a set of $m$ generators $g_{1}$, \ldots, $g_{m}$.
Consider $g_{1}$. It either commutes with all other generators or it
anticommutes with at least one other generator. Remove it from the set if it
commutes with all other generators. Now suppose that it does not, i.e., it
anticommutes with one other generator $g_{i}$. Relabel $g_{2}$ as $g_{i}$ and
vice versa. Recall that we can multiply generators without changing the
error-correcting properties of the set of generators. So we perform several
multiplications to change the commutation relations to have the standard
commutation relations in (\ref{eq:comm-relations}). We perform the following
manipulations on the generators $g_{3}$, \ldots, $g_{m}$:%
\[
g_{j}=g_{j}\cdot g_{1}^{f\left(  g_{2},g_{j}\right)  }\cdot g_{2}^{f\left(
g_{1},g_{j}\right)  }\ \ \ \ \ \ \forall\ i\in\left\{  3,\ldots,m\right\}
\]
where the function $f$ is equal to zero if its two arguments commute and it is
equal to one if its two arguments anticommute. (This
function is the symplectic product.) Remove the generators $g_{1}$
and $g_{2}$ from the set. Repeat the algorithm on the remaining generators.
When the algorithm finishes, all the ``removed' generators constitute
the generators for the code and have the standard
commutation relations in (\ref{eq:comm-relations}).

We can partition the nonabelian group $\mathcal{S}$ into two subgroups: the
isotropic subgroup $\mathcal{S}_{I}$\ and the entanglement subgroup
$\mathcal{S}_{E}$. The isotropic subgroup $\mathcal{S}_{I}$ is a commuting
subgroup of $\mathcal{S}$ and thus corresponds to ancilla
qubits:\ $\mathcal{S}_{I}=\left\{  \bar{Z}_{1},\ldots,\bar{Z}_{a}\right\}  $.
The elements of the entanglement subgroup $\mathcal{S}_{E}$ come in
anticommuting pairs and thus correspond to halves of ebits: $\mathcal{S}%
_{E}=\left\{  \bar{Z}_{a+1},\ldots,\bar{Z}_{a+c},\bar{X}_{a+1},\ldots,\bar
{X}_{a+c}\right\}  $. The two subgroups $\mathcal{S}_{I}$ and $\mathcal{S}%
_{E}$\ play a role in the error-correcting conditions for the
entanglement-assisted stabilizer formalism. An entanglement-assisted code
corrects errors in a set $\mathcal{E}\subset\Pi^{n}$ if
\[
\forall E_{1},E_{2}\in\mathcal{E}\ \ \ \ \ \ \ \ E_{1}^{\dag}E_{2}%
\in\mathcal{S}_{I}\text{ or }E_{1}^{\dag}E_{2}\in\Pi^{n}-\mathcal{Z}\left(
\left\langle \mathcal{S}_{I},\mathcal{S}_{E}\right\rangle \right)  .
\]
The conditions correspond to error pairs $E_{1},E_{2}$ in an error set $\mathcal{E}$.
The first condition corresponds to the passive error-correcting capability of the code,
and the second condition corresponds to its active error-correcting
capability.%

\begin{figure}
[ptb]
\begin{center}
\includegraphics[
natheight=4.266100in,
natwidth=10.639800in,
height=2.3004in,
width=5.7155in
]%
{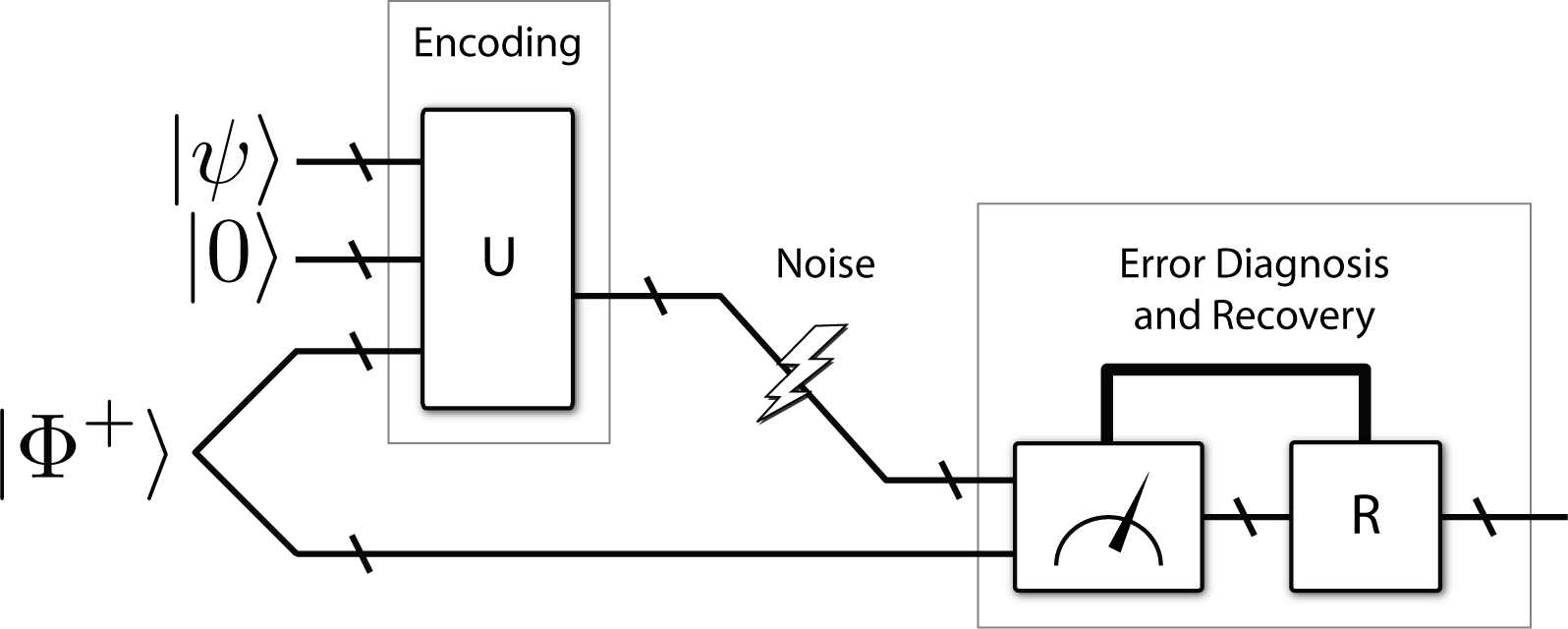}%
\caption{The operation of an entanglement-assisted quantum error-correcting
code. The sender encodes quantum information in state $\left\vert
\psi\right\rangle $\ with the help of local ancilla qubits $\left\vert
0\right\rangle $\ and her half of a set of shared ebits $\left\vert \Phi
^{+}\right\rangle $. She then sends her qubits over a noisy quantum channel.
The channel does not corrupt the receiver's half of the set of shared ebits.
The receiver performs multi-qubit measurements on all of the qubits to
diagnose the channel error. He performs a recovery unitary $R$ to reverse the
estimated channel error.}%
\label{fig:entanglement-assisted-code}%
\end{center}
\end{figure}

Figure~\ref{fig:entanglement-assisted-code}\ illustrates the operation of an
entanglement-assisted stabilizer code. The following four steps explain the
operation of an $\left[  n,k;c\right]  $\ entanglement-assisted quantum code:

\begin{enumerate}
\item The sender and receiver share $c$ ebits before quantum communication
begins and the sender has $a$ ancilla qubits. The unencoded state is a
simultaneous +1-eigenstate of the following operators:%
\begin{equation}
\left\{  Z_{a+1}|Z_{1},\ldots,Z_{a+c}|Z_{c},X_{a+1}%
|X_{1},\ldots,X_{a+c}|X_{c}, Z_{1},\ldots,Z_{a}\right\}  .
\end{equation}
The operators to the right of the vertical bars indicate the receiver's half
of the shared ebits. The sender encodes her $k$ information qubits with the
help of $a$ ancilla qubits and her half of the $c$ ebits. The encoding
operation is $U$\ in Figure~\ref{fig:entanglement-assisted-code}. The encoding
unitary transforms the unencoded operators to the following encoded operators:%
\begin{equation}
\left\{  \bar{Z}_{a+1}|Z_{1},\ldots,\bar
{Z}_{a+c}|Z_{c},\bar{X}_{a+1}|X_{1},\ldots,\bar{X}_{a+c}|X_{c},\bar{Z}_{1},\ldots,\bar{Z}_{a}\right\}  .
\end{equation}

\item The sender sends her $n$ qubits over a noisy quantum communication
channel. The noisy channel affects these $n$ qubits only and does not affect
the receiver's half of the $c$ ebits.

\item The receiver combines his half of the $c$ ebits with those he receives
from the noisy quantum channel. He performs measurements on all $n+c$ qubits
to diagnose an error that may occur on the $n$ qubits.

\item After estimating which error occurs, the receiver performs a recovery
operation that reverses the estimated error.
\end{enumerate}

\section{The Symplectic Product Matrix}

\label{sec:EA-symp-relations}Let us write the local operators $X^{A}$ and
$Z^{A}$ on the sender Alice's side as respective symplectic vectors $h_{1}$
and $h_{2}$:%
\[
h_{1}=\left[  \left.
\begin{array}
[c]{c}%
0
\end{array}
\right\vert
\begin{array}
[c]{c}%
1
\end{array}
\right]  ,\ \ \ \ \ \ \ \ h_{2}=\left[  \left.
\begin{array}
[c]{c}%
1
\end{array}
\right\vert
\begin{array}
[c]{c}%
0
\end{array}
\right]  .
\]
The \textquotedblleft symplectic product matrix\textquotedblright\ $\Omega
$\ of these two symplectic vectors is as follows:%
\begin{equation}
\Omega=\left[
\begin{array}
[c]{cc}%
0 & 1\\
1 & 0
\end{array}
\right]  .\label{eq:J-matrix}%
\end{equation}
We call this matrix the symplectic product matrix because its entries
enumerate the symplectic products:%
\begin{equation}
\left[  \Omega\right]  _{ij}=h_{i}\odot h_{j}\text{ \ \ \ }\forall
i,j\in\left\{  1,2\right\}  .
\end{equation}
The symplectic product matrix in\ (\ref{eq:J-matrix}) is so special for the
purposes of this thesis that we give it the name $J\equiv\Omega$. It means
that two generators have commutation relations that are equivalent to half of
an ebit.

A general set of generators for a quantum block code can have complicated
commutation relations. For a given set of generators, the commutation
relations for the generators resulting from the symplectic Gram-Schmidt
orthogonalization procedure have a special form. Their symplectic product
matrix $\Omega$\ has the standard form:%
\begin{equation}
\Omega=%
{\displaystyle\bigoplus\limits_{i=1}^{c}}
J\oplus%
{\displaystyle\bigoplus\limits_{j=1}^{a}}
\left[  0\right]  \label{eq:standard-symp-form}%
\end{equation}
where the large and small $\oplus$ correspond to the direct sum operation and
$\left[  0\right]  $ is the one-element null matrix. The standard form above
implies that the commutation relations of the $2c+a$ generators are equivalent
to those of $c$ halves of ebits and $a$ ancilla qubits---the commutation
relations are equivalent to those in (\ref{eq:comm-relations}).

\section{The Rate of an Entanglement-Assisted Quantum Code}

\label{sec:EA-rates}We can interpret the rate of an entanglement-assisted code
in three different ways \cite{science2006brun,arx2006brun,arx2007wildeCED}.
Suppose that an entanglement-assisted quantum code encodes $k$ information
qubits into $n$ physical qubits with the help of $c$ ebits.

\begin{enumerate}
\item The \textquotedblleft entanglement-assisted\textquotedblright\ rate
assumes that entanglement shared between sender and receiver is free. Bennett
et al. make this assumption when deriving the entanglement-assisted capacity
of a quantum channel for sending quantum information
\cite{PhysRevLett.83.3081,ieee2002bennett}. The entanglement-assisted rate is
$k/n$ for a code with the above parameters.

\item The \textquotedblleft trade-off' rate assumes that entanglement is not
free and a rate pair determines performance. The first number in the pair is
the number of noiseless qubits generated per channel use, and the second
number in the pair is the number of ebits consumed per channel use. The rate
pair is $\left(  k/n,c/n\right)  $ for a code with the above parameters.
Quantum information theorists have computed asymptotic trade-off curves that
bound the rate region in which achievable rate pairs lie \cite{arx2005dev}.
Brun et al.'s construction for an entanglement-assisted quantum block code
minimizes the number $c$ of ebits given a fixed number $k$ and $n$ of
respective information qubits and physical qubits
\cite{arx2006brun,science2006brun}.

\item The \textquotedblleft catalytic rate\textquotedblright\ assumes that
bits of entanglement are built up at the expense of transmitted qubits
\cite{arx2006brun,science2006brun}. A noiseless quantum channel or the encoded
use of noisy quantum channel are two different ways to build up entanglement
between a sender and receiver. The catalytic rate of an $\left[  n,k;c\right]
$\ code is $\left(  k-c\right)  /n$.
\end{enumerate}

Which interpretation is most reasonable depends on the context in which we use
the code. In any case, the parameters $n$, $k$, and $c$ ultimately govern
performance, regardless of which definition of the rate we use to interpret
that performance.

\section{Example of an Entanglement-Assisted Code}

\label{ex:science-code}We present an example of an entanglement-assisted code
that corrects an arbitrary single-qubit error \cite{science2006brun}. This
example highlights the main features of the theory given above. Suppose the
sender wants to use the quantum error-correcting properties of the following
nonabelian subgroup of $\Pi^{4}$:%
\begin{equation}%
\begin{array}
[c]{cccc}%
Z & X & Z & I\\
Z & Z & I & Z\\
X & Y & X & I\\
X & X & I & X
\end{array}
\label{eq:science-example}%
\end{equation}
The first two generators anticommute. We obtain a modified third generator by
multiplying the third generator by the second. We then multiply the last
generator by the first, second, and modified third generators. The
error-correcting properties of the generators are invariant under these
operations. The modified generators are as follows:%
\begin{equation}%
\begin{array}
[c]{cccccc}%
g_{1} & = & Z & X & Z & I\\
g_{2} & = & Z & Z & I & Z\\
g_{3} & = & Y & X & X & Z\\
g_{4} & = & Z & Y & Y & X
\end{array}
\end{equation}
The above set of generators have the commutation relations given by the
fundamental theorem of symplectic geometry:%
\[
\left\{  g_{1},g_{2}\right\}  =\left[  g_{1},g_{3}\right]  =\left[
g_{1},g_{4}\right]  =\left[  g_{2},g_{3}\right]  =\left[  g_{2},g_{4}\right]
=\left[  g_{3},g_{4}\right]  =0.
\]
The above set of generators is unitarily equivalent to the following canonical
generators:%
\begin{equation}%
\begin{array}
[c]{cccc}%
X & I & I & I\\
Z & I & I & I\\
I & Z & I & I\\
I & I & Z & I
\end{array}
\label{eq:canonical-Paulis}%
\end{equation}
We can add one ebit to resolve the anticommutativity of the first two
generators:%
\begin{equation}
\left.
\begin{array}
[c]{cccc}%
X & I & I & I\\
Z & I & I & I\\
I & Z & I & I\\
I & I & Z & I
\end{array}
\right\vert
\begin{array}
[c]{c}%
X\\
Z\\
I\\
I
\end{array}
\label{eq:canonical-stabilizer}%
\end{equation}
A +1-eigenstate of the above stabilizer is the following state%
\begin{equation}
\left\vert \Phi^{+}\right\rangle ^{AB}\left\vert 00\right\rangle
^{A}\left\vert \psi\right\rangle ^{A},\label{eq:canonical-state}%
\end{equation}
where $\left\vert \psi\right\rangle ^{A}$ is a qubit that the sender wants to
encode. A local encoding unitary then rotates the generators in
(\ref{eq:canonical-stabilizer})\ to the following set of globally commuting
generators:%
\begin{equation}
\left.
\begin{array}
[c]{cccc}%
Z & X & Z & I\\
Z & Z & I & Z\\
Y & X & X & Z\\
Z & Y & Y & X
\end{array}
\right\vert
\begin{array}
[c]{c}%
X\\
Z\\
I\\
I
\end{array}
\label{eq:encoded-stabilizer}%
\end{equation}
The receiver measures the above generators upon receipt of all qubits to
detect and correct errors.

\section{Encoding Algorithm}

\label{sec:encoding-alg-ent-assist}Suppose we have an arbitrary set of Pauli
matrices in $\Pi^{n}$ whose error-correcting properties we would like to
exploit. The algorithm in this section (in Ref.~\cite{arx2007wildeCED})
determines the encoding and decoding circuit for the set of Pauli generators.

It has two additional benefits. We do not necessarily know beforehand how many
ebits we require for the Pauli matrices to form a commuting set. Several
methods exist such as the Gram-Schmidt algorithm outlined in
Section~\ref{sec:ent-ass-gram} and the others in Refs.
\cite{arx2006brun,science2006brun,isit2007brun,hsieh:062313,arx2007wildeCED},
but the algorithm here also determines the optimal number of ebits and the
measurements the receiver performs to diagnose errors. It \textquotedblleft
kills three birds with one stone.\textquotedblright

We continue with the example in Section~\ref{ex:science-code}. We detail an
algorithm for determining an encoding circuit and the optimal number of ebits
for the entanglement-assisted code. The operators in (\ref{eq:science-example}%
) have the following representation as a binary matrix according to the
P2B\ isomorphism:%
\begin{equation}
H=\left[  \left.
\begin{array}
[c]{cccc}%
1 & 0 & 1 & 0\\
1 & 1 & 0 & 1\\
0 & 1 & 0 & 0\\
0 & 0 & 0 & 0
\end{array}
\right\vert
\begin{array}
[c]{cccc}%
0 & 1 & 0 & 0\\
0 & 0 & 0 & 0\\
1 & 1 & 1 & 0\\
1 & 1 & 0 & 1
\end{array}
\right]  .
\end{equation}
Call the matrix to the left of the vertical bar the \textquotedblleft
Z\textquotedblright\ matrix\ and the matrix to the right of the vertical bar
the \textquotedblleft X\textquotedblright\ matrix.

The algorithm consists of row and column operations on the above matrix. Row
operations do not affect the error-correcting properties of the code but are
crucial for arriving at the optimal decomposition from the fundamental theorem
of symplectic geometry. The operations available for manipulating columns of
the above matrix are Clifford operations (discussed at the end of
Section~\ref{sec:cliff-encoding}). The operations have the following effect on
entries in the binary matrix:

\begin{enumerate}
\item A CNOT\ gate from qubit $i$ to qubit $j$ adds column $i$ to column $j$
in the $X$ matrix and adds column $j$ to column $i$ in the $Z$ matrix.

\item A Hadamard gate on qubit $i$ swaps column $i$ in the $Z$ matrix with
column $i$ in the $X$ matrix and vice versa.

\item A phase gate on qubit $i$\ adds column $i$ in the $X$ matrix to column
$i$ in the $Z$ matrix.

\item Three CNOT\ gates implement a qubit swap operation
\cite{book2000mikeandike}. The effect of a swap on qubits $i$ and $j$ is to
swap columns $i$ and $j$ in both the $X$ and $Z$ matrix.
\end{enumerate}

The algorithm begins by computing the symplectic product between the first row
and all other rows. Leave the matrix as it is if the first row is not
symplectically orthogonal to the second row or if the first row is
symplectically orthogonal to all other rows. Otherwise, swap the second row
with the first available row that is not symplectically orthogonal to the
first row. In our example, the first row is not symplectically orthogonal to
the second so we leave all rows as they are.

Arrange the first row so that the top left entry in the $X$ matrix is one. A
CNOT, swap, Hadamard, or combinations of these operations can achieve this
result. We can have this result in our example by swapping qubits one and two.
The matrix becomes%
\begin{equation}
\left[  \left.
\begin{array}
[c]{cccc}%
0 & 1 & 1 & 0\\
1 & 1 & 0 & 1\\
1 & 0 & 0 & 0\\
0 & 0 & 0 & 0
\end{array}
\right\vert
\begin{array}
[c]{cccc}%
1 & 0 & 0 & 0\\
0 & 0 & 0 & 0\\
1 & 1 & 1 & 0\\
1 & 1 & 0 & 1
\end{array}
\right]  .
\end{equation}

Perform CNOTs to clear the entries in the $X$ matrix in the top row to the
right of the leftmost entry. These entries are already zero in this example so
we need not do anything. Proceed to the clear the entries in the first row of
the $Z$ matrix. Perform a phase gate to clear the leftmost entry in the first
row of the $Z$ matrix if it is equal to one. It is equal to zero in our
example so we need not do anything. We then use Hadamards and CNOTs to clear
the other entries in the first row of the $Z$ matrix.

For our example, perform a Hadamard on qubits two and three. The matrix
becomes%
\begin{equation}
\left[  \left.
\begin{array}
[c]{cccc}%
0 & 0 & 0 & 0\\
1 & 0 & 0 & 1\\
1 & 1 & 1 & 0\\
0 & 1 & 0 & 0
\end{array}
\right\vert
\begin{array}
[c]{cccc}%
1 & 1 & 1 & 0\\
0 & 1 & 0 & 0\\
1 & 0 & 0 & 0\\
1 & 0 & 0 & 1
\end{array}
\right]  .
\end{equation}
Perform a CNOT\ from qubit one to qubit two and from qubit one to qubit three.
The matrix becomes%
\begin{equation}
\left[  \left.
\begin{array}
[c]{cccc}%
0 & 0 & 0 & 0\\
1 & 0 & 0 & 1\\
1 & 1 & 1 & 0\\
1 & 1 & 0 & 0
\end{array}
\right\vert
\begin{array}
[c]{cccc}%
1 & 0 & 0 & 0\\
0 & 1 & 0 & 0\\
1 & 1 & 1 & 0\\
1 & 1 & 1 & 1
\end{array}
\right]  .
\end{equation}
The first row is complete. We now proceed to clear the entries in the second
row. Perform a Hadamard on qubits one and four. The matrix becomes%
\begin{equation}
\left[  \left.
\begin{array}
[c]{cccc}%
1 & 0 & 0 & 0\\
0 & 0 & 0 & 0\\
1 & 1 & 1 & 0\\
1 & 1 & 0 & 1
\end{array}
\right\vert
\begin{array}
[c]{cccc}%
0 & 0 & 0 & 0\\
1 & 1 & 0 & 1\\
1 & 1 & 1 & 0\\
1 & 1 & 1 & 0
\end{array}
\right]  .
\end{equation}
Perform a CNOT\ from qubit one to qubit two and from qubit one to qubit four.
The matrix becomes%
\begin{equation}
\left[  \left.
\begin{array}
[c]{cccc}%
1 & 0 & 0 & 0\\
0 & 0 & 0 & 0\\
0 & 1 & 1 & 0\\
1 & 1 & 0 & 1
\end{array}
\right\vert
\begin{array}
[c]{cccc}%
0 & 0 & 0 & 0\\
1 & 0 & 0 & 0\\
1 & 0 & 1 & 1\\
1 & 0 & 1 & 1
\end{array}
\right]  .
\end{equation}
The first two rows are now complete. They need one ebit to compensate for
their anticommutativity or their nonorthogonality with respect to the
symplectic product.

Now we perform row operations that are similar to the \textquotedblleft
symplectic Gram-Schmidt orthogonalization.\textquotedblright\ Add row one to
any other row that has one as the leftmost entry in its $Z$ matrix. Add row
two to any other row that has one as the leftmost entry in its $X$ matrix. For
our example, we add row one to row four and we add row two to rows three and
four. The matrix becomes%
\begin{equation}
\left[  \left.
\begin{array}
[c]{cccc}%
1 & 0 & 0 & 0\\
0 & 0 & 0 & 0\\
0 & 1 & 1 & 0\\
0 & 1 & 0 & 1
\end{array}
\right\vert
\begin{array}
[c]{cccc}%
0 & 0 & 0 & 0\\
1 & 0 & 0 & 0\\
0 & 0 & 1 & 1\\
0 & 0 & 1 & 1
\end{array}
\right]  .
\end{equation}
The first two rows are now symplectically orthogonal to all other rows per the
fundamental theorem of symplectic geometry.%

\begin{figure}
[ptb]
\begin{center}
\[
\Qcircuit@C=0.5em @R=1.0em  {
&  &  & \qw& \qw& \qw& \qw& \qw& \qw& \qw& \qw& \qw& \qw& \qw\gategroup{1}%
{3}{2}{3}{1.0em}{\{}\\
& \lstick{\raisebox{2em}{$\ket{\Phi^{+}}^{BA}$}} &  & \qw& \qw& \qw& \qw
& \qw& \ctrl{1} \qwx[3] & \gate{H} & \ctrl{1} \qwx[2] & \qw& \qswap
\qwx[1] & \qw\\
& \lstick{\ket{0}^A}      & \gate{H} & \qw& \ctrl{1} & \gate{P} & \ctrl{1}
\qwx[2] & \gate{H} & \targ& \qw& \targ& \gate{H} & \qswap& \qw\\
& \lstick{\ket{0}^A}      & \gate{H} & \ctrl{1} & \targ& \gate{H} & \targ
& \qw& \qw& \qw& \targ& \gate{H} & \qw& \qw\\
& \lstick{\ket{\psi}^A}  & \qw& \targ& \gate{H} & \qw& \targ& \qw
& \targ& \gate{H} & \qw& \qw& \qw& \qw}
\]
\end{center}
\caption{Encoding circuit for the entanglement-assisted code in the example of Section~\ref{ex:science-code}%
. The ``H'' gate is a Hadamard gate and the ``P'' gate is a phase gate.}
\label{fig:encoding-circuit-science}
\end{figure}
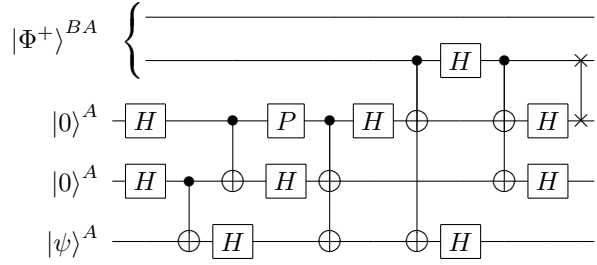

We proceed with the same algorithm on the next two rows. We can deal with the next two rows individually because they are
symplectically orthogonal to each other.
Perform a Hadamard on qubit two. The matrix becomes%
\begin{equation}
\left[  \left.
\begin{array}
[c]{cccc}%
1 & 0 & 0 & 0\\
0 & 0 & 0 & 0\\
0 & 0 & 1 & 0\\
0 & 0 & 0 & 1
\end{array}
\right\vert
\begin{array}
[c]{cccc}%
0 & 0 & 0 & 0\\
1 & 0 & 0 & 0\\
0 & 1 & 1 & 1\\
0 & 1 & 1 & 1
\end{array}
\right]  .
\end{equation}
Perform a CNOT\ from qubit two to qubit three and from qubit two to qubit
four. The matrix becomes%
\begin{equation}
\left[  \left.
\begin{array}
[c]{cccc}%
1 & 0 & 0 & 0\\
0 & 0 & 0 & 0\\
0 & 1 & 1 & 0\\
0 & 1 & 0 & 1
\end{array}
\right\vert
\begin{array}
[c]{cccc}%
0 & 0 & 0 & 0\\
1 & 0 & 0 & 0\\
0 & 1 & 0 & 0\\
0 & 1 & 0 & 0
\end{array}
\right]  .
\end{equation}
Perform a phase gate on qubit two:%
\begin{equation}
\left[  \left.
\begin{array}
[c]{cccc}%
1 & 0 & 0 & 0\\
0 & 0 & 0 & 0\\
0 & 0 & 1 & 0\\
0 & 0 & 0 & 1
\end{array}
\right\vert
\begin{array}
[c]{cccc}%
0 & 0 & 0 & 0\\
1 & 0 & 0 & 0\\
0 & 1 & 0 & 0\\
0 & 1 & 0 & 0
\end{array}
\right]  .
\end{equation}
Perform a Hadamard on qubit three followed by a CNOT\ from qubit two to qubit
three:%
\begin{equation}
\left[  \left.
\begin{array}
[c]{cccc}%
1 & 0 & 0 & 0\\
0 & 0 & 0 & 0\\
0 & 0 & 0 & 0\\
0 & 0 & 0 & 1
\end{array}
\right\vert
\begin{array}
[c]{cccc}%
0 & 0 & 0 & 0\\
1 & 0 & 0 & 0\\
0 & 1 & 0 & 0\\
0 & 1 & 1 & 0
\end{array}
\right]  .
\end{equation}
Add row three to row four and perform a Hadamard on qubit two:%
\begin{equation}
\left[  \left.
\begin{array}
[c]{cccc}%
1 & 0 & 0 & 0\\
0 & 0 & 0 & 0\\
0 & 1 & 0 & 0\\
0 & 0 & 0 & 1
\end{array}
\right\vert
\begin{array}
[c]{cccc}%
0 & 0 & 0 & 0\\
1 & 0 & 0 & 0\\
0 & 0 & 0 & 0\\
0 & 0 & 1 & 0
\end{array}
\right]  .
\end{equation}
Perform a Hadamard on qubit four followed by a CNOT\ from qubit three to qubit
four. End by performing a Hadamard on qubit three:%
\begin{equation}
\left[  \left.
\begin{array}
[c]{cccc}%
1 & 0 & 0 & 0\\
0 & 0 & 0 & 0\\
0 & 1 & 0 & 0\\
0 & 0 & 1 & 0
\end{array}
\right\vert
\begin{array}
[c]{cccc}%
0 & 0 & 0 & 0\\
1 & 0 & 0 & 0\\
0 & 0 & 0 & 0\\
0 & 0 & 0 & 0
\end{array}
\right]  .
\end{equation}
The above matrix now corresponds to the canonical Paulis
in (\ref{eq:canonical-Paulis}). Adding one half of an ebit to the receiver's side
gives the canonical stabilizer in (\ref{eq:canonical-stabilizer}) whose
simultaneous $+1$-eigenstate is the state in (\ref{eq:canonical-state}).

Figure~\ref{fig:encoding-circuit-science} gives the encoding circuit
corresponding to the above operations. The above operations in reverse order
take the canonical stabilizer (\ref{eq:canonical-stabilizer}) to the encoded
stabilizer (\ref{eq:encoded-stabilizer}).

\section{Optimal Entanglement Formulas}

\label{sec:opt-ebit}Entanglement is a valuable resource, and we would like to
minimize the amount that a sender and receiver need to consume for
entanglement-assisted quantum coding. Hsieh, Devetak, and Brun first addressed
this issue by determining a useful formula that gives the optimal number of
ebits required by a Calderbank-Shor-Steane (CSS)\ entanglement-assisted code
\cite{hsieh:062313}. In particular, the number $c$\ of ebits that a
CSS\ entanglement-assisted code requires is%
\begin{equation}
c=\mathrm{rank}\left(  HH^{T}\right)  , \label{eq:CSS-form}%
\end{equation}
where $H$ corresponds to the parity check matrix of a classical binary block
code that we import to correct quantum bit and phase flips. The same authors
also determined an upper bound for the number of ebits that an
entanglement-assisted quantum LDPC\ code requires \cite{aqis2007hsieh}.

In this section, we present several generalizations of the above formula. Our
first theorem gives a formula for the optimal number of ebits that an
arbitrary (non-CSS)\ entanglement-assisted quantum block code requires. We
find special cases of this formula that apply to an entanglement-assisted
quantum block code produced from two arbitrary classical binary block codes,
and to codes from a classical block code over $GF(4)$. Our last formula
applies to \textit{continuous-variable} entanglement-assisted codes
\cite{pra2007wildeEA}.

\begin{theorem}
\label{thm:opt-ebit-formula}Suppose we want to build an entanglement-assisted
quantum code from generators corresponding to the rows in a quantum check
matrix%
\begin{equation}
H=\left[  \left.
\begin{array}
[c]{c}%
H_{Z}%
\end{array}
\right\vert
\begin{array}
[c]{c}%
H_{X}%
\end{array}
\right]  ,
\end{equation}
where $H$ is an $\left(  n-k\right)  \times2n$-dimensional binary matrix
representing the quantum code, and both $H_{Z}$ and $H_{X}$ are $\left(
n-k\right)  \times n$-dimensional binary matrices. Then the resulting code is
an $\left[  \left[  n,k+c;c\right]  \right]  $ entanglement-assisted code and
requires $c$ ebits, where%
\begin{equation}
c=\mathrm{rank}\left(  H_{X}H_{Z}^{T}+H_{Z}H_{X}^{T}\right)  /2,
\label{eq:optimal-ebits}%
\end{equation}
and addition is binary.
\end{theorem}

\begin{proof}%

Consider the symplectic product matrix $\Omega_{H}=H_{X}H_{Z}^{T}+H_{Z}%
H_{X}^{T}$. Its entries are the symplectic products between all rows of $H$ so
that%
\begin{equation}
\left[  \Omega_{H}\right]  _{ij}=h_{i}\odot h_{j},
\end{equation}
where $h_{i}$ is the $i^{\text{th}}$ row of $H$. The matrix $\Omega_{H}$ is a
$\left(  n-k\right)  \times\left(  n-k\right)  $-dimensional binary matrix.

Refs.~\cite{arx2006brun,science2006brun}\ and the algorithm in Section~\ref{sec:ent-ass-gram} outline a symplectic Gram-Schmidt
orthogonalization procedure (SGSOP)\ that uniquely determines the optimal
(i.e., minimal) number of ebits that the code requires, and
Ref.~\cite{unpub2007got} proves that the SGSOP gives the optimal number of
ebits. The code construction in Refs.~\cite{arx2006brun,science2006brun} shows
that the resulting entanglement-assisted quantum code requires at most $c$
ebits. The essence of the argument in Ref.~\cite{unpub2007got} is that the
resulting entanglement-assisted quantum code requires at least $c$ ebits
because any fewer ebits would not be able to resolve the anticommutativity of
the generators on Alice's side of the code.

The SGSOP performs row operations that do not change the error-correcting
properties of the quantum code (because the code is additive), but these row
operations do change the symplectic product relations. These row operations
are either a row swap $S\left(  i,j\right)  $, where $S\left(  i,j\right)  $
is a full-rank $\left(  n-k\right)  \times\left(  n-k\right)  $ matrix that
swaps row $i$ with $j$, or a row addition $A\left(  i,j\right)  $, where
$A\left(  i,j\right)  $ is a full-rank $\left(  n-k\right)  \times\left(
n-k\right)  $ matrix that adds row $i$ to row $j$. These row operations
multiply the matrix $H$ from the left. The SGSOP then is equivalent to a
full-rank $\left(  n-k\right)  \times\left(  n-k\right)  $ matrix $G$ that
contains all of the row operations and produces a new quantum check matrix
$H^{\prime}=GH$ with corresponding symplectic product matrix $\Omega
_{H^{\prime}}=G\left(  H_{X}H_{Z}^{T}+H_{Z}H_{X}^{T}\right)  G^{T}$. In
particular, the resulting symplectic product matrix $\Omega_{H^{\prime}}$\ is
in the standard form in (\ref{eq:standard-symp-form}) so that%
\begin{equation}
\Omega_{H^{\prime}}=%
{\displaystyle\bigoplus\limits_{i=1}^{c}}
J\oplus
{\displaystyle\bigoplus\limits_{j=1}^{n-k-2c}}
\left[  0\right]  .
\end{equation}
Each matrix $J$ in the direct sum corresponds to half of an ebit as described
in Section~\ref{sec:EA-symp-relations} and has rank two. Each matrix
$\left[  0\right]  $ has rank zero and corresponds to an ancilla qubit. The
optimal number of ebits required for the code is $\mathrm{rank}\left(
\Omega_{H^{\prime}}\right)  /2$:%
\begin{equation}
\text{rank}\left(  \Omega_{H^{\prime}}\right)  =\text{rank}\left(
{\displaystyle\bigoplus\limits_{i=1}^{c}}
J\oplus%
{\displaystyle\bigoplus\limits_{j=1}^{n-k-2c}}
\left[  0\right]  \right)  =\sum_{i=1}^{c}\text{rank}\left(  J\right)
+\sum_{j=1}^{n-k-2c}\text{rank}\left(  \left[  0\right]  \right)  =2c.
\end{equation}
The second equality follows because the rank of a direct sum is the sum of the
individual matrix ranks, and the third equality follows from the individual
matrix ranks given above. The number $c$\ of ebits is also equal to
$\mathrm{rank}\left(  \Omega_{H}\right)  /2$ because the matrix $G$ is full
rank. The code is an $\left[  \left[  n,k+c;c\right]  \right]  $
entanglement-assisted quantum block code by the construction in
Refs.~\cite{arx2006brun,science2006brun}.%

\end{proof}%

Our formula (\ref{eq:optimal-ebits}) is equivalent to the formula at the top
of page 14 in Ref.~\cite{arx2006brun}, but it provides the quantum code
designer with a quick method to determine how many ebits an
entanglement-assisted code requires, by simply \textquotedblleft plugging
in\textquotedblright\ the generators of the code.

The formula (\ref{eq:optimal-ebits}), like the CSS\ formula in
(\ref{eq:CSS-form}), is a measure of how far a set of generators is from being
a commuting set, or equivalently, how far it is from giving a standard
stabilizer code.

Corollary~\ref{cor:CSS} below gives a formula for the optimal number of ebits
required by a CSS\ entanglement-assisted quantum code. It is generally a bit
less difficult to compute than the above formula in (\ref{eq:optimal-ebits}).
This reduction in complexity occurs because of the special form of a
CSS\ quantum code and because the size of the matrices involved are generally
smaller for a CSS\ code than for a general code with the same number of
generators and physical qubits.

\begin{corollary}
\label{cor:CSS}Suppose we import two classical $\left[  n,k_{1},d_{1}\right]
$ and $\left[  n,k_{2},d_{2}\right]  $ binary codes with respective parity
check matrices $H_{1}$ and $H_{2}$ to build an entanglement-assisted quantum
code. The resulting code is an $\left[  \left[  n,k_{1}+k_{2}-n+c,\min\left(
d_{1},d_{2}\right)  ;c\right]  \right]  $ entanglement-assisted code, and
requires $c$ ebits where%
\begin{equation}
c=\mathrm{rank}\left(  H_{1}H_{2}^{T}\right)  .
\end{equation}

\end{corollary}

%

\begin{proof}%
The quantum check matrix has the following form:%
\begin{equation}
H=\left[  \left.
\begin{array}
[c]{c}%
H_{1}\\
0
\end{array}
\right\vert
\begin{array}
[c]{c}%
0\\
H_{2}%
\end{array}
\right]  .
\end{equation}
The symplectic product matrix $\Omega_{H}$ is then%
\begin{equation}
\Omega_{H}=%
\begin{bmatrix}
H_{1}\\
0
\end{bmatrix}%
\begin{bmatrix}
0 & H_{2}^{T}%
\end{bmatrix}
+%
\begin{bmatrix}
0\\
H_{2}%
\end{bmatrix}%
\begin{bmatrix}
H_{1}^{T} & 0
\end{bmatrix}
=%
\begin{bmatrix}
0 & H_{1}H_{2}^{T}\\
H_{2}H_{1}^{T} & 0
\end{bmatrix}
.
\end{equation}
The above matrix is equivalent by a full rank permutation matrix to the matrix
$H_{1}H_{2}^{T}\oplus H_{2}H_{1}^{T}$, so the rank of $\Omega_{H}$ is%
\begin{equation}
\text{rank}\left(  \Omega_{H}\right)  =\text{rank}\left(  H_{1}H_{2}^{T}\oplus
H_{2}H_{1}^{T}\right)  =\text{rank}\left(  H_{1}H_{2}^{T}\right)
+\text{rank}\left(  H_{2}H_{1}^{T}\right)  =2~\text{rank}\left(  H_{1}%
H_{2}^{T}\right)
\end{equation}
The second equality follows because the rank of a direct sum is equivalent to
the sum of the individual ranks, and the third equality follows because the
rank is invariant under matrix transposition. The number of ebits required for
the resulting entanglement-assisted quantum code is rank$\left(  H_{1}%
H_{2}^{T}\right)  $, using the result of the previous theorem. The
construction in Refs.~\cite{arx2006brun,science2006brun} produces an $\left[
\left[  n,k_{1}+k_{2}-n+c,\min\left(  d_{1},d_{2}\right)  ;c\right]  \right]
$ entanglement-assisted quantum block code.%

\end{proof}%

\begin{corollary}
Suppose we import an $\left[  n,k,d\right]  _{4}$\ classical code over
$GF\left(  4\right)  $ with parity check matrix $H$ for use as an
entanglement-assisted quantum code according to the construction in
Refs.~\cite{arx2006brun,science2006brun}. Then the resulting quantum code is
an $\left[  \left[  n,2k-n+c;c\right]  \right]  $ entanglement-assisted
quantum code where $c=rank\left(  HH^{\dag}\right)  $ and $\dag$ denotes the
conjugate transpose operation over matrices in $GF\left(  4\right)  $.
\end{corollary}

\begin{proof}%

Ref.~\cite{science2006brun} shows how to produce an entanglement-assisted
quantum code from the parity check matrix $H$\ of a classical code over
$GF\left(  4\right)  $. The resulting quantum parity check matrix $H_{Q}$\ in
symplectic binary form is%
\begin{equation}
H_{Q}=\gamma\left(
\begin{bmatrix}
\omega H\\
\bar{\omega}H
\end{bmatrix}
\right)  ,
\end{equation}
where $\gamma$ denotes the isomorphism between elements of $GF\left(
4\right)  $ and symplectic binary vectors that represent Pauli matrices. We
augment the table in (\ref{eq:z2-pauli}) to include the entries from
$GF\left(  4\right)  $:%
\begin{equation}%
\begin{tabular}
[c]{l|llll}%
$\Pi$ & $I$ & $X$ & $Y$ & $Z$\\\hline
$\left(  \mathbb{Z}_{2}\right)  ^{2}$ & $00$ & $01$ & $11$ & $10$\\\hline
$GF\left(  4\right)  $ & $0$ & $\omega$ & $1$ & $\bar{\omega}$%
\end{tabular}
\ \label{eq:gf4-pauli}%
\end{equation}
The isomorphism $\gamma$ is the mapping from the third row to the second row
in the above table. The symplectic product between binary vectors is
equivalent to the trace product of their $GF\left(  4\right)  $
representations (see, e.g., Ref.~\cite{science2006brun}):%
\begin{equation}
h_{i}\odot h_{j}=\text{tr}\left\{  \gamma^{-1}\left(  h_{i}\right)
\cdot\overline{\gamma^{-1}\left(  h_{j}\right)  }\right\}  ,
\end{equation}
where $h_{i}$ and $h_{j}$ are any two rows of $H_{Q}$, $\cdot$ denotes the
inner product, the overbar denotes the conjugate operation, and tr$\left\{
x\right\}  =x+\bar{x}$ denotes the trace operation over elements of $GF\left(
4\right)  $. We exploit these correspondences to write the symplectic product
matrix $\Omega_{H_{Q}}$ for the quantum check matrix $H_{Q}$ as follows:%
\begin{align}
\Omega_{H_{Q}} &  =\text{tr}\left\{
\begin{bmatrix}
\omega H\\
\bar{\omega}H
\end{bmatrix}%
\begin{bmatrix}
\omega H\\
\bar{\omega}H
\end{bmatrix}
^{\dag}\right\}  =\text{tr}\left\{
\begin{bmatrix}
\omega H\\
\bar{\omega}H
\end{bmatrix}%
\begin{bmatrix}
\bar{\omega}H^{\dag} & \omega H^{\dag}%
\end{bmatrix}
\right\}  \\
&  =\text{tr}\left\{
\begin{bmatrix}
HH^{\dag} & \bar{\omega}HH^{\dag}\\
\omega HH^{\dag} & HH^{\dag}%
\end{bmatrix}
\right\}  =\text{tr}\left\{
\begin{bmatrix}
1 & \bar{\omega}\\
\omega & 1
\end{bmatrix}
\otimes HH^{\dag}\right\}  \\
&  =%
\begin{bmatrix}
1 & \bar{\omega}\\
\omega & 1
\end{bmatrix}
\otimes HH^{\dag}+%
\begin{bmatrix}
1 & \omega\\
\bar{\omega} & 1
\end{bmatrix}
\otimes\bar{H}H^{T}%
\end{align}
where the \textquotedblleft tr\textquotedblright\ operation above is an
element-wise trace operation over $GF\left(  4\right)  $ (it is not the matrix
trace operation.) The matrix $\Omega_{H_{Q}}$ is over the field $GF\left(
2\right)  $, but we can consider it as being over the field $GF\left(
4\right)  $ without changing its rank. Therefore, we can multiply it by
matrices over the field $GF\left(  4\right)  $. Consider the following
full-rank $GF\left(  4\right)  $\ matrices:%
\begin{equation}
A_{1}=%
\begin{bmatrix}
1 & \bar{\omega}\\
0 & 1
\end{bmatrix}
\otimes I,\ \ \ \ A_{2}=%
\begin{bmatrix}
1 & 0\\
1 & 1
\end{bmatrix}
\otimes I.
\end{equation}
We premultiply and postmultiply the matrix $\Omega_{H_{Q}}$ as follows and
obtain a matrix with the same rank as $\Omega_{H_{Q}}$:%
\begin{equation}
A_{2}A_{1}\Omega_{H_{Q}}A_{1}^{\dag}A_{2}^{\dag}=%
\begin{bmatrix}
0 & 0\\
0 & 1
\end{bmatrix}
\otimes HH^{\dag}+%
\begin{bmatrix}
1 & 0\\
0 & 0
\end{bmatrix}
\otimes\bar{H}H^{T}=%
\begin{bmatrix}
\bar{H}H^{T} & 0\\
0 & HH^{\dag}%
\end{bmatrix}
=\bar{H}H^{T}\oplus HH^{\dag}%
\end{equation}
Therefore, the rank of $\Omega_{H_{Q}}$ is%
\begin{equation}
\text{rank}\left(  \Omega_{H_{Q}}\right)  =\text{rank}\left(  \bar{H}%
H^{T}\oplus HH^{\dag}\right)  =\text{rank}\left(  \bar{H}H^{T}\right)
+\text{rank}\left(  HH^{\dag}\right)  =2\text{ rank}\left(  HH^{\dag}\right)
.
\end{equation}
The second equality holds because the rank of a direct sum is the sum of the
individual ranks and the third holds because the rank is invariant under the
matrix transpose operation. Therefore, the resulting entanglement-assisted
quantum code requires $c=~$rank$\left(  HH^{\dag}\right)  $ ebits by applying
the result of the original theorem. The construction in
Refs.~\cite{arx2006brun,science2006brun}\ produces an $\left[  \left[
n,2k-n+c;c\right]  \right]  $ entanglement-assisted quantum code.%

\end{proof}%

\begin{corollary}
We can construct a continuous-variable entanglement-assisted quantum code from
generators corresponding to the rows in quantum check matrix $H=\left[
\left.
\begin{array}
[c]{c}%
H_{Z}%
\end{array}
\right\vert
\begin{array}
[c]{c}%
H_{X}%
\end{array}
\right]  $ where $H$ is $\left(  n-k\right)  \times2n$-dimensional, $H$\ is a
real matrix representing the quantum code \cite{pra2007wildeEA}, and both
$H_{Z}$ and $H_{X}$ are $\left(  n-k\right)  \times n$-dimensional. The
resulting code is an $\left[  \left[  n,k+c;c\right]  \right]  $
continuous-variable entanglement-assisted code and requires $c$ entangled
modes where%
\begin{equation}
c=\mathrm{rank}\left(  H_{X}H_{Z}^{T}-H_{Z}H_{X}^{T}\right)  /2.
\end{equation}

\end{corollary}

\begin{proof}%
The proof is similar to the proof of the first theorem but requires
manipulations of real vectors instead of binary vectors. See
Ref.~\cite{pra2007wildeEA} for details of the symplectic geometry required for
continuous-variable entanglement-assisted codes.%
\end{proof}%

\begin{remark}
A similar formula holds for entanglement-assisted qudit codes by replacing the
subtraction operation above with subtraction modulo $d$. Specifically, we can
construct a qudit entanglement-assisted quantum code from generators
corresponding to the rows in check matrix\ $H=\left[  \left.
\begin{array}
[c]{c}%
H_{Z}%
\end{array}
\right\vert
\begin{array}
[c]{c}%
H_{X}%
\end{array}
\right]  $ whose matrix entries are elements of the finite field
$\mathbb{Z}_{d}$. The code requires $c$ edits (a $d$-dimensional state
$\left(  \sum_{i=0}^{d-1}\left\vert i\right\rangle \left\vert i\right\rangle
\right)  /\sqrt{d}$) where%
\[
c=\mathrm{rank}\left(  H_{X}H_{Z}^{T}\ominus_{d}H_{Z}H_{X}^{T}\right)  /2
\]
and $\ominus_{d}$ is subtraction modulo $d$. We use subtraction modulo $d$
because the symplectic form over $d$-dimensional variables includes
subtraction modulo $d$.
\end{remark}

\section{Closing Remarks}

This chapter reviewed the theory of entanglement-assisted coding for block codes.
We showed how to manipulate these codes with both a Pauli representation and a binary one.
Our contributions included an algorithm for computing an encoding circuit and a method to
determine the optimal number of ebits for an entanglement-assisted code.
Much of this thesis extends the entanglement-assisted technique to the domain of quantum convolutional coding.
We introduce quantum convolutional coding in the next chapter.

\chapter{Quantum Convolutional Codes}

\label{chp:conv}\begin{saying}
Forney, Viterbi, and Elias,\\
To Andy the Trojans are pious,\\
Their methods, we want 'em,\\
Who'd think they'd go quantum?\\
So now they still stupefy us.
\end{saying}The block codes in the previous two chapters are useful
in quantum computing and in quantum communication. One of the drawbacks of the
block-coding technique for quantum communication is that, in general, the
sender must have all her qubits ready before encoding takes place. For a large
block code, this preparation may be a heavy demand on the sender.

Quantum convolutional coding theory
\cite{PhysRevLett.91.177902,arxiv2004olliv,isit2005forney,ieee2007forney}
offers a different paradigm for coding quantum information. The convolutional
structure is useful in a quantum communication scenario where a sender
possesses a stream of qubits to send to a receiver. The encoding circuit for a
quantum convolutional code has a much lower complexity than an encoding
circuit needed for a large block code. It also has a repetitive pattern so
that the same physical devices or the same routines can encode the stream of
quantum information in an online fashion as soon as information qubits are
available for encoding.

Quantum convolutional stabilizer codes borrow heavily from the structure of
their classical counterparts
\cite{PhysRevLett.91.177902,arxiv2004olliv,isit2005forney,ieee2007forney}.
Quantum convolutional codes are similar because some of the qubits feed back
into a repeated encoding unitary and give the code a memory structure like
that of a classical convolutional code. The quantum codes feature online
encoding and decoding of qubits. This feature gives quantum convolutional
codes both their low encoding and decoding complexity.

Our techniques for encoding and decoding are an expansion of previous
techniques from quantum convolutional coding theory. Previous techniques for
encoding and decoding include finite-depth operations only. A finite-depth
operation propagates errors to a finite number of neighboring qubits in the
qubit stream. We introduce an infinite-depth operation to the set of
shift-invariant Clifford operations and explain it in detail in
Section~\ref{sec:infinite-depth-ops}. We must be delicate when using
infinite-depth operations because they can propagate errors to an infinite
number of neighboring qubits in the qubit stream. We explain our assumptions
in detail when we include
infinite-depth operations in entanglement-assisted quantum convolutional
codes in Section~\ref{sec:eaqcc-iefd}\ of the next chapter. An infinite-depth operation gives more flexibility when designing
encoding circuits---similar to the way in which an infinite-impulse response
filter gives more flexibility in the design of classical convolutional
circuits.

We structure this chapter as follows. In Section~\ref{sec:conv-stabilizer}, we
introduce the definition of a quantum convolutional code and discuss how it
operates. Section~\ref{sec:P2B-conv}\ introduces the Pauli-to-binary
isomorphism. This mapping simplifies the mathematics of a quantum
convolutional code by showing how to represent it with a matrix of binary
polynomials. Section~\ref{sec:shifted-symp-prod} discusses the important
\textquotedblleft shifted symplectic product\textquotedblright\ that gives the
commutation relations of a set of convolutional generators. We discuss in
Section~\ref{sec:row-col-ops} how to manipulate the binary polynomial matrix
that represents a quantum convolutional code. In
Sections~\ref{sec:finite-depth-clifford} and \ref{sec:infinite-depth-ops}, we
respectively review finite-depth and infinite-depth operations for encoding a
quantum convolutional code.

\section{Review of the Convolutional Stabilizer Formalism}

\label{sec:conv-stabilizer}We review the theory of convolutional stabilizer
codes by considering a set of Pauli matrices that stabilize a stream of
encoded qubits. We first give the mathematical definition of a quantum
convolutional code and follow by discussing the various steps involved in the
operation of it.

\subsection{Quantum Convolutional Code Definition}

A quantum convolutional stabilizer code acts on a Hilbert space $\mathcal{H}$
that\ is\ a countably infinite tensor product of two-dimensional qubit Hilbert
spaces $\left\{  \mathcal{H}_{i}\right\}  _{i\in\mathbb{Z}^{+}}$ where%
\begin{equation}
\mathcal{H}=%
{\displaystyle\bigotimes\limits_{i=0}^{\infty}}
\ \mathcal{H}_{i}.
\end{equation}
and $\mathbb{Z}^{+}\equiv\left\{  0,1,\ldots\right\}  $. A sequence
$\mathbf{A}$ of Pauli matrices $\left\{  A_{i}\right\}  _{i\in\mathbb{Z}^{+}}%
$, where%
\begin{equation}
\mathbf{A}=%
{\displaystyle\bigotimes\limits_{i=0}^{\infty}}
\ A_{i},
\end{equation}
can act on states in $\mathcal{H}$. Let $\Pi^{\mathbb{Z}^{+}}$ denote the set
of all Pauli sequences. The support supp$\left(  \mathbf{A}\right)  $\ of a
Pauli sequence $\mathbf{A}$ is the set of indices of the entries in
$\mathbf{A}$ that are not equal to the identity. The weight of a sequence
$\mathbf{A}$ is the size $\left\vert \text{supp}\left(  \mathbf{A}\right)
\right\vert $\ of its support. The delay del$\left(  \mathbf{A}\right)  $ of a
sequence $\mathbf{A}$ is the smallest index for an entry not equal to the
identity. The degree deg$\left(  \mathbf{A}\right)  $ of a sequence
$\mathbf{A}$ is the largest index for an entry not equal to the identity.
E.g., the following Pauli sequence%
\begin{equation}%
\begin{array}
[c]{cccccccc}%
I & X & I & Y & Z & I & I & \cdots
\end{array}
,
\end{equation}
has support $\left\{  1,3,4\right\}  $, weight three, delay one, and degree
four. A sequence has finite support if its weight is finite. Let
$F(\Pi^{\mathbb{Z}^{+}})$ denote the set of Pauli sequences with finite
support. The following definition for a quantum convolutional code utilizes
the set $F(\Pi^{\mathbb{Z}^{+}})$ in its description.

\begin{definition}
\label{def:conv-code}A rate $k/n$-convolutional stabilizer code with $0\leq
k\leq n$ is specified by a commuting set $\mathcal{G}$\ of all $n$-qubit shifts of a basic
generator set $\mathcal{G}_{0}$. The commutativity requirement is necessary
for the same reason that standard stabilizer codes require it
\cite{book2000mikeandike}. The basic generator set $\mathcal{G}_{0}$ has $n-k$
Pauli sequences of finite support:%
\begin{equation}
\mathcal{G}_{0}=\left\{  \mathbf{G}_{i}\in F(\Pi^{\mathbb{Z}^{+}}):1\leq i\leq
n-k\right\}  .
\end{equation}
The constraint length $\nu$ of the code is the maximum degree of the
generators in $\mathcal{G}_{0}$. A frame of the code consists of $n$ qubits.
The definition of a quantum convolutional code as $n$-qubit shifts of the
basic set $\mathcal{G}_{0}$\ is what gives the code its periodic structure.
\end{definition}

A quantum convolutional code admits an equivalent definition in terms of the
delay transform or $D$-transform. The $D$-transform captures shifts of the
basic generator set $\mathcal{G}_{0}$. Let us define the $n$-qubit delay
operator $D$ acting on any Pauli sequence $\mathbf{A}\in\Pi^{\mathbb{Z}^{+}}%
$\ as follows:%
\begin{equation}
D\left(  \mathbf{A}\right)  =I^{\otimes n}\otimes\mathbf{A.}
\label{eq:delay-transform}%
\end{equation}
We can write $j$ repeated applications of $D$ as a power of $D$:%
\begin{equation}
D^{j}\left(  \mathbf{A}\right)  =I^{\otimes jn}\otimes\mathbf{A.}%
\end{equation}
Let $D^{j}\left(  \mathcal{G}_{0}\right)  $ be the set of shifts of elements
of $\mathcal{G}_{0}$ by $j$. Then the full stabilizer $\mathcal{G}$ for the
convolutional stabilizer code is%
\begin{equation}
\mathcal{G}=%
{\textstyle\bigcup\limits_{j\in\mathbb{Z}^{+}}}
D^{j}\left(  \mathcal{G}_{0}\right)  .
\end{equation}

\begin{example}
\label{sec:qcc-example}Forney et al. provided an example of a rate-1/3 quantum
convolutional code by importing a particular classical quaternary
convolutional code \cite{isit2005forney,ieee2007forney}. Grassl and
R\"{o}tteler determined a noncatastrophic encoding circuit for Forney et al.'s
rate-1/3 quantum convolutional code \cite{isit2006grassl}. The basic
stabilizer and its first shift are as follows:%
\begin{equation}
\cdots\left\vert
\begin{array}
[c]{c}%
III\\
III\\
III\\
III
\end{array}
\right\vert
\begin{array}
[c]{c}%
XXX\\
ZZZ\\
III\\
III
\end{array}
\left\vert
\begin{array}
[c]{c}%
XZY\\
ZYX\\
XXX\\
ZZZ
\end{array}
\right\vert
\begin{array}
[c]{c}%
III\\
III\\
XZY\\
ZYX
\end{array}
\left\vert
\begin{array}
[c]{c}%
III\\
III\\
III\\
III
\end{array}
\right\vert \cdots\label{eq:qcc-example-stabilizer}%
\end{equation}
The code consists of all three-qubit shifts of the above generators. The
vertical bars are a visual aid to illustrate the three-qubit shifts of the
basic generators. The code can correct for an arbitrary single-qubit error in
every other frame.
\end{example}

\subsection{Quantum Convolutional Code Operation}

Figure~\ref{fig:qcc} illustrates the basic operation of a quantum
convolutional code. The operation of a rate-$k/n$ quantum convolutional code
consists of several steps:%
\begin{figure*}
[ptb]
\begin{center}
\includegraphics[
natheight=4.260100in,
natwidth=18.372900in,
height=1.339in,
width=5.72in
]
{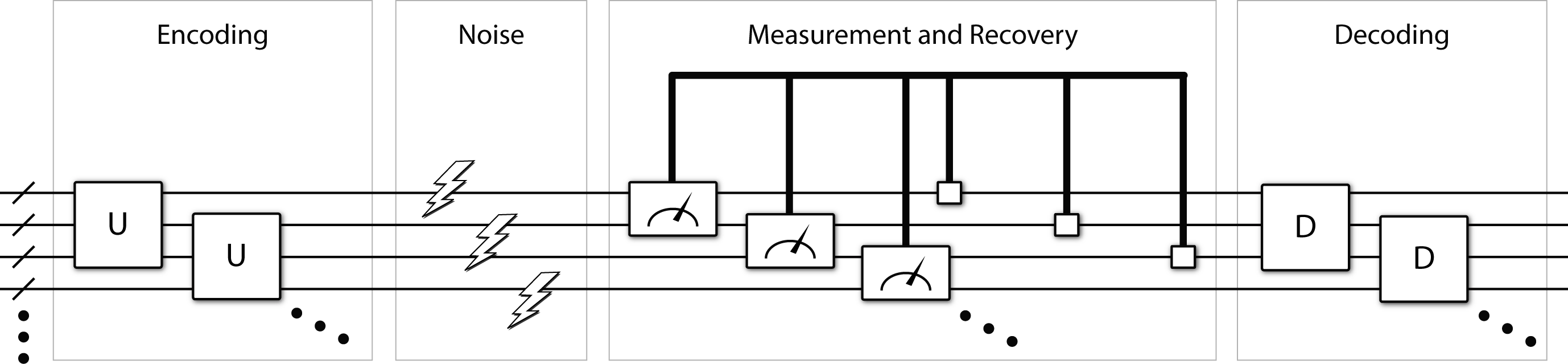}
\caption{The operation of a quantum convolutional code.
The sender applies the same unitary successively to a stream of information qubits and ancilla qubits.
The convolutional structure implies that the unitary overlaps some of the same qubits.
The sender transmits her qubits as soon as the unitary finishes processing them.
The noisy quantum channel corrupts the transmitted qubits.
The receiver performs overlapping multi-qubit measurements to diagnose
channel errors and corrects for them. The receiver performs an online decoding circuit to recover
the sender's original stream of information qubits.}
\label{fig:qcc}
\end{center}
\end{figure*}%

\begin{enumerate}
\item The protocol begins with the sender encoding a stream of information
qubits with an online encoding circuit. The sender encodes $n-k$ ancilla
qubits and $k$ information qubits per frame
\cite{ieee2007grassl,isit2006grassl}. The encoding circuit is
\textquotedblleft online\textquotedblright\ if it acts on a few frames at a time.

\item The sender transmits a set of initial qubits as soon as the first unitary
finishes processing them and continues to transmit later qubits after the next part of the
online encoding circuit finishes processing them.

\item The above basic set $\mathcal{G}_{0}$\ and all of its $n$-qubit shifts act like a
parity check matrix for the quantum convolutional code. The receiver measures
the generators in the stabilizer to determine an error syndrome. It is
important that the generators in $\mathcal{G}_{0}$ have finite weight so that
the receiver can perform the measurements and produce an error syndrome. It is
also important that the generators have a block-band form so that the receiver
can perform the measurements online as the noisy encoded qubits arrive.

\item The receiver processes the error syndromes with an online classical
error estimation algorithm such as the Viterbi algorithm
\cite{itit1967viterbi}\ or any other decoding algorithm \cite{book1999conv} to
determine the most likely error for each frame of quantum data. Syndrome-based
versions of the Viterbi algorithm that are appropriate for quantum coding are
available in Refs.~\cite{PhysRevLett.91.177902,arxiv2004olliv,ollivier:032304}.

\item The receiver performs unitary recovery operations that reverse the
estimated errors.

\item He finally processes the encoded qubits with a decoding circuit to
recover the original stream of information qubits. The qubits decoded from
this convolutional procedure should be error free and ready for quantum
computation at the receiving end.
\end{enumerate}

In the above description, each step is online---the sender and receiver do not
have to employ the above procedure sequentially. Processing of one step can
begin once the previous step has finished some amount of processing. All steps
above are always active during processing once the initial rounds have been completed.

\section{The Pauli-to-Binary Isomorphism for Quantum Convolutional Codes}

\label{sec:P2B-conv}This section contains the most important part of this
quantum convolutional coding review---the isomorphism from the set of Pauli sequences to the module over
the ring of binary polynomials
\cite{arxiv2004olliv,isit2006grassl,ieee2007forney}. We also name it the
Pauli-to-binary (P2B)\ isomorphism (whether we are considering the
P2B\ isomorphism for block codes or convolutional codes should be clear from
the context). The P2B\ isomorphism is important because it is easier to
perform manipulations with vectors of binary polynomials than with Pauli sequences.

We first define the phase-free Pauli group $\left[  \Pi^{\mathbb{Z}}\right]
$\ on a sequence of qubits. Recall that the delay transform $D$ in
(\ref{eq:delay-transform}) shifts a Pauli sequence to the right by $n$. Let us
assume for now that $n=1$. Let $\Pi^{\mathbb{Z}}$ denote the set of all
countably infinite Pauli sequences. The set $\Pi^{\mathbb{Z}}$ is equivalent
to the set of all one-qubit shifts of arbitrary Pauli operators:%
\begin{equation}
\Pi^{\mathbb{Z}}=\left\{
{\textstyle\prod\limits_{i\in\mathbb{Z}}}
D^{i}\left(  A_{i}\right)  :A_{i}\in\Pi\right\}  .
\end{equation}
We remark that $D^{i}\left(  A_{i}\right)  =D^{i}\left(  A_{i}\otimes
I^{\otimes\infty}\right)  $. We make this same abuse of notation in what
follows. We can define the equivalence class $\left[  \Pi^{\mathbb{Z}}\right]
$\ of phase-free Pauli sequences:%
\begin{equation}
\left[  \Pi^{\mathbb{Z}}\right]  =\left\{  \beta\mathbf{A\ }|\ \mathbf{A}%
\in\Pi^{\mathbb{Z}},\beta\in\mathbb{C},\left\vert \beta\right\vert =1\right\}
.
\end{equation}

We develop the Pauli-to-binary (P2B) isomorphism between binary polynomials
and Pauli sequences. The P2B isomorphism is useful for representing the
shifting nature of quantum convolutional codes. Suppose $z\left(  D\right)  $
and $x\left(  D\right)  $ are arbitrary finite-degree and finite-delay
polynomials in $D$\ over $\mathbb{Z}_{2}$:%
\begin{equation}
z\left(  D\right)  =\sum_{i}z_{i}D^{i},\ \ \ \ \ \ \ x\left(  D\right)
=\sum_{i}x_{i}D^{i},\ \ \ \ \ \ \ z_{i},x_{i}\in\mathbb{Z}_{2}\ \ \forall
i\in\mathbb{Z},
\end{equation}
where del$\left(  z\left(  D\right)  \right)  $, del$\left(  x\left(
D\right)  \right)  $, deg$\left(  z\left(  D\right)  \right)  $, deg$\left(
x\left(  D\right)  \right)  <\infty$. Suppose%
\begin{equation}
u\left(  D\right)  =\left(  z\left(  D\right)  ,x\left(  D\right)  \right)
\in\left(  \mathbb{Z}_{2}\left(  D\right)  \right)  ^{2},
\end{equation}
where $\left(  \mathbb{Z}_{2}\left(  D\right)  \right)  ^{2}$ indicates the
Cartesian product $\mathbb{Z}_{2}\left(  D\right)  \times\mathbb{Z}_{2}\left(
D\right)  $. Let us employ the following shorthand:%
\begin{equation}
u\left(  D\right)  =\left[  z\left(  D\right)  |x\left(  D\right)  \right]  .
\end{equation}
Let $N$ be a map from the binary polynomials to the Pauli sequences,
$N:\left(  \mathbb{Z}_{2}\left(  D\right)  \right)  ^{2}\rightarrow
\Pi^{\mathbb{Z}}$, where%
\begin{equation}
N\left(  u\left(  D\right)  \right)  =%
{\textstyle\prod\limits_{i}}
D^{i}\left(  Z^{z_{i}}X^{x_{i}}\right)  .
\end{equation}
Let $v\left(  D\right)  =\left[  z^{\prime}\left(  D\right)  |x^{\prime
}\left(  D\right)  \right]  $ where $v\left(  D\right)  \in\left(
\mathbb{Z}_{2}\left(  D\right)  \right)  ^{2}$. The map $N$ induces an
isomorphism%
\begin{equation}
\left[  N\right]  :\left(  \mathbb{Z}_{2}\left(  D\right)  \right)
^{2}\rightarrow\left[  \Pi^{\mathbb{Z}}\right]  ,
\end{equation}
because addition of binary polynomials is equivalent to multiplication of
Pauli elements up to a global phase:%
\begin{equation}
\left[  N\left(  u\left(  D\right)  +v\left(  D\right)  \right)  \right]
=\left[  N\left(  u\left(  D\right)  \right)  \right]  \left[  N\left(
v\left(  D\right)  \right)  \right]  . \label{eq:isomorphism-poly}%
\end{equation}
The above isomorphism is a powerful way to capture the infiniteness and
shifting nature of convolutional codes with finite-degree and finite-delay
polynomials over the binary field $\mathbb{Z}_{2}$.

A quantum convolutional code in general consists of generators with $n$ qubits
per frame. Therefore, we consider the $n$-qubit extension of the definitions
and isomorphism given above. Let the delay transform $D$ now shift a Pauli
sequence to the right by an arbitrary positive integer $n$. Consider a
$2n$-dimensional vector $\mathbf{u}\left(  D\right)  $ of binary polynomials
where $\mathbf{u}\left(  D\right)  \in\left(  \mathbb{Z}_{2}\left(  D\right)
\right)  ^{2n}$. Let us write $\mathbf{u}\left(  D\right)  $\ as follows%
\[
\mathbf{u}\left(  D\right)  =\left[  \mathbf{z}\left(  D\right)
|\mathbf{x}\left(  D\right)  \right]  =\left[
\begin{array}
[c]{ccc}%
z_{1}\left(  D\right)   & \cdots & z_{n}\left(  D\right)
\end{array}
|%
\begin{array}
[c]{ccc}%
x_{1}\left(  D\right)   & \cdots & x_{n}\left(  D\right)
\end{array}
\right]  ,
\]
where $\mathbf{z}\left(  D\right)  ,\mathbf{x}\left(  D\right)  \in\left(
\mathbb{Z}_{2}\left(  D\right)  \right)  ^{n}$. Suppose%
\begin{equation}
z_{i}\left(  D\right)  =\sum_{j}z_{i,j}D^{j},\ \ \ \ \ \ x_{i}\left(
D\right)  =\sum_{j}x_{i,j}D^{j}.\nonumber
\end{equation}
Define a map $\mathbf{N}:\left(  \mathbb{Z}_{2}\left(  D\right)  \right)
^{2n}\rightarrow\Pi^{\mathbb{Z}}$:%
\begin{equation}
\mathbf{N}\left(  \mathbf{u}\left(  D\right)  \right)  =%
{\textstyle\prod\limits_{j}}
D^{j}\left(  Z^{z_{1,j}}X^{x_{1,j}}\right)  D^{j}\left(  I\otimes Z^{z_{2,j}%
}X^{x_{2,j}}\right)  \cdots D^{j}\left(  I^{\otimes n-1}\otimes Z^{z_{n,j}%
}X^{x_{n,j}}\right)  .\nonumber
\end{equation}
$\mathbf{N}$ is equivalent to the following map (up to a global phase)%
\begin{equation}
\mathbf{N}\left(  \mathbf{u}\left(  D\right)  \right)  =N\left(  u_{1}\left(
D\right)  \right)  \left(  I\otimes N\left(  u_{2}\left(  D\right)  \right)
\right)  \cdots\left(  I^{\otimes n-1}\otimes N\left(  u_{n}\left(  D\right)
\right)  \right)  ,\nonumber
\end{equation}
where%
\begin{equation}
u_{i}\left(  D\right)  =\left[  z_{i}\left(  D\right)  |x_{i}\left(  D\right)
\right]  .
\end{equation}
Suppose%
\begin{equation}
\mathbf{v}\left(  D\right)  =\left[  \mathbf{z}^{\prime}\left(  D\right)
|\mathbf{x}^{\prime}\left(  D\right)  \right]  ,
\end{equation}
where $\mathbf{v}\left(  D\right)  \in\left(  \mathbb{Z}_{2}\left(  D\right)
\right)  ^{2n}$. The map $\mathbf{N}$\ induces an isomorphism $\left[
\mathbf{N}\right]  :\left(  \mathbb{Z}_{2}\left(  D\right)  \right)
^{2n}\rightarrow\left[  \Pi^{\mathbb{Z}}\right]  $ for the same reasons given
in (\ref{eq:isomorphism-poly}):%
\begin{equation}
\left[  \mathbf{N}\left(  \mathbf{u}\left(  D\right)  +\mathbf{v}\left(
D\right)  \right)  \right]  =\left[  \mathbf{N}\left(  \mathbf{u}\left(
D\right)  \right)  \right]  \left[  \mathbf{N}\left(  \mathbf{v}\left(
D\right)  \right)  \right]  .
\end{equation}
The P2B isomorphism $\left[  \mathbf{N}\right]  $ is again useful because it
allows us to perform binary calculations instead of Pauli calculations.

\section{The Shifted Symplectic Product}

\label{sec:shifted-symp-prod}We consider the commutative properties of quantum
convolutional codes in this section. We develop some mathematics for the
important \textquotedblleft shifted symplectic product.\textquotedblright\ The
shifted symplectic product reveals the commutation relations of an arbitrary
number of shifts of a set of Pauli sequences.

Recall from Definition~\ref{def:conv-code} that a commuting set comprising a
basic set of Paulis and all their shifts specifies a quantum convolutional
code. How can we capture the commutation relations of a Pauli sequence and all
of its shifts?\ The \textit{shifted} symplectic product $\odot$, where%
\begin{equation}
\odot:\left(  \mathbb{Z}_{2}\left(  D\right)  \right)  ^{2}\times\left(
\mathbb{Z}_{2}\left(  D\right)  \right)  ^{2}\rightarrow\mathbb{Z}_{2}\left(
D\right)  ,
\end{equation}
is an elegant way to do so. The shifted symplectic product maps two vectors
$u\left(  D\right)  =\left[  z\left(  D\right)  |x\left(  D\right)  \right]  $
and $v\left(  D\right)  =\left[  z^{\prime}\left(  D\right)  |x^{\prime
}\left(  D\right)  \right]  $ to a binary polynomial with finite delay and
finite degree:%
\begin{equation}
\left(  u\odot v\right)  \left(  D\right)  =z\left(  D\right)  x^{\prime
}\left(  D^{-1}\right)  -x\left(  D\right)  z^{\prime}\left(  D^{-1}\right)  .
\end{equation}
The shifted symplectic product is not a proper symplectic product because it
fails to be alternating \cite{book2001symp}. The alternating property requires
that%
\begin{equation}
\left(  u\odot v\right)  \left(  D\right)  =-\left(  v\odot u\right)  \left(
D\right)  ,
\end{equation}
but we find instead that the following holds:%
\begin{equation}
\left(  u\odot v\right)  \left(  D\right)  =-\left(  v\odot u\right)  \left(
D^{-1}\right)  .
\end{equation}
Every vector $u\left(  D\right)  \in\mathbb{Z}_{2}\left(  D\right)  ^{2}$ is
self time-reversal antisymmetric with respect to $\odot$:%
\begin{equation}
\left(  u\odot u\right)  \left(  D\right)  =-\left(  u\odot u\right)  \left(
D^{-1}\right)  \ \ \ \ \forall u\left(  D\right)  \in\mathbb{Z}_{2}\left(
D\right)  ^{2}.
\end{equation}
Every binary vector is also self time-reversal symmetric with respect to
$\odot$\ because addition and subtraction are the same over $\mathbb{Z}_{2}$.
We employ the addition convention from now on and drop the minus signs. The
shifted symplectic product is a binary polynomial in $D$. We write its
coefficients as follows:%
\begin{equation}
\left(  u\odot v\right)  \left(  D\right)  =\sum_{i\in\mathbb{Z}}\left(
u\odot v\right)  _{i}\ D^{i}.
\end{equation}
The coefficient $\left(  u\odot v\right)  _{i}$ captures the commutation
relations of two Pauli sequences for $i$ one-qubit shifts of one of the
sequences:%
\begin{equation}
N\left(  u\left(  D\right)  \right)  D^{i}\left(  N\left(  v\left(  D\right)
\right)  \right)  =\left(  -1\right)  ^{\left(  u\odot v\right)  _{i}}%
D^{i}\left(  N\left(  v\left(  D\right)  \right)  \right)  N\left(  u\left(
D\right)  \right)  .\label{eq:shifted-symp-comm}%
\end{equation}
Thus two Pauli sequences $N\left(  u\left(  D\right)  \right)  $ and $N\left(
v\left(  D\right)  \right)  $\ commute for all shifts if and only if the
shifted symplectic product $\left(  u\odot v\right)  \left(  D\right)  $ vanishes.

The next example highlights the main features of the shifted symplectic
product and further emphasizes the relationship between Pauli commutation and
orthogonality of the shifted symplectic product.

\begin{example}
Consider two vectors of binary polynomials:%
\begin{equation}
u\left(  D\right)  =\left[
\begin{array}
[c]{c}%
D
\end{array}
|%
\begin{array}
[c]{c}%
1+D^{3}%
\end{array}
\right]  ,\ \ \ \ \ v\left(  D\right)  =\left[
\begin{array}
[c]{c}%
1+D
\end{array}
|%
\begin{array}
[c]{c}%
D^{3}%
\end{array}
\right]  .\nonumber
\end{equation}
The P2B isomorphism maps the above polynomials to the following Pauli
sequences:%
\begin{equation}
N\left(  u\left(  D\right)  \right)  =\left(  \cdots|I|X|Z|I|X|I|\cdots
\right)  ,\ \ \ \ \ \ N\left(  v\left(  D\right)  \right)  =\left(
\cdots|I|Z|Z|I|X|I|\cdots\right)  .
\end{equation}
The vertical bars between every Pauli in the sequence indicate that we are
considering one-qubit shifts. We determine the commutation relations of the
above sequences by inspection. $N\left(  u\left(  D\right)  \right)  $
anticommutes with a shift of itself by one or two to the left or right and
commutes with all other shifts of itself. $N\left(  v\left(  D\right)
\right)  $ anticommutes with a shift of itself by two or three to the left or
right and commutes with all other shifts of itself. $N\left(  u\left(
D\right)  \right)  $ anticommutes with $N\left(  v\left(  D\right)  \right)  $
shifted to the left by one or two, with the zero-shifted $N\left(  v\left(
D\right)  \right)  $, and with $N\left(  v\left(  D\right)  \right)  $ shifted
to the right by two or three. The following shifted symplectic products give
us the same information:%
\begin{align}
\left(  u\odot u\right)  \left(  D\right)   &  =D^{-2}+D^{-1}+D+D^{2}%
,\ \ \ \ \ \left(  v\odot v\right)  \left(  D\right)  =D^{-3}+D^{-2}%
+D^{2}+D^{3},\nonumber\\
\left(  v\odot u\right)  \left(  D\right)   &  =D^{-3}+D^{-2}+1+D+D^{2}.
\end{align}
The nonzero coefficients indicate the commutation relations just as
(\ref{eq:shifted-symp-comm}) claims.
\end{example}

We can again define a shifted symplectic product for the case of $n$-qubits
per frame. Let $\odot$ denote the shifted symplectic product between vectors
of binary polynomials:%
\begin{equation}
\odot:\left(  \mathbb{Z}_{2}\left(  D\right)  \right)  ^{2n}\times\left(
\mathbb{Z}_{2}\left(  D\right)  \right)  ^{2n}\rightarrow\mathbb{Z}_{2}\left(
D\right)  .
\end{equation}
It maps $2n$-dimensional vectors $\mathbf{u}\left(  D\right)  $ and
$\mathbf{v}\left(  D\right)  $ of binary polynomials to a finite-degree and
finite-delay binary polynomial%
\begin{equation}
\left(  \mathbf{u}\odot\mathbf{v}\right)  \left(  D\right)  =\sum_{i=1}%
^{n}\left(  u_{i}\odot v_{i}\right)  \left(  D\right)  ,
\end{equation}
where%
\[
u_{i}\left(  D\right)  =\left[  z_{i}\left(  D\right)  |x_{i}\left(  D\right)
\right]  ,\ \ \ \ \ \ v_{i}\left(  D\right)  =\left[  z_{i}^{\prime}\left(
D\right)  |x_{i}^{\prime}\left(  D\right)  \right]  .
\]
The standard inner product gives an alternative way to define the shifted
symplectic product:%
\begin{equation}
\left(  \mathbf{u}\odot\mathbf{v}\right)  \left(  D\right)  =\mathbf{z}\left(
D\right)  \cdot\mathbf{x}^{\prime}\left(  D^{-1}\right)  -\mathbf{x}\left(
D\right)  \cdot\mathbf{z}^{\prime}\left(  D^{-1}\right)  .
\end{equation}
Every vector $\mathbf{u}\left(  D\right)  \in\mathbb{Z}_{2}\left(  D\right)
^{2n}$ is self time-reversal symmetric with respect to $\odot$:%
\begin{equation}
\left(  \mathbf{u}\odot\mathbf{u}\right)  \left(  D\right)  =\left(
\mathbf{u}\odot\mathbf{u}\right)  \left(  D^{-1}\right)  \ \ \ \ \forall
\mathbf{u}\left(  D\right)  \in\mathbb{Z}_{2}\left(  D\right)  ^{2n}%
.\label{eq:self-time-reversal-sym}%
\end{equation}
The shifted symplectic product for vectors of binary polynomials is a binary
polynomial in $D$. We write its coefficients as follows:%
\begin{equation}
\left(  \mathbf{u}\odot\mathbf{v}\right)  \left(  D\right)  =\sum
_{i\in\mathbb{Z}}\left(  \mathbf{u}\odot\mathbf{v}\right)  _{i}\ D^{i}.
\end{equation}
The coefficient $\left(  \mathbf{u}\odot\mathbf{v}\right)  _{i}$\ captures the
commutation relations of two Pauli sequences for $i$ $n$-qubit shifts of one
of the sequences:%
\begin{equation}
\mathbf{N}\left(  \mathbf{u}\left(  D\right)  \right)  D^{i}\left(
\mathbf{N}\left(  \mathbf{v}\left(  D\right)  \right)  \right)  =\left(
-1\right)  ^{\left(  \mathbf{u}\odot\mathbf{v}\right)  _{i}}D^{i}\left(
\mathbf{N}\left(  \mathbf{v}\left(  D\right)  \right)  \right)  \mathbf{N}%
\left(  \mathbf{u}\left(  D\right)  \right)  .\nonumber
\end{equation}
The shifted symplectic product between two vectors of binary polynomials
vanishes if and only if their corresponding Pauli sequences commute.

\begin{example}
We consider the case where the frame size $n=4$. Consider the following
vectors of polynomials:%
\begin{equation}
\left[
\begin{array}
[c]{c}%
\mathbf{u}\left(  D\right)  \\
\mathbf{v}\left(  D\right)
\end{array}
\right]  =\left[  \left.
\begin{array}
[c]{cccc}%
1+D & D & 1 & D\\
0 & 1 & 0 & 0
\end{array}
\right\vert
\begin{array}
[c]{cccc}%
0 & 1 & 0 & 0\\
1+D & 1+D & 1 & D
\end{array}
\right]  .
\end{equation}
The P2B isomorphism maps $\mathbf{u}\left(  D\right)  $ and $\mathbf{v}\left(
D\right)  $ to the following Pauli sequences:%
\begin{align}
\mathbf{N}\left(  \mathbf{u}\left(  D\right)  \right)   &  =\left(
\cdots|IIII|ZXZI|ZZIZ|IIII|\cdots\right)  ,\nonumber\\
\mathbf{N}\left(  \mathbf{v}\left(  D\right)  \right)   &  =\left(
\cdots|IIII|XYXI|XXIX|IIII|\cdots\right)  .
\end{align}
We can determine the commutation relations by inspection of the above Pauli
sequences. $\mathbf{N}\left(  \mathbf{u}\left(  D\right)  \right)  $
anticommutes with itself shifted by one to the left or right, $\mathbf{N}%
\left(  \mathbf{v}\left(  D\right)  \right)  $ anticommutes with itself
shifted by one to the left or right, and $\mathbf{N}\left(  \mathbf{u}\left(
D\right)  \right)  $ anticommutes with $\mathbf{N}\left(  \mathbf{v}\left(
D\right)  \right)  $ shifted by one to the left. The following shifted
symplectic products confirm the above commutation relations:%
\begin{equation}
\left(  \mathbf{u}\odot\mathbf{u}\right)  \left(  D\right)  =D^{-1}%
+D,\ \ \ \ \ \ \left(  \mathbf{v}\odot\mathbf{v}\right)  \left(  D\right)
=D^{-1}+D,\ \ \ \ \ \ \left(  \mathbf{u}\odot\mathbf{v}\right)  \left(
D\right)  =D^{-1}.\nonumber
\end{equation}

\end{example}

We note two useful properties of the shifted symplectic product $\odot$.
Suppose $f\left(  D\right)  \in\mathbb{Z}_{2}\left(  D\right)  $ with
$\deg\left(  f\right)  \geq0$. Let us denote scalar polynomial multiplication
as follows:%
\begin{equation}
\left(  f\ \mathbf{u}\right)  \left(  D\right)  =%
\begin{bmatrix}
f\left(  D\right)  u_{1}\left(  D\right)  & \cdots & f\left(  D\right)
u_{n}\left(  D\right)
\end{bmatrix}
.
\end{equation}
The following identities hold.%
\begin{align}
\left(  \left(  f\ \mathbf{u}\right)  \odot\mathbf{v}\right)  \left(
D\right)   &  =f\left(  D\right)  \left(  \mathbf{u}\odot\mathbf{v}\right)
\left(  D\right)  ,\\
\left(  \mathbf{u}\odot\left(  f\ \mathbf{v}\right)  \right)  \left(
D\right)   &  =f\left(  D^{-1}\right)  \left(  \mathbf{u}\odot\mathbf{v}%
\right)  \left(  D\right)  .
\end{align}
We also remark that
\[
\left(  \mathbf{u}\odot\mathbf{v}\right)  \left(  D\right)  =\left(
\mathbf{v}\odot\mathbf{u}\right)  \left(  D\right)  ,
\]
iff%
\[
\left(  \mathbf{u}\odot\mathbf{v}\right)  \left(  D\right)  =\left(
\mathbf{u}\odot\mathbf{v}\right)  \left(  D^{-1}\right)  .
\]

The shifted symplectic product $\odot$ vanishes for all generators in a
quantum convolutional code because of the commutativity requirement in
Definition~\ref{def:conv-code}. The cases where the shifted symplectic product
does not vanish (where the two Pauli sequences anticommute for one or more
shifts) are important for constructing entanglement-assisted quantum
convolutional codes.

In general, we represent a rate-$k/n$ quantum convolutional code with an
$\left(  n-k\right)  \times2n$-dimensional quantum check matrix $H\left(
D\right)  $ according to the P2B isomorphism. The entries of $H\left(
D\right)  $ are binary polynomials where%
\[
H\left(  D\right)  =\left[  \left.
\begin{array}
[c]{c}%
Z\left(  D\right)
\end{array}
\right\vert
\begin{array}
[c]{c}%
X\left(  D\right)
\end{array}
\right]  ,
\]
and $Z\left(  D\right)  $ and $X\left(  D\right)  $ are both $\left(
n-k\right)  \times n$-dimensional binary polynomial matrices. The following
matrix $\Omega\left(  D\right)  $ captures the commutation relations of the
generators in $H\left(  D\right)  $:%
\[
\Omega\left(  D\right)  =Z\left(  D\right)  X^{T}\left(  D^{-1}\right)
+X\left(  D\right)  Z^{T}\left(  D^{-1}\right)  .
\]
The reader can verify that the matrix elements $\left[  \Omega\left(
D\right)  \right]  _{ij}$ of $\Omega\left(  D\right)  $ are the shifted
symplectic products between the $i^{\text{th}}$ and $j^{\text{th}}$ respective
rows $h_{i}\left(  D\right)  $ and $h_{j}\left(  D\right)  $ of $H\left(
D\right)  $:%
\[
\left[  \Omega\left(  D\right)  \right]  _{ij}=\left(  h_{i}\odot
h_{j}\right)  \left(  D\right)  .
\]
We call the matrix $\Omega\left(  D\right)  $ the \textit{shifted symplectic
product matrix} because it encodes all of the shifted symplectic products or
the commutation relations of the code. This matrix is equal to the null matrix
for a valid quantum convolutional code because all generators commute with
themselves and with all $n$-qubit shifts of themselves and each other. For a
general set of generators, $\Omega\left(  D\right)  $ is not equal to the null
matrix and obeys the symmetry:\ $\Omega\left(  D\right)  =\Omega^{T}\left(
D^{-1}\right)  $.

\section{Row and Column Operations}

\label{sec:row-col-ops}We can perform row operations on binary polynomial
matrices for quantum convolutional codes. A row operation is merely a
\textquotedblleft mental\textquotedblright\ operation that has no effect on
the states in the codespace or on the error-correcting properties of the code.
It just changes the rows of the check matrix for a code. We have three types
of row operations:

\begin{enumerate}
\item An elementary row operation multiplies a row times an arbitrary binary
polynomial and adds the result to another row. This additive invariance holds
for any code that admits a description within the stabilizer formalism.
Additive codes are invariant under multiplication of the stabilizer generators
in the \textquotedblleft Pauli\ picture\textquotedblright\ or under row
addition in the \textquotedblleft
binary-polynomial\ picture.\textquotedblright

\item Another type of row operation is to multiply a row by an arbitrary power
of $D$. Ollivier and Tillich discuss such row operations as \textquotedblleft
multiplication of a line by $D$\textquotedblright\ and use them to find
encoding operations for their quantum convolutional codes
\cite{arxiv2004olliv}. Grassl and R\"{o}tteler use this type of operation to
find a subcode of a given quantum convolutional code with an equivalent
asymptotic rate and equivalent error-correcting properties
\cite{isit2006grassl}.

\item We also employ row operations that multiply a row by an arbitrary
polynomial (not necessarily a power of $D$). This type of row operation occurs
when we have generators with infinite weight that we would like to reduce to
finite weight so that the receiver can perform measurements in an online
fashion as qubits arrive from the noisy channel.
\end{enumerate}

We can encode all of the above types of row operations in a full-rank matrix
$R\left(  D\right)  $ with rational polynomial entries. Let $H^{\prime}\left(
D\right)  $ denote the resulting check matrix after performing a set of row
operations in the matrix $R\left(  D\right)  $ where $H^{\prime}\left(
D\right)  =R\left(  D\right)  H\left(  D\right)  $. The resulting effect on
the shifted symplectic product matrix $\Omega\left(  D\right)  $ is to change
it to another shifted symplectic product matrix $\Omega^{\prime}\left(
D\right)  $ related to $\Omega\left(  D\right)  $ by%
\begin{equation}
\Omega^{\prime}\left(  D\right)  =R\left(  D\right)  \Omega\left(  D\right)
R^{T}\left(  D^{-1}\right)  .\label{eq:symp-row-op}%
\end{equation}
Row operations do not change the commutation relations of a valid quantum
convolutional code because its shifted symplectic product matrix is equal to
the null matrix. But row operations do change the commutation relations of a
set of generators whose corresponding shifted symplectic product matrix is not
equal to the null matrix. This ability to change the commutation relations
through row operations is crucial for constructing entanglement-assisted
quantum convolutional codes from an arbitrary set of generators. We use
entanglement to resolve any anticommutativity in the generators.

We can also perform column operations on binary polynomial matrices for
quantum convolutional codes. Column operations do change the error-correcting
properties of the code and are important for realizing a periodic encoding
circuit for the code. We have two types of column operations:

\begin{enumerate}
\item An elementary column operation multiplies one column by an arbitrary
binary polynomial and adds the result to another column. We implement
elementary column operations with gates from the shift-invariant Clifford
group \cite{ieee2007grassl,isit2006grassl}.

\item Another column operation is to multiply column $i$\ in both the
\textquotedblleft X\textquotedblright\ and \textquotedblleft
Z\textquotedblright\ matrix by $D^{l}$ where$\ l\in\mathbb{Z}$. We perform
this operation by delaying or advancing the processing of qubit $i$\ by $l$
frames relative to the original frame.

\item An infinite-depth column operation multiplies one column in the
\textquotedblleft X\textquotedblright\ matrix by a rational polynomial whose
numerator is one and multiplies the corresponding column in the
\textquotedblleft Z\textquotedblright\ matrix by a corresponding finite polynomial.
\end{enumerate}

A column operation implemented on the \textquotedblleft X\textquotedblright%
\ side of the binary polynomial matrix has a corresponding effect on the
\textquotedblleft Z\textquotedblright\ side of the binary polynomial matrix.
This corresponding effect is a manifestation of the Heisenberg uncertainty
principle because commutation relations remain invariant with respect to the
action of unitary quantum gates. The shifted symplectic product is therefore
invariant with respect to column operations from the shift-invariant Clifford
group. The next two sections describe possible column operations for
implementing encoding circuits.

\section{Finite-Depth Clifford Operations}

\label{sec:finite-depth-clifford}One of the main advantages of a quantum
convolutional code is that its encoding circuit has a periodic form. We can
encode a stream of quantum information with the same physical routines or
devices and therefore reduce encoding and decoding complexity.

Ollivier and Tillich were the first to discuss encoding circuits for quantum
convolutional codes \cite{PhysRevLett.91.177902,arxiv2004olliv}. They provided
a set of necessary and sufficient conditions to determine when an encoding
circuit is noncatastrophic. A \textit{noncatastrophic encoding circuit }does
not propagate uncorrected errors infinitely through the decoded information
qubit stream. Classical convolutional coding theory has a well developed
theory of noncatastrophic encoding circuits \cite{book1999conv}.

Grassl and R\"{o}tteler later showed that Ollivier and Tillich's conditions
for a circuit to be noncatastrophic are too restrictive
\cite{isit2006grassl,ieee2006grassl,ieee2007grassl}. They found subcodes of
quantum convolutional codes that admit noncatastrophic encoders. The
noncatastrophic encoders are a sequence of Clifford circuits with finite
depth. They developed a formalism for encapsulating the periodic structure of
an encoding circuit by decomposing the encoding circuit as a set of elementary
column operations. Their decoding circuits are exact inverses of their
encoding circuits because their decoding circuits perform the encoding
operations in reverse order.

\begin{definition}
A \textit{finite-depth operation} transforms every finite-weight\ stabilizer
generator to one with finite weight.
\end{definition}

This property is important because we do not want the decoding circuit to
propagate uncorrected errors into the information qubit stream
\cite{book1999conv}. A finite-depth decoding circuit corresponding to any
stabilizer for a quantum convolutional code exists by the algorithm given in
Ref.~\cite{isit2006grassl}.

\label{sec:finite-depth-ops}We review the finite-depth operations in the
shift-invariant Clifford group
\cite{isit2006grassl,ieee2006grassl,ieee2007grassl}. The shift-invariant
Clifford group is an extension of the Clifford group operations mentioned in
Section~\ref{sec:cliff-encoding}. We describe how finite-depth operations in
the shift-invariant Clifford group affect the binary polynomial matrix for a
code. Each of the following operations acts on every frame of a quantum
convolutional code.

\begin{enumerate}
\item The sender performs a CNOT\ from qubit $i$ to qubit $j$ in every frame
where qubit $j$ is in a frame delayed by $k$. The effect on the binary
polynomial matrix is to multiply column $i$ by $D^{k}$ and add the result to
column $j$ in the \textquotedblleft X\textquotedblright\ matrix and to
multiply column$\ j$ by $D^{-k}$ and add the result to column $i$ in the
\textquotedblleft Z\textquotedblright\ matrix.

\item A Hadamard on qubit $i$ swaps column $i$ in the \textquotedblleft
X\textquotedblright\ matrix with column $i$ in the \textquotedblleft
Z\textquotedblright\ matrix.

\item A phase gate on qubit $i$ adds column $i$ from the \textquotedblleft
X\textquotedblright\ matrix to column $i$ in the \textquotedblleft
Z\textquotedblright\ matrix.

\item A controlled-phase gate from qubit $i$ to qubit $j$ in a frame delayed
by $k$ multiplies column $i$ in the \textquotedblleft X\textquotedblright%
\ matrix by $D^{k}$ and adds the result to column $j$ in the \textquotedblleft
Z\textquotedblright\ matrix. It also multiplies column $j$ in the
\textquotedblleft X\textquotedblright\ matrix by $D^{-k}$ and adds the result
to column $i$ in the \textquotedblleft Z\textquotedblright\ matrix.

\item A controlled-phase gate from qubit $i$ to qubit $i$ in a frame delayed
by $k$ multiplies column $i$ in the \textquotedblleft X\textquotedblright%
\ matrix by $D^{k}+D^{-k}$ and adds the result to column $i$ in the
\textquotedblleft Z\textquotedblright\ matrix.
\end{enumerate}

We use finite-depth operations extensively in the next few chapters.
Figure~\ref{fig:example-eaqcc-fefd}\ gives an example of an
entanglement-assisted quantum convolutional code that employs several
finite-depth operations. The circuit encodes a stream of information qubits
with the help of ebits shared between sender and receiver.

Multiple CNOT\ gates can realize an elementary column operation as described
at the end of the previous section. Suppose the elementary column operation
multiplies column $i$ in the \textquotedblleft X\textquotedblright\ matrix by
$f\left(  D\right)  $ and adds the result to column $j$. Polynomial $f\left(
D\right)  $ is a summation of some finite set $\left\{  l_{1},\ldots
,l_{n}\right\}  $ of powers of $D$:%
\[
f\left(  D\right)  =D^{l_{1}}+\cdots+D^{l_{n}}.
\]
We can realize $f\left(  D\right)  $ by performing a CNOT\ gate from qubit $i$
to qubit $j$ in a frame delayed by $l_{i}$ for each $i\in\left\{
1,\ldots,n\right\}  $.

\section{Infinite-Depth Clifford Operations}

\label{sec:infinite-depth-ops}We now introduce an
infinite-depth operation to the set of operations in the shift-invariant
Clifford group available for encoding and decoding quantum convolutional
codes.

\begin{definition}
An \textit{infinite-depth operation} can transform a finite-weight stabilizer
generator to one with infinite weight (but does not necessarily do so to every
finite-weight generator).
\end{definition}

A decoding circuit with infinite-depth operations on qubits sent over the
noisy channel is undesirable because it spreads uncorrected errors infinitely
into the decoded information qubit stream. But an encoding circuit with
infinite-depth operations is acceptable if we assume a communication paradigm
in which the only noisy process is the noisy quantum channel.

The next chapter shows several examples of circuits that include
infinite-depth operations. Infinite-depth operations expand the possibilities
for quantum convolutional circuits in much the same way that incorporating
feedback expands the possibilities for classical convolutional circuits.

We illustrate the details of several infinite-depth operations by first
providing some specific examples of infinite-depth operations and then show
how to realize an arbitrary infinite-depth operation.

We consider both the stabilizer and the logical operators for the information
qubits in our analysis. Tracking both of these sets of generators is necessary
for determining the proper decoding circuit when including infinite-depth operations.

\subsection{Examples of Infinite-Depth Operations}

Our first example of an infinite-depth operation involves a stream of
information qubits and ancilla qubits. We divide the stream into frames of
three qubits where each frame has two ancilla qubits and one information
qubit. The following two generators and each of their three-qubit shifts
stabilize the initial qubit stream:%
\begin{equation}
\cdots\left\vert
\begin{array}
[c]{ccc}%
I & I & I\\
I & I & I
\end{array}
\right\vert
\begin{array}
[c]{ccc}%
Z & I & I\\
I & Z & I
\end{array}
\left\vert
\begin{array}
[c]{ccc}%
I & I & I\\
I & I & I
\end{array}
\right\vert \cdots\label{eq:id-unencoded-Paulis}%
\end{equation}
The binary polynomial matrix corresponding to this stabilizer is as follows:%
\begin{equation}
\left[  \left.
\begin{array}
[c]{ccc}%
1 & 0 & 0\\
0 & 1 & 0
\end{array}
\right\vert
\begin{array}
[c]{ccc}%
0 & 0 & 0\\
0 & 0 & 0
\end{array}
\right]  . \label{eq:id-unencoded-qubits}%
\end{equation}
We obtain any Pauli sequence in the stabilizer by multiplying the above rows
by a power of $D$ and applying the inverse of the P2B isomorphism. The logical
operators for the information qubits are as follows:%
\[
\cdots\left\vert
\begin{array}
[c]{ccc}%
I & I & I\\
I & I & I
\end{array}
\right\vert
\begin{array}
[c]{ccc}%
I & I & X\\
I & I & Z
\end{array}
\left\vert
\begin{array}
[c]{ccc}%
I & I & I\\
I & I & I
\end{array}
\right\vert \cdots
\]
They also admit a description with a binary polynomial matrix:%
\begin{equation}
\left[  \left.
\begin{array}
[c]{ccc}%
0 & 0 & 0\\
0 & 0 & 1
\end{array}
\right\vert
\begin{array}
[c]{ccc}%
0 & 0 & 1\\
0 & 0 & 0
\end{array}
\right]  . \label{eq-unencoded-info-qubits}%
\end{equation}
We refer to the above matrix as the \textquotedblleft information-qubit
matrix.\textquotedblright

\subsubsection{Encoding}

Suppose we would like to encode the above stream so that the following
generators stabilize it:%
\[
\cdots\left\vert
\begin{array}
[c]{ccc}%
I & I & I\\
I & I & I
\end{array}
\right\vert
\begin{array}
[c]{ccc}%
X & X & X\\
Z & Z & I
\end{array}
\left\vert
\begin{array}
[c]{ccc}%
X & X & I\\
I & I & I
\end{array}
\right\vert \cdots,
\]
or equivalently, the following binary polynomial matrix stabilizes it:%
\begin{equation}
\left[  \left.
\begin{array}
[c]{ccc}%
0 & 0 & 0\\
1 & 1 & 0
\end{array}
\right\vert
\begin{array}
[c]{ccc}%
D+1 & D+1 & 1\\
0 & 0 & 0
\end{array}
\right]  . \label{eq:desired-stabilizer}%
\end{equation}

We encode the above stabilizer using a combination of finite-depth operations
and an infinite-depth operation. We perform a Hadamard on the first qubit in
each frame and follow with a CNOT\ from the first qubit to the second and
third qubits in each frame. These operations transform the matrix in
(\ref{eq:id-unencoded-qubits}) to the following matrix%
\[
\left[  \left.
\begin{array}
[c]{ccc}%
0 & 0 & 0\\
1 & 1 & 0
\end{array}
\right\vert
\begin{array}
[c]{ccc}%
1 & 1 & 1\\
0 & 0 & 0
\end{array}
\right]  ,
\]
or equivalently transform the generators in (\ref{eq:id-unencoded-Paulis})\ to
the following generators:%
\[
\cdots\left\vert
\begin{array}
[c]{ccc}%
I & I & I\\
I & I & I
\end{array}
\right\vert
\begin{array}
[c]{ccc}%
X & X & X\\
Z & Z & I
\end{array}
\left\vert
\begin{array}
[c]{ccc}%
I & I & I\\
I & I & I
\end{array}
\right\vert \cdots
\]
The information-qubit matrix becomes%
\[
\left[  \left.
\begin{array}
[c]{ccc}%
0 & 0 & 0\\
1 & 0 & 1
\end{array}
\right\vert
\begin{array}
[c]{ccc}%
0 & 0 & 1\\
0 & 0 & 0
\end{array}
\right]  .
\]
We now perform an infinite-depth operation: a CNOT\ from the third qubit in
one frame to the third qubit in a delayed frame and repeat this operation for
all following frames. Figure~\ref{fig:inf-depth-simple}\ shows this operation
acting on our stream of qubits with three qubits per frame.%
\begin{figure}
[ptb]
\begin{center}
\includegraphics[
natheight=9.786200in,
natwidth=3.673700in,
height=4.7755in,
width=1.8092in
]%
{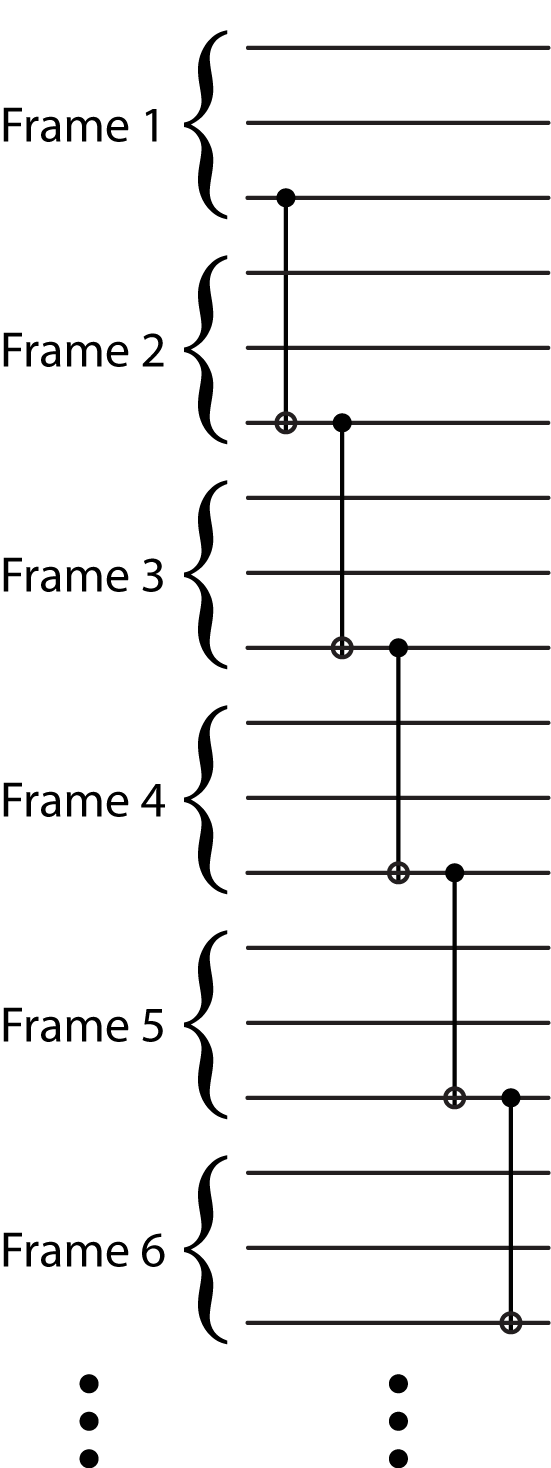}%
\caption{An example of an infinite-depth operation. A sequence of CNOT\ gates
acts on the third qubit of every frame. This infinite-depth operation
effectively multiplies the third column of the \textquotedblleft
X\textquotedblright\ side of the binary polynomial matrix by the rational
polynomial $1/\left(  1+D\right)  $ and multiplies the third column of the
\textquotedblleft Z\textquotedblright\ side of the binary polynomial matrix by
$1+D^{-1}$.}%
\label{fig:inf-depth-simple}%
\end{center}
\end{figure}
The effect of this operation is to translate the above stabilizer generators
as follows:%
\[
\cdots\left\vert
\begin{array}
[c]{ccc}%
I & I & I\\
I & I & I
\end{array}
\right\vert
\begin{array}
[c]{ccc}%
X & X & X\\
Z & Z & I
\end{array}
\left\vert
\begin{array}
[c]{ccc}%
I & I & X\\
I & I & I
\end{array}
\right\vert \left.
\begin{array}
[c]{ccc}%
I & I & X\\
I & I & I
\end{array}
\right\vert \cdots
\]
The first generator above and each of its three-qubits shifts is an
infinite-weight generator if the above sequence of CNOT gates acts on the
entire countably-infinite qubit stream. We represent the above stabilizer with
the binary \textit{rational} polynomial matrix%
\begin{equation}
\left[  \left.
\begin{array}
[c]{ccc}%
0 & 0 & 0\\
1 & 1 & 0
\end{array}
\right\vert
\begin{array}
[c]{ccc}%
1 & 1 & 1/\left(  1+D\right)  \\
0 & 0 & 0
\end{array}
\right]  ,\label{eq:id-encoded-stabilizer}%
\end{equation}
where $1/\left(  1+D\right)  =1+D+D^{2}+\ldots$ is a repeating fraction. The
operation is infinite-depth because it translates the original finite-weight
stabilizer generator to one with infinite weight.

It is possible to perform a row operation that multiplies the first row by
$D+1$. This operation gives a stabilizer matrix that is equivalent to the
desired stabilizer in (\ref{eq:desired-stabilizer}). The receiver of the
encoded qubits measures the finite-weight stabilizer generators in
(\ref{eq:desired-stabilizer}) to diagnose errors. These measurements do not
disturb the information qubits because they also stabilize the encoded stream.

The above encoding operations transform the information-qubit matrix as
follows:%
\begin{equation}
\left[  \left.
\begin{array}
[c]{ccc}%
0 & 0 & 0\\
1 & 0 & 1+D^{-1}%
\end{array}
\right\vert
\begin{array}
[c]{ccc}%
0 & 0 & 1/\left(  1+D\right)  \\
0 & 0 & 0
\end{array}
\right]  .\label{eq:id-encoded-info-qubits}%
\end{equation}
The infinite-depth operation on the third qubit has an effect on the
\textquotedblleft Z\textquotedblright\ or left side of the information-qubit
matrix as illustrated in the second row of the above matrix. The effect is to
multiply the third column of the \textquotedblleft Z\textquotedblright\ matrix
by $f\left(  D^{-1}\right)  $ if the operation multiplies the third column of
the \textquotedblleft X\textquotedblright\ matrix by $1/f\left(  D\right)  $.
This corresponding action on the \textquotedblleft Z\textquotedblright\ side
occurs because the commutation relations of the Pauli operators remain
invariant under quantum gates (due to the Heisenberg uncertainty principle),
or equivalently, the shifted symplectic product remains invariant under column
operations. The original shifted symplectic product for the logical operators
is one, and it remains as one because $f((D^{-1})^{-1})/f\left(  D\right)  =1$.

\subsubsection{Decoding}

We perform finite-depth operations to decode the stream of information qubits.
Begin with the stabilizer and information-qubit matrix in
(\ref{eq:id-encoded-stabilizer})\ and (\ref{eq:id-encoded-info-qubits})
respectively. Perform a CNOT\ from the first qubit to the second qubit. The
stabilizer becomes%
\[
\left[  \left.
\begin{array}
[c]{ccc}%
0 & 0 & 0\\
0 & 1 & 0
\end{array}
\right\vert
\begin{array}
[c]{ccc}%
1 & 0 & 1/\left(  1+D\right)  \\
0 & 0 & 0
\end{array}
\right]  ,
\]
and the information-qubit matrix does not change. Perform a CNOT\ from the
third qubit to the first qubit in the same frame and in a delayed frame. These
gates multiply column three in the \textquotedblleft X\textquotedblright%
\ matrix by $1+D$ and add the result to column one. The gates also multiply
column one in the \textquotedblleft Z\textquotedblright\ matrix by $1+D^{-1}$
and add the result to column three. The effect is as follows on both the
stabilizer%
\begin{equation}
\left[  \left.
\begin{array}
[c]{ccc}%
0 & 0 & 0\\
0 & 1 & 0
\end{array}
\right\vert
\begin{array}
[c]{ccc}%
0 & 0 & 1/\left(  1+D\right)  \\
0 & 0 & 0
\end{array}
\right]  ,\label{eq:id-decoded-stabilizer}%
\end{equation}
and the information-qubit matrix%
\begin{equation}
\left[  \left.
\begin{array}
[c]{ccc}%
0 & 0 & 0\\
1 & 0 & 0
\end{array}
\right\vert
\begin{array}
[c]{ccc}%
1 & 0 & 1/\left(  1+D\right)  \\
0 & 0 & 0
\end{array}
\right]  .\label{eq:id-decoded-info}%
\end{equation}
We can multiply the logical operators by any element of the stabilizer and
obtain an equivalent logical operator \cite{thesis97gottesman}. We perform
this multiplication in the \textquotedblleft binary-polynomial
picture\textquotedblright\ by adding the first row of the stabilizer in
(\ref{eq:id-decoded-stabilizer}) to the first row of (\ref{eq:id-decoded-info}%
). The information-qubit matrix becomes%
\begin{equation}
\left[  \left.
\begin{array}
[c]{ccc}%
0 & 0 & 0\\
1 & 0 & 0
\end{array}
\right\vert
\begin{array}
[c]{ccc}%
1 & 0 & 0\\
0 & 0 & 0
\end{array}
\right]  ,\label{eq:encoded-info-qubit-1st-ex}%
\end{equation}
so that the resulting logical operators act only on the first qubit of every
frame. We have successfully decoded the information qubits with finite-depth
operations. The information qubits teleport coherently
\cite{prl2004harrow,wilde:060303}\ from being the third qubit of each frame as
in (\ref{eq-unencoded-info-qubits}) to being the first qubit of each frame as
in (\ref{eq:encoded-info-qubit-1st-ex}). We exploit the above method of
encoding with infinite-depth operations and decoding with finite-depth
operations for the class of entanglement-assisted quantum convolutional codes
in Section~\ref{sec:coh-tele-EAQCCs}.

\subsection{General Infinite-Depth Operations}

We discuss the action of a general infinite-depth operation on two weight-one
\textquotedblleft X\textquotedblright\ and \textquotedblleft
Z\textquotedblright\ Pauli sequences where each frame has one Pauli matrix.
Our analysis then determines the effect of an infinite-depth operation on an
arbitrary stabilizer or information-qubit matrix. The generators in the
\textquotedblleft Pauli picture\textquotedblright\ are as follows:%
\begin{equation}
\cdots\left\vert
\begin{array}
[c]{c}%
I\\
I
\end{array}
\right.  \left\vert
\begin{array}
[c]{c}%
X\\
Z
\end{array}
\right\vert \left.
\begin{array}
[c]{c}%
I\\
I
\end{array}
\right\vert \cdots,\label{eq:ie-unencoded-Paulis}%
\end{equation}
or as follows in the \textquotedblleft binary-polynomial
picture\textquotedblright:%
\[
\left[  \left.
\begin{array}
[c]{c}%
0\\
1
\end{array}
\right\vert
\begin{array}
[c]{c}%
1\\
0
\end{array}
\right]  .
\]
An infinite-depth $1/f\left(  D\right)  $ operation, where $f\left(  D\right)
$ is an arbitrary polynomial, transforms the above matrix to the following
one:%
\[
\left[  \left.
\begin{array}
[c]{c}%
0\\
f\left(  D^{-1}\right)
\end{array}
\right\vert
\begin{array}
[c]{c}%
1/f\left(  D\right)  \\
0
\end{array}
\right]  .
\]
A circuit that performs this transformation preserves the shifted symplectic
product because $f\left(  D^{-1}\right)  \cdot1/f\left(  D^{-1}\right)  =1$.
The circuit operates on a few qubits at a time and is shift-invariant---the
same device or physical routines implement it.

First perform the long division expansion of binary rational polynomial
$1/f\left(  D\right)  $. This expansion has a particular repeating pattern
with period $l$. For example, suppose that $f\left(  D\right)  =1+D+D^{3}$.
Its long-division expansion is $1+D+D^{2}+D^{4}+D^{7}+D^{8}+D^{9}%
+D^{11}+\ldots$ and exhibits a repeating pattern with period seven. We want a
circuit that realizes the following Pauli generators%
\begin{equation}
\cdots\left\vert
\begin{array}
[c]{c}%
I\\
Z
\end{array}
\right.  \left\vert
\begin{array}
[c]{c}%
I\\
I
\end{array}
\right\vert
\begin{array}
[c]{c}%
I\\
Z
\end{array}
\left\vert
\begin{array}
[c]{c}%
X\\
Z
\end{array}
\right\vert
\begin{array}
[c]{c}%
X\\
I
\end{array}
\left\vert
\begin{array}
[c]{c}%
X\\
I
\end{array}
\right\vert
\begin{array}
[c]{c}%
I\\
I
\end{array}
\left\vert
\begin{array}
[c]{c}%
X\\
I
\end{array}
\right\vert \left.
\begin{array}
[c]{c}%
I\\
I
\end{array}
\right\vert \cdots, \label{eq:ie-desired-paulis}%
\end{equation}
where the pattern in the $X$ matrices is the same as the repeating polynomial
$1/f\left(  D\right)  $ and continues infinitely to the right, and the pattern
on the $Z$ matrices is the same as that in $f\left(  D^{-1}\right)  $\ and
terminates at the left. The above Pauli sequence is equivalent to the
following binary rational polynomial matrix:%
\[
\left[  \left.
\begin{array}
[c]{c}%
0\\
1+D^{-1}+D^{-3}%
\end{array}
\right\vert
\begin{array}
[c]{c}%
1/\left(  1+D+D^{3}\right) \\
0
\end{array}
\right]  .
\]

We now discuss a method that realizes an arbitrary rational polynomial
$1/f\left(  D\right)  $ as an infinite-depth operation. Our method for
encoding the generators in (\ref{eq:ie-desired-paulis}) from those in
(\ref{eq:ie-unencoded-Paulis}) consists of a \textquotedblleft
sliding-window\textquotedblright\ technique that determines transformation
rules for the circuit. The circuit is an additive, shift-invariant filtering
operation. It resembles an infinite-impulse response filter because the
resulting sequence extends infinitely. In general, the number $N$\ of qubits
that the encoding unitary operates on is as follows%
\[
N=\deg\left(  f\left(  D\right)  \right)  -\mathrm{del}\left(  f\left(
D\right)  \right)  +1,
\]
where $\deg\left(  f\left(  D\right)  \right)  $ and $\mathrm{del}\left(
f\left(  D\right)  \right)  $ are the respective highest and lowest powers of
polynomial $f\left(  D\right)  $. Therefore, our exemplary encoding unitary
operates on four qubits at a time. We delay the original sequence
in\ (\ref{eq:ie-unencoded-Paulis}) by three frames. These initial frames are
\textquotedblleft scratch\textquotedblright\ frames that give the encoding
unitary enough \textquotedblleft room\textquotedblright\ to generate the
desired Paulis in (\ref{eq:ie-desired-paulis}). The first set of
transformation rules is as follows%
\begin{equation}
\left.
\begin{array}
[c]{c}%
I\\
I
\end{array}
\right\vert
\begin{array}
[c]{c}%
I\\
I
\end{array}
\left\vert
\begin{array}
[c]{c}%
I\\
I
\end{array}
\right\vert
\begin{array}
[c]{c}%
X\\
Z
\end{array}
\rightarrow\left.
\begin{array}
[c]{c}%
I\\
Z
\end{array}
\right\vert
\begin{array}
[c]{c}%
I\\
I
\end{array}
\left\vert
\begin{array}
[c]{c}%
I\\
Z
\end{array}
\right\vert
\begin{array}
[c]{c}%
X\\
Z
\end{array}
,\label{eq:first-rule}%
\end{equation}
and generates the first four elements of the pattern in
(\ref{eq:ie-desired-paulis}). Now that the encoding unitary has acted on the
first four frames, we need to shift our eyes to the right by one frame in the
sequence in (\ref{eq:ie-desired-paulis})\ to determine the next set of rules.
So we shift the above outputs by one frame to the \textit{left} (assuming that
only identity matrices lie to the right) and determine the next set of
transformation rules that generate the next elements of the sequence in
(\ref{eq:ie-desired-paulis}):%
\[
\left.
\begin{array}
[c]{c}%
I\\
I
\end{array}
\right\vert
\begin{array}
[c]{c}%
I\\
Z
\end{array}
\left\vert
\begin{array}
[c]{c}%
X\\
Z
\end{array}
\right\vert
\begin{array}
[c]{c}%
I\\
I
\end{array}
\rightarrow\left.
\begin{array}
[c]{c}%
I\\
I
\end{array}
\right\vert
\begin{array}
[c]{c}%
I\\
Z
\end{array}
\left\vert
\begin{array}
[c]{c}%
X\\
Z
\end{array}
\right\vert
\begin{array}
[c]{c}%
X\\
I
\end{array}
.
\]
Shift the above outputs to the left by one frame to determine the next set of
transformation rules:%
\[
\left.
\begin{array}
[c]{c}%
I\\
Z
\end{array}
\right\vert
\begin{array}
[c]{c}%
X\\
Z
\end{array}
\left\vert
\begin{array}
[c]{c}%
X\\
I
\end{array}
\right\vert
\begin{array}
[c]{c}%
I\\
I
\end{array}
\rightarrow\left.
\begin{array}
[c]{c}%
I\\
Z
\end{array}
\right\vert
\begin{array}
[c]{c}%
X\\
Z
\end{array}
\left\vert
\begin{array}
[c]{c}%
X\\
I
\end{array}
\right\vert
\begin{array}
[c]{c}%
X\\
I
\end{array}
.
\]
We obtain the rest of the transformation rules by continuing this sliding
process, and we stop when the pattern in the sequence in
(\ref{eq:ie-desired-paulis}) begins to repeat:%
\[
\left.
\begin{array}
[c]{c}%
X\\
Z\\
X\\
X\\
I\\
X
\end{array}
\right\vert
\begin{array}
[c]{c}%
X\\
I\\
X\\
I\\
X\\
I
\end{array}
\left\vert
\begin{array}
[c]{c}%
X\\
I\\
I\\
X\\
I\\
I
\end{array}
\right\vert
\begin{array}
[c]{c}%
I\\
I\\
I\\
I\\
I\\
I
\end{array}
\rightarrow\left.
\begin{array}
[c]{c}%
X\\
Z\\
X\\
X\\
I\\
X
\end{array}
\right\vert
\begin{array}
[c]{c}%
X\\
I\\
X\\
I\\
X\\
I
\end{array}
\left\vert
\begin{array}
[c]{c}%
X\\
I\\
I\\
X\\
I\\
I
\end{array}
\right\vert
\begin{array}
[c]{c}%
I\\
I\\
X\\
I\\
I\\
X
\end{array}
.
\]
The above set of rules determines the encoding unitary and only a few of them
are actually necessary. We can multiply the rules together to form equivalent
rules because the circuit obeys additivity (in the \textquotedblleft
binary-polynomial picture\textquotedblright). The rules become as follows
after rearranging into a standard form:%
\[
\left.
\begin{array}
[c]{c}%
Z\\
I\\
I\\
I\\
X\\
I\\
I\\
I
\end{array}
\right\vert
\begin{array}
[c]{c}%
I\\
Z\\
I\\
I\\
I\\
X\\
I\\
I
\end{array}
\left\vert
\begin{array}
[c]{c}%
I\\
I\\
Z\\
I\\
I\\
I\\
X\\
I
\end{array}
\right\vert
\begin{array}
[c]{c}%
I\\
I\\
I\\
Z\\
I\\
I\\
I\\
X
\end{array}
\rightarrow\left.
\begin{array}
[c]{c}%
Z\\
I\\
I\\
Z\\
X\\
I\\
I\\
I
\end{array}
\right\vert
\begin{array}
[c]{c}%
I\\
Z\\
I\\
I\\
I\\
X\\
I\\
I
\end{array}
\left\vert
\begin{array}
[c]{c}%
I\\
I\\
Z\\
Z\\
I\\
I\\
X\\
I
\end{array}
\right\vert
\begin{array}
[c]{c}%
I\\
I\\
I\\
Z\\
X\\
I\\
X\\
X
\end{array}
.
\]
A CNOT\ from qubit one to qubit four and a CNOT\ from qubit three to qubit
four suffice to implement this circuit. We repeatedly apply these operations
shifting by one frame at a time to implement the infinite-depth operation. We
could have observed that these gates suffice to implement the
\textquotedblleft Z\textquotedblright\ transformation in the first set of
transformation rules in (\ref{eq:first-rule}), but we wanted to show how this
method generates the full periodic \textquotedblleft X\textquotedblright%
\ sequence in (\ref{eq:ie-desired-paulis}). Figure~\ref{fig:inf-depth-example}%
\ shows how the above encoding unitary acts on a stream of quantum
information.%
\begin{figure}
[ptb]
\begin{center}
\includegraphics[
natheight=4.633700in,
natwidth=5.233800in,
height=2.7484in,
width=3.1012in
]%
{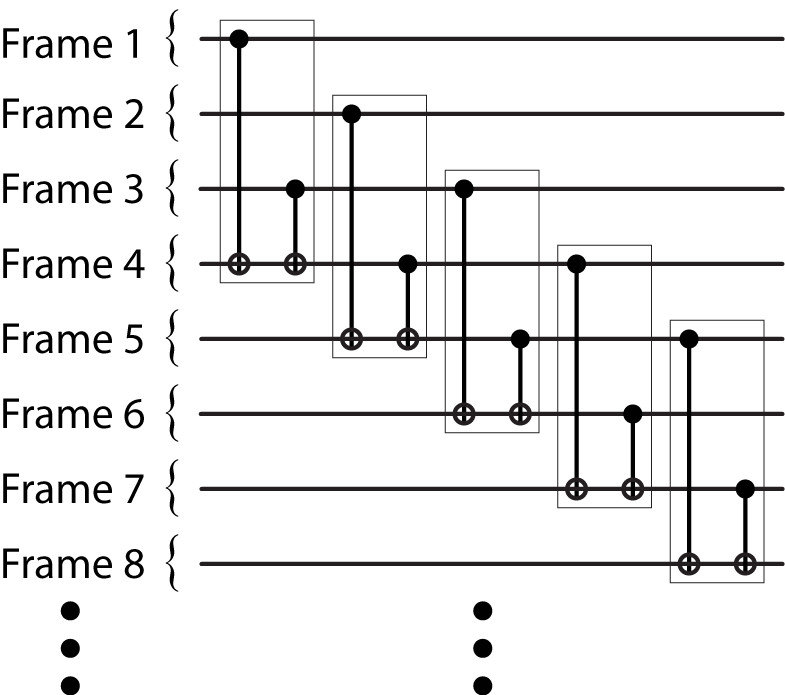}%
\caption{Another example of an infinite-depth operation. An infinite-depth
operation acts on qubit $i$ in every frame. This particular infinite-depth
operation multiplies column $i$ on the \textquotedblleft X\textquotedblright%
\ side of the binary polynomial matrix by $1/\left(  1+D+D^{3}\right)  $ and
multiplies column $i$ on the \textquotedblleft Z\textquotedblright\ side of
the binary polynomial matrix by $1+D^{-1}+D^{-3}$.}%
\label{fig:inf-depth-example}%
\end{center}
\end{figure}

We can determine the encoding unitary for an arbitrary rational polynomial
$1/f\left(  D\right)  $ using a similar method. Suppose that $\mathrm{del}%
\left(  f\left(  D\right)  \right)  =n$ and suppose $n\neq0$ as in the above
case. First delay or advance the frames if $n>0$ or if $n<0$ respectively.
Determine the CNOT\ gates that transform the \textquotedblleft
Z\textquotedblright\ Pauli sequence%
\[
\left[  \left.
\begin{array}
[c]{c}%
1
\end{array}
\right\vert
\begin{array}
[c]{c}%
0
\end{array}
\right]
\]
to%
\[
\left[  \left.
\begin{array}
[c]{c}%
D^{n}f\left(  D^{-1}\right)
\end{array}
\right\vert
\begin{array}
[c]{c}%
0
\end{array}
\right]  .
\]
These CNOT\ gates form the encoding circuit that transform both the
\textquotedblleft X\textquotedblright\ and \textquotedblleft
Z\textquotedblright\ Pauli sequences. We perform the encoding unitary, shift
by one frame, perform it again, and keep repeating. Our method encodes any
arbitrary polynomial $1/f\left(  D\right)  $ on the \textquotedblleft
X\textquotedblright\ side and $f\left(  D^{-1}\right)  $ on the
\textquotedblleft Z\textquotedblright\ side.

We can implement the \textquotedblleft time-reversed\textquotedblright%
\ polynomial $1/f\left(  D^{-1}\right)  $ on the \textquotedblleft
X\textquotedblright\ side by first delaying the frames by $m=\deg\left(
f\left(  D\right)  \right)  -\mathrm{del}\left(  f\left(  D\right)  \right)  $
frames and performing the circuit corresponding to $1/D^{m}\left(  f\left(
D^{-1}\right)  \right)  $. These operations implement the circuit $D^{m}%
/D^{m}\left(  f\left(  D^{-1}\right)  \right)  =1/f\left(  D^{-1}\right)  $.

\subsection{Infinite-Depth Operations in Practice}

We assume above that each of the infinite-depth operations acts on the entire
countably-infinite stream of qubits. In practice, each infinite-depth
operation acts on a finite number of qubits at a time so that the encoding and
decoding circuits operate in an \textquotedblleft online\textquotedblright%
\ manner. Therefore, each infinite-depth operation approximates its
corresponding rational polynomial. This approximation does not pose a barrier
to implementation. We can implement each of the above infinite-depth
operations by padding the initial qubits of the information qubit stream with
some \textquotedblleft scratch\textquotedblright\ qubits. We first transmit
these \textquotedblleft scratch\textquotedblright\ qubits that contain no
useful quantum information so that the later information qubits enjoy the full
protection of the code. These scratch qubits do not affect the asymptotic rate
of the code and merely serve as a convenience for implementing the
infinite-depth operations. From now on, we adhere to describing infinite-depth
operations with binary rational polynomials because it is more convenient to
do so mathematically.

\section{Closing Remarks}

This chapter reviewed in detail the method of quantum convolutional coding. We
discussed how to represent a quantum convolutional code with a matrix of
binary polynomials. The shifted symplectic product is crucial in several of
the forthcoming chapters because it determines the commutation relations for
an arbitrary set of convolutional generators. In general, a set of generators
do not form a commuting set and we can use entanglement to resolve their
anticommutativity. In this chapter, we also developed many of the operations
that manipulate quantum convolutional codes. We repeatedly use these
operations in the next few chapters.

\chapter{Entanglement-Assisted Quantum Convolutional Coding: The CSS Case}

\label{chp:EAQCC-CSS}\begin{saying}
O Calderbank, good Shor, and bright Steane,\\
Methinks your construction is keen,\\
But it makes me tense,\\
this Pauli sequence,\\
I just want our quantum bits clean.
\end{saying}In this chapter, we develop a theory of entanglement-assisted
quantum convolutional coding for a broad class of codes. Our major result is
that we can produce an entanglement-assisted quantum convolutional code from
two \textit{arbitrary} classical binary convolutional codes. The resulting
quantum convolutional codes admit a Calderbank-Shor-Steane (CSS)\ structure
\cite{PhysRevA.54.1098,PhysRevLett.77.793,book2000mikeandike}. The rates and
error-correcting properties of the two binary classical convolutional codes
directly determine the corresponding properties of the entanglement-assisted
quantum convolutional code.

Our CSS\ entanglement-assisted quantum convolutional codes divide into two
classes based on certain properties of the classical codes from which we
produce them. These properties of the classical codes determine the structure
of the encoding and decoding circuit for the code, and the structure of the
encoding and decoding circuit in turn determines the class of the
entanglement-assisted quantum convolutional code.

\begin{enumerate}
\item Codes in the first class admit both a finite-depth encoding and decoding circuit.

\item Codes in the second class have an encoding circuit that employs both
finite-depth and infinite-depth operations. Their decoding circuits have
finite-depth operations only.
\end{enumerate}

We structure this chapter as follows. We outline the operation of an
entanglement-assisted quantum convolutional code and present our main theorem
in Section~\ref{sec:eaqcc}. This theorem shows how to produce a CSS
entanglement-assisted quantum convolutional code from two arbitrary classical
binary convolutional codes. The theorem gives the rate and error-correcting
properties of a CSS\ entanglement-assisted quantum convolutional code as a
function of the parameters of the classical convolutional codes.
Section~\ref{sec:eaqcc-fefd} completes the proof of the theorem for our first
class of entanglement-assisted quantum convolutional codes.
Section~\ref{sec:eaqcc-iefd} completes the proof of our theorem for the second
class of entanglement-assisted quantum convolutional codes. We discuss the
implications of the assumptions for the different classes of
entanglement-assisted quantum convolutional codes while developing the
constructions. Our theory produces high-performance quantum convolutional
codes by importing high-performance classical convolutional codes.

\section{Operation of an Entanglement-Assisted Quantum Convolutional Code}

\label{sec:eaqcc}An entanglement-assisted quantum convolutional code operates
similarly to a standard quantum convolutional code. The main difference is
that the sender and receiver share entanglement in the form of ebits before quantum communication begins. An
$[[n,k;c]]$ entanglement-assisted quantum convolutional code encodes $k$
information qubits per frame with the help of $c$ ebits and $n-k-c$ ancilla
qubits per frame. Figure~\ref{fig:eaqcc} highlights the main features of the
operation of an entanglement-assisted quantum convolutional code. An
entanglement-assisted quantum convolutional code operates as follows:%
\begin{figure*}
[ptb]
\begin{center}
\includegraphics[
natheight=8.879900in,
natwidth=19.139999in,
height=2.809in,
width=5.834in
]
{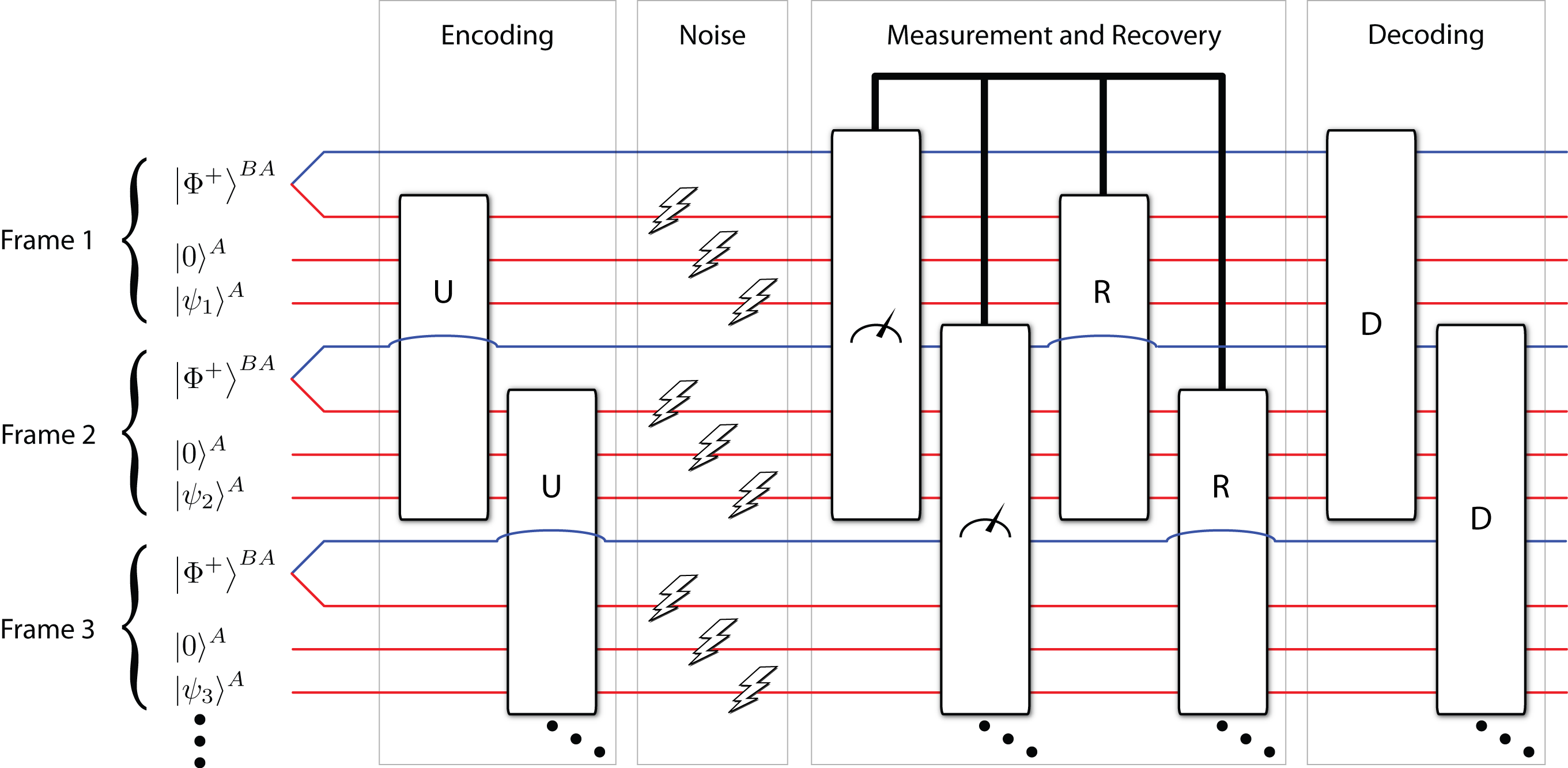}
\end{center}
\caption
{(Color online) An entanglement-assisted quantum convolutional code operates on a stream of qubits
partitioned into a countable number of frames. The sender encodes
the frames of information qubits, ancilla qubits, and half of shared ebits with
a repeated, overlapping encoding circuit $U$. The noisy channel affects the sender's encoded qubits
but does not affect the receiver's half of the shared ebits. The receiver performs overlapping
measurements on both the encoded qubits and his half of the shared ebits. These measurements
produce an error syndrome which the receiver can process to determine the most likely error.
The receiver reverses the errors on the noisy qubits from the sender. The final decoding circuit operates
on all qubits in a frame and recovers
the original stream of information qubits.}
\label{fig:eaqcc}
\end{figure*}%

\begin{enumerate}
\item The sender encodes a stream of quantum information using both additional
ancillas and ebits. The sender performs the encoding operations on her qubits
only (i.e., not including the halves of the ebits in possession of the
receiver). The encoding operations have a periodic structure so that the same
operations act on qubits in different frames and give the code a memory
structure. The sender can perform these encoding operations in an online
manner as she places more qubits in the unencoded qubit stream.

\item The sender transmits her encoded qubits over the noisy quantum
communication channel. The noisy channel does not affect the receiver's half
of the shared ebits.

\item The receiver combines the received noisy qubits with his half of the
ebits and performs measurements to diagnose errors that may occur. These
measurements may overlap on some of the same qubits.

\item The receiver then diagnoses errors using a classical error estimation
algorithm such as Viterbi error estimation \cite{itit1967viterbi} and reverses
the estimates of the errors.

\item The receiver finally performs an online decoding circuit that outputs
the original information qubit stream.
\end{enumerate}

\section{The CSS\ Construction}

Our main theorem below allows us to import two arbitrary classical
convolutional codes for use as a CSS entanglement-assisted quantum
convolutional code. Grassl and R\"{o}tteler were the first to construct
CSS\ quantum convolutional codes from two classical binary convolutional codes
that satisfy an orthogonality constraint---the polynomial parity check
matrices $H_{1}\left(  D\right)  $ and $H_{2}\left(  D\right)  $ of the two
classical codes are orthogonal with respect to the shifted symplectic product:%
\begin{equation}
H_{1}\left(  D\right)  H_{2}^{T}\left(  D^{-1}\right)  =0.
\end{equation}
The resulting symplectic code has a self-orthogonal parity-check matrix when
we join them together using the CSS\ construction. Our theorem generalizes the
work of Grassl and R\"{o}tteler because we can import two \textit{arbitrary}
classical binary convolutional codes---the codes do not necessarily have to
obey the self-orthogonality constraint.

The theorem gives a direct way to compute the amount of entanglement that the
code requires. The number of ebits required is equal to the rank of a
particular matrix derived from the check matrices of the two classical codes.
It generalizes the earlier theorems discussed in Section~\ref{sec:opt-ebit}
that determine the amount of entanglement required for an
entanglement-assisted quantum block code.

Theorem~\ref{thm:main-CSS} also provides a formula to compute the performance
parameters of the entanglement-assisted quantum convolutional code from the
performance parameters of the two classical codes. This formula ensures that
high-rate classical convolutional codes produce high-rate
entanglement-assisted quantum convolutional codes. Our constructions also
ensure high performance for the \textquotedblleft trade-off\textquotedblright%
\ and \textquotedblleft catalytic\textquotedblright\ rates by minimizing the
number of ebits that the codes require.

We begin the proof of the theorem in this section and complete it in different
ways for each of our two classes of entanglement-assisted quantum
convolutional codes in Sections~\ref{sec:eaqcc-fefd} and \ref{sec:eaqcc-iefd}.
The proofs detail how to encode a stream of information qubits, ancilla
qubits, and shared ebits into a code that has the CSS\ structure.

\begin{theorem}
\label{thm:main-CSS}Let $H_{1}\left(  D\right)  $ and $H_{2}\left(  D\right)
$ be the respective check matrices corresponding to noncatastrophic,
delay-free encoders for classical binary convolutional codes $C_{1}$ and
$C_{2}$. Suppose that classical code $C_{i}$ encodes $k_{i}$ information bits
with $n$ bits per frame where $i=1,2$. The respective\ dimensions of
$H_{1}\left(  D\right)  $ and $H_{2}\left(  D\right)  $ are thus $\left(
n-k_{1}\right)  \times n$ and $\left(  n-k_{2}\right)  \times n$. Then the
resulting entanglement-assisted quantum convolutional code encodes
$k_{1}+k_{2}-n+c$ information qubits per frame and is an $\left[  \left[
n,k_{1}+k_{2}-n+c;c\right]  \right]  $ entanglement-assisted quantum
convolutional code. The code requires $c$ ebits per frame where $c$ is equal
to the rank of $H_{1}\left(  D\right)  H_{2}^{T}\left(  D^{-1}\right)  $.
\end{theorem}

Let us begin the proof of the above theorem by constructing an
entanglement-assisted quantum convolutional code. Consider the following
quantum check matrix in CSS\ form:%
\begin{equation}
\left[  \left.
\begin{array}
[c]{c}%
H_{1}\left(  D\right)  \\
0
\end{array}
\right\vert
\begin{array}
[c]{c}%
0\\
H_{2}\left(  D\right)
\end{array}
\right]  .\label{eq:orig-gens}%
\end{equation}
We label the above matrix as a \textquotedblleft quantum check
matrix\textquotedblright\ for now because it does not necessarily correspond
to a commuting stabilizer. The quantum check matrix corresponds to a set of
Pauli sequences whose error-correcting properties are desirable. The Pauli
sequences inherit these desirable properties from the classical codes $C_{1}$
and $C_{2}$.

The following lemma begins the proof of the above theorem. It details an
initial decomposition of the above quantum check matrix for each of our two
classes of entanglement-assisted quantum convolutional codes.

\begin{lemma}
\label{lemma:general-ops}Elementary row and column operations relate the
quantum check matrix\ in (\ref{eq:orig-gens}) to the following matrix%
\[
\left[  \left.
\begin{array}
[c]{cc}%
E\left(  D\right)  & F\left(  D\right) \\
0 & 0
\end{array}
\right\vert
\begin{array}
[c]{cc}%
0 & 0\\
I & 0
\end{array}
\right]  .
\]
where $E\left(  D\right)  $ is dimension $\left(  n-k_{1}\right)
\times\left(  n-k_{2}\right)  $, $F\left(  D\right)  $ is $\left(
n-k_{1}\right)  \times k_{2}$, the identity matrix is $\left(  n-k_{2}\right)
\times\left(  n-k_{2}\right)  $, and the null matrix on the right is $\left(
n-k_{2}\right)  \times k_{2}$. We give a definition of $E\left(  D\right)  $
and $F\left(  D\right)  $ in the following proof.
\end{lemma}

%

\begin{proof}%
The Smith form \cite{book1999conv} of $H_{i}\left(  D\right)  $ for each
$i=1,2$\ is%
\begin{equation}
H_{i}\left(  D\right)  =A_{i}\left(  D\right)  \left[
\begin{array}
[c]{cc}%
I & 0
\end{array}
\right]  B_{i}\left(  D\right)  ,\label{eq:first-Smith-form}%
\end{equation}
where $A_{i}\left(  D\right)  $ is $\left(  n-k_{i}\right)  \times\left(
n-k_{i}\right)  $, the matrix in brackets is $\left(  n-k_{i}\right)  \times
n$, and $B_{i}\left(  D\right)  $ is $n\times n$. The Smith decomposition is
somewhat analogous to the SVD\ decomposition of a matrix over the reals.
The binary polynomials along the diagonal of the center
matrix of the Smith decomposition are the \textit{invariant factors}. 
The matrices $A_{i}\left(  D\right)  $ and $B_{i}\left(
D\right)  $ are a product of a sequence of elementary row and column
operations respectively \cite{book1999conv}.
Let
$B_{ia}\left(  D\right)  $\ be the first $n-k_{i}$ rows of $B_{i}\left(
D\right)  $ and let $B_{ib}\left(  D\right)  $ be the last $k_{i}$ rows of
$B_{i}\left(  D\right)  $:%
\[
B_{i}\left(  D\right)  =\left[
\begin{array}
[c]{c}%
B_{ia}\left(  D\right)  \\
B_{ib}\left(  D\right)
\end{array}
\right]  .
\]
The $\left(  n-k_{i}\right)  \times\left(  n-k_{i}\right)  $ identity matrix
in brackets in (\ref{eq:first-Smith-form}) indicates that the invariant
factors of $H_{i}\left(  D\right)  $ for each $i=1,2$ are all equal to one
\cite{book1999conv}. The
invariant factors are all unity for both check matrices because the check
matrices correspond to noncatastrophic, delay-free encoders
\cite{book1999conv}. 

Premultiplying $H_{i}\left(  D\right)  $ by $A_{i}^{-1}\left(  D\right)  $
gives a check matrix $H_{i}^{\prime}\left(  D\right)  $ for each $i=1,2$.
Matrix $H_{i}^{\prime}\left(  D\right)  $ is a check matrix for code $C_{i}$
with equivalent error-correcting properties as $H_{i}\left(  D\right)
$\ because row operations relate the two matrices. This new check matrix
$H_{i}^{\prime}\left(  D\right)  $\ is equal to the first $n-k_{i}$ rows of
matrix $B_{i}\left(  D\right)  $:%
\[
H_{i}^{\prime}\left(  D\right)  =B_{ia}\left(  D\right)  .
\]

The invariant factors of $H_{1}\left(  D\right)  H_{2}^{T}\left(
D^{-1}\right)  $ are equivalent to those of $H_{1}^{\prime}\left(  D\right)
H_{2}^{\prime T}\left(  D^{-1}\right)  $ because they are related by elementary row and
column operations \cite{book1999conv}:%
\begin{equation}
H_{1}\left(  D\right)  H_{2}^{T}\left(  D^{-1}\right)  =A_{1}\left(  D\right)
H_{1}^{\prime}\left(  D\right)  H_{2}^{\prime T}\left(  D^{-1}\right)
A_{2}^{T}\left(  D^{-1}\right)  . \label{eq:check-relation}%
\end{equation}

We now decompose the above quantum check matrix into a basic form using
elementary row and column operations. The row operations have no effect on the
error-correcting properties of the code, and the column operations correspond
to elements of an encoding circuit. We later show how to incorporate ebits so
that the quantum check matrix forms a valid commuting stabilizer.

Perform the row operations in matrices $A_{i}^{-1}\left(  D\right)  $ for both
check matrices $H_{i}\left(  D\right)  $. The quantum check matrix becomes%
\begin{equation}
\left[  \left.
\begin{array}
[c]{c}%
B_{1a}\left(  D\right)  \\
0
\end{array}
\right\vert
\begin{array}
[c]{c}%
0\\
B_{2a}\left(  D\right)
\end{array}
\right]  .\label{eq:orig-qcm}%
\end{equation}
The error-correcting properties of the above generators are equivalent to
those of the generators in (\ref{eq:orig-gens}) because row operations relate
the two sets of generators. The matrix $B_{2}\left(  D\right)  $ corresponds
to a sequence $\left\{  B_{2,i}\left(  D\right)  \right\}  _{i=1}^{l}$\ of
elementary column operations:%
\[
B_{2}\left(  D\right)  =B_{2,1}\left(  D\right)  \cdots B_{2,l}\left(
D\right)  =\prod_{i=1}^{l}B_{2,i}\left(  D\right)  .
\]
The inverse matrix $B_{2}^{-1}\left(  D\right)  $ is therefore equal to the
above sequence of operations in reverse order:%
\[
B_{2}^{-1}\left(  D\right)  =B_{2,l}\left(  D\right)  \cdots B_{2,1}\left(
D\right)  =\prod_{i=l}^{1}B_{2,i}\left(  D\right)  .
\]
Perform the elementary column operations in $B_{2}^{-1}\left(  D\right)  $
with CNOT\ and SWAP\ gates \cite{isit2006grassl}. The effect of each
elementary column operation $B_{2,i}\left(  D\right)  $ is to postmultiply the
\textquotedblleft X\textquotedblright\ matrix by $B_{2,i}\left(  D\right)  $
and to postmultiply the \textquotedblleft Z\textquotedblright\ matrix by
$B_{2,i}^{T}\left(  D^{-1}\right)  $. Therefore the effect of all elementary
operations is to postmultiply the \textquotedblleft Z\textquotedblright%
\ matrix by $B_{2}^{T}\left(  D^{-1}\right)  $ because%
\[
\prod_{i=l}^{1}B_{2,i}^{T}\left(  D^{-1}\right)  =\left(  \prod_{i=1}%
^{l}B_{2,i}\left(  D^{-1}\right)  \right)  ^{T}=B_{2}^{T}\left(
D^{-1}\right)  .
\]
The quantum check matrix in (\ref{eq:orig-qcm})\ becomes%
\begin{equation}
\left[  \left.
\begin{array}
[c]{c}%
B_{1a}\left(  D\right)  B_{2}^{T}\left(  D^{-1}\right)  \\
0
\end{array}
\right\vert
\begin{array}
[c]{cc}%
0 & 0\\
I & 0
\end{array}
\right]  .\label{eq:lemma-temp-QCM}%
\end{equation}
Let $E\left(  D\right)  $ be equal to the first $n-k_{1}$ rows and $n-k_{2}$
columns of the \textquotedblleft Z\textquotedblright\ matrix:%
\[
E\left(  D\right)  \equiv B_{1,a}\left(  D\right)  B_{2,a}^{T}\left(
D^{-1}\right)  ,
\]
and let $F\left(  D\right)  $ be equal to the first $n-k_{1}$ rows and last
$k_{2}$ columns of the \textquotedblleft Z\textquotedblright\ matrix:%
\[
F\left(  D\right)  \equiv B_{1,a}\left(  D\right)  B_{2,b}^{T}\left(
D^{-1}\right)  .
\]
The quantum check matrix in (\ref{eq:lemma-temp-QCM}) is then equivalent to
the following matrix%
\begin{equation}
\left[  \left.
\begin{array}
[c]{cc}%
E\left(  D\right)   & F\left(  D\right)  \\
0 & 0
\end{array}
\right\vert
\begin{array}
[c]{cc}%
0 & 0\\
I & 0
\end{array}
\right]  ,\label{eq:last-eq-fefd}%
\end{equation}
where each matrix above has the dimensions stated in the theorem above.%
\end{proof}%

The above operations end the initial set of operations that each of our two
classes of entanglement-assisted quantum convolutional codes employs. We
outline the remaining operations for each class of codes in what follows.

\section{Codes with Finite-Depth Encoding and Decoding Circuits}

\label{sec:eaqcc-fefd}This section details entanglement-assisted quantum
convolutional codes in our first class. Codes in the first class admit an
encoding and decoding circuit that employ finite-depth operations only. The
check matrices for codes in this class have a property that allows this type
of encoding and decoding. The following lemma gives the details of this
property, and the proof outlines how to encode and decode this class of
entanglement-assisted quantum convolutional codes.

\begin{lemma}
\label{lemma:fefd}Suppose the Smith form\ of $H_{1}\left(  D\right)  H_{2}%
^{T}\left(  D^{-1}\right)  $ is%
\[
H_{1}\left(  D\right)  H_{2}^{T}\left(  D^{-1}\right)  =A\left(  D\right)
\left[
\begin{array}
[c]{cc}%
\Gamma\left(  D\right)  & 0\\
0 & 0
\end{array}
\right]  B\left(  D\right)  ,
\]
where $A\left(  D\right)  $ is an $\left(  n-k_{1}\right)  \times\left(
n-k_{1}\right)  $ matrix, $B\left(  D\right)  $ is an $\left(  n-k_{2}\right)
\times\left(  n-k_{2}\right)  $ matrix, $\Gamma\left(  D\right)  $ is a
diagonal $c\times c$ matrix whose entries are powers of $D$, and the matrix in
brackets has dimension $\left(  n-k_{1}\right)  \times\left(  n-k_{2}\right)
$. Then the resulting entanglement-assisted quantum convolutional code has
both a finite-depth encoding and decoding circuit.
\end{lemma}

%

\begin{proof}%
We begin the proof of this lemma by continuing where the proof of
Lemma~\ref{lemma:general-ops} ends. The crucial assumption for the above lemma
is that the invariant factors of $H_{1}\left(  D\right)  H_{2}^{T}\left(
D^{-1}\right)  $ are all powers of $D$. The Smith form of $E\left(  D\right)
$ in (\ref{eq:last-eq-fefd}) therefore becomes%
\[
A_{1}^{-1}\left(  D\right)  A\left(  D\right)  \left[
\begin{array}
[c]{cc}%
\Gamma\left(  D\right)  & 0\\
0 & 0
\end{array}
\right]  B\left(  D\right)  A_{2}^{-1}\left(  D\right)  ,
\]
by employing the hypothesis of Lemma~\ref{lemma:fefd} and
(\ref{eq:check-relation}). The rank of both $H_{1}\left(  D\right)  H_{2}%
^{T}\left(  D^{-1}\right)  $ and $E\left(  D\right)  $ is equal to $c$.

Perform the inverse of the row operations in $A_{1}^{-1}\left(  D\right)
A\left(  D\right)  $ on the first $n-k_{1}$ rows of the quantum check matrix
in (\ref{eq:last-eq-fefd}). Perform the inverse of the column operations in
matrix $B\left(  D\right)  A_{2}^{-1}\left(  D\right)  $ on the first
$n-k_{2}$ columns of the quantum check matrix in (\ref{eq:last-eq-fefd}). We
execute these column operations with Hadamard, CNOT,\ and SWAP\ gates. These
column operations have a corresponding effect on columns in the
\textquotedblleft X\textquotedblright\ matrix, but we can exploit the identity
matrix in the last $n-k_{2}$ rows of the \textquotedblleft X\textquotedblright%
\ matrix to counteract this effect. We perform row operations on the last
$n-k_{2}$ rows of the matrix that act as the inverse of the column operations,
and therefore the quantum check matrix in (\ref{eq:last-eq-fefd})\ becomes%
\[
\left[  \left.
\begin{array}
[c]{ccc}%
\Gamma\left(  D\right)  & 0 & F_{1}\left(  D\right) \\
0 & 0 & F_{2}\left(  D\right) \\
0 & 0 & 0\\
0 & 0 & 0
\end{array}
\right\vert
\begin{array}
[c]{ccc}%
0 & 0 & 0\\
0 & 0 & 0\\
I & 0 & 0\\
0 & I & 0
\end{array}
\right]  ,
\]
where $F_{1}\left(  D\right)  $ and $F_{2}\left(  D\right)  $ are the first
$c$ and $n-k_{1}-c$\ respective rows of $A^{-1}\left(  D\right)  A_{1}\left(
D\right)  F\left(  D\right)  $. We perform Hadamard and CNOT\ gates to clear
the entries in $F_{1}\left(  D\right)  $ in the \textquotedblleft
Z\textquotedblright\ matrix above. The quantum check matrix becomes%
\begin{equation}
\left[  \left.
\begin{array}
[c]{ccc}%
\Gamma\left(  D\right)  & 0 & 0\\
0 & 0 & F_{2}\left(  D\right) \\
0 & 0 & 0\\
0 & I & 0
\end{array}
\right\vert
\begin{array}
[c]{ccc}%
0 & 0 & 0\\
0 & 0 & 0\\
I & 0 & 0\\
0 & 0 & 0
\end{array}
\right]  . \label{eq:init-example-inter-step}%
\end{equation}

The Smith form of $F_{2}\left(  D\right)  $ is%
\[
F_{2}\left(  D\right)  =A_{F}\left(  D\right)  \left[
\begin{array}
[c]{cc}%
\Gamma_{F}\left(  D\right)  & 0
\end{array}
\right]  B_{F}\left(  D\right)  ,
\]
where $\Gamma_{F}\left(  D\right)  $ is a diagonal matrix whose entries are
powers of $D$, $A_{F}\left(  D\right)  $ is $\left(  n-k_{1}-c\right)
\times\left(  n-k_{1}-c\right)  $, and $B_{F}\left(  D\right)  $ is
$k_{2}\times k_{2}$. The Smith form of $F_{2}\left(  D\right)  $ takes this
particular form because the original check matrix $H_{2}\left(  D\right)  $ is
noncatastrophic and column operations with Laurent polynomials change the
invariant factors only up to powers of $D$.

Perform row operations corresponding to $A_{F}^{-1}\left(  D\right)  $ on the
second set of $n-k_{1}-c$ rows with $F_{2}\left(  D\right)  $ in
(\ref{eq:init-example-inter-step}). Perform column operations corresponding to
$B_{F}^{-1}\left(  D\right)  $ on columns $n-k_{2}+1,\ldots,n$ with Hadamard,
CNOT, and SWAP\ gates. The resulting quantum check matrix has the following
form:%
\begin{equation}
\left[  \left.
\begin{array}
[c]{cccc}%
\Gamma\left(  D\right)  & 0 & 0 & 0\\
0 & 0 & \Gamma_{F}\left(  D\right)  & 0\\
0 & 0 & 0 & 0\\
0 & I & 0 & 0
\end{array}
\right\vert
\begin{array}
[c]{cccc}%
0 & 0 & 0 & 0\\
0 & 0 & 0 & 0\\
I & 0 & 0 & 0\\
0 & 0 & 0 & 0
\end{array}
\right]  . \label{eq:final-quantum-check-matrix}%
\end{equation}

We have now completed the decomposition of the original quantum check matrix
in (\ref{eq:orig-gens}) for this class of entanglement-assisted quantum
convolutional codes. It is not possible to perform row or column operations to
decompose the above matrix any further. The problem with the above quantum
check matrix is that it does not form a valid quantum convolutional code. The
first set of rows with matrix $\Gamma\left(  D\right)  $ are not orthogonal
under the shifted symplectic product to the third set of rows with the
identity matrix on the \textquotedblleft X\textquotedblright\ side.
Equivalently, the set of Pauli sequences corresponding to the above quantum
check matrix do not form a commuting stabilizer. We can use entanglement
shared between sender and receiver to solve this problem. Entanglement adds
columns to the above quantum check matrix to resolve the issue. The additional
columns correspond to qubits on the receiver's side. We next show in detail
how to incorporate ancilla qubits, ebits, and information qubits to obtain a
valid stabilizer code. The result is that we can exploit the error-correcting
properties of the original code to protect the sender's qubits.

Consider the following check matrix corresponding to a commuting stabilizer:%
\begin{equation}
\left[  \left.
\begin{array}
[c]{ccccc}%
I & I & 0 & 0 & 0\\
0 & 0 & 0 & I & 0\\
0 & 0 & 0 & 0 & 0\\
0 & 0 & I & 0 & 0
\end{array}
\right\vert
\begin{array}
[c]{ccccc}%
0 & 0 & 0 & 0 & 0\\
0 & 0 & 0 & 0 & 0\\
I & I & 0 & 0 & 0\\
0 & 0 & 0 & 0 & 0
\end{array}
\right]  , \label{eq:bare-1st-stab}%
\end{equation}
where the identity matrices in the first and third sets of rows each have
dimension $c\times c$, the identity matrix in the second set of rows has
dimension $\left(  n-k_{1}-c\right)  \times\left(  n-k_{1}-c\right)  $, and
the identity matrix in the fourth set of rows has dimension $\left(
n-k_{2}-c\right)  \times\left(  n-k_{2}-c\right)  $. The first and third sets
of $c$\ rows stabilize a set of $c$ ebits shared between Alice and Bob. Bob
possesses the \textquotedblleft left\textquotedblright\ $c$ qubits and Alice
possesses the \textquotedblleft right\textquotedblright\ $n$ qubits. The
second and fourth sets of rows stabilize a set of $2\left(  n-c\right)
-k_{1}-k_{2}$ ancilla qubits that Alice possesses. The stabilizer therefore
stabilizes a set of $c$ ebits, $2\left(  n-c\right)  -k_{1}-k_{2}$ ancilla
qubits, and $k_{1}+k_{2}-n+c$ information qubits.

Observe that the last $n$ columns of the \textquotedblleft Z\textquotedblright%
\ and \textquotedblleft X\textquotedblright\ matrices in the above stabilizer
are similar in their layout to the entries in
(\ref{eq:final-quantum-check-matrix}). We can delay the rows of the above
stabilizer by an arbitrary amount to obtain the desired stabilizer. So the
above stabilizer is a subcode of the following stabilizer in the sense of
Ref.~\cite{isit2006grassl}:%
\[
\left[  \left.
\begin{array}
[c]{ccccc}%
\Gamma\left(  D\right)  & \Gamma\left(  D\right)  & 0 & 0 & 0\\
0 & 0 & 0 & \Gamma_{F}\left(  D\right)  & 0\\
0 & 0 & 0 & 0 & 0\\
0 & 0 & I & 0 & 0
\end{array}
\right\vert
\begin{array}
[c]{ccccc}%
0 & 0 & 0 & 0 & 0\\
0 & 0 & 0 & 0 & 0\\
I & I & 0 & 0 & 0\\
0 & 0 & 0 & 0 & 0
\end{array}
\right]  .
\]
The stabilizer in (\ref{eq:bare-1st-stab}) has equivalent error-correcting
properties to and the same asymptotic rate as the above desired stabilizer.
The above stabilizer matrix is an augmented version of the quantum check
matrix in (\ref{eq:final-quantum-check-matrix}) that includes entanglement.
The sender performs all of the encoding column operations detailed in the
proofs of this lemma and Lemma~\ref{lemma:general-ops} in reverse order. The
result of these operations is an $\left[  \left[  n,k_{1}+k_{2}-n+c;c\right]
\right]  $ entanglement-assisted quantum convolutional code with the same
error-correcting properties as the quantum check matrix in (\ref{eq:orig-gens}%
). The receiver decodes the original information-qubit stream by performing
the column operations in the order presented. The information qubits appear as
the last $k_{1}+k_{2}-n+c$ in each frame of the stream (corresponding to the
$k_{1}+k_{2}-n+c$ columns of zeros in both the \textquotedblleft
Z\textquotedblright\ and \textquotedblleft X\textquotedblright\ matrices
above).%
\end{proof}%

\begin{example}
\label{ex:fefd-example}Consider a classical convolutional code with the
following check matrix:%
\[
H\left(  D\right)  =\left[
\begin{array}
[c]{cc}%
1+D^{2} & 1+D+D^{2}%
\end{array}
\right]  .
\]
We can use $H\left(  D\right)  $ in an entanglement-assisted quantum
convolutional code to correct for both bit-flip errors and phase-flip errors.
We form the following quantum check matrix:%
\begin{equation}
\left[  \left.
\begin{array}
[c]{cc}%
1+D^{2} & 1+D+D^{2}\\
0 & 0
\end{array}
\right\vert
\begin{array}
[c]{cc}%
0 & 0\\
1+D^{2} & 1+D+D^{2}%
\end{array}
\right]  . \label{eq:desired-QCM-first-example}%
\end{equation}
This code falls in the first class of entanglement-assisted quantum
convolutional codes because $H\left(  D\right)  H^{T}\left(  D^{-1}\right)
=1$.\newline\newline We do not show the decomposition of the above check
matrix as outlined in Lemma~\ref{lemma:fefd}, but instead show how to encode
it starting from a stream of information qubits and ebits. Each frame has one
ebit and one information qubit.%
\begin{figure}
[ptb]
\begin{center}
\includegraphics[
natheight=9.753400in,
natwidth=7.920000in,
height=4.0335in,
width=3.2802in
]%
{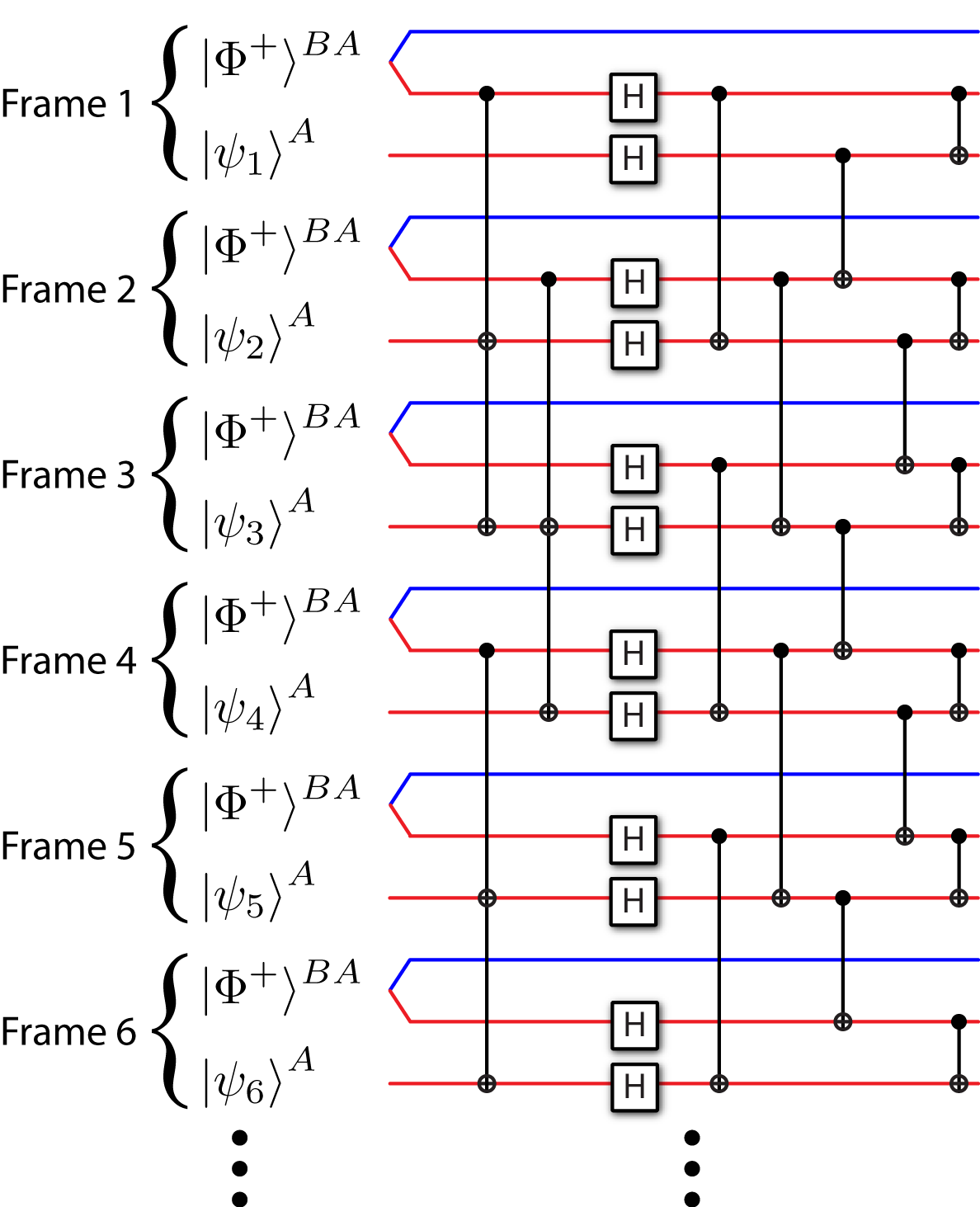}%
\caption{(Color online) The finite-depth encoding circuit for the
entanglement-assisted quantum convolutional code in
Example~\ref{ex:fefd-example}. The above operations in reverse order give a
valid decoding circuit.}%
\label{fig:example-eaqcc-fefd}%
\end{center}
\end{figure}
Let us begin with a polynomial matrix that stabilizes the unencoded state:%
\[
\left[  \left.
\begin{array}
[c]{ccc}%
1 & 1 & 0\\
0 & 0 & 0
\end{array}
\right\vert
\begin{array}
[c]{ccc}%
0 & 0 & 0\\
1 & 1 & 0
\end{array}
\right]  .
\]
Alice possesses the two qubits on the \textquotedblleft
right\textquotedblright\ and Bob possesses the qubit on the \textquotedblleft
left.\textquotedblright\ We label the middle qubit as \textquotedblleft qubit
one\textquotedblright\ and the rightmost qubit as \textquotedblleft qubit
two.\textquotedblright\ Alice performs a CNOT\ from qubit one to qubit two in
a delayed frame and a CNOT\ from qubit one to qubit two in a frame delayed by
two. The stabilizer becomes%
\[
\left[  \left.
\begin{array}
[c]{ccc}%
1 & 1 & 0\\
0 & 0 & 0
\end{array}
\right\vert
\begin{array}
[c]{ccc}%
0 & 0 & 0\\
1 & 1 & D+D^{2}%
\end{array}
\right]  .
\]
Alice performs Hadamard gates on both of her qubits. The stabilizer becomes%
\[
\left[  \left.
\begin{array}
[c]{ccc}%
1 & 0 & 0\\
0 & 1 & D+D^{2}%
\end{array}
\right\vert
\begin{array}
[c]{ccc}%
0 & 1 & 0\\
1 & 0 & 0
\end{array}
\right]  .
\]
Alice performs a CNOT\ from qubit one to qubit two in a delayed frame. The
stabilizer becomes%
\[
\left[  \left.
\begin{array}
[c]{ccc}%
1 & 0 & 0\\
0 & D & D+D^{2}%
\end{array}
\right\vert
\begin{array}
[c]{ccc}%
0 & 1 & D\\
1 & 0 & 0
\end{array}
\right]  .
\]
Alice performs a CNOT\ from qubit two to qubit one in a delayed frame. The
stabilizer becomes%
\[
\left[  \left.
\begin{array}
[c]{ccc}%
1 & 0 & 0\\
0 & D & 1+D+D^{2}%
\end{array}
\right\vert
\begin{array}
[c]{ccc}%
0 & 1+D^{2} & D\\
1 & 0 & 0
\end{array}
\right]  .
\]
Alice performs a CNOT\ from qubit one to qubit two. The stabilizer becomes%
\[
\left[  \left.
\begin{array}
[c]{ccc}%
1 & 0 & 0\\
0 & 1+D^{2} & 1+D+D^{2}%
\end{array}
\right\vert
\begin{array}
[c]{ccc}%
0 & 1+D^{2} & 1+D+D^{2}\\
1 & 0 & 0
\end{array}
\right]  .
\]
A row operation that switches the first row with the second row gives the
following stabilizer:%
\[
\left[  \left.
\begin{array}
[c]{ccc}%
0 & 1+D^{2} & 1+D+D^{2}\\
1 & 0 & 0
\end{array}
\right\vert
\begin{array}
[c]{ccc}%
1 & 0 & 0\\
0 & 1+D^{2} & 1+D+D^{2}%
\end{array}
\right]  .
\]
The entries on Alice's side of the above stabilizer have equivalent
error-correcting properties to the quantum check matrix in
(\ref{eq:desired-QCM-first-example}). Figure~\ref{fig:example-eaqcc-fefd}%
\ illustrates how the above operations encode a stream of ebits and
information qubits for our example.
\end{example}

\subsubsection{Discussion}

Codes in the first class are more useful in practice than those in the second
because their encoding and decoding circuits are finite depth. An uncorrected
error propagates only to a finite number of information qubits in the decoded
qubit stream. Codes in the first class therefore do not require any
assumptions about noiseless encoding or decoding.

The assumption about the invariant factors in the Smith form of $H_{1}\left(
D\right)  H_{2}^{T}\left(  D^{-1}\right)  $\ holds only for some classical
check matrices. Only a subclass of classical codes satisfy this assumption,
but it still expands the set of available quantum codes beyond those whose
check matrices $H_{1}\left(  D\right)  $ and $H_{2}\left(  D\right)  $ are
orthogonal. We need further techniques to handle the classical codes for which
this assumption does not hold. The following sections provide these further
techniques to handle a larger class of entanglement-assisted quantum
convolutional codes.

\section{Codes with Infinite-Depth Encoding and Finite-Depth Decoding
Circuits}

\label{sec:eaqcc-iefd}This section details codes whose encoding circuits have
both infinite-depth and finite-depth operations. We therefore assume that
encoding is noiseless to eliminate the possibility of encoding errors
spreading infinitely into the encoded qubit stream. We later briefly discuss
the effects of relaxing this assumption in a realistic system. Their decoding
circuits require finite-depth operations only. Some of the decoding circuits
are not the exact inverse of their corresponding encoding circuits, but the
decoding circuits invert the effect of the encoding circuits because they
produce the original stream of information qubits at their output.

Just as with the previous class, this class of codes is determined by the
properties of their corresponding classical check matrices, as described in
the following lemma.

\begin{lemma}
\label{lemma:ieid}Suppose the Smith form of $E\left(  D\right)  $ does not
admit the form from Lemma~\ref{lemma:fefd}. Then the entanglement-assisted
quantum convolutional code has an encoding circuit with both infinite-depth
and finite-depth operations. Its decoding circuit has finite-depth operations.
\end{lemma}

%

\begin{proof}%
We perform all of the operations from Lemma~\ref{lemma:general-ops}. The Smith
form of $E\left(  D\right)  $ is in general as follows%
\[
A_{E}\left(  D\right)
\begin{bmatrix}
\Gamma_{1}\left(  D\right)  & 0 & 0\\
0 & \Gamma_{2}\left(  D\right)  & 0\\
0 & 0 & 0
\end{bmatrix}
B_{E}\left(  D\right)  ,
\]
where $A_{E}\left(  D\right)  $ is $\left(  n-k_{1}\right)  \times\left(
n-k_{1}\right)  $, $\Gamma_{1}\left(  D\right)  $ is an $s\times s$ diagonal
matrix whose entries are powers of $D$, $\Gamma_{2}\left(  D\right)  $ is a
$\left(  c-s\right)  \times\left(  c-s\right)  $ diagonal matrix whose entries
are arbitrary polynomials, and $B_{E}\left(  D\right)  $ is $\left(
n-k_{2}\right)  \times\left(  n-k_{2}\right)  $. Perform the row operations in
$A_{E}^{-1}\left(  D\right)  $ and the column operations in $B_{E}^{-1}\left(
D\right)  $ on the quantum check matrix in (\ref{eq:last-eq-fefd}). Counteract
the effect of the column operations on the identity matrix in the
\textquotedblleft X\textquotedblright\ matrix by performing row operations.
The quantum check matrix in (\ref{eq:last-eq-fefd}) becomes%
\[
\left[  \left.
\begin{array}
[c]{cccc}%
\Gamma_{1}\left(  D\right)  & 0 & 0 & F_{1}\left(  D\right) \\
0 & \Gamma_{2}\left(  D\right)  & 0 & F_{2}\left(  D\right) \\
0 & 0 & 0 & F_{3}\left(  D\right) \\
0 & 0 & 0 & 0
\end{array}
\right\vert
\begin{array}
[c]{cc}%
0 & 0\\
0 & 0\\
0 & 0\\
I & 0
\end{array}
\right]  ,
\]
where $F_{1}\left(  D\right)  $, $F_{2}\left(  D\right)  $, and $F_{3}\left(
D\right)  $ are the respective $s$, $c-s$, and $n-k_{1}-c$ rows of $A_{E}%
^{-1}\left(  D\right)  F\left(  D\right)  $. The Smith form of $F_{3}\left(
D\right)  $ is as follows%
\[
F_{3}\left(  D\right)  =A_{F_{3}}\left(  D\right)
\begin{bmatrix}
\Gamma_{F_{3}}\left(  D\right)  & 0
\end{bmatrix}
B_{F_{3}}\left(  D\right)  ,
\]
where $A_{F_{3}}\left(  D\right)  $ is $\left(  n-k_{1}-c\right)
\times\left(  n-k_{1}-c\right)  $, $\Gamma_{F_{3}}\left(  D\right)  $ is an
$\left(  n-k_{1}-c\right)  \times\left(  n-k_{1}-c\right)  $ diagonal matrix
whose entries are powers of $D$, and $B_{F_{3}}\left(  D\right)  $ is
$k_{2}\times k_{2}$. The entries of $\Gamma_{F_{3}}\left(  D\right)  $ are
powers of $D$ because the original check matrix $H_{2}\left(  D\right)  $ is
noncatastrophic and column and row operations with Laurent polynomials change
the invariant factors only by a power of $D$. Perform the row operations in
$A_{F_{3}}^{-1}\left(  D\right)  $ and the column operations in $B_{F_{3}%
}^{-1}\left(  D\right)  $. The quantum check matrix becomes%
\[
\left[  \left.
\begin{array}
[c]{ccccc}%
\Gamma_{1}\left(  D\right)  & 0 & 0 & F_{1a}^{^{\prime}}\left(  D\right)  &
F_{1b}^{^{\prime}}\left(  D\right) \\
0 & \Gamma_{2}\left(  D\right)  & 0 & F_{2a}^{^{\prime}}\left(  D\right)  &
F_{2b}^{^{\prime}}\left(  D\right) \\
0 & 0 & 0 & \Gamma_{F_{3}}\left(  D\right)  & 0\\
0 & 0 & 0 & 0 & 0
\end{array}
\right\vert
\begin{array}
[c]{cc}%
0 & 0\\
0 & 0\\
0 & 0\\
I & 0
\end{array}
\right]  ,
\]
where $F_{1a}^{^{\prime}}\left(  D\right)  $, $F_{1b}^{^{\prime}}\left(
D\right)  $, $F_{2a}^{^{\prime}}\left(  D\right)  $, $F_{2b}^{^{\prime}%
}\left(  D\right)  $ are the matrices resulting from the column operations in
$B_{F_{3}}^{-1}\left(  D\right)  $. Perform row operations from the entries in
$\Gamma_{F_{3}}\left(  D\right)  $ to the rows above it to clear the entries
in $F_{1a}^{^{\prime}}\left(  D\right)  $ and $F_{2a}^{^{\prime}}\left(
D\right)  $. Use Hadamard and CNOT\ gates to clear the entries in
$F_{1b}^{^{\prime}}\left(  D\right)  $. The quantum check matrix becomes%
\[
\left[  \left.
\begin{array}
[c]{ccccc}%
\Gamma_{1}\left(  D\right)  & 0 & 0 & 0 & 0\\
0 & \Gamma_{2}\left(  D\right)  & 0 & 0 & F_{2b}^{^{\prime}}\left(  D\right)
\\
0 & 0 & 0 & \Gamma_{F_{3}}\left(  D\right)  & 0\\
0 & 0 & 0 & 0 & 0
\end{array}
\right\vert
\begin{array}
[c]{cc}%
0 & 0\\
0 & 0\\
0 & 0\\
I & 0
\end{array}
\right]  .
\]
We can reduce $F_{2b}^{^{\prime}}\left(  D\right)  $ to a lower triangular
form with an algorithm consisting of column operations only. The algorithm
operates on the last $k_{2}+k_{1}-n+c$ columns. It is similar to the Smith
algorithm but does not involve row operations. Consider the first row of
$F_{2b}^{^{\prime}}\left(  D\right)  $. Perform column operations between the
different elements of the row to reduce it to one non-zero entry. Swap this
non-zero entry to the leftmost position. Perform the same algorithm on
elements $2,\ldots,k_{2}+k_{1}-n+c$ of the second row. Continue on for all
rows of $F_{2b}^{^{\prime}}\left(  D\right)  $ to reduce it to a matrix of the
following form\bigskip%
\[
F_{2b}^{^{\prime}}\left(  D\right)  \rightarrow%
\begin{bmatrix}%
\raisebox{0ex}[1.5ex]{$\overbrace{L(D)}^{c-s}$}%
&
\raisebox{0ex}[1.5ex]{$\overbrace{0}^{k_1+k_2-n+s}$}%
\end{bmatrix}
,
\]
where $L\left(  D\right)  $ is a lower triangular matrix. The above quantum
check matrix becomes%
\[
\left[  \left.
\begin{array}
[c]{cccccc}%
\Gamma_{1}\left(  D\right)  & 0 & 0 & 0 & 0 & 0\\
0 & \Gamma_{2}\left(  D\right)  & 0 & 0 & L\left(  D\right)  & 0\\
0 & 0 & 0 & \Gamma_{F_{3}}\left(  D\right)  & 0 & 0\\
0 & 0 & 0 & 0 & 0 & 0
\end{array}
\right\vert
\begin{array}
[c]{cc}%
0 & 0\\
0 & 0\\
0 & 0\\
I & 0
\end{array}
\right]  .
\]
We have completed decomposition of the first set of $s$ rows with $\Gamma
_{1}\left(  D\right)  $, the third set of $n-k_{1}-c$ rows with $\Gamma
_{F_{3}}\left(  D\right)  $, and rows $n-k_{1}+1,\ldots,n-k_{1}+s$ with the
identity matrix on the \textquotedblleft X\textquotedblright\ side.

We now consider an algorithm with infinite-depth operations to encode the
following submatrix of the above quantum check matrix:%
\begin{equation}
\left[  \left.
\begin{array}
[c]{cc}%
\Gamma_{2}\left(  D\right)  & L\left(  D\right) \\
0 & 0
\end{array}
\right\vert
\begin{array}
[c]{cc}%
0 & 0\\
I & 0
\end{array}
\right]  . \label{eq:ieid-desired-matrix}%
\end{equation}
We begin with a set of $c-s$ ebits and $c-s$ information qubits. The following
matrix stabilizes the ebits%
\[
\left[  \left.
\begin{array}
[c]{ccc}%
I & I & 0\\
0 & 0 & 0
\end{array}
\right\vert
\begin{array}
[c]{ccc}%
0 & 0 & 0\\
I & I & 0
\end{array}
\right]  ,
\]
and the following matrix represents the information qubits%
\[
\left[  \left.
\begin{array}
[c]{ccc}%
0 & 0 & I\\
0 & 0 & 0
\end{array}
\right\vert
\begin{array}
[c]{ccc}%
0 & 0 & 0\\
0 & 0 & I
\end{array}
\right]  ,
\]
where all matrices have dimension $\left(  c-s\right)  \times\left(
c-s\right)  $ and Bob possesses the $c-s$ qubits on the \textquotedblleft
left\textquotedblright\ and Alice possesses the $2\left(  c-s\right)  $ qubits
on the \textquotedblleft right.\textquotedblright\ We track both the
stabilizer and the information qubits as they progress through some encoding
operations. Alice performs CNOT\ and Hadamard gates on her $2\left(
c-s\right)  $ qubits. These gates multiply the middle $c-s$ columns of the
\textquotedblleft Z\textquotedblright\ matrix by $L\left(  D\right)  $ and add
the result to the last $c-s$ columns and multiply the last $c-s$ columns of
the \textquotedblleft X\textquotedblright\ matrix by $L^{T}\left(
D^{-1}\right)  $ and add the result to the middle $c-s$ columns. The
stabilizer becomes%
\[
\left[  \left.
\begin{array}
[c]{ccc}%
I & I & L\left(  D\right) \\
0 & 0 & 0
\end{array}
\right\vert
\begin{array}
[c]{ccc}%
0 & 0 & 0\\
I & I & 0
\end{array}
\right]  ,
\]
and the information-qubit matrix becomes%
\[
\left[  \left.
\begin{array}
[c]{ccc}%
0 & 0 & I\\
0 & 0 & 0
\end{array}
\right\vert
\begin{array}
[c]{ccc}%
0 & 0 & 0\\
0 & L^{T}\left(  D^{-1}\right)  & I
\end{array}
\right]  .
\]
Alice performs infinite-depth operations on her first $c-s$ qubits
corresponding to the rational polynomials $\gamma_{2,1}^{-1}\left(
D^{-1}\right)  $, $\ldots$, $\gamma_{2,c-s}^{-1}\left(  D^{-1}\right)  $ in
$\Gamma_{2}^{-1}\left(  D^{-1}\right)  $. The stabilizer matrix becomes%
\[
\left[  \left.
\begin{array}
[c]{ccc}%
I & \Gamma_{2}\left(  D\right)  & L\left(  D\right) \\
0 & 0 & 0
\end{array}
\right\vert
\begin{array}
[c]{ccc}%
0 & 0 & 0\\
I & \Gamma_{2}^{-1}\left(  D^{-1}\right)  & 0
\end{array}
\right]  ,
\]
and the information-qubit matrix becomes%
\[
\left[  \left.
\begin{array}
[c]{ccc}%
0 & 0 & I\\
0 & 0 & 0
\end{array}
\right\vert
\begin{array}
[c]{ccc}%
0 & 0 & 0\\
0 & L^{T}\left(  D^{-1}\right)  \Gamma_{2}^{-1}\left(  D^{-1}\right)  & I
\end{array}
\right]  .
\]
Alice's part of the above stabilizer matrix is equivalent to the quantum check
matrix in (\ref{eq:ieid-desired-matrix}) by row operations (premultiplying the
second set of rows in the stabilizer by $\Gamma_{2}\left(  D\right)  $.) Bob
can therefore make stabilizer measurements that have finite weight and that
are equivalent to the desired stabilizer.

We now describe a method to decode the above encoded stabilizer and
information-qubit matrix so that the information qubits appear at the output
of the decoding circuit. Bob performs Hadamard gates on his first and third
sets of $c-s$ qubits, performs CNOT\ gates from the first set of qubits to the
third set of qubits corresponding to the entries in $L\left(  D\right)  $, and
performs the Hadamard gates again. The stabilizer becomes%
\begin{equation}
\left[  \left.
\begin{array}
[c]{ccc}%
I & \Gamma_{2}\left(  D\right)  & 0\\
0 & 0 & 0
\end{array}
\right\vert
\begin{array}
[c]{ccc}%
0 & 0 & 0\\
I & \Gamma_{2}^{-1}\left(  D^{-1}\right)  & 0
\end{array}
\right]  , \label{eq:decoding-iefd}%
\end{equation}
and the information-qubit matrix becomes%
\[
\left[  \left.
\begin{array}
[c]{ccc}%
0 & 0 & I\\
0 & 0 & 0
\end{array}
\right\vert
\begin{array}
[c]{ccc}%
0 & 0 & 0\\
L^{T}\left(  D^{-1}\right)  & L^{T}\left(  D^{-1}\right)  \Gamma_{2}%
^{-1}\left(  D^{-1}\right)  & I
\end{array}
\right]  .
\]
Bob finishes decoding at this point because we can equivalently express the
information-qubit matrix as follows%
\[
\left[  \left.
\begin{array}
[c]{ccc}%
0 & 0 & I\\
0 & 0 & 0
\end{array}
\right\vert
\begin{array}
[c]{ccc}%
0 & 0 & 0\\
0 & 0 & I
\end{array}
\right]  ,
\]
by multiplying the last $c-s$ rows of the stabilizer by $L^{T}\left(
D^{-1}\right)  $ and adding to the last $c-s$ rows of the information-qubit matrix.

The overall procedure for encoding is to begin with a set of $c$ ebits,
$2\left(  n-c\right)  -k_{1}-k_{2}$ ancilla qubits, and $k_{1}+k_{2}-n+c$
information qubits. Alice performs the infinite-depth operations detailed in the
paragraph with\ (\ref{eq:ieid-desired-matrix}) for $c-s$ of the ebits. Alice then
performs the finite-depth operations detailed in the proofs of this lemma and
Lemma~\ref{lemma:general-ops} in reverse order. The resulting stabilizer has
equivalent error-correcting properties to the quantum check matrix in
(\ref{eq:orig-gens}).

The receiver decodes by first performing all of the finite-depth operations in
the encoding circuit in reverse order. The receiver then decodes the
infinite-depth operations by the procedure listed in the paragraph with
(\ref{eq:decoding-iefd}) so that the original $k_{1}+k_{2}-n+c$ information
qubits per frame are available for processing at the receiving end.%
\end{proof}%

\subsection{Special Case with Coherent Teleportation Decoding}

\label{sec:coh-tele-EAQCCs}We now detail a special case of the above codes in
this final section. These codes are interesting because the information qubits
teleport coherently to other physical qubits when encoding and decoding is complete.

\begin{lemma}
\label{lemma:iefd}Suppose that the Smith form of $F\left(  D\right)  $ in
(\ref{eq:last-eq-fefd}) is%
\[
F\left(  D\right)  =A_{F}\left(  D\right)
\begin{bmatrix}
\Gamma_{F}\left(  D\right)  & 0
\end{bmatrix}
B_{F}\left(  D\right)  ,
\]
where $A_{F}\left(  D\right)  $ is $\left(  n-k_{1}\right)  \times\left(
n-k_{1}\right)  $, $\Gamma_{F}\left(  D\right)  $ is an $\left(
n-k_{1}\right)  \times\left(  n-k_{1}\right)  $\ diagonal matrix whose entries
are powers of $D$, and $B_{F}\left(  D\right)  $ is $k_{2}\times k_{2}$. Then
the resulting entanglement-assisted code admits an encoding circuit with both
infinite-depth and finite-depth operations and admits a decoding circuit with
finite-depth operations only. The information qubits also teleport coherently
to other physical qubits for this special case of codes.
\end{lemma}

%

\begin{proof}%
We perform all the operations in Lemma~\ref{lemma:general-ops} to obtain the
quantum check matrix in (\ref{eq:last-eq-fefd}). Then perform the row
operations in $A_{F}^{-1}\left(  D\right)  $ and the column operations in
$B_{F}^{-1}\left(  D\right)  $. The quantum check matrix becomes%
\[
\left[  \left.
\begin{array}
[c]{ccc}%
E^{\prime}\left(  D\right)  & \Gamma_{F}\left(  D\right)  & 0\\
0 & 0 & 0
\end{array}
\right\vert
\begin{array}
[c]{ccc}%
0 & 0 & 0\\
I & 0 & 0
\end{array}
\right]  ,
\]
where $E^{\prime}\left(  D\right)  =A_{F}^{-1}\left(  D\right)  E\left(
D\right)  $. The Smith form of $E^{\prime}\left(  D\right)  $ is%
\[
E^{\prime}\left(  D\right)  =A_{E^{\prime}}\left(  D\right)
\begin{bmatrix}
\Gamma_{1}\left(  D\right)  & 0 & 0\\
0 & \Gamma_{2}\left(  D\right)  & 0\\
0 & 0 & 0
\end{bmatrix}
B_{E^{\prime}}\left(  D\right)  ,
\]
where $A_{E^{\prime}}\left(  D\right)  $ is $\left(  n-k_{1}\right)
\times\left(  n-k_{1}\right)  $, $\Gamma_{1}\left(  D\right)  $ is an $s\times
s$ diagonal matrix whose entries are powers of $D$, $\Gamma_{2}\left(
D\right)  $ is a $\left(  c-s\right)  \times\left(  c-s\right)  $ diagonal
matrix whose entries are arbitrary polynomials, and $B_{E^{\prime}}\left(
D\right)  $ is $\left(  n-k_{2}\right)  \times\left(  n-k_{2}\right)  $.

Now perform the row operations in $A_{E^{\prime}}^{-1}\left(  D\right)  $ and
the column operations in $B_{E^{\prime}}^{-1}\left(  D\right)  $. It is
possible to counteract the effect of the row operations on $\Gamma_{F}\left(
D\right)  $ by performing column operations, and it is possible to counteract
the effect of the column operations on the identity matrix in the
\textquotedblleft X\textquotedblright\ matrix by performing row operations.
The quantum check matrix becomes%
\[
\left[  \left.
\begin{array}
[c]{ccccccc}%
\Gamma_{1}(D) & 0 & 0 & \Gamma_{1}^{\prime}(D) & 0 & 0 & 0\\
0 & \Gamma_{2}(D) & 0 & 0 & \Gamma_{2}^{\prime}(D) & 0 & 0\\
0 & 0 & 0 & 0 & 0 & \Gamma_{3}^{\prime}(D) & 0\\
0 & 0 & 0 & 0 & 0 & 0 & 0
\end{array}
\right\vert
\begin{array}
[c]{ccc}%
0 & 0 & 0\\
0 & 0 & 0\\
0 & 0 & 0\\
I & 0 & 0
\end{array}
\right]  ,
\]
where $\Gamma_{1}^{\prime}(D)$, $\Gamma_{2}^{\prime}(D)$, and $\Gamma
_{3}^{\prime}(D)$ represent the respective $s\times s$, $\left(  c-s\right)
\times\left(  c-s\right)  $, and $\left(  n-k_{1}-c\right)  \times\left(
n-k_{1}-c\right)  $ diagonal matrices resulting from counteracting the effect
of row operations $A_{E^{\prime}}^{-1}\left(  D\right)  $ on $\Gamma
_{F}\left(  D\right)  $.  We use Hadamard and CNOT\ gates to clear the entries
in $\Gamma_{1}^{\prime}\left(  D\right)  $. The quantum check matrix
becomes%
\[
\left[  \left.
\begin{array}
[c]{ccccccc}%
\Gamma_{1}(D) & 0 & 0 & 0 & 0 & 0 & 0\\
0 & \Gamma_{2}(D) & 0 & 0 & \Gamma_{2}^{\prime}(D) & 0 & 0\\
0 & 0 & 0 & 0 & 0 & \Gamma_{3}^{\prime}(D) & 0\\
0 & 0 & 0 & 0 & 0 & 0 & 0
\end{array}
\right\vert
\begin{array}
[c]{ccc}%
0 & 0 & 0\\
0 & 0 & 0\\
0 & 0 & 0\\
I & 0 & 0
\end{array}
\right]  .
\]
The first $s$ rows with $\Gamma_{1}(D)$ and rows $n-k_{1}-c+1,\ldots,n-k_{1}-c+s$
with the identity matrix on the \textquotedblleft X\textquotedblright\ side
stabilize a set of $s$ ebits. The $n-k_{1}-c$ rows with $\Gamma_{3}^{\prime}(D)$
and the $n-k_{2}-c$ rows with identity in the \textquotedblleft
X\textquotedblright\ matrix stabilize a set of $2\left(  n-c\right)
-k_{1}-k_{2}$\ ancilla qubits (up to Hadamard gates). The $s$ and
$k_{2}-n+k_{1}$ columns with zeros in both the \textquotedblleft
Z\textquotedblright\ and \textquotedblleft X\textquotedblright\ matrices
correspond to information qubits. The decomposition of these rows is now complete.

We need to finish processing the $c-s$ rows with $\Gamma_{2}\left(  D\right)
$ and $\Gamma_{2}^{\prime}\left(  D\right)  $ as entries and the $c-s$ rows
of the identity in the \textquotedblleft X\textquotedblright\ matrix. We
construct a submatrix of the above quantum check matrix:%
\begin{equation}
\left[  \left.
\begin{array}
[c]{cc}%
\Gamma_{2}\left(  D\right)  & \Gamma_{2}^{\prime}\left(  D\right) \\
0 & 0
\end{array}
\right\vert
\begin{array}
[c]{cc}%
0 & 0\\
I & 0
\end{array}
\right]  . \label{eq:submatrix}%
\end{equation}

We describe a procedure to encode the above entries with $c-s$ ebits and $c-s$
information qubits using infinite-depth operations. Consider the following
stabilizer matrix%
\begin{equation}
\left[  \left.
\begin{array}
[c]{ccc}%
I & I & 0\\
0 & 0 & 0
\end{array}
\right\vert
\begin{array}
[c]{ccc}%
0 & 0 & 0\\
I & I & 0
\end{array}
\right]  , \label{eq:initial-stab-submatrix}%
\end{equation}
where all identity and null matrices are $\left(  c-s\right)  \times\left(
c-s\right)  $. The above matrix stabilizes a set of $c-s$ ebits and $c-s$
information qubits. Bob's half of the ebits are the $c-s$ columns on the left
in both the \textquotedblleft Z\textquotedblright\ and \textquotedblleft
X\textquotedblright\ matrices and Alice's half are the next $c-s$ columns. We
also track the logical operators for the information qubits to verify that the
circuit encodes and decodes properly. The information-qubit matrix is as
follows%
\begin{equation}
\left[  \left.
\begin{array}
[c]{ccc}%
0 & 0 & I\\
0 & 0 & 0
\end{array}
\right\vert
\begin{array}
[c]{ccc}%
0 & 0 & 0\\
0 & 0 & I
\end{array}
\right]  , \label{eq:init-info-matrix-sub}%
\end{equation}
where all matrices are again $\left(  c-s\right)  \times\left(  c-s\right)  $.
Alice performs Hadamard gates on her first $c-s$ qubits and then performs CNOT
gates from her first $c-s$ qubits to her last $c-s$ qubits to transform
(\ref{eq:initial-stab-submatrix}) to the following stabilizer%
\[
\left[  \left.
\begin{array}
[c]{ccc}%
I & 0 & 0\\
0 & I & 0
\end{array}
\right\vert
\begin{array}
[c]{ccc}%
0 & I & \Gamma_{2}^{\prime}\left(  D\right) \\
I & 0 & 0
\end{array}
\right]  .
\]
The information-qubit matrix in (\ref{eq:init-info-matrix-sub})\ becomes%
\[
\left[  \left.
\begin{array}
[c]{ccc}%
0 & \Gamma_{2}^{\prime}\left(  D^{-1}\right)  & I\\
0 & 0 & 0
\end{array}
\right\vert
\begin{array}
[c]{ccc}%
0 & 0 & 0\\
0 & 0 & I
\end{array}
\right]  .
\]
Alice then performs infinite-depth operations on her last $c-s$ qubits. These
infinite-depth operations correspond to the elements of $\Gamma_{2}%
^{-1}\left(  D\right)  $. She finally performs Hadamard gates on her $2\left(
c-s\right)  $ qubits. The stabilizer becomes%
\begin{equation}
\left[  \left.
\begin{array}
[c]{ccc}%
I & I & \Gamma_{2}^{-1}\left(  D\right)  \Gamma_{2}^{\prime}\left(
D\right) \\
0 & 0 & 0
\end{array}
\right\vert
\begin{array}
[c]{ccc}%
0 & 0 & 0\\
I & I & 0
\end{array}
\right]  , \label{eq:next-stab-submatrix}%
\end{equation}
and the information-qubit matrix becomes%
\begin{equation}
\left[  \left.
\begin{array}
[c]{ccc}%
0 & 0 & 0\\
0 & 0 & \Gamma_{2}^{-1}\left(  D\right)
\end{array}
\right\vert
\begin{array}
[c]{ccc}%
0 & \Gamma_{2}^{\prime}\left(  D^{-1}\right)  & \Gamma_{2}\left(
D^{-1}\right) \\
0 & 0 & 0
\end{array}
\right]  . \label{eq:final-encoded-info-sub}%
\end{equation}
The stabilizer in (\ref{eq:next-stab-submatrix})\ is equivalent to the
following stabilizer by row operations (premultiplying the first $c-s$ rows by
$\Gamma_{2}\left(  D\right)  $):%
\begin{equation}
\left[  \left.
\begin{array}
[c]{ccc}%
\Gamma_{2}\left(  D\right)  & \Gamma_{2}\left(  D\right)  & \Gamma
_{2}^{^{\prime}}\left(  D\right) \\
0 & 0 & 0
\end{array}
\right\vert
\begin{array}
[c]{ccc}%
0 & 0 & 0\\
I & I & 0
\end{array}
\right]  . \label{eq:augmented-submatrix}%
\end{equation}
The measurements that Bob performs have finite weight because the row
operations are multiplications of the rows by the arbitrary polynomials in
$\Gamma_{2}\left(  D\right)  $. Alice thus encodes a code equivalent to the
desired quantum check matrix in\ (\ref{eq:submatrix})\ using $c-s$ ebits and
$c-s$ information qubits.

We now discuss decoding the stabilizer in (\ref{eq:next-stab-submatrix}) and
information qubits. Bob performs CNOTs from the first $c-s$ qubits to the next
$c-s$ qubits. The stabilizer becomes%
\begin{equation}
\left[  \left.
\begin{array}
[c]{ccc}%
0 & I & \Gamma_{2}^{-1}\left(  D\right)  \Gamma_{2}^{\prime}\left(
D\right) \\
0 & 0 & 0
\end{array}
\right\vert
\begin{array}
[c]{ccc}%
0 & 0 & 0\\
I & 0 & 0
\end{array}
\right]  , \label{eq:first-decode-stab-submatrix}%
\end{equation}
and the information-qubit matrix does not change. Bob uses Hadamard and
finite-depth CNOT\ gates to multiply the last $c-s$ columns in the
\textquotedblleft Z\textquotedblright\ matrix by $\Gamma_{2}^{^{\prime}%
}\left(  D^{-1}\right)  \Gamma_{2}\left(  D\right)  $ and add the result to
the middle $c-s$\ columns. It is possible to use finite-depth operations
because the entries of $\Gamma_{2}^{\prime}\left(  D\right)  $ are all
powers of $D$ so that $\Gamma_{2}^{\prime}\left(  D^{-1}\right)
=\Gamma_{2}^{^{\prime}-1}\left(  D\right)  $. The stabilizer in
(\ref{eq:first-decode-stab-submatrix})\ becomes%
\[
\left[  \left.
\begin{array}
[c]{ccc}%
0 & 0 & \Gamma_{2}^{-1}\left(  D\right)  \Gamma_{2}^{\prime}\left(
D\right) \\
0 & 0 & 0
\end{array}
\right\vert
\begin{array}
[c]{ccc}%
0 & 0 & 0\\
I & 0 & 0
\end{array}
\right]  ,
\]
and the information-qubit matrix in (\ref{eq:final-encoded-info-sub})\ becomes%
\[
\left[  \left.
\begin{array}
[c]{ccc}%
0 & 0 & 0\\
0 & \Gamma_{2}^{\prime}\left(  D^{-1}\right)  & \Gamma_{2}^{-1}\left(
D\right)
\end{array}
\right\vert
\begin{array}
[c]{ccc}%
0 & \Gamma_{2}^{\prime}\left(  D^{-1}\right)  & 0\\
0 & 0 & 0
\end{array}
\right]  .
\]
We premultiply the first $c-s$ rows of the stabilizer by $\Gamma_{2}%
^{^{\prime}}\left(  D^{-1}\right)  $ and add the result to the second $c-s$
rows of the information-qubit matrix. These row operations from the stabilizer
to the information-qubit matrix result in the information-qubit matrix having
pure logical operators for the middle $c-s$ qubits. Perform Hadamard gates on
the second set of $c-s$ qubits. The resulting information-qubit matrix is as
follows%
\begin{equation}
\left[  \left.
\begin{array}
[c]{ccc}%
0 & \Gamma_{2}^{\prime}\left(  D^{-1}\right)  & 0\\
0 & 0 & 0
\end{array}
\right\vert
\begin{array}
[c]{ccc}%
0 & 0 & 0\\
0 & \Gamma_{2}^{\prime}\left(  D^{-1}\right)  & 0
\end{array}
\right]  , \label{eq:second-class-info-decoded}%
\end{equation}
so that the information qubits are available at the end of decoding.
Processing may delay or advance them with respect to their initial locations
because the matrix $\Gamma_{2}^{\prime}\left(  D^{-1}\right)  $ is diagonal
with powers of $D$. We can determine that the information qubits teleport
coherently\ from the last set of $c-s$ qubits to the second set of $c-s$
qubits in every frame by comparing (\ref{eq:second-class-info-decoded}) to
(\ref{eq:init-info-matrix-sub}).

The overall procedure for encoding is to begin with a set of $c$ ebits,
$2\left(  n-c\right)  -k_{1}-k_{2}$ ancilla qubits, and $k_{1}+k_{2}-n+c$
information qubits. Alice performs the infinite-depth operations detailed
in\ (\ref{eq:submatrix}-\ref{eq:augmented-submatrix}) for $c-s$ of the ebits.
Alice then performs the finite-depth operations detailed in the proofs of this
lemma and Lemma~\ref{lemma:general-ops} in reverse order. The resulting
stabilizer has equivalent error-correcting properties to the quantum check
matrix in (\ref{eq:orig-gens}).

The receiver decodes by first performing all of the finite-depth operations in
reverse order. The receiver then decodes the infinite-depth operations by the
procedure listed in (\ref{eq:first-decode-stab-submatrix}%
-\ref{eq:second-class-info-decoded}) so that the original $k_{1}+k_{2}-n+c$
information qubits per frame are available for processing at the receiving
end.%
\end{proof}%

\begin{example}
\label{ex:brun-example}Consider a classical convolutional code with the
following check matrix:%
\[
H\left(  D\right)  =\left[
\begin{array}
[c]{cc}%
1 & 1+D
\end{array}
\right]  .
\]
We can use the above check matrix in an entanglement-assisted quantum
convolutional code to correct for both bit flips and phase flips. We form the
following quantum check matrix:%
\begin{equation}
\left[  \left.
\begin{array}
[c]{cc}%
1 & 1+D\\
0 & 0
\end{array}
\right\vert
\begin{array}
[c]{cc}%
0 & 0\\
1 & 1+D
\end{array}
\right]  . \label{eq:brun-desired-stab}%
\end{equation}
We first perform some manipulations to put the above quantum check matrix into
a standard form. Perform a CNOT\ from qubit one to qubit two in the same frame
and in the next frame. The above matrix becomes%
\[
\left[  \left.
\begin{array}
[c]{cc}%
D^{-1}+1+D & 1+D\\
0 & 0
\end{array}
\right\vert
\begin{array}
[c]{cc}%
0 & 0\\
1 & 0
\end{array}
\right]  .
\]
Perform a Hadamard gate on qubits one and two. The matrix becomes%
\[
\left[  \left.
\begin{array}
[c]{cc}%
0 & 0\\
1 & 0
\end{array}
\right\vert
\begin{array}
[c]{cc}%
D^{-1}+1+D & 1+D\\
0 & 0
\end{array}
\right]  .
\]
Perform a CNOT\ from qubit one to qubit two. The matrix becomes%
\[
\left[  \left.
\begin{array}
[c]{cc}%
0 & 0\\
1 & 0
\end{array}
\right\vert
\begin{array}
[c]{cc}%
D^{-1}+1+D & D^{-1}\\
0 & 0
\end{array}
\right]  .
\]
Perform a row operation that delays the first row by $D$. Perform a Hadamard
on both qubits. The stabilizer becomes%
\[
\left[  \left.
\begin{array}
[c]{cc}%
1+D+D^{2} & 1\\
0 & 0
\end{array}
\right\vert
\begin{array}
[c]{cc}%
0 & 0\\
1 & 0
\end{array}
\right]  .
\]
The above matrix is now in standard form. The matrix $F\left(  D\right)  =1$
as in (\ref{eq:last-eq-fefd}) so that its only invariant factor is equal to
one. The code falls into the second class of entanglement-assisted quantum
convolutional codes. We begin encoding with one ebit and one information qubit
per frame. The stabilizer matrix for the unencoded stream is as follows:%
\[
\left[  \left.
\begin{array}
[c]{ccc}%
1 & 1 & 0\\
0 & 0 & 0
\end{array}
\right\vert
\begin{array}
[c]{ccc}%
0 & 0 & 0\\
1 & 1 & 0
\end{array}
\right]  ,
\]
and the information-qubit matrix is as follows:%
\[
\left[  \left.
\begin{array}
[c]{ccc}%
0 & 0 & 0\\
0 & 0 & 1
\end{array}
\right\vert
\begin{array}
[c]{ccc}%
0 & 0 & 1\\
0 & 0 & 0
\end{array}
\right]  .
\]
Perform a Hadamard on qubit two and a CNOT\ from qubit two to qubit three so
that the above stabilizer becomes%
\[
\left[  \left.
\begin{array}
[c]{ccc}%
1 & 0 & 0\\
0 & 1 & 0
\end{array}
\right\vert
\begin{array}
[c]{ccc}%
0 & 1 & 1\\
1 & 0 & 0
\end{array}
\right]  ,
\]
and the information-qubit matrix becomes%
\[
\left[  \left.
\begin{array}
[c]{ccc}%
0 & 0 & 0\\
0 & 1 & 1
\end{array}
\right\vert
\begin{array}
[c]{ccc}%
0 & 0 & 1\\
0 & 0 & 0
\end{array}
\right]  .
\]%
Perform an infinite-depth operation corresponding to the rational polynomial
$1/\left(  1+D+D^{2}\right)  $ on qubit three. Follow with a Hadamard gate on
qubits two and three. The stabilizer matrix becomes
\begin{equation}
\left[  \left.
\begin{array}
[c]{ccc}%
1 & 1 & 1/\left(  1+D+D^{2}\right) \\
0 & 0 & 0
\end{array}
\right\vert
\begin{array}
[c]{ccc}%
0 & 0 & 0\\
1 & 1 & 0
\end{array}
\right]  , \label{eq:brun-stab}%
\end{equation}
and the information-qubit matrix becomes%
\begin{equation}
\left[  \left.
\begin{array}
[c]{ccc}%
0 & 0 & 1/\left(  1+D+D^{2}\right) \\
0 & 0 & 0
\end{array}
\right\vert
\begin{array}
[c]{ccc}%
0 & 0 & 0\\
0 & 1 & 1+D^{-1}+D^{-2}%
\end{array}
\right]  . \label{eq:brun-info}%
\end{equation}
Perform the finite-depth operations above in reverse order so that the
stabilizer becomes%
\[
\left[  \left.
\begin{array}
[c]{ccc}%
D^{-1} & \frac{1}{1+D+D^{2}} & \frac{1+D}{1+D+D^{2}}\\
0 & 0 & 0
\end{array}
\right\vert
\begin{array}
[c]{ccc}%
0 & 0 & 0\\
1 & 1 & 1+D
\end{array}
\right]  ,
\]
and the information-qubit matrix becomes%
\[
\left[  \left.
\begin{array}
[c]{ccc}%
0 & \frac{D^{-1}+D^{-2}}{1+D+D^{2}} & \frac{1}{1+D+D^{2}}\\
0 & 0 & 0
\end{array}
\right\vert
\begin{array}
[c]{ccc}%
0 & 0 & 0\\
0 & D^{-1}+D^{-2} & D^{-1}%
\end{array}
\right]  .
\]
The above stabilizer is equivalent to the desired quantum check matrix in
(\ref{eq:brun-desired-stab}) by a row operation that multiplies its first row
by $1+D+D^{2}$.\newline\newline The receiver decodes by performing the
finite-depth encoding operations in reverse order and gets the stabilizer in
(\ref{eq:brun-stab})\ and the information-qubit matrix in (\ref{eq:brun-info}%
). The receiver performs a CNOT\ from qubit one to qubit two and follows with
a CNOT from qubit two to qubit three in the same frame, in an advanced frame,
and in a twice-advanced frame. Finally perform a Hadamard gate on qubits two
and three. The stabilizer becomes%
\[
\left[  \left.
\begin{array}
[c]{ccc}%
0 & 0 & 0\\
0 & 0 & 0
\end{array}
\right\vert
\begin{array}
[c]{ccc}%
0 & 0 & 1/\left(  1+D+D^{2}\right) \\
0 & 0 & 0
\end{array}
\right]  ,
\]
and the information-qubit matrix becomes%
\[
\left[  \left.
\begin{array}
[c]{ccc}%
0 & 0 & 0\\
0 & 1 & 0
\end{array}
\right\vert
\begin{array}
[c]{ccc}%
0 & 1 & 1/\left(  1+D+D^{2}\right) \\
0 & 0 & 0
\end{array}
\right]  .
\]
The receiver decodes the information qubits successfully because a row
operation from the first row of the stabilizer to the first row of the
information-qubit matrix gives the proper logical operators for the
information qubits. Figure~7 of Ref.~\cite{arx2007wildeEAQCC} details the above encoding
and decoding operations for this entanglement-assisted quantum convolutional code.
\end{example}

\subsection{Discussion}

\label{sec:2nd-class-assumpts}This second class of codes assumes that
noiseless encoding is available. We require this assumption because the
encoding circuit employs infinite-depth encoding operations.

If an error does occur during the encoding process, it can propagate
infinitely through the encoded qubit stream. The result of a single encoding
error can distort both the encoded quantum information, the syndromes that
result from measurements, and the final recovery operations based on the syndromes.

We may be able to relax the noiseless encoding assumption if nearly noiseless
encoding is available. The probability of an error would have to be negligible
in order to ensure that the probability for a catastrophic failure is
negligible. One way to lower the probability of an encoding error is to encode
first with a quantum block code and then further encode with our quantum
convolutional coding method. Many classical coding systems exploit this
technique, the most popular of which is a Reed-Solomon encoder followed by a
convolutional encoder.

\section{Closing Remarks}

This chapter develops the theory of entanglement-assisted quantum
convolutional coding for CSS\ codes. We show several methods for importing two
arbitrary classical binary convolutional codes for use in an
entanglement-assisted quantum convolutional code. Our methods outline
different ways for encoding and decoding our entanglement-assisted quantum
convolutional codes.

Our first class of codes employs only finite-depth operations in their
encoding and decoding procedures. These codes are the most useful in practice
because they do not have the risk of catastrophic error propagation. An error
that occurs during encoding, measurement, recovery, or decoding propagates
only to a finite number of neighboring qubits.

Our second class of codes uses infinite-depth operations during encoding. This
assumption is reasonable only if noiseless encoding is available. The method
of concatenated coding is one way to approach nearly noiseless encoding in practice.

We suggest several lines of inquiry from here. Our codes are not only useful
for quantum communication, but also should be useful for private classical
communication because of the well-known connection between a quantum channel
and private classical channel \cite{ieee2005dev}. It may make sense from a
practical standpoint to begin investigating the performance of our codes for
encoding secret classical messages. The commercial success of quantum key
distribution \cite{bb84}\ for the generation of a private shared secret key
motivates this investigation. It is also interesting to determine which
entanglement-assisted codes can correct for errors on the receiver's side.
Codes that possess this property will be more useful in practice.

We hope that our theory of entanglement-assisted quantum convolutional coding
provides a step in the direction of finding quantum codes that approach the
ultimate capacity of an entanglement-assisted quantum channel.

\chapter{Entanglement-Assisted Quantum Convolutional Coding: The General Case}

\label{chp:general-case}\begin{saying}
These patterns that anticommute,\\
They laugh and they think it's so cute,\\
Expand until fit,\\
Then ole Gram and Schmidt,\\
They'll line up with a grand salute.
\end{saying}In this chapter, we show how to encode and decode an
entanglement-assisted quantum convolutional code that does not necessarily
possess the CSS\ structure. The methods in this chapter represent a
significant extension of our work on CSS\ entanglement-assisted quantum
convolutional codes in the previous chapter.

In particular, we develop a set of techniques that make the commutation
relations for the generators of a general code the same as those of entangled
qubits (ebits) and ancilla qubits. This procedure first \textquotedblleft
expands\textquotedblright\ the check matrix for a quantum convolutional code
and applies an extended version of the symplectic Gram-Schmidt
orthogonalization procedure \cite{arx2006brun}\ that incorporates the binary
polynomials representing quantum convolutional codes. We then show how to
encode a stream of information qubits, ancilla qubits, and ebits to have the
error-correcting properties of the desired code. We follow by detailing the
decoding circuit that Bob employs at the receiving end of the channel. The
algorithms for encoding and decoding use similar techniques to those outlined
in the previous chapter. The encoding circuits incorporate both finite-depth
and infinite-depth\ operations as discussed in respective
Sections~\ref{sec:finite-depth-ops}\ and \ref{sec:infinite-depth-ops}\ and the
decoding circuits incorporate finite-depth operations only.

One benefit of the techniques developed in this chapter is that the quantum
code designer can produce an entanglement-assisted quantum convolutional code
from a classical quaternary convolutional code. More generally, quantum
convolutional code designers now have the freedom to design quantum
convolutional codes to have desirable error-correcting properties without
having to search for codes that satisfy a restrictive commutativity constraint.

We structure this chapter as follows. In Section \ref{sec:expand}, we present
an example that demonstrates how to expand the check matrix of a set of quantum
convolutional generators and then show how to expand an arbitrary check
matrix. We show in Section~\ref{sec:GS}\ how to compute the symplectic
Gram-Schmidt orthogonalization procedure for the example in
Section~\ref{sec:expand} and then generalize this procedure to work for an
arbitrary set of quantum convolutional generators. The Gram-Schmidt
orthogonalization technique reduces the quantum check matrix to have the same
commutation relations as a set of ebits and ancilla qubits. This technique is
essential for determining an encoding circuit that encodes a stream of
information qubits, ancilla qubits, and ebits. Section \ref{sec:encode-decode}%
\ gives the algorithms for computing the encoding and decoding circuits for an
arbitrary entanglement-assisted quantum convolutional code. We present a
detailed example in Section~\ref{sec:examples}\ that illustrates all of the
procedures in this chapter. We finish with a discussion of the practical
issues involved in using these entanglement-assisted quantum convolutional
codes and present some concluding remarks.

The focus of the next two sections is to elucidate techniques for reducing a
set of quantum convolutional generators to have the standard commutation
relations in (\ref{eq:comm-relations}), or equivalently, the standard
symplectic relations in (\ref{eq:standard-symp-form}). The difference between
the current techniques and the former techniques
\cite{arx2006brun,science2006brun,arx2007wildeCED,prep2007shaw} is that the
current techniques operate on generators that have a convolutional form rather
than a block form. Section~\ref{sec:expand}\ shows how to expand a set of
quantum convolutional generators to simplify their commutation relations.
Section~\ref{sec:GS} outlines a symplectic Gram-Schmidt orthogonalization
algorithm that uses binary polynomial operations to reduce the commutation
relations of the expanded generators to have the standard form
in\ (\ref{eq:standard-symp-form}).

\section{The Expansion of Quantum Convolutional Generators}

\label{sec:expand}We begin this section by demonstrating how to expand a
particular generator that we eventually incorporate in an
entanglement-assisted quantum convolutional code. We later generalize this
example and the expansion technique to an arbitrary set of generators. This
technique is important for determining how to utilize entanglement in the form
of ebits in a quantum convolutional code.

\subsection{Example of the Expansion}

Let us first suppose that we have one convolutional generator:%
\begin{equation}
\cdots\left\vert
\begin{array}
[c]{c}%
I
\end{array}
\right.  \left\vert
\begin{array}
[c]{c}%
X
\end{array}
\right\vert
\begin{array}
[c]{c}%
Z
\end{array}
\left\vert
\begin{array}
[c]{c}%
I
\end{array}
\right\vert \cdots\label{eq:pauli-conv-simple}%
\end{equation}
This generator does not yet represent a valid quantum convolutional code
because it anticommutes with a shift of itself to the left or to the right by
one qubit. We are not concerned with the error-correcting properties of this
generator but merely want to illustrate the technique of expanding it.

Let us for now consider a block version of this code that operates on eight
physical qubits with six total generators. The generators are as follows:%
\begin{equation}
\left.
\begin{array}
[c]{c}%
X\\
I\\
I\\
I\\
I\\
I
\end{array}
\right\vert
\begin{array}
[c]{c}%
Z\\
X\\
I\\
I\\
I\\
I
\end{array}
\left\vert
\begin{array}
[c]{c}%
I\\
Z\\
X\\
I\\
I\\
I
\end{array}
\right\vert
\begin{array}
[c]{c}%
I\\
I\\
Z\\
X\\
I\\
I
\end{array}
\left\vert
\begin{array}
[c]{c}%
I\\
I\\
I\\
Z\\
X\\
I
\end{array}
\right\vert
\begin{array}
[c]{c}%
I\\
I\\
I\\
I\\
Z\\
X
\end{array}
\left\vert
\begin{array}
[c]{c}%
I\\
I\\
I\\
I\\
I\\
Z
\end{array}
\right\vert
\begin{array}
[c]{c}%
I\\
I\\
I\\
I\\
I\\
I
\end{array}
\label{eq:pauli-block-simple}%
\end{equation}
We still use the vertical bars in the above block code to denote that the
frame size of the code is one. Observe that we can view the frame size of the
code as two without changing any of the error-correcting properties of the
block code:%
\begin{equation}
\left.
\begin{array}
[c]{cc}%
X & Z\\
I & X\\
I & I\\
I & I\\
I & I\\
I & I
\end{array}
\right\vert
\begin{array}
[c]{cc}%
I & I\\
Z & I\\
X & Z\\
I & X\\
I & I\\
I & I
\end{array}
\left\vert
\begin{array}
[c]{cc}%
I & I\\
I & I\\
I & I\\
Z & I\\
X & Z\\
I & X
\end{array}
\right\vert
\begin{array}
[c]{cc}%
I & I\\
I & I\\
I & I\\
I & I\\
I & I\\
Z & I
\end{array}
. \label{eq:two-expanded-code}%
\end{equation}
The frame is merely a way to organize our qubits so that we send one frame at
a time over the channel after the online encoding circuit has finished
processing them. We can extend the above block code with frame size two to
have a convolutional structure with the following two convolutional
generators:%
\[
\cdots\left\vert
\begin{array}
[c]{cc}%
I & I\\
I & I
\end{array}
\right.  \left\vert
\begin{array}
[c]{cc}%
X & Z\\
I & X
\end{array}
\right\vert
\begin{array}
[c]{cc}%
I & I\\
Z & I
\end{array}
\left\vert
\begin{array}
[c]{cc}%
I & I\\
I & I
\end{array}
\right\vert \cdots
\]
The above two generators with frame size two have equivalent error-correcting
properties to the original generator in (\ref{eq:pauli-conv-simple}) by the
arguments above. We say that we have expanded the original generator by a
factor of two or that the above code is a two-expanded version of the original
generator. We can also extend the original generator in
(\ref{eq:pauli-conv-simple})\ to have a frame size of three and we require
three convolutional generators so that they have equivalent error-correcting
properties to the original generator:%
\[
\cdots\left\vert
\begin{array}
[c]{ccc}%
I & I & I\\
I & I & I\\
I & I & I
\end{array}
\right.  \left\vert
\begin{array}
[c]{ccc}%
X & Z & I\\
I & X & Z\\
I & I & X
\end{array}
\right\vert
\begin{array}
[c]{ccc}%
I & I & I\\
I & I & I\\
Z & I & I
\end{array}
\left\vert
\begin{array}
[c]{ccc}%
I & I & I\\
I & I & I\\
I & I & I
\end{array}
\right\vert \cdots
\]

The representation of the original generator in (\ref{eq:pauli-conv-simple})
as a quantum check matrix in the polynomial formalism is as follows:%
\begin{equation}
g\left(  D\right)  =\left[  \left.
\begin{array}
[c]{c}%
D
\end{array}
\right\vert
\begin{array}
[c]{c}%
1
\end{array}
\right]  .\label{eq:poly-conv-simple}%
\end{equation}
The two-expanded check matrix is as follows:%
\begin{equation}
G_{2}\left(  D\right)  =\left[  \left.
\begin{array}
[c]{cc}%
0 & 1\\
D & 0
\end{array}
\right\vert
\begin{array}
[c]{cc}%
1 & 0\\
0 & 1
\end{array}
\right]  ,\label{eq:two-exp-check}%
\end{equation}
and the three-expanded check matrix is as follows:%
\[
G_{3}\left(  D\right)  =\left[  \left.
\begin{array}
[c]{ccc}%
0 & 1 & 0\\
0 & 0 & 1\\
D & 0 & 0
\end{array}
\right\vert
\begin{array}
[c]{ccc}%
1 & 0 & 0\\
0 & 1 & 0\\
0 & 0 & 1
\end{array}
\right]  .
\]

An alternative method for obtaining the polynomial representation of the
two-expanded check matrix consists of two steps. We first multiply $g\left(
D\right)  $ as follows:%
\[
G_{2}^{\prime}\left(  D\right)  =\left[
\begin{array}
[c]{c}%
1\\
D
\end{array}
\right]  \left[  g\left(  D\right)  \right]  \left[
\begin{array}
[c]{cccc}%
1 & D^{-1} & 0 & 0\\
0 & 0 & 1 & D^{-1}%
\end{array}
\right]  .
\]
We then \textquotedblleft plug in\textquotedblright\ the fractional delay
operator $D^{1/2}$ and apply the flooring operation $\left\lfloor
\cdot\right\rfloor $ that nulls the coefficients of any fractional power of
\thinspace$D$:%
\[
G_{2}\left(  D\right)  =\left\lfloor G_{2}^{\prime}\left(  D^{1/2}\right)
\right\rfloor .
\]
A similar technique applies to find the check matrix of the three-expanded
matrix. We first multiply $g\left(  D\right)  $ as follows:%
\[
G_{3}\left(  D\right)  =\left[
\begin{array}
[c]{c}%
1\\
D\\
D^{2}%
\end{array}
\right]  \left[  g\left(  D\right)  \right]  \left[
\begin{array}
[c]{cccccc}%
1 & \frac{1}{D} & \frac{1}{D^{2}} & 0 & 0 & 0\\
0 & 0 & 0 & 1 & \frac{1}{D} & \frac{1}{D^{2}}%
\end{array}
\right]
\]
We then \textquotedblleft plug in\textquotedblright\ the fractional delay
operator $D^{1/3}$ and apply the flooring operation $\left\lfloor
\cdot\right\rfloor $:%
\[
G_{3}\left(  D\right)  =\left\lfloor G_{3}^{\prime}\left(  D^{1/3}\right)
\right\rfloor
\]
We discuss the general method for expanding an arbitrary check matrix in the
next subsection.

\subsection{General Technique for Expansion}

We generalize the above example to determine how to expand an arbitrary
quantum check matrix by a factor of $l$. Suppose that we have an $n-k\times2n$
quantum check matrix $H\left(  D\right)  $ where%
\[
H\left(  D\right)  =\left[  \left.
\begin{array}
[c]{c}%
Z\left(  D\right)
\end{array}
\right\vert
\begin{array}
[c]{c}%
X\left(  D\right)
\end{array}
\right]  .
\]
Let $\mathbf{D}$ denote a diagonal matrix whose diagonal entries are the delay
operator $D$. We take the convention that $\mathbf{D}^{0}$ is the identity
matrix and $\mathbf{D}^{m}$ is the matrix $\mathbf{D}$ multiplied $m$ times so
that its diagonal entries are $D^{m}$. Let $R_{l}\left(  D\right)  $ and
$C_{l}\left(  D\right)  $ denote the following matrices:%
\begin{align*}
R_{l}\left(  D\right)   &  \equiv\left[
\begin{array}
[c]{cccc}%
\mathbf{D}^{0} & \mathbf{D}^{1} & \cdots & \mathbf{D}^{l-1}%
\end{array}
\right]  ^{T}\\
C_{l}\left(  D\right)   &  \equiv\left[
\begin{array}
[c]{cccccc}%
\mathbf{D}^{0} & \cdots & \mathbf{D}^{-\left(  l-1\right)  } & 0 & \cdots &
0\\
0 & \cdots & 0 & \mathbf{D}^{0} & \cdots & \mathbf{D}^{-\left(  l-1\right)  }%
\end{array}
\right]
\end{align*}
where the diagonal $\mathbf{D}$ matrices in $R_{l}\left(  D\right)  $ and
$C_{l}\left(  D\right)  $ have respective dimensions $n-k\times n-k$ and
$n\times n$. We premultiply and postmultiply the matrix $H\left(  D\right)  $
by respective matrices $R_{l}\left(  D\right)  $ and $C_{l}\left(  D\right)
$, evaluate the resulting matrix at a fractional power $1/l$ of the delay
operator $D$, and perform the flooring operation $\left\lfloor \cdot
\right\rfloor $ to null the coefficients of any fractional power of
\thinspace$D$. The $l$-expanded check matrix $H_{l}\left(  D\right)  $ is as
follows:%
\[
H_{l}\left(  D\right)  =\left\lfloor R_{l}\left(  D^{1/l}\right)  H\left(
D^{1/l}\right)  C_{l}\left(  D^{1/l}\right)  \right\rfloor .
\]
The $l$-expanded quantum check matrix $H_{l}\left(  D\right)  $ has equivalent
error-correcting properties to the original check matrix.

\section{Polynomial Symplectic Gram-Schmidt Orthogonalization Procedure}

\label{sec:GS}In general, a given set of generators may have complicated
commutation relations. We have to simplify the commutation relations so that
we can encode information qubits with the help of ancilla qubits and halves of
ebits shared with the receiver. In this section, we begin with an arbitrary
set of convolutional generators. We show how to determine a set of generators
with equivalent error-correcting properties and commutation relations that are
the same as those of halves of ebits and ancilla qubits. We first show an example of the
technique by illustrating it for Pauli sequences, for the
quantum check matrix, and with the shifted symplectic product matrix. We then
state a polynomial symplectic Gram-Schmidt orthogonalization algorithm that
performs this action for an arbitrary set of quantum convolutional generators.

\subsection{Example of the Procedure}

\subsubsection{Pauli Picture}

Let us consider again our example from the previous section. Specifically,
consider the respective expressions in (\ref{eq:pauli-conv-simple})\ and
(\ref{eq:two-expanded-code}) for the convolutional generator and the block
code. Recall that we can multiply the generators in a block code without
changing the error-correcting properties of the code \cite{book2000mikeandike}%
. Therefore, we can multiply the sixth generator in
(\ref{eq:two-expanded-code}) to the fourth. We can then multiply the modified
fourth to the second to yield the following equivalent code:%
\begin{equation}
\left.
\begin{array}
[c]{cc}%
X & Z\\
I & X\\
I & I\\
I & I\\
I & I\\
I & I
\end{array}
\right\vert
\begin{array}
[c]{cc}%
I & I\\
Z & X\\
X & Z\\
I & X\\
I & I\\
I & I
\end{array}
\left\vert
\begin{array}
[c]{cc}%
I & I\\
Z & X\\
I & I\\
Z & X\\
X & Z\\
I & X
\end{array}
\right\vert
\begin{array}
[c]{cc}%
I & I\\
Z & I\\
I & I\\
Z & I\\
I & I\\
Z & I
\end{array}
\label{eq:equiv-set}%
\end{equation}
We have manipulated the two-expanded matrix rather than the code in
(\ref{eq:pauli-block-simple})\ because the commutation relations of the above
code are equivalent to the commutation relations of the following operators:%
\[
\left.
\begin{array}
[c]{cc}%
I & Z\\
I & X\\
I & I\\
I & I\\
I & I\\
I & I
\end{array}
\right\vert
\begin{array}
[c]{cc}%
I & I\\
I & I\\
I & Z\\
I & X\\
I & I\\
I & I
\end{array}
\left\vert
\begin{array}
[c]{cc}%
I & I\\
I & I\\
I & I\\
I & I\\
I & Z\\
I & X
\end{array}
\right\vert
\begin{array}
[c]{cc}%
I & I\\
I & I\\
I & I\\
I & I\\
I & I\\
I & I
\end{array}
.
\]
We can use three ebits to encode the set of generators in (\ref{eq:equiv-set})
because they have the same commutation relations as the above operators and
the above operators correspond to halves of three ebits. We resolve the
anticommutation relations by using the following entanglement-assisted code:%
\[
\left.
\begin{array}
[c]{ccc}%
Z & X & Z\\
X & I & X\\
I & I & I\\
I & I & I\\
I & I & I\\
I & I & I
\end{array}
\right\vert
\begin{array}
[c]{ccc}%
I & I & I\\
I & Z & X\\
Z & X & Z\\
X & I & X\\
I & I & I\\
I & I & I
\end{array}
\left\vert
\begin{array}
[c]{ccc}%
I & I & I\\
I & Z & X\\
I & I & I\\
I & Z & X\\
Z & X & Z\\
X & I & X
\end{array}
\right\vert
\begin{array}
[c]{ccc}%
I & I & I\\
I & Z & I\\
I & I & I\\
I & Z & I\\
I & I & I\\
I & Z & I
\end{array}
\]
The convention above is that the first qubit of each frame belongs to Bob and
corresponds to half of an ebit. The second two qubits of each frame belong to
Alice. The overall code forms a commuting stabilizer so that it corresponds to
a valid quantum code. Bob could measure the above operators to diagnose errors
or he could measure the following operators that are equivalent by row
operations:%
\[
\left.
\begin{array}
[c]{ccc}%
Z & X & Z\\
X & I & X\\
I & I & I\\
I & I & I\\
I & I & I\\
I & I & I
\end{array}
\right\vert
\begin{array}
[c]{ccc}%
I & I & I\\
X & Z & I\\
Z & X & Z\\
X & I & X\\
I & I & I\\
I & I & I
\end{array}
\left\vert
\begin{array}
[c]{ccc}%
I & I & I\\
I & I & I\\
I & I & I\\
I & Z & I\\
Z & X & Z\\
X & I & X
\end{array}
\right\vert
\begin{array}
[c]{ccc}%
I & I & I\\
I & I & I\\
I & I & I\\
I & I & I\\
I & I & I\\
I & Z & I
\end{array}
\]
One can check that the operators corresponding to the second two qubits of
each frame are equivalent to the desired generators in
(\ref{eq:two-expanded-code}).

\subsubsection{Polynomial Picture}

Let us consider the convolutional generator in (\ref{eq:pauli-conv-simple}).
We now use the polynomial formalism because it is easier to perform the
manipulations for general codes in this picture rather than in the Pauli picture.

We extend the row operations from the above block code to the polynomial
picture. Each row operation multiplied the even-numbered generators by
themselves shifted by two qubits. Extending this operation to act on an
infinite Pauli sequence corresponds to multiplying the second generator in
(\ref{eq:two-exp-check})\ by the rational polynomial $1/\left(  1+D\right)  $.
Consider the two-expanded check matrix with the operation described above
applied to the second generator:%
\[
\left[
\begin{array}
[c]{c}%
g_{1}\left(  D\right)  \\
g_{2}\left(  D\right)
\end{array}
\right]  =\left[  \left.
\begin{array}
[c]{cc}%
0 & 1\\
\frac{D}{1+D} & 0
\end{array}
\right\vert
\begin{array}
[c]{cc}%
1 & 0\\
0 & \frac{1}{1+D}%
\end{array}
\right]  ,
\]
The shifted symplectic products $\left(  g_{1}\odot g_{1}\right)  \left(
D\right)  =0$, $\left(  g_{1}\odot g_{2}\right)  \left(  D\right)  =1$,
$\ $and $\left(  g_{2}\odot g_{2}\right)  \left(  D\right)  =0$ capture the
commutation relations of the resulting code for all frames. We can therefore
resolve this anticommutativity by prepending a column that corresponds to Bob
possessing half of an ebit:%
\[
\left[  \left.
\begin{array}
[c]{ccc}%
1 & 0 & 1\\
0 & \frac{D}{1+D} & 0
\end{array}
\right\vert
\begin{array}
[c]{ccc}%
0 & 1 & 0\\
1 & 0 & \frac{1}{1+D}%
\end{array}
\right]  .
\]
Bob possesses the qubit corresponding to column one of both the
\textquotedblleft Z\textquotedblright\ and \textquotedblleft
X\textquotedblright\ matrix and Alice possesses the two qubits corresponding
to the second and third columns. The code above has equivalent
error-correcting properties to those of the desired generators.

The generators in the above polynomial setting correspond to Pauli
sequences with infinite weight. Infinite-weight generators are undesirable
because Bob cannot measure an infinite number of qubits. There is a simple
solution to this problem and it is similar to what we did at the end of the
previous subsection. We multiply the second row of the above check matrix by
$1+D$ to obtain the following check matrix that has equivalent
error-correcting properties to the above one:%
\[
\left[  \left.
\begin{array}
[c]{ccc}%
1 & 0 & 1\\
0 & D & 0
\end{array}
\right\vert
\begin{array}
[c]{ccc}%
0 & 1 & 0\\
1+D & 0 & 1
\end{array}
\right]  .
\]
Bob can measure operators that have finite weight because the above check
matrix corresponds to Pauli sequences with finite weight.

\subsubsection{Shifted Symplectic Product Matrix Picture}

We can also determine the row operations that simplify the commutation
relations by looking only at the shifted symplectic product matrix. It is
easier to determine the row operations that resolve anticommutativity by
considering the shifted symplectic product matrix. We would like to perform
row operations to reduce the shifted symplectic product matrix to the standard
form in (\ref{eq:standard-symp-form}) so that it has commutation relations
equivalent to those of halves of ebits and ancilla qubits.

The shifted symplectic product matrix corresponding to the check matrix in
(\ref{eq:poly-conv-simple}) is a one-element matrix $\Omega_{g}\left(
D\right)  $ where%
\[
\Omega_{g}\left(  D\right)  =\left[  D^{-1}+D\right]  .
\]
It is clear that $\Omega_{g}\left(  D\right)  $ is not reducible by row
operations to the $2\times2$ matrix $J$ in (\ref{eq:J-matrix})\ because
$\Omega_{g}\left(  D\right)  $ is a one-element matrix. It is also not
reducible to the one-element null matrix $\left[  0\right]  $ because there is
no valid row operation that can zero out the term $D^{-1}+D$.

We therefore consider the two-expanded check matrix in (\ref{eq:two-exp-check}%
) to determine if we can reduce it with row operations to the standard form in
(\ref{eq:standard-symp-form}). The shifted symplectic product matrix
$\Omega_{G_{2}}\left(  D\right)  $ for the two-expanded check matrix
$G_{2}\left(  D\right)  $ is as follows:%
\[
\Omega_{G_{2}}\left(  D\right)  =\left[
\begin{array}
[c]{cc}%
0 & 1+D^{-1}\\
1+D & 0
\end{array}
\right]
\]
We can formulate the row operation of multiplying the second generator by
$1/\left(  1+D\right)  $ as a matrix $R\left(  D\right)  $ where%
\[
R\left(  D\right)  =\left[
\begin{array}
[c]{cc}%
1 & 0\\
0 & 1/\left(  1+D\right)
\end{array}
\right]  .
\]
The effect on the matrix $\Omega_{G_{2}}\left(  D\right)  $ is to change it to%
\[
R\left(  D\right)  \Omega_{G_{2}}\left(  D\right)  R^{T}\left(  D^{-1}\right)
=J
\]
as described in (\ref{eq:symp-row-op}). The above matrix $J$ has equivalent
commutation relations to half of an ebit. We can therefore use one ebit per
two qubits to encode this code.

In a later section, we show how to devise encoding circuits beginning from $c$
ebits per frame and $a$ ancilla qubits per frame and also give the algorithm
for their online decoding circuits. It is important for us right now to devise
a procedure to reduce the commutation relations of any check matrix to those
of ebits and ancilla qubits.

\subsection{The Procedure for General Codes}

We detail a polynomial version of the symplectic Gram-Schmidt
orthogonalization procedure in this section. It is a generalized version of
the procedure we developed for the above example. Before detailing the
algorithm, we first prove a lemma that shows how to determine the shifted
symplectic product matrix for an $l$-expanded version of the generators by
starting from the shifted symplectic product matrix of the original generators.

\subsubsection{The Shifted Symplectic Product Matrix for an $l$-Expanded Code}

\begin{lemma}
Suppose the shifted symplectic product matrix $\Omega\left(  D\right)  $ of a
given check matrix $H\left(  D\right)  $ is as follows:%
\[
\Omega\left(  D\right)  =Z\left(  D\right)  X^{T}\left(  D^{-1}\right)
+X\left(  D\right)  Z^{T}\left(  D^{-1}\right)  .
\]
The shifted symplectic product matrix $\Omega_{l}\left(  D\right)  $ of the
$l$-expanded check matrix $H_{l}\left(  D\right)  $ is as follows:%
\[
\Omega_{l}\left(  D\right)  =\left\lfloor R_{l}\left(  D^{1/l}\right)
\Omega\left(  D^{1/l}\right)  R_{l}^{T}\left(  D^{-1/l}\right)  \right\rfloor
\]
where the flooring operation $\left\lfloor \cdot\right\rfloor $ nulls the
coefficients of any fractional power of \thinspace$D$.
\end{lemma}

%

\begin{proof}%
Consider that the \textquotedblleft X\textquotedblright\ matrix $X_{l}\left(
D\right)  $ of the $l$-expanded check matrix $H_{l}\left(  D\right)  $ is as
follows:%
\[
X_{l}\left(  D\right)  =\left\lfloor R_{l}\left(  D^{1/l}\right)  X\left(
D^{1/l}\right)  C_{l}^{\prime}\left(  D^{1/l}\right)  \right\rfloor ,
\]
where
\[
C_{l}^{\prime}\left(  D^{1/l}\right)  \equiv\left[
\begin{array}
[c]{cccc}%
\mathbf{D}^{0} & \mathbf{D}^{-1} & \cdots & \mathbf{D}^{-\left(  l-1\right)  }%
\end{array}
\right]  ,
\]
and each diagonal $\mathbf{D}$ matrix is $n\times n$-dimensional. It is also
then true that the \textquotedblleft Z\textquotedblright\ matrix of
$H_{l}\left(  D\right)  $ is as follows:%
\[
Z_{l}\left(  D\right)  =\left\lfloor R_{l}\left(  D^{1/l}\right)  Z\left(
D^{1/l}\right)  C_{l}^{\prime}\left(  D^{1/l}\right)  \right\rfloor .
\]
The matrix transpose operation, the time reversal operation (substituting
$D^{-1}$ for $D$), and matrix addition are all invariant under the flooring
operation $\left\lfloor \cdot\right\rfloor $ for arbitrary matrices $M\left(
D\right)  $ and $N\left(  D\right)  $:%
\begin{align*}
\left\lfloor M^{T}\left(  D^{1/l}\right)  \right\rfloor  &  =\left\lfloor
M\left(  D^{1/l}\right)  \right\rfloor ^{T},\\
\left\lfloor M\left(  D^{-1/l}\right)  \right\rfloor  &  =\left.  \left\lfloor
M\left(  D^{1/l}\right)  \right\rfloor \right\vert _{D=D^{-1}},\\
\left\lfloor M\left(  D^{1/l}\right)  +N\left(  D^{1/l}\right)  \right\rfloor
&  =\left\lfloor M\left(  D^{1/l}\right)  \right\rfloor +\left\lfloor N\left(
D^{1/l}\right)  \right\rfloor .
\end{align*}
Additionally, the following property holds for two arbitrary binary
polynomials $f\left(  D\right)  $ and $g\left(  D\right)  $:%
\begin{equation}
\left\lfloor f\left(  D^{1/l}\right)  g\left(  D^{1/l}\right)  \right\rfloor
=\sum_{i=0}^{l-1}\left\lfloor D^{-i/l}f\left(  D^{1/l}\right)  \right\rfloor
\left\lfloor D^{i/l}g\left(  D^{1/l}\right)  \right\rfloor
.\label{eq:mult-prop-floor}%
\end{equation}
Now consider the product $X_{l}\left(  D\right)  Z_{l}^{T}\left(
D^{-1}\right)  $:%
\begin{align*}
X_{l}\left(  D\right)  Z_{l}^{T}\left(  D^{-1}\right)   &  =\left\lfloor
R_{l}\left(  D^{1/l}\right)  X\left(  D^{1/l}\right)  C_{l}^{\prime}\left(
D^{1/l}\right)  \right\rfloor \left.  \left(  \left\lfloor R_{l}\left(
D^{1/l}\right)  Z\left(  D^{1/l}\right)  C_{l}^{\prime}\left(  D^{1/l}\right)
\right\rfloor \right)  ^{T}\right\vert _{D=D^{-1}}\\
&  =\left\lfloor R_{l}\left(  D^{1/l}\right)  X\left(  D^{1/l}\right)
C_{l}^{\prime}\left(  D^{1/l}\right)  \right\rfloor \left\lfloor C_{l}^{\prime
T}\left(  D^{-1/l}\right)  Z^{T}\left(  D^{-1/l}\right)  R_{l}^{T}\left(
D^{-1/l}\right)  \right\rfloor \\
&  =\sum_{i=0}^{l-1}\left\lfloor D^{-i/l}R_{l}\left(  D^{1/l}\right)  X\left(
D^{1/l}\right)  \right\rfloor \left\lfloor D^{i/l}Z^{T}\left(  D^{-1/l}%
\right)  R_{l}^{T}\left(  D^{-1/l}\right)  \right\rfloor \\
&  =\left\lfloor R_{l}\left(  D^{1/l}\right)  X\left(  D^{1/l}\right)
Z^{T}\left(  D^{-1/l}\right)  R_{l}^{T}\left(  D^{-1/l}\right)  \right\rfloor
,
\end{align*}
where the second line uses the invariance of the flooring operation with
respect to matrix transposition and time reversal, the third line expands the
matrix multiplications using the matrix $C_{l}^{\prime}\left(  D\right)  $
defined above, and the last line uses the matrix generalization of the
multiplication property defined in (\ref{eq:mult-prop-floor}). Our final step
is to use the invariance of the flooring operation with respect to matrix
addition:%
\begin{align*}
\Omega_{l}\left(  D\right)   &  =X_{l}\left(  D\right)  Z_{l}^{T}\left(
D^{-1}\right)  +Z_{l}\left(  D\right)  X_{l}^{T}\left(  D^{-1}\right)  \\
&  =\left\lfloor R_{l}\left(  D^{\frac{1}{l}}\right)  X\left(  D^{\frac{1}{l}%
}\right)  Z^{T}\left(  D^{\frac{-1}{l}}\right)  R_{l}^{T}\left(  D^{\frac
{-1}{l}}\right)  \right\rfloor +\left\lfloor R_{l}\left(  D^{\frac{1}{l}%
}\right)  Z\left(  D^{\frac{1}{l}}\right)  X^{T}\left(  D^{\frac{-1}{l}%
}\right)  R_{l}^{T}\left(  D^{\frac{-1}{l}}\right)  \right\rfloor \\
&  =\left\lfloor R_{l}\left(  D^{\frac{1}{l}}\right)  X\left(  D^{\frac{1}{l}%
}\right)  Z^{T}\left(  D^{\frac{-1}{l}}\right)  R_{l}^{T}\left(  D^{\frac
{-1}{l}}\right)  +R_{l}\left(  D^{\frac{1}{l}}\right)  Z\left(  D^{\frac{1}%
{l}}\right)  X^{T}\left(  D^{\frac{-1}{l}}\right)  R_{l}^{T}\left(
D^{\frac{-1}{l}}\right)  \right\rfloor \\
&  =\left\lfloor R_{l}\left(  D^{1/l}\right)  \Omega\left(  D^{1/l}\right)
R_{l}^{T}\left(  D^{-1/l}\right)  \right\rfloor .
\end{align*}%
\end{proof}%

\subsubsection{The Gram-Schmidt Procedure}

We now present the polynomial symplectic Gram-Schmidt orthogonalization
procedure that reduces the commutation relations of a given set of
convolutional generators to have the standard form in
(\ref{eq:standard-symp-form}).

Consider the following $n-k\times2n$-dimensional quantum check matrix
$H\left(  D\right)  $:%
\[
H\left(  D\right)  =\left[  \left.
\begin{array}
[c]{c}%
Z\left(  D\right)
\end{array}
\right\vert
\begin{array}
[c]{c}%
X\left(  D\right)
\end{array}
\right]  .
\]
Label each row as $h_{i}\left(  D\right)  =\left[  \left.
\begin{array}
[c]{c}%
z_{i}\left(  D\right)
\end{array}
\right\vert
\begin{array}
[c]{c}%
x_{i}\left(  D\right)
\end{array}
\right]  $ for all $i\in\left\{  1,\ldots,n-k\right\}  $.

We state the Gram-Schmidt procedure in terms of its effect on the shifted
symplectic product matrix. It is easier to see how the algorithm proceeds by
observing the shifted symplectic product matrix rather than by tracking the
generators in the check matrix.

Let $l$ denote the amount by which we expand the check matrix $H\left(
D\right)  $. Suppose first that $l=1$ (we do not expand the check matrix). Let
us say that the check matrix has $r$ generators (we take a number different
from $n-k$ because the number of generators may not be equal to $n-k$ for
future iterations). There are the following possibilities:

\begin{enumerate}
\item There is a generator $h_{i}\left(  D\right)  $ such that $\left(
h_{i}\odot h_{j}\right)  \left(  D\right)  =0$ for all $j\in\left\{
1,\ldots,r\right\}  $. In this case, the generator is already decoupled from
all the others and corresponds to an ancilla qubit because an ancilla and it
share the same commutation relations. Swap $h_{i}\left(  D\right)  $ to be the
first row of the matrix so that it is $h_{1}\left(  D\right)  $. The shifted
symplectic product matrix then has the form:%
\[%
\begin{bmatrix}
0 & 0 & \cdots & 0\\
0 & h_{2,2} & \cdots & h_{2,r}\\
\vdots & \vdots & \ddots & \vdots\\
0 & h_{r,2} & \cdots & h_{r,r}%
\end{bmatrix}
=%
\begin{bmatrix}
0
\end{bmatrix}
\oplus%
\begin{bmatrix}
h_{2,2} & \cdots & h_{2,r}\\
\vdots & \ddots & \vdots\\
h_{r,2} & \cdots & h_{r,r}%
\end{bmatrix}
\]
where $\left[  0\right]  $ is the one-element zero matrix and we use the
shorthand $h_{i,j}=\left(  h_{i}\odot h_{j}\right)  \left(  D\right)  $. We
remove generator $h_{1}\left(  D\right)  $ from check matrix $H\left(
D\right)  $ and continue to step two below for the remaining generators in the matrix.

\item There are two generators $h_{i}\left(  D\right)  $ and $h_{j}\left(
D\right)  $ for which $\left(  h_{i}\odot h_{j}\right)  \left(  D\right)
=D^{m}$ for some integer $m$ and $\left(  h_{i}\odot h_{i}\right)  \left(
D\right)  =\left(  h_{j}\odot h_{j}\right)  \left(  D\right)  =0$. In this
case these generators correspond exactly to half of an ebit. Multiply
generator $h_{j}\left(  D\right)  $ by $D^{m}$. This row operation has the
effect of delaying (or advancing) the generator by an amount $m$ and changes
the shifted symplectic product to be $\left(  h_{i}\odot h_{j}\right)  \left(
D\right)  =1$. These two generators considered by themselves now have the
commutation relations of half of an ebit. Swap the generators $h_{i}\left(
D\right)  $ and $h_{j}\left(  D\right)  $ to be the first and second
respective rows of the check matrix $H\left(  D\right)  $. Call them
$h_{1}\left(  D\right)  $ and $h_{2}\left(  D\right)  $ respectively. The
shifted symplectic product matrix is then as follows:%
\[%
\begin{bmatrix}
0 & 1 & h_{1,3} & \cdots & h_{1,r}\\
1 & 0 & h_{2,3} & \cdots & h_{2,r}\\
h_{3,1} & h_{3,2} & h_{3,3} & \cdots & h_{3,r}\\
\vdots & \vdots & \vdots & \ddots & \vdots\\
h_{r,2} & h_{r,2} & h_{r,3} & \cdots & h_{r,r}%
\end{bmatrix}
.
\]
We use the following row operations to decouple the other generators from
these two generators:%
\[
h_{i}^{\prime}\left(  D\right)  \equiv h_{i}\left(  D\right)  +\left[\left(  h_{i}\odot
h_{2}\right)  \left(  D\right)\right]   h_{1}\left(  D\right)  + \left[\left(
h_{i}\odot h_{1}\right)  \left(  D\right) \right]  h_{2}\left(  D\right)  \text{
for all }i\in\left\{  3,\ldots,r\right\}.
\]
The shifted symplectic product matrix becomes as follows under these row
operations:%
\[%
\begin{bmatrix}
0 & 1 & 0 & \cdots & 0\\
1 & 0 & 0 & \cdots & 0\\
0 & 0 & h_{3,3}^{\prime} & \cdots & h_{3,r}^{\prime}\\
\vdots & \vdots & \vdots & \ddots & \vdots\\
0 & 0 & h_{r,3}^{\prime} & \cdots & h_{r,r}^{\prime}%
\end{bmatrix}
=
\begin{bmatrix}
J
\end{bmatrix}
\oplus%
\begin{bmatrix}
h_{3,3}^{\prime} & \cdots & h_{3,r}^{\prime}\\
\vdots & \ddots & \vdots\\
h_{r,3}^{\prime} & \cdots & h_{r,r}^{\prime}%
\end{bmatrix}
.
\]
The first two generators are now decoupled from the other generators and have
the commutation relations of half of an ebit. We remove the first two
generators from the check matrix so that it now consists of generators
$h_{3}^{\prime}\left(  D\right)  $, $\ldots$, $h_{r}^{\prime}\left(  D\right)
$. We check to see if the conditions at the beginning of the previous step or
this step hold for any other generators. If so, repeat the previous step or
this step on the remaining generators. If not, see if the conditions for step
three hold.

\item There are two generators $h_{i}\left(  D\right)  $ and $h_{j}\left(
D\right)  $ for which $\left(  h_{i}\odot h_{i}\right)  \left(  D\right)
=\left(  h_{j}\odot h_{j}\right)  \left(  D\right)  =0$ but $\left(
h_{i}\odot h_{j}\right)  \left(  D\right)  \neq D^{m}$ for all $m$ and
$\left(  h_{i}\odot h_{j}\right)  \left(  D\right)  \neq0$. Multiply generator
$h_{j}\left(  D\right)  $ by $1/\left(  h_{j}\odot h_{i}\right)  \left(
D\right)  $. Generator $h_{j}\left(  D\right)  $ becomes infinite weight
because $\left(  h_{j}\odot h_{i}\right)  \left(  D\right)  $ is a polynomial
with two or more powers of $D$ with non-zero coefficients. Now the shifted
symplectic product relations are as follows: $\left(  h_{i}\odot h_{i}\right)
\left(  D\right)  =\left(  h_{j}\odot h_{j}\right)  \left(  D\right)  =0$ and
$\left(  h_{i}\odot h_{j}\right)  \left(  D\right)  =1$. We handle this case
as we did the previous case after the two generators there had the commutation
relations of half of an ebit.

\item None of these conditions hold. In this case, we stop this iteration of
the algorithm and expand the check matrix by the next factor $l:=l+1$ and
repeat the above steps.
\end{enumerate}

We have not proven that this procedure converges on all codes; however, it
does converge on all the codes we have tried. We conjecture that this
procedure converges for all codes, but even if this is true, in principle it
might require expansion to a large number of generators. A simple and
practical convergence condition is as follows. In practice, convolutional
codes do not act on an infinite stream of qubits but instead act on a finite
number of qubits. It may be that there are codes for which this procedure
either does not converge or must be expanded until the frame size of the
expanded code exceeds the number of qubits that the code acts on. In this
case, we would not employ this procedure, and instead would treat the code as
a block code, where we could employ the methods from
Refs.~\cite{science2006brun,arx2006brun} for encoding and decoding. It is
unclear if this practical convergence condition will ever really be necessary.
This procedure does converge for a large number of useful codes, so that the
frame size of the expanded code is much less than the number of qubits that
the code acts on, and we have not found an example where this procedure fails.

\section{Encoding and Decoding Circuits}

\label{sec:encode-decode}This section proves the main theorem of this chapter.
The theorem assumes that we have already processed an arbitrary check matrix
with the Gram-Schmidt algorithm and that the shifted symplectic product matrix
corresponding to the processed check matrix has the standard form in
(\ref{eq:standard-symp-form}). The theorem shows how to encode a set of
information qubits, ancilla qubits, and halves of ebits into a code that has
equivalent error-correcting properties to those of a desired set of
convolutional generators. The theorem uses both finite-depth and
infinite-depth operations in the encoding circuit and finite-depth operations
in the decoding circuit.

\begin{theorem}
\label{thm:main}Suppose we have a set of quantum convolutional generators in
the $n-k\times2n$-dimensional matrix $H\left(  D\right)  $ where%
\[
H\left(  D\right)  =\left[  \left.
\begin{array}
[c]{c}%
Z\left(  D\right)
\end{array}
\right\vert
\begin{array}
[c]{c}%
X\left(  D\right)
\end{array}
\right]  .
\]
Its shifted symplectic product matrix $\Omega\left(  D\right)  $ is as
follows:%
\[
\Omega\left(  D\right)  =Z\left(  D\right)  X^{T}\left(  D^{-1}\right)
+X\left(  D\right)  Z^{T}\left(  D^{-1}\right)  .
\]
Suppose check matrix $H\left(  D\right)  $ is the matrix resulting from
processing with the polynomial symplectic Gram-Schmidt orthogonalization
procedure. Therefore, $\Omega\left(  D\right)  $ has the standard form in
(\ref{eq:standard-symp-form}) with parameters $c$ and $a=n-k-2c$. Then there
exists an online encoding circuit for the code that uses finite-depth and
infinite-depth operations in the shift-invariant Clifford group and there
exists an online decoding circuit for the code that uses finite-depth
operations. The code encodes $k+c$ information qubits per frame with the help
of $c$ ebits and $a=n-k-2c$ ancilla qubits.
\end{theorem}

%

\begin{proof}%
We prove the theorem by giving an algorithm to compute both the encoding
circuit and the decoding circuit. The shifted symplectic product matrix
$\Omega\left(  D\right)  $ for check matrix $H\left(  D\right)  $ is in
standard form so that the first $2c$ rows have the commutation relations of
$c$ halves of ebits and the last $a$ rows have the commutation relations of
$a$ ancilla qubits. Perform the algorithm outlined in
Refs.~\cite{ieee2006grassl,isit2006grassl}\ on the last $a$ generators that
correspond to the ancilla qubits. The algorithm uses finite-depth CNOT\ gates,
Hadamard gates, and phase gates. The resulting check matrix has the following
form:%
\begin{equation}
\left[  \left.
\begin{array}
[c]{cc}%
Z^{\prime\prime}\left(  D\right)  & Z^{\prime}\left(  D\right) \\
\Gamma\left(  D\right)  & 0
\end{array}
\right\vert
\begin{array}
[c]{cc}%
0 & X^{\prime}\left(  D\right) \\
0 & 0
\end{array}
\right]  , \label{eq:start-finite-depth-encode}%
\end{equation}
where the matrix $Z^{\prime\prime}\left(  D\right)  $ and the null matrix at
the top left of the \textquotedblleft X\textquotedblright\ matrix each have
dimension $2c\times a$, the matrices $Z^{\prime}\left(  D\right)  $ and
$X^{\prime}\left(  D\right)  $ each have dimension $2c\times n-a$, and all
matrices in the second set of rows each have $a$ rows. The matrix
$\Gamma\left(  D\right)  $ may have entries that are rational polynomials. If
so, replace each of these entries with a \textquotedblleft1\textquotedblright%
\ so that the resulting matrix has the following form:%
\[
\left[  \left.
\begin{array}
[c]{cc}%
Z^{\prime\prime}\left(  D\right)  & Z^{\prime}\left(  D\right) \\
I & 0
\end{array}
\right\vert
\begin{array}
[c]{cc}%
0 & X^{\prime}\left(  D\right) \\
0 & 0
\end{array}
\right]  .
\]
This replacement is equivalent to taking a subcode of the original that has
equivalent error-correcting properties and rate~\cite{isit2006grassl}. We can
also think of the replacement merely as row operations with rational
polynomials. We then perform row operations from the last $a$ rows to the
first $2c$ rows to obtain the following check matrix:%
\[
\left[  \left.
\begin{array}
[c]{cc}%
0 & Z^{\prime}\left(  D\right) \\
I & 0
\end{array}
\right\vert
\begin{array}
[c]{cc}%
0 & X^{\prime}\left(  D\right) \\
0 & 0
\end{array}
\right]  .
\]
The shifted symplectic product matrix still has the standard form in
(\ref{eq:standard-symp-form}) because these last row operations do not change
its entries. We now focus exclusively on the first $2c$ rows because the
previous steps decoupled the last $a$ rows from the first $2c$ rows. Consider
the following submatrix:%
\[
H^{\prime}\left(  D\right)  =\left[  \left.
\begin{array}
[c]{c}%
Z\left(  D\right)
\end{array}
\right\vert
\begin{array}
[c]{c}%
X\left(  D\right)
\end{array}
\right]  ,
\]
where we have reset variable labels so that $Z\left(  D\right)  =Z^{\prime
}\left(  D\right)  $ and $X\left(  D\right)  =X^{\prime}\left(  D\right)  $.
Perform row permutations on the above matrix so that the shifted symplectic
product matrix for $H^{\prime}\left(  D\right)  $ changes from%
\[%
{\displaystyle\bigoplus\limits_{i=1}^{c}}
J,
\]
to become%
\begin{equation}
\left[
\begin{array}
[c]{cc}%
0 & I\\
I & 0
\end{array}
\right]  , \label{eq:alt-symp-rels}%
\end{equation}
where each identity and null matrix in the above matrix are $c\times
c$-dimensional. We can employ the algorithm from Ref.~\cite{ieee2006grassl} on
the first $c$ generators because the first $c$ rows of the resulting check
matrix form a commuting set (we do not use the row operations in that
algorithm). The algorithm employs finite-depth CNOT\ gates, Hadamard gates,
and phase gates and reduces the check matrix to have the following form:%
\[
\left[  \left.
\begin{array}
[c]{cc}%
0 & 0\\
U\left(  D\right)  & Z_{2}\left(  D\right)
\end{array}
\right\vert
\begin{array}
[c]{cc}%
L\left(  D\right)  & 0\\
X_{1}\left(  D\right)  & X_{2}\left(  D\right)
\end{array}
\right]  ,
\]
where $L\left(  D\right)  $ is a $c\times c$ lower triangular matrix and
$U\left(  D\right)  $ is a $c\times c$ upper triangular matrix. The
$i^{\text{th}}$ diagonal entry $u_{ii}\left(  D\right)  $ of $U\left(
D\right)  $ is equal to $1/l_{ii}\left(  D^{-1}\right)  $ where $l_{ii}\left(
D\right)  $ is the $i^{\text{th}}$ diagonal entry of $L\left(  D\right)  $.
This relationship holds because of the shifted symplectic relations in
(\ref{eq:alt-symp-rels}) and because gates in the shift-invariant Clifford
group do not affect the shifted symplectic relations. We now employ several
row operations whose net effect is to preserve the shifted symplectic
relations in (\ref{eq:alt-symp-rels})---we can therefore include them as a
part of the original polynomial symplectic Gram-Schmidt orthogonalization
procedure. Multiply row $i$ of the above check matrix by $1/l_{ii}\left(
D\right)  $ and multiply row $i+c$ by $1/u_{ii}\left(  D\right)  $ for all
$i\in\left\{  1,\ldots,c\right\}  $. Then use row operations to cancel all the
off-diagonal entries in both $L\left(  D\right)  $ and $U\left(  D\right)  $.
The resulting check matrix has the following form:%
\[
\left[  \left.
\begin{array}
[c]{cc}%
0 & 0\\
I & Z_{2}^{\prime}\left(  D\right)
\end{array}
\right\vert
\begin{array}
[c]{cc}%
I & 0\\
X_{1}^{\prime}\left(  D\right)  & X_{2}^{\prime}\left(  D\right)
\end{array}
\right]  ,
\]
where the primed matrices result from all the row operations. One can check
that the shifted symplectic relations of the above matrix are equivalent to
those in (\ref{eq:alt-symp-rels}). Perform Hadamard gates on the first $c$
qubits. The check matrix becomes%
\begin{equation}
\left[  \left.
\begin{array}
[c]{cc}%
I & 0\\
X_{1}^{\prime}\left(  D\right)  & Z_{2}^{\prime}\left(  D\right)
\end{array}
\right\vert
\begin{array}
[c]{cc}%
0 & 0\\
I & X_{2}^{\prime}\left(  D\right)
\end{array}
\right]  . \label{eq:desired-check-matrix}%
\end{equation}
We show how to encode the above matrix starting from $c$ ebits and $k+c$
information qubits. The following matrix stabilizes a set of $c$ ebits:%
\begin{equation}
\left[  \left.
\begin{array}
[c]{ccc}%
I & I & 0\\
0 & 0 & 0
\end{array}
\right\vert
\begin{array}
[c]{ccc}%
0 & 0 & 0\\
I & I & 0
\end{array}
\right]  , \label{eq:start-encode}%
\end{equation}
where each identity matrix is $c\times c$ and the last column of zeros in each
matrix is $c\times\left(  k+c\right)  $. The receiver Bob possesses the first
$c$ qubits and the sender Alice possesses the last $k+2c$ qubits. The
following matrix is the information-qubit:%
\[
\left[  \left.
\begin{array}
[c]{ccc}%
0 & 0 & I\\
0 & 0 & 0
\end{array}
\right\vert
\begin{array}
[c]{ccc}%
0 & 0 & 0\\
0 & 0 & I
\end{array}
\right]
\]
where each identity matrix is $\left(  k+c\right)  \times\left(  k+c\right)  $
and each column of zeros is $\left(  k+c\right)  \times c$. It is important to
track the information-qubit matrix throughout encoding and decoding so that we
can determine at the end of the process if we have truly decoded the
information qubits. Perform finite-depth CNOT\ operations from the first $c$
ebits to the last $k+c$ qubits to encode the numerators of the entries in
matrix $Z_{2}^{\prime}\left(  D\right)  $. Let $Z_{2,N}^{\prime}\left(
D\right)  $ denote this matrix of the numerators of the entries in
$Z_{2}^{\prime}\left(  D\right)  $. The stabilizer matrix becomes%
\[
\left[  \left.
\begin{array}
[c]{ccc}%
I & I & 0\\
0 & 0 & 0
\end{array}
\right\vert
\begin{array}
[c]{ccc}%
0 & 0 & 0\\
I & I & Z_{2,N}^{\prime}\left(  D\right)
\end{array}
\right]  ,
\]
and the information-qubit matrix becomes%
\[
\left[  \left.
\begin{array}
[c]{ccc}%
0 & Z_{2,N}^{\prime\prime}\left(  D\right)  & I\\
0 & 0 & 0
\end{array}
\right\vert
\begin{array}
[c]{ccc}%
0 & 0 & 0\\
0 & 0 & I
\end{array}
\right]  ,
\]
where $Z_{2,N}^{\prime\prime}\left(  D\right)  $ is the matrix that results on
the \textquotedblleft Z\textquotedblright\ side after performing the CNOT
operations corresponding to the entries in $Z_{2,N}^{\prime}\left(  D\right)
$. Perform Hadamard gates on the last $k+c$ qubits. The stabilizer matrix
becomes%
\[
\left[  \left.
\begin{array}
[c]{ccc}%
I & I & 0\\
0 & 0 & Z_{2,N}^{\prime}\left(  D\right)
\end{array}
\right\vert
\begin{array}
[c]{ccc}%
0 & 0 & 0\\
I & I & 0
\end{array}
\right]  ,
\]
and the information-qubit matrix becomes%
\[
\left[  \left.
\begin{array}
[c]{ccc}%
0 & Z_{2,N}^{\prime\prime}\left(  D\right)  & 0\\
0 & 0 & I
\end{array}
\right\vert
\begin{array}
[c]{ccc}%
0 & 0 & I\\
0 & 0 & 0
\end{array}
\right]  .
\]
Let $X_{2,N}^{\prime}\left(  D\right)  $ denote the matrix whose entries are
the numerators of the entries in $X_{2}^{\prime}\left(  D\right)  $. Perform
CNOT\ gates from the first $c$ qubits to the last $k+c$ qubits corresponding
to the entries in $X_{2,N}^{\prime}\left(  D\right)  $. The stabilizer matrix
becomes%
\[
\left[  \left.
\begin{array}
[c]{ccc}%
I & I & 0\\
0 & A\left(  D\right)  & Z_{2,N}^{\prime}\left(  D\right)
\end{array}
\right\vert
\begin{array}
[c]{ccc}%
0 & 0 & 0\\
I & I & X_{2,N}^{\prime}\left(  D\right)
\end{array}
\right]  ,
\]
where $A\left(  D\right)  =Z_{2,N}^{\prime}\left(  D\right)  X_{2,N}%
^{\prime\prime}\left(  D\right)  $ and $X_{2,N}^{\prime\prime}\left(
D\right)  $ is the matrix that results on the \textquotedblleft
Z\textquotedblright\ side after performing the CNOT operations on the
\textquotedblleft X\textquotedblright\ side corresponding to the entries in
$X_{2,N}^{\prime}\left(  D\right)  $. The information-qubit matrix becomes%
\[
\left[  \left.
\begin{array}
[c]{ccc}%
0 & Z_{2,N}^{\prime\prime}\left(  D\right)  & 0\\
0 & X_{2,N}^{\prime\prime}\left(  D\right)  & I
\end{array}
\right\vert
\begin{array}
[c]{ccc}%
0 & 0 & I\\
0 & 0 & 0
\end{array}
\right]  .
\]
Let $\Gamma\left(  D\right)  $ be a diagonal matrix whose $i^{\text{th}}$
diagonal entry is the denominator of the $i^{\text{th}}$ row of $Z_{2}%
^{\prime}\left(  D\right)  $ and $X_{2}^{\prime}\left(  D\right)  $. We
perform infinite-depth operations corresponding to
the entries in $\Gamma\left(  D\right)  $. The stabilizer matrix becomes%
\[
\left[  \left.
\begin{array}
[c]{ccc}%
I & \Gamma^{-1}\left(  D^{-1}\right)  & 0\\
0 & A\left(  D\right)  & Z_{2,N}^{\prime}\left(  D\right)
\end{array}
\right\vert
\begin{array}
[c]{ccc}%
0 & 0 & 0\\
I & \Gamma\left(  D\right)  & X_{2,N}^{\prime}\left(  D\right)
\end{array}
\right]  ,
\]
and the information-qubit matrix becomes%
\[
\left[  \left.
\begin{array}
[c]{ccc}%
0 & Z_{2,N}^{\prime\prime}\left(  D\right)  \Gamma^{-1}\left(  D^{-1}\right)
& 0\\
0 & X_{2,N}^{\prime\prime}\left(  D\right)  \Gamma^{-1}\left(  D^{-1}\right)
& I
\end{array}
\right\vert
\begin{array}
[c]{ccc}%
0 & 0 & I\\
0 & 0 & 0
\end{array}
\right]  .
\]
The above stabilizer matrix is equivalent to the desired one
in\ (\ref{eq:desired-check-matrix})\ by several row operations. We premultiply
the first set of rows by $\Gamma\left(  D^{-1}\right)  $ and multiply the
second set of rows by $\Gamma^{-1}\left(  D\right)  $. We can also use the
resulting identity matrix in the first set of rows to perform row operations
from the first set of rows to the second set of rows to realize the matrix
$X_{1}^{\prime}\left(  D\right)  $. The operators that Bob would really
measure need to have finite weight so he would measure the operators
corresponding to the entries in the following stabilizer matrix:%
\begin{equation}
\left[  \left.
\begin{array}
[c]{ccc}%
\Gamma\left(  D^{-1}\right)  & I & 0\\
0 & A\left(  D\right)  & Z_{2,N}^{\prime}\left(  D\right)
\end{array}
\right\vert
\begin{array}
[c]{ccc}%
0 & 0 & 0\\
I & \Gamma\left(  D\right)  & X_{2,N}^{\prime}\left(  D\right)
\end{array}
\right]  . \label{eq:end-encode}%
\end{equation}
We are done with the encoding algorithm. Alice begins with a set of ebits and
performs the encoding operations detailed in (\ref{eq:start-encode}%
-\ref{eq:end-encode}) and then performs the finite-depth operations detailed
in (\ref{eq:start-finite-depth-encode}-\ref{eq:desired-check-matrix})\ in
reverse order. We now detail the steps of the decoding algorithm. Perform CNOT
gates corresponding to the entries in $X_{2,N}^{\prime}\left(  D\right)  $
from the first set of $c$ qubits to the last set of $k+c$ qubits. The
stabilizer matrix becomes%
\[
\left[  \left.
\begin{array}
[c]{ccc}%
I & \Gamma^{-1}\left(  D^{-1}\right)  & 0\\
B\left(  D\right)  & A\left(  D\right)  & Z_{2,N}^{\prime}\left(  D\right)
\end{array}
\right\vert
\begin{array}
[c]{ccc}%
0 & 0 & 0\\
I & \Gamma\left(  D\right)  & 0
\end{array}
\right]  ,
\]
where $B\left(  D\right)  \equiv Z_{2,N}^{\prime}\left(  D\right)
X_{2,N}^{\prime\prime}\left(  D\right)  $. The information-qubit matrix
becomes%
\[
\left[  \left.
\begin{array}
[c]{ccc}%
0 & Z_{2,N}^{\prime\prime}\left(  D\right)  \Gamma^{-1}\left(  D^{-1}\right)
& 0\\
X_{2,N}^{\prime\prime}\left(  D\right)  & X_{2,N}^{\prime\prime}\left(
D\right)  \Gamma^{-1}\left(  D^{-1}\right)  & I
\end{array}
\right\vert
\begin{array}
[c]{ccc}%
0 & 0 & I\\
0 & 0 & 0
\end{array}
\right]  .
\]
Perform Hadamard gates on the last set of $k+c$ qubits. The stabilizer matrix
becomes%
\[
\left[  \left.
\begin{array}
[c]{ccc}%
I & \Gamma^{-1}\left(  D^{-1}\right)  & 0\\
B\left(  D\right)  & A\left(  D\right)  & 0
\end{array}
\right\vert
\begin{array}
[c]{ccc}%
0 & 0 & 0\\
I & \Gamma\left(  D\right)  & Z_{2,N}^{\prime}\left(  D\right)
\end{array}
\right]  ,
\]
and the information-qubit matrix becomes%
\[
\left[  \left.
\begin{array}
[c]{ccc}%
0 & Z_{2,N}^{\prime\prime}\left(  D\right)  \Gamma^{-1}\left(  D^{-1}\right)
& I\\
X_{2,N}^{\prime\prime}\left(  D\right)  & X_{2,N}^{\prime\prime}\left(
D\right)  \Gamma^{-1}\left(  D^{-1}\right)  & 0
\end{array}
\right\vert
\begin{array}
[c]{ccc}%
0 & 0 & 0\\
0 & 0 & I
\end{array}
\right]  .
\]
Perform CNOT gates from the first set of $c$ qubits to the last set of $k+c$
qubits. These CNOT\ gates correspond to the entries in $Z_{2,N}^{\prime
}\left(  D\right)  $. The stabilizer matrix becomes%
\[
\left[  \left.
\begin{array}
[c]{ccc}%
I & \Gamma^{-1}\left(  D^{-1}\right)  & 0\\
B\left(  D\right)  & A\left(  D\right)  & 0
\end{array}
\right\vert
\begin{array}
[c]{ccc}%
0 & 0 & 0\\
I & \Gamma\left(  D\right)  & 0
\end{array}
\right]  ,
\]
and the information-qubit matrix becomes%
\[
\left[  \left.
\begin{array}
[c]{ccc}%
Z_{2,N}^{\prime\prime}\left(  D\right)  & Z_{2,N}^{\prime\prime}\left(
D\right)  \Gamma^{-1}\left(  D^{-1}\right)  & I\\
X_{2,N}^{\prime\prime}\left(  D\right)  & X_{2,N}^{\prime\prime}\left(
D\right)  \Gamma^{-1}\left(  D^{-1}\right)  & 0
\end{array}
\right\vert
\begin{array}
[c]{ccc}%
0 & 0 & 0\\
0 & 0 & I
\end{array}
\right]  .
\]
Row operations from the first set of rows of the stabilizer to each of the two
sets of rows in the information-qubit matrix reduce the information-qubit
matrix to the following form:%
\[
\left[  \left.
\begin{array}
[c]{ccc}%
0 & 0 & I\\
0 & 0 & 0
\end{array}
\right\vert
\begin{array}
[c]{ccc}%
0 & 0 & 0\\
0 & 0 & I
\end{array}
\right]  .
\]
Then we perform the finite-depth operations detailed in
(\ref{eq:start-finite-depth-encode}-\ref{eq:desired-check-matrix}). We have
now finished the algorithm for the decoding circuit because the logical
operators for the information qubits appear in their original form.%
\end{proof}%

\subsection{Discussion}

Similar practical issues arise in these circuits as we discussed in the
previous chapter. Encoding circuits with infinite-depth operations are
acceptable if we assume that noiseless encoding is possible. Otherwise,
infinite-depth operations could lead to catastrophic propagation of
uncorrected errors. Noiseless encoding is difficult to achieve in practice but
we may be able to come close to it by concatenation of codes at the encoder.

There is a dichotomy of these codes similar to that in the previous chapter.
Some of the codes may have a simpler form in which the encoding circuit
requires finite-depth operations only. These codes fall within the first class
of codes discussed in the previous chapter and will be more useful in practice
because they propagate errors in the encoding circuit to a finite number of
qubits only. The remaining codes that do not have this structure fall within
the second class of codes whose encoding circuits have both finite-depth and
infinite-depth operations and whose decoding circuits have finite-depth
operations only.

\subsection{Importing Classical Convolutional Codes over $GF\left(  4\right)
$}

One benefit of the new entanglement-assisted quantum convolutional codes is
that we can produce one from an arbitrary classical convolutional code over
$GF\left(  4\right)  $. The error-correcting properties of the classical
convolutional code translate to the resulting quantum convolutional code. It
is less clear how the rate translates because we use the expansion technique.
We know that the term $\left(  2k-n\right)  /n$ lower bounds the
\textquotedblleft entanglement-assisted\textquotedblright\ rate where $n$ and
$k$ are the parameters from the imported classical code. The rate should get a
significant boost from entanglement---the rate boosts by the number of ebits
that the code requires.

The construction for importing an $\left[  n,k\right]  $ classical
convolutional code over $GF\left(  4\right)  $ is as follows. Suppose the
check matrix for the classical code is an $n-k\times n$-dimensional matrix
$H\left(  D\right)  $ whose entries are polynomials over $GF\left(  4\right)
$. We construct the quantum check matrix $\tilde{H}\left(  D\right)  $
according to the following formula:%
\[
\tilde{H}\left(  D\right)  =\gamma\left(  \left[
\begin{array}
[c]{c}%
\omega H\left(  D\right)  \\
\bar{\omega}H\left(  D\right)
\end{array}
\right]  \right)
\]
where $\gamma$ denotes the isomorphism between elements of $GF(4)$ and
symplectic binary vectors detailed in (\ref{eq:gf4-pauli}). We use this
construction for the example in the next section.

\section{Example}

\label{sec:examples}We take the convolutional generators from
Ref.~\cite{arx2007wildeCED} as our example. Ref.~\cite{science2006brun}
originally used these generators in a block code. We import the following
classical convolutional code over $GF\left(  4\right)  $:%
\begin{equation}
\left(  \cdots|0000|1\bar{\omega}10|1101|0000|\cdots\right)  .
\end{equation}
We produce two quantum convolutional generators by multiplying the above
generator by $\omega$ and $\bar{\omega}$ and applying the map in (\ref{eq:gf4-pauli}%
). The resulting quantum convolutional generators are as follows:%
\begin{equation}
\left(  \cdots|IIII|ZXZI|ZZIZ|IIII|\cdots\right)  ,\ \ \ \ \ \ \left(
\cdots|IIII|XYXI|XXIX|IIII|\cdots\right)  .\nonumber
\end{equation}
These generators have the following representation in the polynomial
formalism:%
\[
\left[  \left.
\begin{array}
[c]{cccc}%
1+D & D & 1 & D\\
0 & 1 & 0 & 0
\end{array}
\right\vert
\begin{array}
[c]{cccc}%
0 & 1 & 0 & 0\\
1+D & 1+D & 1 & D
\end{array}
\right]  .
\]
The shifted symplectic product matrix $\Omega\left(  D\right)  $ for the above
code is as follows:%
\[
\Omega\left(  D\right)  =%
\begin{bmatrix}
D+D^{-1} & D^{-1}\\
D & D+D^{-1}%
\end{bmatrix}
.
\]
The above matrix is not reducible to the standard form by any row operations.
We therefore expand the code by a factor of two to give four generators with a
frame size of eight. The two-expanded check matrix $H\left(  D\right)  $ is as
follows:%
\[
H\left(  D\right)  =\left[  \left.
\begin{array}
[c]{cccccccc}%
1 & 0 & 1 & 0 & 1 & 1 & 0 & 1\\
0 & 1 & 0 & 0 & 0 & 0 & 0 & 0\\
D & D & 0 & D & 1 & 0 & 1 & 0\\
0 & 0 & 0 & 0 & 0 & 1 & 0 & 0
\end{array}
\right\vert
\begin{array}
[c]{cccccccc}%
0 & 1 & 0 & 0 & 0 & 0 & 0 & 0\\
1 & 1 & 1 & 0 & 1 & 1 & 0 & 1\\
0 & 0 & 0 & 0 & 0 & 1 & 0 & 0\\
D & D & 0 & D & 1 & 1 & 1 & 0
\end{array}
\right]  .
\]
The shifted symplectic product matrix $\Omega_{2}\left(  D\right)  $ of the
two-expanded check matrix is as follows:%
\[
\Omega_{2}\left(  D\right)  =%
\begin{bmatrix}
0 & 0 & 1+D^{-1} & D^{-1}\\
0 & 0 & 1 & 1+D^{-1}\\
1+D & 1 & 0 & 0\\
D & 1+D & 0 & 0
\end{bmatrix}
.
\]
We proceed with the Gram-Schmidt procedure because this matrix satisfies its
initial requirements. We swap generators two and three to be the first and
second generators of the check matrix because they have the commutation
relations of half of an ebit. The shifted symplectic product matrix becomes%
\[%
\begin{bmatrix}
0 & 1 & 0 & 1+D^{-1}\\
1 & 0 & 1+D & 0\\
0 & 1+D^{-1} & 0 & D^{-1}\\
1+D & 0 & D & 0
\end{bmatrix}
.
\]
Multiply generator two by $1+D$ and add to generator four. Multiply generator
one by $1+D^{-1}$ and add to generator three. The shifted symplectic matrix
becomes%
\[%
\begin{bmatrix}
0 & 1 & 0 & 0\\
1 & 0 & 0 & 0\\
0 & 0 & 0 & 1+D^{-1}+D^{-2}\\
0 & 0 & 1+D+D^{2} & 0
\end{bmatrix}
.
\]
We finally divide generator four by $1+D+D^{2}$ and the shifted symplectic
product matrix then becomes%
\[%
{\displaystyle\bigoplus\limits_{i=1}^{2}}
J\text{,}%
\]
so that it has the commutation relations of halves of two ebits. The check
matrix resulting from these operations is as follows:%
\begin{equation}
H_{2}\left(  D\right)  =\left[  \left.
\begin{array}
[c]{c}%
Z_{2}\left(  D\right)
\end{array}
\right\vert
\begin{array}
[c]{c}%
X_{2}\left(  D\right)
\end{array}
\right]  , \label{eq:example-code}%
\end{equation}
where%
\begin{align*}
Z_{2}\left(  D\right)   &  =\left[
\begin{array}
[c]{cccccccc}%
0 & 1 & 0 & 0 & 0 & 0 & 0 & 0\\
D & D & 0 & D & 1 & 0 & 1 & 0\\
1 & \frac{1}{D}+1 & 1 & 0 & 1 & 1 & 0 & 1\\
\frac{D^{2}+D}{D^{2}+D+1} & \frac{D^{2}+D}{D^{2}+D+1} & 0 & \frac{D^{2}%
+D}{D^{2}+D+1} & \frac{D+1}{D^{2}+D+1} & \frac{1}{D^{2}+D+1} & \frac
{D+1}{D^{2}+D+1} & 0
\end{array}
\right]  ,\\
X_{2}\left(  D\right)   &  =\left[
\begin{array}
[c]{cccccccc}%
1 & 1 & 1 & 0 & 1 & 1 & 0 & 1\\
0 & 0 & 0 & 0 & 0 & 1 & 0 & 0\\
\frac{1}{D}+1 & \frac{1}{D} & \frac{1}{D}+1 & 0 & \frac{1}{D}+1 & \frac{1}%
{D}+1 & 0 & \frac{1}{D}+1\\
\frac{D}{D^{2}+D+1} & \frac{D}{D^{2}+D+1} & 0 & \frac{D}{D^{2}+D+1} & \frac
{1}{D^{2}+D+1} & \frac{D}{D^{2}+D+1} & \frac{1}{D^{2}+D+1} & 0
\end{array}
\right]  .
\end{align*}
The error-correcting properties of the above check matrix are equivalent to
the error-correcting properties of the original two generators.

This code sends six information qubits and consumes two ebits per eight
channel uses. The rate pair for this code is therefore $\left(
3/4,1/4\right)  $.

We can now apply the algorithm in Theorem~\ref{thm:main} to determine the
encoding and decoding circuits for this code. The encoding circuit begins from
a set of two ebits and eight information qubits per frame with the following
stabilizer matrix:%
\[
H_{0}\left(  D\right)  =\left[  \left.
\begin{array}
[c]{cccccccccc}%
1 & 0 & 1 & 0 & 0 & 0 & 0 & 0 & 0 & 0\\
0 & 1 & 0 & 1 & 0 & 0 & 0 & 0 & 0 & 0\\
0 & 0 & 0 & 0 & 0 & 0 & 0 & 0 & 0 & 0\\
0 & 0 & 0 & 0 & 0 & 0 & 0 & 0 & 0 & 0
\end{array}
\right\vert
\begin{array}
[c]{cccccccccc}%
0 & 0 & 0 & 0 & 0 & 0 & 0 & 0 & 0 & 0\\
0 & 0 & 0 & 0 & 0 & 0 & 0 & 0 & 0 & 0\\
1 & 0 & 1 & 0 & 0 & 0 & 0 & 0 & 0 & 0\\
0 & 1 & 0 & 1 & 0 & 0 & 0 & 0 & 0 & 0
\end{array}
\right]  .
\]
We label the eight qubits on the right side of each matrix above as
$1,\ldots,8$. We label Bob's two qubits on the left as $B1$ and $B2$. Perform
the following finite-depth operations (in order from left to right and then
top to bottom):%
\begin{align*}
&  C\left(  1,4,D+D^{2}\right)  C\left(  1,5,1+D^{2}\right)  C\left(
1,6,1\right)  C\left(  1,7,1+D\right)  C\left(  2,4,D\right)  C\left(
2,5,1+D\right)  C\left(  2,6,1\right) \\
&  H\left(  3,\ldots,8\right)  C\left(  1,4,D+D^{2}+D^{4}\right)  C\left(
1,5,D^{2}\right)  C\left(  1,6,1+D\right)  C\left(  1,7,D^{2}\right)  C\left(
2,4,D+D^{2}\right) \\
&  C\left(  2,5,1+D\right)  C\left(  2,6,1\right)  C\left(  2,7,1+D\right)
\end{align*}
where we use the notation $C\left(  q_{1},q_{2},f\left(  D\right)  \right)  $
to represent a finite-depth CNOT\ gate from qubit one to qubit two that
implements the polynomial $f\left(  D\right)  $, $H\left(  q_{i},\ldots
,q_{j}\right)  $ is a sequence of Hadamard gates applied to qubits $q_{i}$
through $q_{j}$ in each frame, $P\left(  q\right)  $ is a phase gate applied
to qubit $q$ in each frame, and $C\left(  q,1/f\left(  D\right)  \right)  $ is
an infinite-depth CNOT\ gate implementing the rational polynomial $1/f\left(
D\right)  $ on qubit $q$. Alice performs the following infinite-depth operations:%
\[
H\left(  1,2\right)  C\left(  1,\frac{1}{1+D^{-1}+D^{-2}}\right)  C\left(
2,\frac{1}{1+D^{-1}+D^{-2}}\right)  H\left(  1,2\right)  .
\]
She then finishes the encoding circuit with the following finite-depth
operations:%
\begin{align}
&  H\left(  1,2\right)  C\left(  2,3,1\right)  C\left(  2,5,1\right)  C\left(
2,6,1\right)  C\left(  2,8,1\right)  P\left(  2\right)  H\left(
3,\ldots,8\right)  S\left(  2,3\right) \label{eq:example-finite-depth}\\
&  C\left(  1,2,1\right)  C\left(  1,3,1\right)  C\left(  1,5,1\right)
C\left(  1,6,1\right)  C\left(  1,8,1\right)  P\left(  2\right) \nonumber
\end{align}
The code she encodes has equivalent error-correcting properties to the check
matrix in (\ref{eq:example-code}).

Bob performs the following operations in the decoding circuit. He first
performs the operations in (\ref{eq:example-finite-depth})\ in reverse order.
He then performs the following finite-depth operations:%
\begin{align*}
&  C\left(  B1,4,D+D^{2}+D^{4}\right)  C\left(  B1,5,D^{2}\right)  C\left(
B1,6,1+D\right)  C\left(  B1,7,D^{2}\right)  C\left(  B2,4,D+D^{2}\right) \\
&  C\left(  B2,5,1+D\right)  C\left(  B2,6,1\right)  C\left(  B2,7,1+D\right)
H\left(  3,\ldots,8\right) \\
&  C\left(  B1,4,D+D^{2}\right)  C\left(  B1,5,1+D^{2}\right)  C\left(
B1,6,1\right)  C\left(  B1,7,1+D\right)  C\left(  B2,4,D\right) \\
&  C\left(  B2,5,1+D\right)  C\left(  B2,6,1\right)
\end{align*}
The information qubits then appear at the output of this online decoding circuit.

\section{Optimal Entanglement Formulas}

We conjecture two formulas for the optimal number of ebits that a general
(non-CSS) entanglement-assisted quantum convolutional code or one imported
from a classical quaternary convolutional code need to consume per frame of
operation. We show that this conjecture holds for a particular example.

\begin{conjecture}
The optimal number\ $c$ of ebits necessary per frame for an
entanglement-assisted quantum convolutional code is%
\begin{equation}
c=\mathrm{rank}\left(  H_{X}(D)H_{Z}^{T}(D^{-1})+H_{Z}(D)H_{X}^{T}%
(D^{-1})\right)  /2
\end{equation}
where $H(D)=\left[  \left.
\begin{array}
[c]{c}%
H_{Z}(D)
\end{array}
\right\vert
\begin{array}
[c]{c}%
H_{X}(D)
\end{array}
\right]  $ represents the parity check matrix for a set of quantum
convolutional generators that do not necessarily form a commuting set.
\end{conjecture}

It is clear that the above formula holds after we have expanded an original
set of generators and their commutation relations are reducible by the
polynomial symplectic Gram-Schmidt orthogonalization procedure to the standard
form in (\ref{eq:standard-symp-form}). The proof technique follows from the
proof technique outlined in Section~\ref{sec:opt-ebit}. But we are not sure
how to apply this formula to an initial set of unexpanded generators.

\begin{conjecture}
The optimal number $c$ of ebits required per frame for an
entanglement-assisted quantum convolutional code imported from a classical
quaternary convolutional code with parity check matrix $H\left(  D\right)  $
is%
\begin{equation}
c=\mathrm{rank}\left(  H(D)H^{\dag}(D^{-1})\right)  .
\end{equation}

\end{conjecture}

We know the number of ebits required for a CSS entanglement-assisted quantum
convolutional code is$\ \mathrm{rank}\left(  H_{1}(D)H_{2}^{T}(D^{-1})\right)
$ where $H_{1}(D)$ and $H_{2}(D)$ correspond to the classical binary
convolutional codes that we import to correct respective bit and phase flips.
Comparing the CSS\ convolutional formula, the CSS\ block formula
in\ Corollary~\ref{cor:CSS}, and the general block formula in
Theorem~\ref{thm:opt-ebit-formula}, the above conjecture seems natural.

We finally provide an example of the above conjecture. It is a slight
modification of the code presented in Ref.~\cite{arx2008wildeUQCC}.

\begin{example}
Consider the quantum convolutional code with quantum check matrix as follows:%
\[
\left[  \left.
\begin{array}
[c]{ccccc}%
0 & 0 & 0 & 0 & 0\\
h\left(  D\right)  & D & 0 & 1 & h\left(  D\right) \\
0 & 0 & D & D & D\\
0 & \frac{1}{D} & 1 & \frac{1}{D} & 0\\
0 & \frac{1}{D} & 0 & 0 & 0
\end{array}
\right\vert
\begin{array}
[c]{ccccc}%
h\left(  D\right)  & 0 & D & 1 & h\left(  D\right) \\
0 & 0 & 0 & 0 & 0\\
0 & 1 & 0 & 1 & 1\\
0 & 0 & 0 & 0 & 0\\
0 & 0 & 1 & 0 & 0
\end{array}
\right]  ,
\]
where $h\left(  D\right)  =1+D$. This code requires two ebits and one ancilla
qubit for quantum redundancy and encodes two information qubits. The shifted
symplectic product matrix $H_{X}(D)H_{Z}^{T}(D^{-1})+H_{Z}(D)H_{X}^{T}%
(D^{-1})$ \cite{arx2007wildeCED,arx2007wildeEAQCC}\ for this code is as
follows:%
\begin{equation}%
\begin{bmatrix}
0 & 1 & 0 & 0 & 0\\
1 & 0 & 0 & 0 & 0\\
0 & 0 & 0 & 0 & 0\\
0 & 0 & 0 & 0 & 1\\
0 & 0 & 0 & 1 & 0
\end{bmatrix}
\end{equation}
The rank of the above matrix is four and the code requires two ebits.
Therefore, the conjecture holds for this example.
\end{example}

\section{Closing Remarks}

\label{sec:conclusion}There are several differences between the methods used
for general, non-CSS\ codes discussed in this chapter and the CSS\ codes used
in the previous chapter. It was more straightforward to determine how to use
ebits efficiently in CSS entanglement-assisted quantum convolutional codes,
but we have had to introduce the expansion technique in
Section~\ref{sec:expand}\ in order to determine how to use ebits efficiently
for codes in this chapter. There was also no need for an explicit Gram-Schmidt
orthogonalization procedure in the previous chapter. The Smith algorithm
implicitly produced the row operations necessary for symplectic orthogonalization.

We do have some methods in the next chapter that do not require expansion of a check matrix or an
explicit Gram-Schmidt procedure but these methods do not make efficient use of
entanglement and have a lower rate of quantum communication and
higher rate of entanglement consumption than the codes discussed in this
chapter. Nonetheless, we have determined ways to make these other codes more
useful by encoding classical information in the extra entanglement with a
superdense-coding-like effect \cite{PhysRevLett.69.2881}. These other codes
are grandfather codes in the sense of
Refs.~\cite{arx2008wildeUQCC,arx2008kremsky} because they consume entanglement
to send both quantum and classical information. We discuss these techniques in
the next chapter.

One negative implication of the expansion of a code is that the expanded code
requires more qubits per frame. The expanded code then requires a larger
buffer at both the sender's and receiver's local stations. The increased buffer
will be a concern right now because it is difficult to build large quantum
memories. This issue will become less of a concern as quantum technology
advances. The entanglement-inefficient codes in the next chapter have the
advantage that they do not require expansion and thus require smaller buffers.
It therefore should be of interest to find solutions in between the
entanglement-efficient codes discussed in this chapter and the
entanglement-inefficient codes discussed in the previous paragraph.

Some outstanding issues remain. We have not proven the convergence of the
polynomial symplectic Gram-Schmidt orthogonalization procedure and have
instead provided a practical stopping condition. We have a conjecture for how
to proceed with proving convergence. Suppose that we would like to construct a
code consisting of one generator that does not commute with shifts of itself.
We have found for many examples that the correct expansion factor for the
generator is equal to the period of the inverse polynomial of the generator's
shifted symplectic product. We do not have a proof that this factor is the
correct one and we are not sure what the expansion factor should be when we
would like to construct a code starting from more than one generator.

The techniques developed in this chapter represent a useful way for encoding
quantum information. This framework will most likely lead to codes with
reasonable performance but they will most likely not come close to achieving
the quantum capacity of an entanglement-assisted channel. The next step should
be to combine the theory in this chapter with Poulin et al.'s recent theory of
quantum serial-turbo coding~\cite{arx2007poulin}. This combination might lead
to entanglement-assisted quantum turbo codes that come close to achieving capacity.

\chapter{Entanglement-Assisted Quantum Convolutional Coding: Free
Entanglement}

\label{chp:free-ent}\begin{saying}
Sometimes we've a free cater,\\
Perhaps on the house we've a waiter,\\
But free entanglement you say?\\
Not a price to pay?\\
Oh assume it for now and pay later.
\end{saying}We mentioned in the previous chapter that the expansion of a set
of quantum convolutional generators increases the frame size of a code. This
increase implies that each round of transmission in the protocol sends a
larger number of encoded qubits and requires a larger quantum memory for its
operation. Thus increasing the frame size is somewhat undesirable.

We offer the \textquotedblleft free-entanglement\textquotedblright\ codes in
this chapter as an alternative option to the codes in the previous chapter.
The free-entanglement codes have a basic generator set that includes an
arbitrary set of Pauli sequences. The free-entanglement codes consume more
entanglement than is necessary and have a lower quantum information rate, but
the benefit is that they do not require a larger quantum memory for each round
of transmission. It is up to the quantum coding engineer to decide which
option is more desirable for the particular coding application: using less
entanglement and having a higher rate of quantum transmission or using a
smaller quantum memory to send fewer qubits for each round of transmission.

We make the additional assumption that shared entanglement is freely
available. Quantum information theorists make the \textquotedblleft free
entanglement\textquotedblright\ assumption when deriving the
entanglement-assisted capacity of a quantum channel
\cite{PhysRevLett.83.3081,ieee2002bennett}. This model makes sense when the
sender and receiver can use the noisy channel and entanglement distillation
protocols \cite{PhysRevLett.76.722,PhysRevA.54.3824}\ during off-peak times to
generate a large amount of noiseless entanglement.

The \textquotedblleft entanglement-assisted\textquotedblright\ rate
(Section~\ref{sec:EA-rates}) of a quantum code is the ratio of the number of
encoded information qubits to the number of physical qubits assuming that
entanglement is available for free. The entanglement-assisted rate of our
codes is at least $k/n$. Parameter $k$ is the number of information qubits
that the code encodes and parameter $n$ is the number of qubits used for
encoding. A basic set of $n-k$ generators specifies all of our
entanglement-assisted quantum convolutional codes and yields the above
entanglement-assisted rate.

We say that the entanglement-assisted rate is at least $k/n$ because the
example in Section~\ref{sec:free-ent-example} discusses a method to boost the
rate. This method encodes some classical information into the extra
entanglement along the lines of the \textquotedblleft
grandfather\textquotedblright\ coding technique discussed in the next chapter
and in Ref.~\cite{arx2008kremsky}. The sender can use this classical
information and the entanglement to teleport \cite{PhysRevLett.70.1895}%
\ additional information qubits and increase the rate of quantum transmission.

\section{Construction}

The proof of our main theorem below outlines how to encode a stream of
information qubits, ancilla qubits, and shared ebits so that the encoded
qubits have the error-correcting properties of an arbitrary set of Paulis. The
receiver may employ an error estimation algorithm such as Viterbi decoding
\cite{itit1967viterbi}\ to determine the most likely errors that the noisy
quantum communication channel induces on the encoded stream. We then show how
to decode the encoded qubit stream so that the information qubits become
available at the receiving end of the channel.

The encoding circuits in the proof of our theorem employ both finite-depth and
infinite-depth operations. The decoding circuits employ finite-depth
operations only. Infinite-depth operations can lead to catastrophic error
propagation as discussed in previous chapters. In our proof below, we restrict
infinite-depth operations to act on qubits before sending them over the noisy
channel. Catastrophic error propagation does not occur under the ideal
circumstance when the operations in the encoding circuit are noiseless.

Our theorem below begins with a \textquotedblleft quantum check
matrix\textquotedblright\ that consists of a set of Pauli sequences with
desirable error-correcting properties. This quantum check matrix does not
necessarily correspond to a commuting stabilizer. The proof of the theorem
shows how to incorporate ebits so that the sender realizes the same quantum
check matrix for her qubits and the sender and receiver's set of generators
form a valid commuting stabilizer.

\begin{theorem}
Suppose the following quantum check matrix%
\[
S\left(  D\right)  =\left[  \left.
\begin{array}
[c]{c}%
Z\left(  D\right)
\end{array}
\right\vert
\begin{array}
[c]{c}%
X\left(  D\right)
\end{array}
\right]  \in\mathbb{F}_{2}\left[  D\right]  ^{\left(  n-k\right)  \times2n},
\]
where $S\left(  D\right)  $ is of full rank and does not necessarily form a
commuting stabilizer and $\mathbb{F}_{2}\left[  D\right]  $ is the field of
binary polynomials. Then an entanglement-assisted quantum
convolutional code exists that has the same error-correcting properties as the
above quantum check matrix $S\left(  D\right)  $. The entanglement-assisted
rate of the code is at least $k/n$.
\end{theorem}

%

\begin{proof}%
Suppose that the Smith form \cite{book1999conv}\ of $X\left(  D\right)  $ is
as follows%
\begin{equation}
X\left(  D\right)  =A\left(  D\right)  \left[
\begin{array}
[c]{ccc}%
\Gamma_{1}\left(  D\right)  & 0 & 0\\
0 & \Gamma_{2}\left(  D\right)  & 0\\
0 & 0 & 0
\end{array}
\right]  B\left(  D\right)  , \label{eq:first-finite-depth-encode}%
\end{equation}
where $A\left(  D\right)  $ is $\left(  n-k\right)  \times\left(  n-k\right)
$, $B\left(  D\right)  $ is $n\times n$, $\Gamma_{1}\left(  D\right)  $ is an
$s\times s$ diagonal matrix whose entries are powers of $D$, and $\Gamma
_{2}\left(  D\right)  $ is a $\left(  c-s\right)  \times\left(  c-s\right)  $
diagonal matrix whose entries are arbitrary polynomials. Perform the row
operations in $A^{-1}\left(  D\right)  $ and the column operations in
$B^{-1}\left(  D\right)  $ on $S\left(  D\right)  $. The quantum check matrix
$S\left(  D\right)  $\ becomes%
\begin{equation}
\left[
\begin{array}
[c]{c}%
E\left(  D\right)
\end{array}
\left\vert
\begin{array}
[c]{ccc}%
\Gamma_{1}\left(  D\right)  & 0 & 0\\
0 & \Gamma_{2}\left(  D\right)  & 0\\
0 & 0 & 0
\end{array}
\right.  \right]  ,
\end{equation}
where $E\left(  D\right)  =A^{-1}\left(  D\right)  Z\left(  D\right)
B^{T}\left(  D^{-1}\right)  $. Suppose $E_{1}\left(  D\right)  $ is the first
$c$ columns of $E\left(  D\right)  $ and $E_{2}\left(  D\right)  $ is the next
$n-c$ columns of $E\left(  D\right)  $ so that the quantum check matrix is as
follows:%
\begin{equation}
\left[
\begin{array}
[c]{cc}%
E_{1}\left(  D\right)  & E_{2}\left(  D\right)
\end{array}
\left\vert
\begin{array}
[c]{ccc}%
\Gamma_{1}\left(  D\right)  & 0 & 0\\
0 & \Gamma_{2}\left(  D\right)  & 0\\
0 & 0 & 0
\end{array}
\right.  \right]  .
\end{equation}
Perform Hadamard gates on the last $n-c$ qubits so that the quantum check
matrix becomes%
\begin{equation}
\left[
\begin{array}
[c]{cc}%
E_{1}\left(  D\right)  & 0
\end{array}
\left\vert
\begin{array}
[c]{ccc}%
\Gamma_{1}\left(  D\right)  & 0 & E_{2,1}\left(  D\right) \\
0 & \Gamma_{2}\left(  D\right)  & E_{2,2}\left(  D\right) \\
0 & 0 & E_{2,3}\left(  D\right)
\end{array}
\right.  \right]  ,
\end{equation}
where%
\begin{equation}
E_{2}\left(  D\right)  =\left[
\begin{array}
[c]{c}%
E_{2,1}\left(  D\right) \\
E_{2,2}\left(  D\right) \\
E_{2,3}\left(  D\right)
\end{array}
\right]  .
\end{equation}
Perform CNOT\ operations from the first $s$ qubits to the last $n-c$ qubits to
clear the entries in $E_{2,1}\left(  D\right)  $. The quantum check matrix
becomes%
\begin{equation}
\left[
\begin{array}
[c]{cc}%
E_{1}\left(  D\right)  & 0
\end{array}
\left\vert
\begin{array}
[c]{ccc}%
\Gamma_{1}\left(  D\right)  & 0 & 0\\
0 & \Gamma_{2}\left(  D\right)  & E_{2,2}\left(  D\right) \\
0 & 0 & E_{2,3}\left(  D\right)
\end{array}
\right.  \right]  .
\end{equation}
The Smith form of $E_{2,3}\left(  D\right)  $ is as follows%
\begin{equation}
E_{2,3}\left(  D\right)  =A_{E}\left(  D\right)  \left[
\begin{array}
[c]{cc}%
\Gamma\left(  D\right)  & 0
\end{array}
\right]  B_{E}\left(  D\right)  ,
\end{equation}
where $A_{E}\left(  D\right)  $ is $\left(  n-k-c\right)  \times\left(
n-k-c\right)  $, $B_{E}\left(  D\right)  $ is $\left(  n-c\right)
\times\left(  n-c\right)  $, and $\Gamma\left(  D\right)  $ is a $\left(
n-k-c\right)  \times\left(  n-k-c\right)  $\ diagonal matrix whose entries are
polynomials. The Smith form of $E_{2,3}\left(  D\right)  $\ is full rank
because\ the original quantum check matrix $S\left(  D\right)  $\ is full
rank. Perform the row operations in $A_{E}^{-1}\left(  D\right)  $ and the
column operations in $B_{E}^{-1}\left(  D\right)  $. The quantum check matrix
becomes%
\begin{equation}
\left[
\begin{array}
[c]{cc}%
E_{1}^{\prime}\left(  D\right)  & 0
\end{array}
\left\vert
\begin{array}
[c]{cccc}%
\Gamma_{1}\left(  D\right)  & 0 & 0 & 0\\
0 & \Gamma_{2}\left(  D\right)  & E_{2,2a}^{^{\prime}}\left(  D\right)  &
E_{2,2b}^{^{\prime}}\left(  D\right) \\
0 & 0 & \Gamma\left(  D\right)  & 0
\end{array}
\right.  \right]  ,
\end{equation}
where%
\begin{align}
E_{1}^{\prime}\left(  D\right)   &  =\left[
\begin{array}
[c]{cc}%
I & 0\\
0 & A_{E}^{-1}\left(  D\right)
\end{array}
\right]  E_{1}\left(  D\right)  ,\\
E_{2,2}^{^{\prime}}\left(  D\right)   &  =E_{2,2}\left(  D\right)  B_{E}%
^{-1}\left(  D\right)  =\left[
\begin{array}
[c]{cc}%
E_{2,2a}^{^{\prime}}\left(  D\right)  & E_{2,2b}^{^{\prime}}\left(  D\right)
\end{array}
\right]  .
\end{align}
Perform a modified version of the Smith algorithm to reduce the $\left(
c-s\right)  \times\left(  n-c\right)  $ matrix $E_{2,2b}^{^{\prime}}\left(
D\right)  $ to a lower triangular form (discussed in
Chapter~\ref{chp:EAQCC-CSS}). This modified algorithm uses only column
operations to transform%
\begin{equation}
E_{2,2b}^{^{\prime}}\left(  D\right)  \rightarrow\left[
\begin{array}
[c]{cc}%
L\left(  D\right)  & 0
\end{array}
\right]  ,
\end{equation}
where $L\left(  D\right)  $ is $\left(  c-s\right)  \times\left(  c-s\right)
$ and the null matrix is $\left(  c-s\right)  \times\left(  n+s-2c\right)  $.
The quantum check matrix becomes%
\begin{equation}
\left[
\begin{array}
[c]{cc}%
E_{1}^{\prime}\left(  D\right)  & 0
\end{array}
\left\vert
\begin{array}
[c]{ccccc}%
\Gamma_{1}\left(  D\right)  & 0 & 0 & 0 & 0\\
0 & \Gamma_{2}\left(  D\right)  & E_{2,2a}^{^{\prime}}\left(  D\right)  &
L\left(  D\right)  & 0\\
0 & 0 & \Gamma\left(  D\right)  & 0 & 0
\end{array}
\right.  \right]  . \label{eq:big-QCM}%
\end{equation}
We have now completed the decomposition of the quantum check matrix with
column and row operations.

We turn to showing how to encode and decode a certain quantum check matrix
that proves to be useful in encoding the above quantum check matrix. Consider
the following quantum check matrix%
\begin{equation}
\left[
\begin{array}
[c]{cc}%
I & 0\\
0 & 0
\end{array}
\left\vert
\begin{array}
[c]{cc}%
0 & 0\\
\Gamma_{2}\left(  D\right)  & L\left(  D\right)
\end{array}
\right.  \right]  , \label{eq:desired-QCM}%
\end{equation}
where $\Gamma_{2}\left(  D\right)  $ and $L\left(  D\right)  $ are from the
matrix in (\ref{eq:big-QCM})\ and each of them, the identity matrix, and the
null matrices have dimension $\left(  c-s\right)  \times\left(  c-s\right)  $.
We use a method for encoding and decoding the quantum check matrix in
(\ref{eq:desired-QCM})\ similar to the method in Chapter~\ref{chp:EAQCC-CSS}%
\ for the second class of CSS\ entanglement-assisted quantum convolutional
codes. We begin with a set of $c-s$ ebits and $c-s$ information qubits. The
following matrix stabilizes the ebits%
\begin{equation}
\left[  \left.
\begin{array}
[c]{ccc}%
I & I & 0\\
0 & 0 & 0
\end{array}
\right\vert
\begin{array}
[c]{ccc}%
0 & 0 & 0\\
I & I & 0
\end{array}
\right]  , \label{eq:first-inf-depth-encode}%
\end{equation}
where Bob possesses the $c-s$ qubits on the \textquotedblleft
left,\textquotedblright\ Alice possesses the $2\left(  c-s\right)  $ qubits on
the \textquotedblleft right,\textquotedblright\ and each matrix is $\left(
c-s\right)  \times\left(  c-s\right)  $. The following matrix represents the
information qubits:%
\begin{equation}
\left[  \left.
\begin{array}
[c]{ccc}%
0 & 0 & I\\
0 & 0 & 0
\end{array}
\right\vert
\begin{array}
[c]{ccc}%
0 & 0 & 0\\
0 & 0 & I
\end{array}
\right]  . \label{eq:first-inf-depth-encode-info}%
\end{equation}
The above information-qubit matrix represents the logical operators for the
information qubits and gives a useful way of tracking the information qubits
while processing them. Tracking the information-qubit matrix helps to confirm
that the information qubits decode properly at the receiver's end
\cite{arx2007wildeEAQCC}. We track both the above stabilizer and the
information-qubit matrix as they progress through some encoding operations.
Alice then performs CNOT\ gates from her first $c-s$ qubits to her next $c-s$
qubits. These gates multiply the middle $c-s$ columns of the \textquotedblleft
X\textquotedblright\ matrix by $L\left(  D\right)  $ and add the result to the
last $c-s$ columns and multiply the last $c-s$ columns of the
\textquotedblleft Z\textquotedblright\ matrix by $L^{T}\left(  D^{-1}\right)
$ and add the result to the last $c-s$ columns. The stabilizer becomes%
\begin{equation}
\left[  \left.
\begin{array}
[c]{ccc}%
I & I & 0\\
0 & 0 & 0
\end{array}
\right\vert
\begin{array}
[c]{ccc}%
0 & 0 & 0\\
I & I & L\left(  D\right)
\end{array}
\right]  ,
\end{equation}
and the information-qubit matrix becomes%
\begin{equation}
\left[  \left.
\begin{array}
[c]{ccc}%
0 & L^{T}\left(  D^{-1}\right)  & I\\
0 & 0 & 0
\end{array}
\right\vert
\begin{array}
[c]{ccc}%
0 & 0 & 0\\
0 & 0 & I
\end{array}
\right]  .
\end{equation}
Alice performs infinite-depth operations on her first $c-s$ qubits
corresponding to the rational polynomials $\gamma_{2,1}^{-1}\left(
D^{-1}\right)  $, $\ldots$, $\gamma_{2,c-s}^{-1}\left(  D^{-1}\right)  $ in
$\Gamma_{2}^{-1}\left(  D^{-1}\right)  $. These operations multiply the middle
$c-s$ columns of the \textquotedblleft Z\textquotedblright\ matrix by
$\Gamma_{2}^{-1}\left(  D^{-1}\right)  $ and multiply the middle $c-s$ columns
of the \textquotedblleft X\textquotedblright\ matrix by $\Gamma_{2}\left(
D\right)  $. The stabilizer matrix becomes%
\begin{equation}
\left[  \left.
\begin{array}
[c]{ccc}%
I & \Gamma_{2}^{-1}\left(  D^{-1}\right)  & 0\\
0 & 0 & 0
\end{array}
\right\vert
\begin{array}
[c]{ccc}%
0 & 0 & 0\\
I & \Gamma_{2}\left(  D\right)  & L\left(  D\right)
\end{array}
\right]  , \label{eq:last-inf-depth-encode-stab}%
\end{equation}
and the information-qubit matrix becomes%
\begin{equation}
\left[  \left.
\begin{array}
[c]{ccc}%
0 & L^{T}\left(  D^{-1}\right)  \Gamma_{2}^{-1}\left(  D^{-1}\right)  & I\\
0 & 0 & 0
\end{array}
\right\vert
\begin{array}
[c]{ccc}%
0 & 0 & 0\\
0 & 0 & I
\end{array}
\right]  . \label{eq:last-inf-depth-encode}%
\end{equation}
Alice's part of the above stabilizer matrix is equivalent to the quantum check
matrix in (\ref{eq:desired-QCM}) by row operations (premultiplying the first
set of rows by $\Gamma_{2}\left(  D^{-1}\right)  $.)

We now illustrate a way to decode the encoded stabilizer
in\ (\ref{eq:last-inf-depth-encode-stab}) and information-qubit matrix in
(\ref{eq:last-inf-depth-encode})\ so that the information qubits appear at the
output of the decoding circuit. Bob performs CNOT\ gates from the first set of
qubits to the third set of qubits corresponding to the entries in $L\left(
D\right)  $. The stabilizer becomes%
\begin{equation}
\left[  \left.
\begin{array}
[c]{ccc}%
I & \Gamma_{2}^{-1}\left(  D^{-1}\right)  & 0\\
0 & 0 & 0
\end{array}
\right\vert
\begin{array}
[c]{ccc}%
0 & 0 & 0\\
I & \Gamma_{2}\left(  D\right)  & 0
\end{array}
\right]  , \label{eq:first-inf-depth-decode}%
\end{equation}
and the information-qubit matrix becomes%
\begin{equation}
\left[  \left.
\begin{array}
[c]{ccc}%
L^{T}\left(  D^{-1}\right)  & L^{T}\left(  D^{-1}\right)  \Gamma_{2}%
^{-1}\left(  D^{-1}\right)  & I\\
0 & 0 & 0
\end{array}
\right\vert
\begin{array}
[c]{ccc}%
0 & 0 & 0\\
0 & 0 & I
\end{array}
\right]  .
\end{equation}
Bob finishes decoding at this point because we can equivalently express the
information-qubit matrix as follows%
\begin{equation}
\left[  \left.
\begin{array}
[c]{ccc}%
0 & 0 & I\\
0 & 0 & 0
\end{array}
\right\vert
\begin{array}
[c]{ccc}%
0 & 0 & 0\\
0 & 0 & I
\end{array}
\right]  , \label{eq:last-inf-depth-decode}%
\end{equation}
by multiplying the first $c-s$ rows of the stabilizer by $L^{T}\left(
D^{-1}\right)  $ and adding to the first $c-s$ rows of the information-qubit
matrix. The information qubits are available at the receiving end of the
channel because the above information-qubit matrix is equivalent to the
original one in (\ref{eq:first-inf-depth-encode-info}).

We show how to encode the quantum check matrix in\ (\ref{eq:big-QCM}) using
ebits, ancilla qubits, and information qubits. We employ the encoding
technique for the submatrix listed above and use some other techniques as
well. Suppose that we have the following matrix that stabilizes a set of $c$
ebits per frame, $n-k-c$ ancilla qubits per frame, and $k$ information qubits
per frame:%
\begin{equation}
\left[
\begin{array}
[c]{ccccccc}%
I & 0 & I & 0 & 0 & 0 & 0\\
0 & I & 0 & I & 0 & 0 & 0\\
0 & 0 & 0 & 0 & 0 & 0 & 0\\
0 & 0 & 0 & 0 & 0 & 0 & 0\\
0 & 0 & 0 & 0 & 0 & 0 & 0
\end{array}
\left\vert
\begin{array}
[c]{ccccccc}%
0 & 0 & 0 & 0 & 0 & 0 & 0\\
0 & 0 & 0 & 0 & 0 & 0 & 0\\
I & 0 & I & 0 & 0 & 0 & 0\\
0 & I & 0 & I & 0 & 0 & 0\\
0 & 0 & 0 & 0 & I & 0 & 0
\end{array}
\right.  \right]  .
\end{equation}
The first and third sets of rows have $s$ rows and correspond to $s$ ebits per
frame, the second and fourth sets of rows have $c-s$ rows and correspond to
$c-s$ ebits per frame, and the last set of $n-k-c$ rows corresponds to $n-k-c$
ancilla qubits per frame. The above matrix has $n+c$ columns on both the
\textquotedblleft Z\textquotedblright\ and \textquotedblleft
X\textquotedblright\ side so that the above matrix stabilizes $k$ information
qubits per frame. Bob possesses the first $c$ qubits and Alice possesses the
next $n$ qubits. Alice performs the encoding operations in
(\ref{eq:first-inf-depth-encode}-\ref{eq:last-inf-depth-encode}) to get the
following stabilizer:%
\begin{equation}
\left[
\begin{array}
[c]{ccccccc}%
I & 0 & I & 0 & 0 & 0 & 0\\
0 & I & 0 & \Gamma_{2}^{-1}\left(  D^{-1}\right)  & 0 & 0 & 0\\
0 & 0 & 0 & 0 & 0 & 0 & 0\\
0 & 0 & 0 & 0 & 0 & 0 & 0\\
0 & 0 & 0 & 0 & 0 & 0 & 0
\end{array}
\left\vert
\begin{array}
[c]{ccccccc}%
0 & 0 & 0 & 0 & 0 & 0 & 0\\
0 & 0 & 0 & 0 & 0 & 0 & 0\\
I & 0 & I & 0 & 0 & 0 & 0\\
0 & I  & 0 & \Gamma_{2}\left(  D\right)  & 0 &
L\left(  D\right)  & 0\\
0 & 0 & 0 & 0 & I & 0 & 0
\end{array}
\right.  \right]  . \label{eq:start-decoding}%
\end{equation}
We perform several row operations to get the quantum check matrix
in\ (\ref{eq:big-QCM}). Premultiply the middle set of rows by $\Gamma
_{1}\left(  D\right)  $. Premultiply the last set of rows by $E_{2,2a}%
^{^{\prime}}\left(  D\right)  $ and add the result to the set of rows above
the last set. Premultiply the last set of rows by $\Gamma\left(  D\right)  $.
Finally, premultiply the first two sets of rows by $E_{1}^{\prime\prime
}\left(  D\right)  =E_{1}^{\prime}\left(  D\right)  \left[  I\oplus\Gamma
_{2}\left(  D^{-1}\right)  \right]  $ and add the result to the last three
sets of rows. The quantum check matrix becomes%
\begin{equation}
\left[  \left.
\begin{tabular}
[c]{ccccccc}%
$I$ & $0$ & $I$ & $0$ & $0$ & $0$ & $0$\\
$0$ & $I$ & $0$ & $\Gamma_{2}^{-1}\left(  D^{-1}\right)  $ & $0$ & $0$ & $0$\\
&  &  &  & $0$ & $0$ & $0$\\
\multicolumn{2}{c}{$E_{1}^{\prime\prime}\left(  D\right)  $} &
\multicolumn{2}{c}{$E_{1}^{\prime}\left(  D\right)  $} & $0$ & $0$ & $0$\\
&  &  &  & $0$ & $0$ & $0$%
\end{tabular}
\ \right\vert
\begin{array}
[c]{ccccccc}%
0 & 0 & 0 & 0 & 0 & 0 & 0\\
0 & 0 & 0 & 0 & 0 & 0 & 0\\
\Gamma_{1}\left(  D\right)  & 0 & \Gamma_{1}\left(  D\right)  & 0 & 0 & 0 &
0\\
0 & I  & 0 & \Gamma_{2}\left(  D\right)  &
E_{2,2a}^{^{\prime}}\left(  D\right)  & L\left(  D\right)  & 0\\
0 & 0 & 0 & 0 & \Gamma\left(  D\right)  & 0 & 0
\end{array}
\right]  .
\end{equation}%
\end{proof}%
Alice's part of the above quantum check matrix and the last three sets of rows
are equivalent to the quantum check matrix in (\ref{eq:big-QCM}). Alice then
performs all finite-depth encoding operations (column operations)\ in
(\ref{eq:first-finite-depth-encode}-\ref{eq:big-QCM}) in reverse order to
obtain the desired quantum check matrix in the statement of the theorem.
Decoding consists of performing all the operations in
(\ref{eq:first-finite-depth-encode}-\ref{eq:big-QCM}) and then applying the
decoding operations in (\ref{eq:first-inf-depth-decode}%
-\ref{eq:last-inf-depth-decode}). The entanglement-assisted rate of the above
code is $k/n$ because the code uses a noisy quantum communication channel $n$
times per frame to send $k$ information qubits per frame.

\section{Example}

\label{sec:free-ent-example}We now present an example that begins with the
same generators as those in the example from the previous chapter. We begin
with the following two Pauli generators:%
\begin{equation}
\left(  \cdots|IIII|ZXZI|ZZIZ|IIII|\cdots\right)  ,\ \ \ \ \ \ \left(
\cdots|IIII|XYXI|XXIX|IIII|\cdots\right)  .\nonumber
\end{equation}
We write the above two generators as a quantum check matrix:%
\begin{equation}
\left[  \left.
\begin{array}
[c]{cccc}%
1+D & D & 1 & D\\
0 & 1 & 0 & 0
\end{array}
\right\vert
\begin{array}
[c]{cccc}%
0 & 1 & 0 & 0\\
1+D & 1+D & 1 & D
\end{array}
\right]  \label{eq:multigen-stab}%
\end{equation}
We encode two information qubits per frame with the help of two ebits. The
stabilizer matrix for the unencoded qubit stream is as follows:%
\begin{equation}
\left[  \left.
\begin{array}
[c]{cccccc}%
0 & 1 & 1 & 0 & 0 & 0\\
1 & 0 & 0 & 1 & 0 & 0\\
0 & 0 & 0 & 0 & 0 & 0\\
0 & 0 & 0 & 0 & 0 & 0
\end{array}
\right\vert
\begin{array}
[c]{cccccc}%
0 & 0 & 0 & 0 & 0 & 0\\
0 & 0 & 0 & 0 & 0 & 0\\
0 & 1 & 1 & 0 & 0 & 0\\
1 & 0 & 0 & 1 & 0 & 0
\end{array}
\right]
\end{equation}
Rows one and three correspond to one ebit and rows two and four correspond to
the other. Multiply row one by $D$ and add the result to row three, multiply
row one by $1+D^{-1}+D^{2}$ and add the result to row four, and multiply row
two by $1+D^{-2}$ and add the result to row four. These row operations give
the following equivalent stabilizer:%
\begin{equation}
\left[  \left.
\begin{array}
[c]{cccccc}%
0 & 1 & 1 & 0 & 0 & 0\\
1 & 0 & 0 & 1 & 0 & 0\\
0 & D & D & 0 & 0 & 0\\
1+D^{-2} & 1+D^{-1}+D^{2} & 1+D^{-1}+D^{2} & 1+D^{-2} & 0 & 0
\end{array}
\right\vert
\begin{array}
[c]{cccccc}%
0 & 0 & 0 & 0 & 0 & 0\\
0 & 0 & 0 & 0 & 0 & 0\\
0 & 1 & 1 & 0 & 0 & 0\\
1 & 0 & 0 & 1 & 0 & 0
\end{array}
\right]  . \label{eq:unencoded-free-ent-example}%
\end{equation}
Figure~\ref{fig:science-conv}\ illustrates the operations that transform the
unencoded stabilizer to the encoded one in an online encoding circuit.
\begin{figure}
[ptb]
\begin{center}
\includegraphics[
natheight=11.846200in,
natwidth=10.133900in,
height=5.3074in,
width=4.8092in
]%
{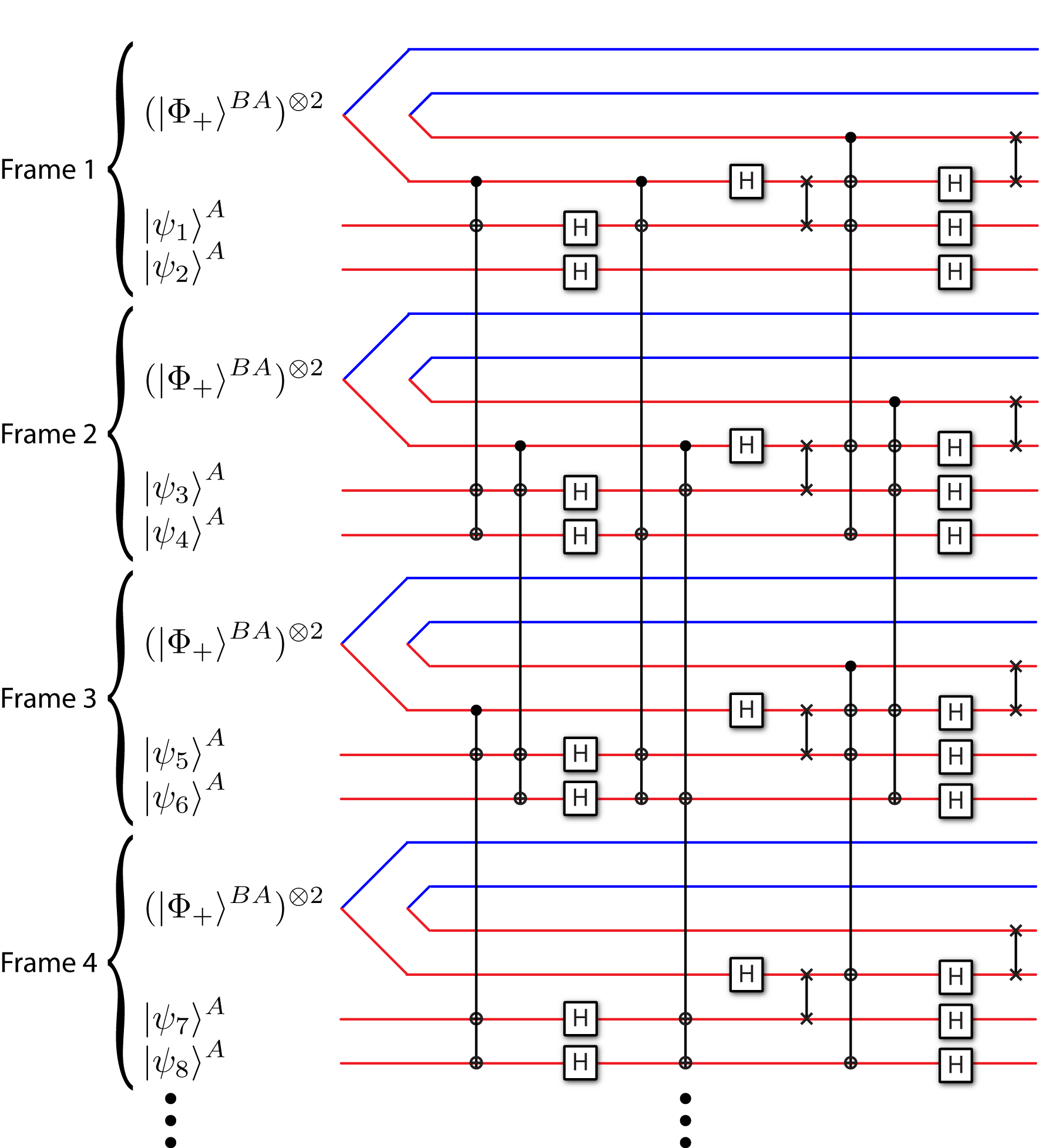}%
\caption{(Color online) An online encoding circuit for an
entanglement-assisted quantum convolutional code. The receiver Bob possesses
the first two qubits in the two ebits and the sender Alice possesses the
second two qubits in the two ebits. The sender encodes two information qubits
per frame with the help of her half of the two ebits.}%
\label{fig:science-conv}%
\end{center}
\end{figure}
The final stabilizer is as follows%
\[
\left[  \left.
\begin{array}
[c]{cccccc}%
0 & 1 & 0 & 1 & 0 & 0\\
1 & 0 & 0 & 0 & 0 & 0\\
0 & D & 1+D & D & 1 & D\\
1+D^{-2} & D^{-1}+1+D & 0 & 1 & 0 & 0
\end{array}
\right\vert
\begin{array}
[c]{cccccc}%
0 & 0 & 0 & 0 & 0 & 0\\
0 & 0 & 1 & 0 & 0 & 0\\
0 & 1 & 0 & 1 & 0 & 0\\
1 & 0 & 1+D & 1+D & 1 & D
\end{array}
\right]  .
\]
Compare Alice's Paulis in the last two rows of the above matrix to the quantum
check matrix in (\ref{eq:multigen-stab}). We have constructed a code with the
same error-correcting properties because these two matrices are equivalent.
The entanglement-assisted rate of the above code is 1/2 because it encodes two
information qubits for every four uses of the noisy quantum channel.

Consider the following two operators:%
\[
\left[  \left.
\begin{array}
[c]{cccccc}%
0 & 0 & D^{-1} & 1+D+D^{-2} & 0 & 0\\
0 & 0 & 0 & 1+D^{2} & 0 & 0
\end{array}
\right\vert
\begin{array}
[c]{cccccc}%
0 & 0 & 1 & 0 & 0 & 0\\
0 & 0 & 0 & 1 & 0 & 0
\end{array}
\right]  .
\]
The first row anticommutes with the first row in
(\ref{eq:unencoded-free-ent-example}) and commutes with all other rows in
(\ref{eq:unencoded-free-ent-example}). The second row anticommutes with the
second row in (\ref{eq:unencoded-free-ent-example}) and commutes with all
other rows in (\ref{eq:unencoded-free-ent-example}). These commutation
relations imply that the above operators are useful for encoding classical
information in a superdense-coding-like fashion. These operators encode two
classical bits into the code and make use of the first two rows in
(\ref{eq:unencoded-free-ent-example}) instead of just \textquotedblleft
wasting\textquotedblright\ them. Measuring the first two rows in
(\ref{eq:unencoded-free-ent-example}) reveals the values of the two classical
bits. We can determine the encoded versions of these \textquotedblleft
classical-information-encoding\textquotedblright\ operators by tracing how the
operators change in the Heisenberg picture through the rest of the encoding
circuit. We can use these two classical bits and consume one ebit to teleport
an additional information qubit. This technique boosts the
entanglement-assisted rate of this code from $1/2$\ to $3/4$ because we can
now encode three information qubits for every four uses of the noisy quantum
communication channel.

\section{Closing Remarks}

The ``free entanglement'' approach of this chapter uses entanglement less efficiently than the protocol
in the previous chapter. It does not require expanding a set of generators
and therefore does not require a heuristic convergence argument as do the codes of the previous chapter. The free entanglement method 
always works and results in encoding and decoding
circuits that act on smaller numbers of qubits than the circuits in the previous chapter.

\chapter{Unified Quantum Convolutional Coding}

\label{chp:unified}\begin{saying}
``Never chase after a bus, a girl, or a unifying theory.\\There'll always be another one coming along.''\\
---John Archibald Wheeler (1911--2008)
\end{saying}In this chapter, we design a framework for \textquotedblleft
grandfather\textquotedblright\ quantum convolutional codes. Our grandfather
codes are useful for the simultaneous transmission of classical and quantum
information. Rather than using block codes for this purpose, we design quantum
convolutional codes.\footnote{Kremsky, Hsieh, and Brun address the
formulation of grandfather block codes in Ref.~\cite{arx2008kremsky}.} Our
technique incorporates many of the known techniques for quantum
coding:\ subsystem codes \cite{kribs:180501,poulin:230504},
entanglement-assisted codes \cite{science2006brun}, convolutional codes
\cite{arxiv2004olliv,isit2006grassl,ieee2007forney}, and classical-quantum
mixed coding \cite{cmp2005dev,beny:100502,arx2008kremsky}. The goal of our
technique is to provide a formalism for designing codes that approach the
optimal triple trade-off rates in the grandfather resource inequality in
(\ref{eq:GF}).

We structure this chapter as follows. Section~\ref{sec:hybrid} details our
\textquotedblleft grandfather\textquotedblright\ quantum convolutional codes.
We explicitly show how to encode a stream of classical-quantum information
using finite-depth operations and discuss the error-correcting properties of
our codes. We end with an example of a grandfather quantum convolutional code.
We discuss which errors the code corrects actively and others that it corrects passively.

\section{Grandfather Quantum Convolutional Codes}

\label{sec:hybrid}We detail the\ stabilizer formalism for our grandfather
quantum convolutional codes and describe how they operate. This formalism is a
significant extension of the entanglement-assisted formalism.

An $\left[  n,k,l;r,c\right]  $ grandfather quantum convolutional code encodes
$k$ information qubits and $l$ information classical bits with the help of $c$
ebits, $a=n-k-l-c-r$ ancilla qubits, and $r$ gauge qubits. Each input frame
includes the following:

\begin{enumerate}
\item Alice's half of $c$ ebits in the state $\left\vert \Phi^{+}\right\rangle
$.

\item $a=n-k-c-l-r$ ancilla qubits in the state $\left\vert 0\right\rangle $.

\item $r$ gauge qubits (which can be in any arbitrary state $\sigma$).

\item $l$ classical information bits $x^{1}\cdots x^{l}$, given by a
computational basis state $\left\vert x\right\rangle =X^{x^{1}}\otimes
\cdots\otimes X^{x^{l}}\left\vert 0\right\rangle ^{\otimes l}$.

\item $k$ information qubits in a state $\left\vert \psi\right\rangle
$.\footnote{This statement is not entirely true because the information qubits
can be entangled across multiple frames, or with an external system, but we
use it to illustrate the idea.}
\end{enumerate}

The left side of Figure~\ref{fig:example}\ shows an example initial qubit
stream before an encoding circuit operates on it.

The stabilizer matrix $S_{0}\left(  D\right)  $\ for the initial qubit stream
is as follows:%
\begin{equation}
S_{0}\left(  D\right)  =\left[  \left.
\begin{array}
[c]{cccccc}%
I & I & 0 & 0 & 0 & 0\\
0 & 0 & 0 & 0 & 0 & 0\\
0 & 0 & I & 0 & 0 & 0
\end{array}
\right\vert
\begin{array}
[c]{cccccc}%
0 & 0 & 0 & 0 & 0 & 0\\
I & I & 0 & 0 & 0 & 0\\
0 & 0 & 0 & 0 & 0 & 0
\end{array}
\right]  ,\label{eq:init-grand-stab}%
\end{equation}
where all identity matrices in the first two sets of rows are $c\times c$, the
identity matrix in the last row is $a\times a$, the three columns of all zeros
in both the \textquotedblleft Z\textquotedblright\ and \textquotedblleft
X\textquotedblright\ matrices are respectively $\left(  a+2c\right)  \times
r$, $\left(  a+2c\right)  \times l$, and $\left(  a+2c\right)  \times k$. The
first two sets of rows stabilize a set of $c$ ebits and the last set of rows
stabilize a set of $a$ ancilla qubits. The first $c$ columns of both the
\textquotedblleft Z\textquotedblright\ and \textquotedblleft
X\textquotedblright\ matrix correspond to halves of ebits that Bob possesses
and the last $n$ columns in both matrices correspond to the qubits that Alice possesses.

Different generators for the grandfather code are important in active error
correction, in passive error correction, and for the identification of the $l$
classical information bits. We first write the unencoded generators that act
on the initial qubit stream. The first subgroup of generators is the
entanglement subgroup $\mathcal{S}_{E,0}$ with the following generators:%
\begin{equation}
S_{E,0}\left(  D\right)  =\left[  \left.
\begin{array}
[c]{ccccc}%
I & 0 & 0 & 0 & 0\\
0 & 0 & 0 & 0 & 0
\end{array}
\right\vert
\begin{array}
[c]{ccccc}%
0 & 0 & 0 & 0 & 0\\
I & 0 & 0 & 0 & 0
\end{array}
\right]  . \label{eq:grand-egroup}%
\end{equation}
The above generators are equivalent to the first two sets of rows in
(\ref{eq:init-grand-stab})\ acting on Alice's $n$ qubits. The next subgroup is
the isotropic subgroup $\mathcal{S}_{I,0}$ with the following generators:%
\begin{equation}
S_{I,0}\left(  D\right)  =\left[  \left.
\begin{array}
[c]{ccccc}%
0 & I & 0 & 0 & 0
\end{array}
\right\vert
\begin{array}
[c]{ccccc}%
0 & 0 & 0 & 0 & 0
\end{array}
\right]  .
\end{equation}
The above generators are equivalent to the last set of rows in
(\ref{eq:init-grand-stab})\ acting on Alice's $n$ qubits. The encoded versions
of both of the above two matrices are important in the active correction of
errors. The next subgroup is the gauge subgroup $\mathcal{S}_{G,0}$ whose
generators are as follows:%
\begin{equation}
S_{G,0}\left(  D\right)  =\left[  \left.
\begin{array}
[c]{ccccc}%
0 & 0 & I & 0 & 0\\
0 & 0 & 0 & 0 & 0
\end{array}
\right\vert
\begin{array}
[c]{ccccc}%
0 & 0 & 0 & 0 & 0\\
0 & 0 & I & 0 & 0
\end{array}
\right]  .
\end{equation}
The generators in $\mathcal{S}_{G,0}$ correspond to quantum operations that
have no effect on the encoded quantum information and therefore represent a
set of errors to which the code is immune. The last subgroup is the classical
subgroup $\mathcal{S}_{C,0}$ with generators%
\begin{equation}
S_{C,0}\left(  D\right)  =\left[  \left.
\begin{array}
[c]{ccccc}%
0 & 0 & 0 & I & 0
\end{array}
\right\vert
\begin{array}
[c]{ccccc}%
0 & 0 & 0 & 0 & 0
\end{array}
\right]  . \label{eq:grand-cgroup}%
\end{equation}
The grandfather code passively corrects errors corresponding to the encoded
version of the above generators because the initial qubit stream is immune to
the action of operators in $\mathcal{S}_{C,0}$ (up to a global phase). Alice
could measure the generators in $\mathcal{S}_{C,0}$ to determine the classical
information in each frame. Unlike quantum information, it is possible to
measure classical information without disturbing it.

Alice performs an encoding circuit with finite-depth operations to encode her
stream of qubits before sending them over the noisy quantum channel. The
encoding circuit transforms the initial stabilizer $S_{0}\left(  D\right)  $
to the encoded stabilizer $S\left(  D\right)  $\ as follows:%
\begin{equation}
S\left(  D\right)  =\left[  \left.
\begin{array}
[c]{cc}%
I & Z_{E1}\left(  D\right) \\
0 & Z_{E2}\left(  D\right) \\
0 & Z_{I}\left(  D\right)
\end{array}
\right\vert
\begin{array}
[c]{cc}%
0 & X_{E1}\left(  D\right) \\
I & X_{E2}\left(  D\right) \\
0 & X_{I}\left(  D\right)
\end{array}
\right]  , \label{eq:grand-stab}%
\end{equation}
where $Z_{E1}\left(  D\right)  $, $X_{E1}\left(  D\right)  $, $Z_{E2}\left(
D\right)  $ and $X_{E2}\left(  D\right)  $ are $c\times n$-dimensional,
$Z_{I}\left(  D\right)  $ and $X_{I}\left(  D\right)  $ are $a\times
n$-dimensional. The encoding circuit affects only the rightmost $n$ entries in
both the \textquotedblleft Z\textquotedblright\ and \textquotedblleft
X\textquotedblright\ matrix of $S_{0}\left(  D\right)  $ because these are the
qubits in Alice's possession. It transforms $S_{E,0}\left(  D\right)  $,
$S_{I,0}\left(  D\right)  $, $S_{G,0}\left(  D\right)  $, and $S_{C,0}\left(
D\right)  $ as follows:%
\begin{align}
S_{E}\left(  D\right)   &  =\left[  \left.
\begin{array}
[c]{c}%
Z_{E1}\left(  D\right) \\
Z_{E2}\left(  D\right)
\end{array}
\right\vert
\begin{array}
[c]{c}%
X_{E1}\left(  D\right) \\
X_{E2}\left(  D\right)
\end{array}
\right]  ,\ \ \ S_{I}\left(  D\right)  =\left[  \left.
\begin{array}
[c]{c}%
Z_{I}\left(  D\right)
\end{array}
\right\vert
\begin{array}
[c]{c}%
X_{I}\left(  D\right)
\end{array}
\right]  ,\\
S_{G}\left(  D\right)   &  =\left[  \left.
\begin{array}
[c]{c}%
Z_{G1}\left(  D\right) \\
Z_{G2}\left(  D\right)
\end{array}
\right\vert
\begin{array}
[c]{c}%
X_{G1}\left(  D\right) \\
X_{G2}\left(  D\right)
\end{array}
\right]  ,\ \ \ S_{C}\left(  D\right)  =\left[  \left.
\begin{array}
[c]{c}%
Z_{C}\left(  D\right)
\end{array}
\right\vert
\begin{array}
[c]{c}%
X_{C}\left(  D\right)
\end{array}
\right]  ,
\end{align}
where $Z_{G1}\left(  D\right)  $, $X_{G1}\left(  D\right)  $, $Z_{G2}\left(
D\right)  $ and $X_{G2}\left(  D\right)  $ are $r\times n$-dimensional and
$Z_{C}\left(  D\right)  $ and $X_{C}\left(  D\right)  $ are $l\times
n$-dimensional. The above polynomial matrices have the same commutation
relations as their corresponding unencoded polynomial matrices in
(\ref{eq:grand-egroup}-\ref{eq:grand-cgroup}) and respectively generate the
entanglement subgroup $\mathcal{S}_{E}$, the isotropic subgroup $\mathcal{S}%
_{I}$, the gauge subgroup $\mathcal{S}_{G}$, and the classical subgroup
$\mathcal{S}_{C}$.

A grandfather quantum convolutional code operates as follows.

\begin{enumerate}
\item Alice begins with an initial qubit stream as above. She performs the
finite-depth encoding operations corresponding to a specific grandfather
quantum convolutional code.

\item She sends the encoded qubits online over the noisy quantum communication
channel. The code passively protects against errors in $\left\langle
\mathcal{S}_{I},\mathcal{S}_{G},\mathcal{S}_{C}\right\rangle $.

\item Bob combines the received qubits with his half of the ebits in each
frame. He obtains the error syndrome by measuring the generators in
(\ref{eq:grand-stab}). He processes these syndrome bits with a classical error
estimation algorithm to diagnose errors and applies recovery operations to
reverse the errors.

\item He then performs the inverse of the encoding circuit to recover the
initial qubit stream with the information qubits and the classical information
bits. He recovers the classical information bits either by measuring the
generators in $\mathcal{S}_{C}$ before decoding or the generators in
$\mathcal{S}_{C,0}$ after decoding.
\end{enumerate}

A grandfather quantum convolutional code corrects errors in a Pauli error set
$\mathcal{E}$ that obey one of the following conditions $\forall E_{a}%
,E_{b}\in\mathcal{E}$:%
\[
\exists\ g\in\left\langle \mathcal{S}_{I},\mathcal{S}_{E}\right\rangle
:\left\{  g,E_{a}^{\dag}E_{b}\right\}  =0\text{ \ or \ }E_{a}^{\dag}E_{b}%
\in\left\langle \mathcal{S}_{I},\mathcal{S}_{G},\mathcal{S}_{C}\right\rangle
.
\]
It corrects errors that anticommute with generators in $\left\langle
\mathcal{S}_{I},\mathcal{S}_{E}\right\rangle $ by employing a classical error
estimation algorithm. The code passively protects against errors in the group
$\left\langle \mathcal{S}_{I},\mathcal{S}_{G},\mathcal{S}_{C}\right\rangle $.

Our scheme for quantum convolutional coding incorporates many of the known
techniques for quantum error correction. It can take full advantage of the
benefits of these different techniques.

\section{Example}

\label{sec:example}%
We present an example of a grandfather quantum convolutional code in this
section. The code protects one information qubit and classical bit with the
help of an ebit, an ancilla qubit, and a gauge qubit. The first frame of input
qubits has the state%
\begin{equation}
\rho_{0}=\left\vert \Phi^{+}\right\rangle \left\langle \Phi^{+}\right\vert
\otimes\left\vert 0\right\rangle \left\langle 0\right\vert \otimes\sigma
_{0}\otimes\left\vert x_{0}\right\rangle \left\langle x_{0}\right\vert
\otimes\left\vert \psi_{0}\right\rangle \left\langle \psi_{0}\right\vert ,
\end{equation}
where $\left\vert \Phi^{+}\right\rangle $ is the ebit, $\left\vert
0\right\rangle $ is the ancilla qubit, $\sigma_{0}$ is an arbitrary state for
the gauge qubit, $\left\vert x_{0}\right\rangle $ is a classical bit
represented by state $\left\vert 0\right\rangle $ or $\left\vert
1\right\rangle $, and $\left\vert \psi_{0}\right\rangle $ is one information
qubit equal to $\alpha_{0}\left\vert 0\right\rangle +\beta_{0}\left\vert
1\right\rangle $. The states of the other input frames have a similar form
though recall that information qubits can be entangled across multiple frames.

The initial stabilizer for the code is as follows:%
\[
S_{0}\left(  D\right)  =\left[  \left.
\begin{array}
[c]{cccccc}%
1 & 1 & 0 & 0 & 0 & 0\\
0 & 0 & 0 & 0 & 0 & 0\\
0 & 0 & 1 & 0 & 0 & 0
\end{array}
\right\vert
\begin{array}
[c]{cccccc}%
0 & 0 & 0 & 0 & 0 & 0\\
1 & 1 & 0 & 0 & 0 & 0\\
0 & 0 & 0 & 0 & 0 & 0
\end{array}
\right]  .
\]
The first two rows stabilize the ebit shared between Alice and Bob. Bob
possesses the half of the ebit in column one and Alice possesses the half of
the ebit in column two in both the left and right matrix. The third row
stabilizes the ancilla qubit. We name Alice's qubits one through five (they
are actually two through six in the above matrix from the left to the right).

The generators for the initial entanglement subgroup $\mathcal{S}_{E,0}$,
isotropic subgroup $\mathcal{S}_{I,0}$, gauge subgroup $\mathcal{S}_{G,0}$,
and classical subgroup $\mathcal{S}_{C,0}$ are respectively as follows:%
\begin{align*}
S_{E,0}\left(  D\right)   &  =\left[  \left.
\begin{array}
[c]{ccccc}%
1 & 0 & 0 & 0 & 0\\
0 & 0 & 0 & 0 & 0
\end{array}
\right\vert
\begin{array}
[c]{ccccc}%
0 & 0 & 0 & 0 & 0\\
1 & 0 & 0 & 0 & 0
\end{array}
\right]  ,\\
S_{I,0}\left(  D\right)   &  =\left[  \left.
\begin{array}
[c]{ccccc}%
0 & 1 & 0 & 0 & 0
\end{array}
\right\vert
\begin{array}
[c]{ccccc}%
0 & 0 & 0 & 0 & 0
\end{array}
\right]  ,\\
S_{G,0}\left(  D\right)   &  =\left[  \left.
\begin{array}
[c]{ccccc}%
0 & 0 & 1 & 0 & 0\\
0 & 0 & 0 & 0 & 0
\end{array}
\right\vert
\begin{array}
[c]{ccccc}%
0 & 0 & 0 & 0 & 0\\
0 & 0 & 1 & 0 & 0
\end{array}
\right]  ,\\
S_{C,0}\left(  D\right)   &  =\left[  \left.
\begin{array}
[c]{ccccc}%
0 & 0 & 0 & 1 & 0
\end{array}
\right\vert
\begin{array}
[c]{ccccc}%
0 & 0 & 0 & 0 & 0
\end{array}
\right]  .
\end{align*}
The sender performs the following finite-depth operations (order is from left
to right and top to bottom):%
\begin{align*}
&  H\left(  2\right)  \ C\left(  2,3,D\right)  \ C\left(  2,4,1+D\right)
\ C\left(  2,5,D\right)  \ H\left(  3,4,5\right)  \\
&  C\left(  2,3,D\right)  \ C\left(  2,5,D\right)  \ H\left(  2\right)
\ C\left(  1,2,D\right)  \ C\left(  1,4,1+D\right)  \\
&  C\left(  1,5,1+D\right)  \ H\left(  1,2,3,4,5\right)  \ C\left(
1,3,D\right)  \ C\left(  1,4,1+D\right)  \\
&  C\left(  1,5,1+D\right)  \ S\left(  1,4\right)  .
\end{align*}
where the notation for the above encoding operations was established in
Section~\ref{sec:finite-depth-clifford}.
The initial stabilizer matrix
$S_{0}\left(  D\right)  $ transforms to $S\left(  D\right)  $ under these
encoding operations, where%
\begin{equation}
S\left(  D\right)  =\left[  \left.
\begin{array}
[c]{cccccc}%
1 & 0 & 0 & 0 & 0 & 0\\
0 & 1+D & D & 0 & 1 & 1+D\\
0 & 0 & 0 & D & D & D
\end{array}
\right\vert
\begin{array}
[c]{cccccc}%
0 & 1+D & 0 & D & 1 & 1+D\\
1 & 0 & 0 & 0 & 0 & 0\\
0 & 0 & 1 & 0 & 1 & 1
\end{array}
\right]  .
\end{equation}
The generators for the different subgroups transform respectively as follows:%
\begin{align*}
S_{E}\left(  D\right)   &  =\left[  \left.
\begin{array}
[c]{ccccc}%
0 & 0 & 0 & 0 & 0\\
1+D & D & 0 & 1 & 1+D
\end{array}
\right\vert
\begin{array}
[c]{ccccc}%
1+D & 0 & D & 1 & 1+D\\
0 & 0 & 0 & 0 & 0
\end{array}
\right]  ,\\
S_{I}\left(  D\right)   &  =\left[  \left.
\begin{array}
[c]{ccccc}%
0 & 0 & D & D & D
\end{array}
\right\vert
\begin{array}
[c]{ccccc}%
0 & 1 & 0 & 1 & 1
\end{array}
\right]  ,\\
S_{G}\left(  D\right)   &  =\left[  \left.
\begin{array}
[c]{ccccc}%
0 & \frac{1}{D} & 1 & \frac{1}{D} & 0\\
0 & \frac{1}{D} & 0 & 0 & 0
\end{array}
\right\vert
\begin{array}
[c]{ccccc}%
0 & 0 & 0 & 0 & 0\\
0 & 0 & 1 & 0 & 0
\end{array}
\right]  ,\\
S_{C}\left(  D\right)   &  =\left[  \left.
\begin{array}
[c]{ccccc}%
1 & 1+D^{-1} & 0 & 1+D^{-1} & 0
\end{array}
\right\vert
\begin{array}
[c]{ccccc}%
0 & 0 & 0 & 0 & 0
\end{array}
\right]  .
\end{align*}
The code actively protects against an arbitrary single-qubit error in every
other frame. One can check that the syndromes of the stabilizer in $S\left(
D\right)  $ satisfy this property. Consider the Pauli generators corresponding
to the generators in the entanglement subgroup and the isotropic subgroup:%
\begin{equation}
\cdots\left\vert
\begin{array}
[c]{ccccc}%
X & I & I & X & X\\
Z & I & I & Z & Z\\
I & X & I & X & X
\end{array}
\right\vert \left.
\begin{array}
[c]{ccccc}%
X & I & X & I & X\\
Z & Z & I & I & Z\\
I & I & Z & Z & Z
\end{array}
\right\vert \cdots,
\end{equation}
where all other entries in the left and right directions are tensor products
of the identity. We can use a table-lookup syndrome-based algorithm to
determine the error-correcting capability of the code. The method is similar
to the technique originally outlined in detail in Ref. \cite{ieee2007forney}.
The syndrome vector $s$\ consists of six bits where $s=s_{1}\cdots s_{6}$. The
first bit $s_{1}$ is one if the error anticommutes with the operator $XIIXX$
in the first part of the first generator above and zero otherwise. The second
bit $s_{2}$\ is one if the error anticommutes with the operator $XIXIX$ in the
delayed part of the first generator above and zero otherwise. The third
through sixth bits follow a similar pattern for the second and third
generators above. Table~\ref{TableKey}\ lists all single-qubit errors over
five qubits and their corresponding syndromes. The code corrects an arbitrary
single-qubit error in every other frame using this algorithm because the
syndromes are all unique. A syndrome-based Viterbi algorithm might achieve
better performance than the simple syndrome table-lookup algorithm outlined
above.%
\begin{table}[tbp] \centering
\begin{tabular}
[c]{l|l||l|l||l|l}\hline\hline
\textbf{Error} & \textbf{Syndrome} & \textbf{Error} & \textbf{Syndrome} &
\textbf{Error} & \textbf{Syndrome}\\\hline\hline
$X_{1}$ & 001100 & $X_{3}$ & 000001 & $X_{5}$ & 001101\\\hline
$Y_{1}$ & 111100 & $Y_{3}$ & 010001 & $Y_{5}$ & 111111\\\hline
$Z_{1}$ & 110000 & $Z_{3}$ & 010000 & $Z_{5}$ & 110010\\\hline
$X_{2}$ & 000100 & $X_{4}$ & 001001 &  & \\\hline
$Y_{2}$ & 000110 & $Y_{4}$ & 101011 &  & \\\hline
$Z_{2}$ & 000010 & $Z_{4}$ & 100010 &  & \\\hline\hline
\end{tabular}
\caption{A list of possible single-qubit errors in a particular frame and the corresponding
syndrome vector. The syndrome corresponding to any single-qubit error is unique. The code therefore
corrects an arbitrary single-qubit error in every other frame.}\label{TableKey}%
\end{table}%

This code also has passive protection against errors in $\left\langle
\mathcal{S}_{I},\mathcal{S}_{G},\mathcal{S}_{C}\right\rangle $. The Pauli form
of the errors in this group span over three frames and are as follows:%
\begin{equation}
\cdots\left\vert
\begin{array}
[c]{l}%
IIIII\\
IZIZI\\
IZIII\\
IZIZI
\end{array}
\right\vert
\begin{array}
[c]{l}%
IXIXX\\
IIZII\\
IIXII\\
ZZIZI
\end{array}
\left\vert
\begin{array}
[c]{l}%
IIZZZ\\
IIIII\\
IIIII\\
IIIII
\end{array}
\right\vert \cdots
\end{equation}
The smallest weight errors in this group have weight two and three. The code
passively corrects the above errors or any product of them or any five-qubit
shift of them.

There is a trade-off between passive error correction and the ability to
encode quantum information as discussed in Ref.~\cite{hsieh:062313}. One can
encode more quantum information by dropping the gauge group and instead
encoding extra information qubits. The gauge generators then become logical $X$ and $Z$
operators for the extra encoded qubits. One can also turn classical bits
into information qubits by dropping the generators in the classical subgroup. These
generators then become logical $Z$ operators for the extra encoded qubits. By
making the above replacements, the code loses some of its ability to correct passively, but we gain
a higher quantum information rate.
On the other hand, if we replace a gauge qubit with an
ancilla qubit or an ebit, we gain the ability to correct extra errors.
The trade-off now is that replacing gauge qubits with ebits or ancillas
enhances the active error-correcting capability of the code, but increases the overall complexity of error
correction.

\section{Closing Remarks}

We have presented a framework and a representative example for grandfather
quantum convolutional codes. We have explicitly shown how these codes operate,
and how to encode and decode a classical-quantum information stream by using
ebits, ancilla qubits, and gauge qubits for quantum redundancy. The ultimate
goal for this theory is to find quantum convolutional codes that might play an
integral part in larger quantum codes that approach the grandfather capacity
\cite{prep2008dev}. One useful line of investigation may be to combine this
theory with the recent quantum turbo-coding theory \cite{arx2007poulin}.

\chapter{Convolutional Entanglement Distillation}

\label{chp:ced}\begin{saying}
A chap at the ``Entanglement Distillery,''\\
Was drunk so they gave him the pillory,\\
They'd foul dirty ebits,\\
He said, ``Convolutional circuits!''\\
So they augmented the quantum artillery.
\end{saying}The goal of entanglement distillation resembles the goal of
quantum error correction \cite{PhysRevLett.76.722,PhysRevA.54.3824}. An
entanglement distillation protocol extracts noiseless, maximally-entangled
ebits from a larger set of noisy ebits. A sender and receiver can use these
noiseless ebits as a resource for several quantum communication protocols
\cite{PhysRevLett.69.2881,PhysRevLett.70.1895}.

Bennett et al. showed that a strong connection exists between quantum
error-correcting codes and entanglement distillation and demonstrated a method
for converting an arbitrary quantum error-correcting code into a one-way
entanglement distillation protocol \cite{PhysRevA.54.3824}. A one-way
entanglement distillation protocol utilizes one-way classical communication
between sender and receiver to carry out the distillation procedure. Shor and
Preskill\ improved upon Bennett et al.'s method by avoiding the use of ancilla
qubits and gave a simpler method for converting an arbitrary CSS quantum
error-correcting code into an entanglement distillation protocol
\cite{PhysRevLett.85.441}. Nielsen and Chuang showed how to convert a
stabilizer quantum error-correcting code into a stabilizer entanglement
distillation protocol \cite{book2000mikeandike}. Luo and Devetak then
incorporated shared entanglement to demonstrate how to convert an
entanglement-assisted stabilizer code into an entanglement-assisted
entanglement distillation protocol \cite{luo:010303}. All of the above
constructions exploit the relationship between quantum error correction and
entanglement distillation---we further exploit the connection in this chapter
by forming a \textit{convolutional} entanglement distillation protocol.

In this last chapter, our main contribution is a theory of convolutional
entanglement distillation. Our theory allows us to import the entirety of
classical convolutional coding theory for use in entanglement distillation.
The task of finding a good convolutional entanglement distillation protocol
now becomes the well-established task of finding a good classical
convolutional code.

We begin in Section~\ref{sec:conv-ent-dist} by showing how to construct
a\textit{ }convolutional entanglement distillation protocol from an arbitrary
quantum convolutional code. We translate earlier protocols
\cite{PhysRevLett.85.441,book2000mikeandike}\ for entanglement distillation of
a block of noisy ebits to the convolutional setting. A convolutional entanglement distillation
protocol has the benefit of distilling entanglement
\textquotedblleft online.\textquotedblright\ This online property is useful
because the sender and receiver can distill entanglement \textquotedblleft on
the fly\textquotedblright\ as they obtain more noisy ebits. This translation
from a quantum convolutional code to an entanglement distillation protocol is
useful because it paves the way for our major contribution.

Our major advance is a method for constructing a convolutional entanglement
distillation protocol when the sender and receiver initially share some
noiseless ebits. As stated previously, prior quantum convolutional work
requires the code to satisfy the restrictive self-orthogonality constraint,
and authors performed specialized searches for classical convolutional codes
that meet this constraint
\cite{PhysRevLett.91.177902,arxiv2004olliv,isit2005forney,ieee2007forney}. We
lift this constraint by allowing shared noiseless entanglement. The benefit of
convolutional entanglement distillation with entanglement assistance is that
we can import an \textit{arbitrary} classical binary or quaternary
convolutional code for use in a convolutional entanglement distillation
protocol. The error-correcting properties for the convolutional entanglement
distillation protocol follow directly from the properties of the imported
classical code. Thus we can apply the decades of research on classical
convolutional coding theory with many of the benefits of the convolutional
structure carrying over to the quantum domain.

We organize this chapter as follows. We review stabilizer entanglement
distillation in Section \ref{sec:stabilizer-ent-distill} and
entanglement-assisted entanglement distillation in Section
\ref{sec:stabilizer-ent-assist-ent-distill}. In Section
\ref{sec:conv-ent-dist}, we show how to convert an arbitrary quantum
convolutional code into a convolutional entanglement distillation protocol. In
Section \ref{sec:conv-ent-ent-assist}, we provide several methods and examples
for constructing convolutional entanglement distillation protocols where two
parties possess a few initial noiseless ebits. These initial noiseless ebits
act as a catalyst for the convolutional distillation protocol. The
constructions in Section \ref{sec:conv-ent-ent-assist}\ make it possible to
import an arbitrary classical binary or quaternary convolutional code for use
in convolutional entanglement distillation.

\section{Stabilizer Entanglement Distillation without Entanglement Assistance}

\label{sec:stabilizer-ent-distill}The purpose of an $\left[  n,k\right]
$\ entanglement distillation protocol is to distill $k$ pure ebits from $n$
noisy ebits where $0\leq k\leq n$ \cite{PhysRevLett.76.722,PhysRevA.54.3824}.
The yield of such a protocol is $k/n$. Two parties can then use the noiseless
ebits for quantum communication protocols.
Figure~\ref{fig:block-entanglement-distill} illustrates the operation of a
block entanglement distillation protocol.%
\begin{figure}
[ptb]
\begin{center}
\includegraphics[
natheight=7.639800in,
natwidth=10.253200in,
height=2.2329in,
width=2.9914in
]%
{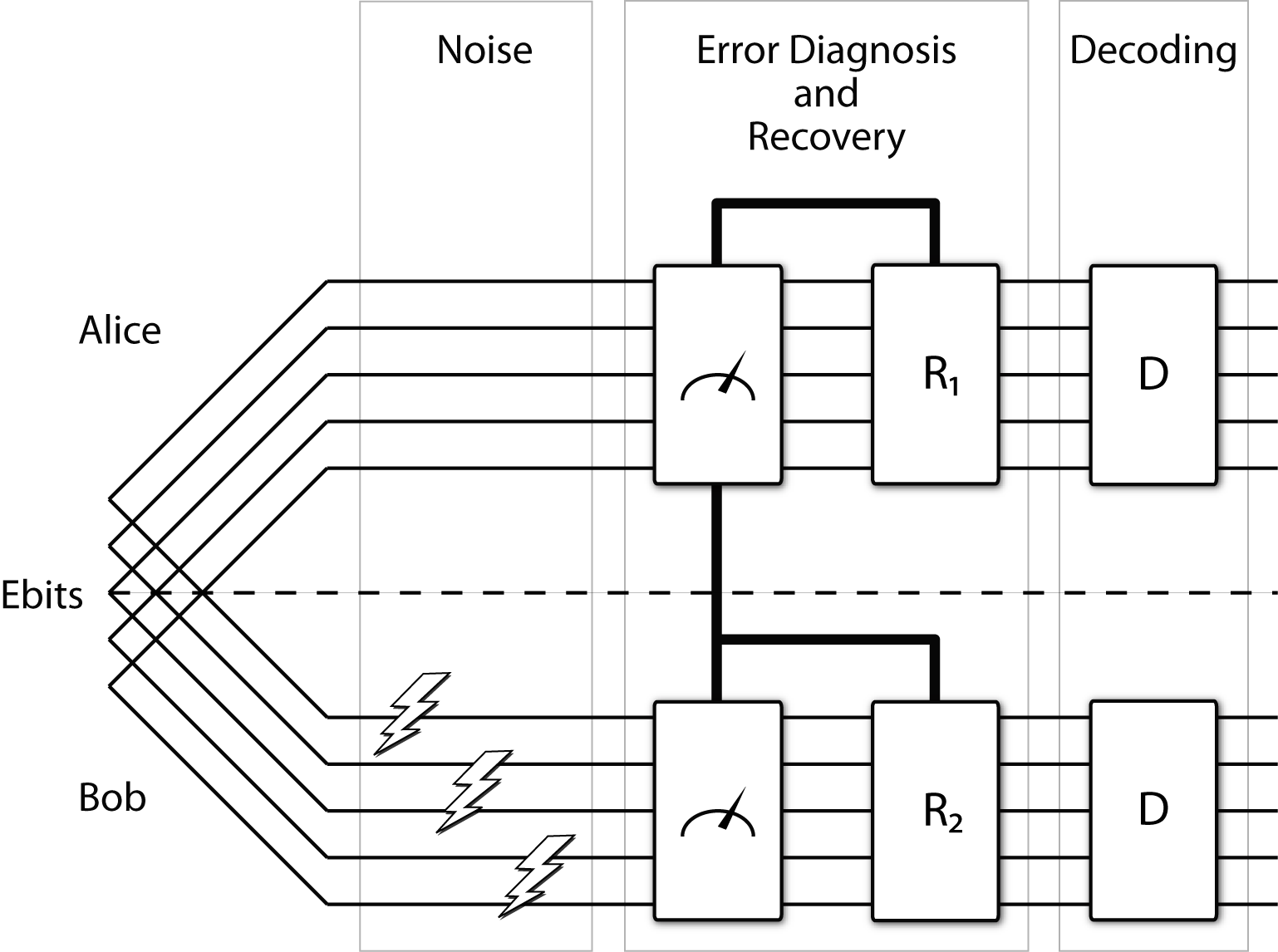}%
\caption{An example of a block entanglement distillation protocol. A sender
creates a set of noisy ebits by sending half of a set of Bell states through a
noisy quantum channel. Both sender and receiver perform multi-qubit
measurements to diagnose channel error. The sender transmits her measurement
results to the receiver over a classical communication channel. Both perform
recovery and decoding operations to obtain a set of noiseless ebits.}%
\label{fig:block-entanglement-distill}%
\end{center}
\end{figure}

The two parties establish a set of shared noisy ebits in the following way.
The sender Alice first prepares $n$ Bell states $\left\vert \Phi
^{+}\right\rangle ^{\otimes n}$ locally. She sends the second qubit of each
pair over a noisy quantum channel to a receiver Bob. Let $\left\vert \Phi
_{n}^{+}\right\rangle $ be the state $\left\vert \Phi^{+}\right\rangle
^{\otimes n}$ rearranged so that all of Alice's qubits are on the left and all
of Bob's qubits are on the right. The noisy channel applies a Pauli error in
the error set $\mathcal{E}\subset\Pi^{n}$ to the set of $n$ qubits sent over
the channel. The sender and receiver then share a set of $n$ noisy ebits of
the form $\left(  \mathbf{I}\otimes\mathbf{A}\right)  \left\vert \Phi_{n}%
^{+}\right\rangle $ where the identity $\mathbf{I}$ acts on Alice's qubits and
$\mathbf{A}$ is some Pauli operator in $\mathcal{E}$ acting on Bob's qubits.

A one-way stabilizer entanglement distillation protocol uses a stabilizer code
for the distillation procedure. Figure~\ref{fig:block-entanglement-distill}
highlights the main features of a stabilizer entanglement distillation
protocol. Suppose the stabilizer $\mathcal{S}$\ for an $\left[  n,k\right]
$\ quantum error-correcting code has generators $g_{1},\ldots,g_{n-k}$. The
distillation procedure begins with Alice measuring the $n-k$ generators in
$\mathcal{S}$. Let $\left\{  \mathbf{P}_{i}\right\}  $ be the set of the
$2^{n-k}$\ projectors that project onto the $2^{n-k}$ orthogonal subspaces
corresponding to the generators in $\mathcal{S}$. The measurement projects
$\left\vert \Phi_{n}^{+}\right\rangle $ randomly onto one of the
$i$\ subspaces. Each $\mathbf{P}_{i}$ commutes with the noisy operator
$\mathbf{A}$\ on Bob's side so that%
\begin{equation}
\left(  \mathbf{P}_{i}\otimes\mathbf{I}\right)  \left(  \mathbf{I}%
\otimes\mathbf{A}\right)  \left\vert \Phi_{n}^{+}\right\rangle =\left(
\mathbf{I}\otimes\mathbf{A}\right)  \left(  \mathbf{P}_{i}\otimes
\mathbf{I}\right)  \left\vert \Phi_{n}^{+}\right\rangle
.\label{eq:distill-proof}%
\end{equation}
The following important \textquotedblleft Bell-state\ matrix
identity\textquotedblright\ holds for an arbitrary matrix $\mathbf{M}$:%
\begin{equation}
\left(  \mathbf{M}\otimes\mathbf{I}\right)  \left\vert \Phi_{n}^{+}%
\right\rangle =\left(  \mathbf{I}\otimes\mathbf{M}^{T}\right)  \left\vert
\Phi_{n}^{+}\right\rangle .
\end{equation}
Then (\ref{eq:distill-proof}) is equal to the following:%
\begin{equation}
\left(  \mathbf{I}\otimes\mathbf{A}\right)  \left(  \mathbf{P}_{i}%
\otimes\mathbf{I}\right)  \left\vert \Phi_{n}^{+}\right\rangle =\left(
\mathbf{I}\otimes\mathbf{A}\right)  \left(  \mathbf{P}_{i}^{2}\otimes
\mathbf{I}\right)  \left\vert \Phi_{n}^{+}\right\rangle =\left(
\mathbf{I}\otimes\mathbf{A}\right)  \left(  \mathbf{P}_{i}\otimes
\mathbf{P}_{i}^{T}\right)  \left\vert \Phi_{n}^{+}\right\rangle .\nonumber
\end{equation}
Therefore each of Alice's projectors $\mathbf{P}_{i}$ projects Bob's qubits
onto a subspace $\mathbf{P}_{i}^{T}$ corresponding to Alice's projected
subspace $\mathbf{P}_{i}$. This operation is one of the \textquotedblleft
weird\textquotedblright\ properties of entanglement because it projects the
set of noisy ebits onto the codespace effectively \textquotedblleft
before\textquotedblright\ the noise acts on it. Alice restores her qubits to
the simultaneous +1-eigenspace of the generators in $\mathcal{S}$. She sends
her measurement results to Bob. Bob measures the generators in $\mathcal{S}$.
Bob combines his measurements with Alice's to determine a syndrome for the
error. He performs a recovery operation on his qubits to reverse the error. He
restores his qubits to the simultaneous +1-eigenspace of the generators in
$\mathcal{S}$. Alice and Bob both perform the decoding unitary corresponding
to stabilizer $\mathcal{S}$\ to convert their $k$ logical ebits to $k$
physical ebits.

\section{Stabilizer Entanglement Distillation with Entanglement Assistance}

\label{sec:stabilizer-ent-assist-ent-distill}Luo and Devetak provided a
straightforward extension of the above protocol \cite{luo:010303}. Their
method converts an entanglement-assisted stabilizer code into an
entanglement-assisted entanglement distillation protocol.

Luo and Devetak form an entanglement distillation protocol that has
entanglement assistance from a few noiseless ebits. The crucial assumption for
an entanglement-assisted entanglement distillation protocol is that Alice and
Bob possess $c$ noiseless ebits in addition to their $n$ noisy ebits. The
total state of the noisy and noiseless ebits is%
\begin{equation}
(\mathbf{I}^{A}\otimes\left(  \mathbf{A\otimes I}\right)  ^{B})\left\vert
\Phi_{n+c}^{+}\right\rangle
\end{equation}
where $\mathbf{I}^{A}$ is the $2^{n+c}\times2^{n+c}$ identity matrix acting on
Alice's qubits and the noisy Pauli operator $\left(  \mathbf{A\otimes
I}\right)  ^{B}$ affects Bob's first $n$ qubits only. Thus the last $c$ ebits
are noiseless, and Alice and Bob have to correct for errors on the first $n$
ebits only.

The protocol proceeds exactly as outlined in the previous section. The only
difference is that Alice and Bob measure the generators in an
entanglement-assisted stabilizer code. Each generator spans over $n+c$ qubits
where the last $c$ qubits are noiseless.

We comment on the yield of this entanglement-assisted entanglement
distillation protocol. An entanglement-assisted code has $n-k$ generators that
each have $n+c$ Pauli entries. These parameters imply that the entanglement
distillation protocol produces $k+c$ ebits. But the protocol consumes $c$
initial noiseless ebits as a catalyst for distillation. Therefore the yield of
this protocol is $k/n$.

In Section \ref{sec:conv-ent-ent-assist}, we exploit this same idea of using a
few noiseless ebits as a catalyst for distillation. The idea is similar in
spirit to that developed in this section, but the mathematics and construction
are different because we perform distillation in a convolutional manner.

\section{Convolutional Entanglement Distillation without Entanglement
Assistance}%

\begin{figure*}
[ptb]
\begin{center}
\includegraphics[
natheight=8.013400in,
natwidth=19.253300in,
height=2.304in,
width=5.518in
]
{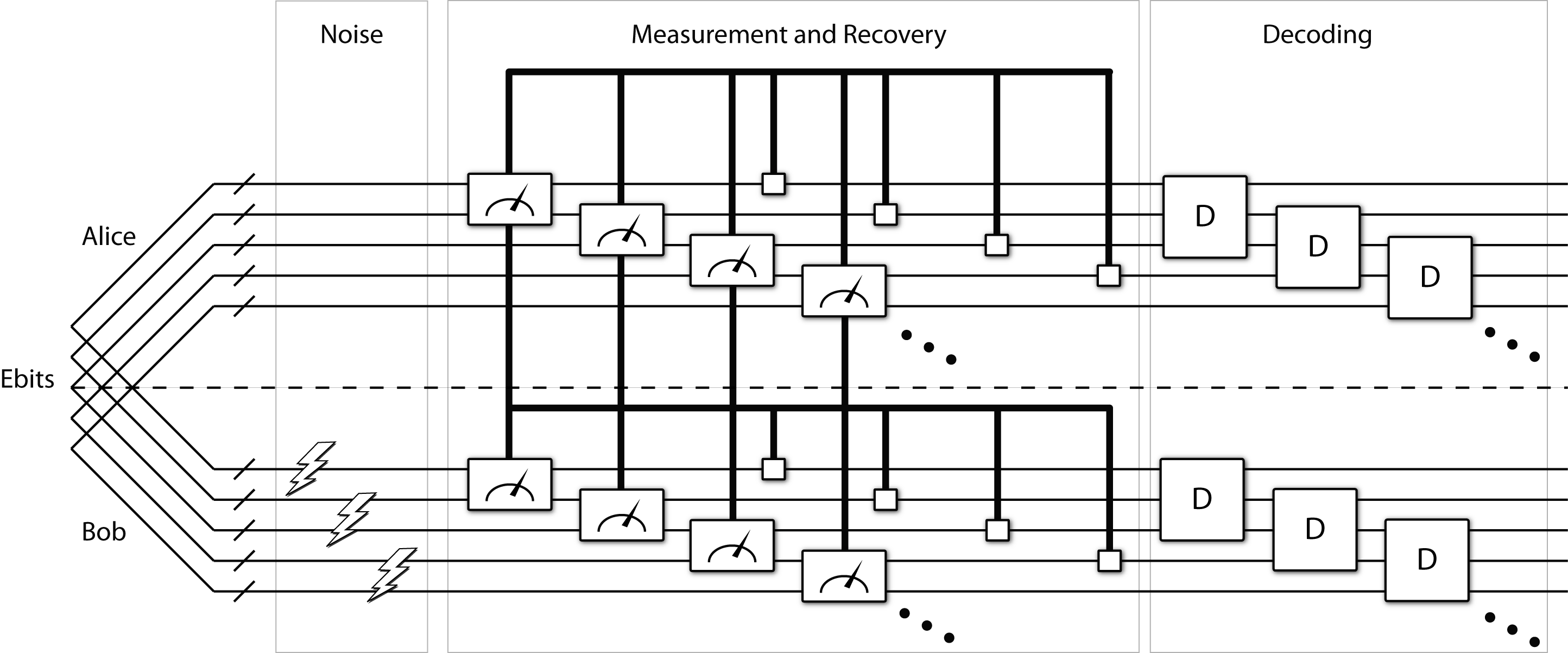}
\caption
{An example of a convolutional entanglement distillation protocol taken
from the quantum convolutional code in Ref. \cite{ieee2007forney}.
The code in Ref. \cite{ieee2007forney} has rate 1/3 and can correct
for single-qubit errors in every other frame. Alice and Bob first measure the
operators in the stabilizer for the quantum convolutional code. Alice
performs conditional unitaries on her qubits to restore them
to the +1 eigenspace of the stabilizer code. Alice forwards her measurement results to Bob. Bob performs a
maximum-likelihood decoding procedure such as Viterbi decoding \cite
{itit1967viterbi} to determine
the qubit errors. He corrects for these errors. He restores his qubits to the +1 eigenspace of the stabilizer code.
Alice and Bob both perform
online decoding to obtain ebits with yield 1/3.}
\label{fig:cedc}
\end{center}
\end{figure*}%
\label{sec:conv-ent-dist}We now show how to convert an arbitrary quantum
convolutional code into a convolutional entanglement distillation protocol.
Figure~\ref{fig:cedc}\ illustrates an example of a yield-1/3 convolutional
entanglement distillation protocol. The protocol has the same benefits as a
quantum convolutional code:\ an online decoder with less decoding complexity
than a block protocol, good error-correcting properties, and higher ebit yield
than a block protocol. The protocol we develop in this section is useful for
our major contribution presented in the next section.

We can think of our protocol in two ways. Our protocol applies when a sender
Alice and a receiver Bob possess a countably infinite number of noisy ebits.
Our protocol also applies as an online protocol when Alice and Bob begin with
a finite number of noisy ebits and establish more as time passes. The
countably infinite and online protocols are equivalent. We would actually
implement the entanglement distillation protocol in the online manner, but we
formulate the forthcoming mathematics with the countably infinite description.
Each step in the protocol does not need to wait for the completion of its
preceding step if Alice and Bob employ the protocol online.

The protocol begins with Alice and Bob establishing a set of noisy ebits.
Alice prepares a countably infinite number of Bell states $\left\vert \Phi
^{+}\right\rangle $ locally. She sends one half of each Bell state through a
noisy quantum channel. Alice and Bob then possess a state $\rho^{AB}$\ that is
a countably infinite number of noisy ebits $\rho_{i}^{AB}$ where%
\begin{equation}
\rho^{AB}=%
{\displaystyle\bigotimes\limits_{i=1}^{\infty}}
\ \rho_{i}^{AB}.
\end{equation}
The state $\rho^{AB}$ is equivalent to the following ensemble%
\begin{equation}
\left\{  p_{i},\left\vert \Phi^{+}\right\rangle _{i}^{AB}\right\}  .
\end{equation}
In the above, $p_{i}$ is the probability that the state is $\left\vert
\Phi^{+}\right\rangle _{i}^{AB}$, where%
\begin{equation}
\left\vert \Phi^{+}\right\rangle _{i}^{AB}\equiv\left(  \mathbf{I}%
\otimes\mathbf{A}_{i}\right)  \left\vert \Phi_{\infty}^{+}\right\rangle ^{AB},
\end{equation}
and $\left\vert \Phi_{\infty}^{+}\right\rangle ^{AB}$ is the state $\left(
\left\vert \Phi^{+}\right\rangle ^{AB}\right)  ^{\otimes\infty}$ rearranged so
that all of Alice's qubits are on the left and all of Bob's are on the right.
$\mathbf{A}_{i}\in\Pi^{\mathbb{Z}^{+}}$ is a Pauli sequence of errors acting
on Bob's side. These errors result from the noisy quantum channel.
$\mathbf{I}$ is a sequence of identity matrices acting on Alice's side
indicating that the noisy channel does not affect her qubits. Alice and Bob
need to correct for a particular error set in order to distill noiseless ebits.

Alice and Bob employ the following strategy to distill noiseless ebits. Alice
measures the $n-k$ generators in the basic set $\mathcal{G}_{0}$. The
measurement operation projects the first $n\left(  \nu+1\right)  $ ebits
($\nu$ is the constraint length) randomly onto one of $2^{n-k}$\ orthogonal
subspaces. Alice places the measurement outcomes in an $\left(  n-k\right)
$-dimensional classical bit vector $\mathbf{a}_{0}$. She restores her half of
the noisy ebits to the simultaneous +1-eigenspace of the generators in
$\mathcal{G}_{0}$ if $\mathbf{a}_{0}$ differs from the all-zero vector. She
sends $\mathbf{a}_{0}$ to Bob over a classical communication channel. Bob
measures the generators in $\mathcal{G}_{0}$ and stores the measurement
outcomes in a classical bit vector $\mathbf{b}_{0}$. Bob compares
$\mathbf{b}_{0}$ to $\mathbf{a}_{0}$ by calculating an error vector
$\mathbf{e}_{0}=\mathbf{a}_{0}\oplus\mathbf{b}_{0}$. He corrects for any
errors that $\mathbf{e}_{0}$ can identify. He may have to wait to receive
later error vectors before determining the full error syndrome. He restores
his half of the noisy ebits to the simultaneous +1-eigenspace of the
generators in $\mathcal{G}_{0}$ if the bit vector $\mathbf{b}_{0}$ indicates
that his logical ebits are not in the +1-space. Alice and Bob repeat the above
procedure for all shifts $D\left(  \mathcal{G}_{0}\right)  $, $D^{2}\left(
\mathcal{G}_{0}\right)  $, \ldots\ of the basic generators in $\mathcal{G}%
_{0}$. Bob obtains a set $\mathcal{E}$ of classical error vectors
$\mathbf{e}_{i}$: $\mathcal{E}=\left\{  \mathbf{e}_{i}:i\in\mathbb{Z}%
^{+}\right\}  $. Bob uses a maximum-likelihood decoding technique such as
Viterbi decoding \cite{itit1967viterbi} or a table-lookup on the error set
$\mathcal{E}$ to determine which errors occur. This error determination
process is a purely classical computation. He reverses the estimated errors
after determining the syndrome.

The states that Alice and Bob possess after the above procedure are encoded
logical ebits. They can extract physical ebits from these logical ebits by
each performing the online decoding circuit for the code $\mathcal{G}$. The
algorithm outlined in Ref.~\cite{isit2006grassl}\ gives a method for
determining the online decoding circuit.

\begin{example}
We use the rate-1/3 quantum convolutional code in
Example~\ref{sec:qcc-example} to produce a yield-1/3 convolutional
entanglement distillation protocol. Alice measures the generators in the
stabilizer in (\ref{eq:qcc-example-stabilizer})\ for every noisy ebit she
shares with Bob. Alice communicates the result of her measurement of the first
two generators to Bob. Alice restores the qubits on her side to be in the
simultaneous +1-eigenspace of the first two generators. Bob measures the same
first two generators. Alice measures the next two generators, communicates her
results, etc. Bob compares his results to Alice's to determine the error bit
vectors. Bob performs Viterbi decoding on the measurement results and corrects
for errors. He rotates his states to the simultaneous +1-eigenspace of the
generators. Alice and Bob perform the above procedure in an online manner
according to Figure~\ref{fig:cedc}. Alice and Bob can decode the first six
qubits after measuring the second two generators. They can decode because
there is no overlap between the first two generators and any two generators
after the second two generators. They use the circuit from
\cite{isit2006grassl} in reverse order to decode physical ebits from logical
ebits. They distill ebits with yield 1/3 by using this convolutional
entanglement distillation protocol. The ebit yield of 1/3 follows directly
from the code rate of 1/3.
\end{example}

\section{Convolutional Entanglement Distillation with Entanglement Assistance}

\label{sec:conv-ent-ent-assist}The convolutional entanglement distillation
protocol that we develop in this section operates identically to the one
developed in the previous section. The measurements, classical communication,
and recovery and decoding operations proceed exactly as Figure~\ref{fig:cedc} indicates.

The difference between the protocol in this section and the previous one is
that we now assume the sender and receiver share a few initial noiseless
ebits. They use these initial ebits as a catalyst to get the protocol started.
The sender and receiver require noiseless ebits for each round of the
convolutional entanglement distillation protocol. They can use the noiseless
ebits generated by earlier rounds for consumption in later rounds. It is
possible to distill noiseless ebits in this way by catalyzing the process with
a few noiseless ebits. The protocol we develop in this section is a more
powerful generalization of the previous section's protocol.

The construction in this section allows sender and receiver to use an
arbitrary set of Paulis for the distillation protocol. The set does not
necessarily have to be a commuting set of Paulis.

The implication of the construction in this section is that we can import an
arbitrary binary or quaternary classical convolutional code for use as a
quantum convolutional code. We explicitly give some examples to highlight the
technique for importing. The error-correcting properties and yield translate
directly from the properties of the classical convolutional code. Thus the
problem of finding a good convolutional entanglement distillation protocol
reduces to that of finding a good classical convolutional code.

\subsection{Yield (n-1)/n Convolutional Entanglement Distillation}

We present our first method for constructing a convolutional entanglement
distillation protocol that uses entanglement assistance. The shifted
symplectic product from Section~\ref{sec:shifted-symp-prod}\ is a crucial
component of our formulation.

Suppose Alice and Bob use one generator $\mathbf{N}\left(  \mathbf{u}\left(
D\right)  \right)  $\ for an entanglement distillation protocol where%
\[
\mathbf{u}\left(  D\right)  =\left[  \mathbf{z}\left(  D\right)
|\mathbf{x}\left(  D\right)  \right]  =\left[
\begin{array}
[c]{ccc}%
z_{1}\left(  D\right)   & \cdots & z_{n}\left(  D\right)
\end{array}
|%
\begin{array}
[c]{ccc}%
x_{1}\left(  D\right)   & \cdots & x_{n}\left(  D\right)
\end{array}
\right]  .
\]
and where $z_{1}\left(  D\right)  $, \ldots, $z_{n}\left(  D\right)  $,
$x_{1}\left(  D\right)  $, \ldots, $x_{n}\left(  D\right)  $ are binary
polynomials. We do not impose a commuting constraint on generator
$\mathbf{N}\left(  \mathbf{u}\left(  D\right)  \right)  $. Alice and Bob
choose generator $\mathbf{N}\left(  \mathbf{u}\left(  D\right)  \right)  $
solely for its error-correcting capability.

The shifted symplectic product helps to produce a commuting generator from a
noncommuting one. The shifted symplectic product of $\mathbf{u}\left(
D\right)  $\ is%
\begin{equation}
\left(  \mathbf{u}\odot\mathbf{u}\right)  \left(  D\right)  =\sum
_{i\in\mathbb{Z}}\left(  \mathbf{u}\odot\mathbf{u}\right)  _{i}\ D^{i}.
\end{equation}
The coefficient $\left(  \mathbf{u}\odot\mathbf{u}\right)  _{0}$ for zero
shifts is equal to zero because every tensor product of Pauli operators
commutes with itself:%
\begin{equation}
\left(  \mathbf{u}\odot\mathbf{u}\right)  _{0}=0.
\end{equation}
Recall that $\mathbf{u}\left(  D\right)  $ is self-time-reversal symmetric
(\ref{eq:self-time-reversal-sym}). We adopt the following notation for a
polynomial that includes the positive-index or negative-index coefficients of
the shifted symplectic product $\left(  \mathbf{u}\odot\mathbf{u}\right)
\left(  D\right)  $:%
\begin{equation}
\left(  \mathbf{u}\odot\mathbf{u}\right)  \left(  D\right)  ^{+}=\sum
_{i\in\mathbb{Z}^{+}}\left(  \mathbf{u}\odot\mathbf{u}\right)  _{i}%
\ D^{i},\ \ \ \ \left(  \mathbf{u}\odot\mathbf{u}\right)  \left(  D\right)
^{-}=\sum_{i\in\mathbb{Z}^{-}}\left(  \mathbf{u}\odot\mathbf{u}\right)
_{i}\ D^{i}.
\end{equation}
The following identity holds:%
\begin{equation}
\left(  \mathbf{u}\odot\mathbf{u}\right)  \left(  D\right)  ^{+}=\left(
\mathbf{u}\odot\mathbf{u}\right)  \left(  D^{-1}\right)  ^{-}.
\end{equation}
Consider the following vector of polynomials:%
\begin{equation}
\mathbf{a}\left(  D\right)  =\left[
\begin{array}
[c]{c}%
\left(  \mathbf{u}\odot\mathbf{u}\right)  \left(  D\right)  ^{+}%
\end{array}
|%
\begin{array}
[c]{c}%
1
\end{array}
\right]  .
\end{equation}
Its relations under the shifted symplectic product are the same as
$\mathbf{u}\left(  D\right)  $:%
\begin{equation}
\left(  \mathbf{a}\odot\mathbf{a}\right)  \left(  D\right)  =\left(
\mathbf{u}\odot\mathbf{u}\right)  \left(  D\right)  ^{+}+\left(
\mathbf{u}\odot\mathbf{u}\right)  \left(  D\right)  ^{-}=\left(
\mathbf{u}\odot\mathbf{u}\right)  \left(  D\right)  .\nonumber
\end{equation}
The vector $\mathbf{a}\left(  D\right)  $ provides a straightforward way to
make $\mathbf{N}\left(  \mathbf{u}\left(  D\right)  \right)  $ commute with
all of its shifts. We augment $\mathbf{u}\left(  D\right)  $ with
$\mathbf{a}\left(  D\right)  $. The augmented generator $\mathbf{u}^{\prime
}\left(  D\right)  $\ is as follows:%
\begin{equation}
\mathbf{u}^{\prime}\left(  D\right)  =\left[
\begin{array}
[c]{cc}%
\mathbf{z}\left(  D\right)   & \left(  \mathbf{u}\odot\mathbf{u}\right)
\left(  D\right)  ^{+}%
\end{array}
|%
\begin{array}
[c]{cc}%
\mathbf{x}\left(  D\right)   & 1
\end{array}
\right]  .\label{eq:augment-conv}%
\end{equation}
The augmented generator $\mathbf{u}^{\prime}\left(  D\right)  $ has vanishing
symplectic product, $\left(  \mathbf{u}^{\prime}\odot\mathbf{u}^{\prime
}\right)  \left(  D\right)  =0$, because the shifted symplectic product of
$\mathbf{a}\left(  D\right)  $ nulls the shifted symplectic product of
$\mathbf{u}\left(  D\right)  $. The augmented generator $\mathbf{N}\left(
\mathbf{u}^{\prime}\left(  D\right)  \right)  $ commutes with itself for every
shift and is therefore useful for convolutional entanglement distillation as
outlined in Section~\ref{sec:conv-ent-dist}.

We can construct an entanglement distillation protocol using an augmented
generator of this form. The first $n$ Pauli entries for every frame of
generator $\mathbf{N}\left(  \mathbf{u}^{\prime}\left(  D\right)  \right)  $
correct errors. Entry $n+1$ for every frame of $\mathbf{N}\left(
\mathbf{u}^{\prime}\left(  D\right)  \right)  $ makes $\mathbf{N}\left(
\mathbf{u}^{\prime}\left(  D\right)  \right)  $ commute with every one of its
shifts. The error-correcting properties of the code do not include errors on
the last (extra) ebit of each frame; therefore, this ebit must be noiseless.
It is necessary to catalyze the distillation procedure with $n\nu$\ noiseless
ebits where $n$ is the frame size and $\nu$ is the constraint length. The
distillation protocol requires this particular amount because it does not
correct errors and generate noiseless ebits until it has finished processing
the first basic set of generators and $\nu-1$ of its shifts. Later frames can
use the noiseless ebits generated from previous frames. Therefore these
initial noiseless ebits are negligible when calculating the yield. This
construction allows us to exploit the error-correcting properties of an
arbitrary set of Pauli matrices for a convolutional entanglement distillation
protocol.%
\begin{table}[tbp] \centering
\label{tbl:syndromes}%
\begin{tabular}
[c]{l|l|l|l|l|l}\hline\hline
$X_{1}$ & $Z_{1}$ & $Y_{1}$ & \thinspace$X_{2}$ & $Z_{2}$ & $Y_{2}%
$\\\hline\hline
$1$ & $0$ & $1$ & $1$ & $0$ & $1$\\
$0$ & $0$ & $0$ & $0$ & $1$ & $1$\\
$0$ & $1$ & $1$ & $1$ & $0$ & $1$\\
$1$ & $0$ & $1$ & $0$ & $0$ & $0$\\\hline\hline
\end{tabular}
\caption{The convolutional entanglement distillation protocol
for Example~\ref{ex:conv-ed-example-one-gen} corrects for a single-qubit error
in every fourth frame. Here we list the syndromes corresponding to
errors $X_1$, $Y_1$, and $Z_1$ on the first qubit and to errors
$X_2$, $Y_2$, and $Z_2$ on the second qubit. The syndromes are
unique so that the receiver can identify which error occurs.}%
\end{table}%

We discuss the yield of such a protocol in more detail. Our construction
employs one generator with $n+1$ qubits per frame. The protocol generates $n$
noiseless ebits for every frame. But it also consumes a noiseless\ ebit for
every frame. Every frame thus produces a net of $n-1$ noiseless ebits, and the
yield of the protocol\ is $\left(  n-1\right)  /n$.

This yield of $\left(  n-1\right)  /n$\ is superior to the yield of an
entanglement distillation protocol taken from the quantum convolutional codes
of Forney et al. \cite{ieee2007forney}. Our construction should also give
entanglement distillation protocols with superior error-correcting properties
because we have no self-orthogonality constraint on the Paulis in the stabilizer.

It is possible to construct an online decoding circuit for the generator
$\mathbf{u}^{\prime}\left(  D\right)  $ by the methods given in
\cite{isit2006grassl}. A circuit satisfies the noncatastrophic property if the
polynomial entries of all of the code generators have a greatest common
divisor that is a power of the delay operator $D$ \cite{isit2006grassl}. The
online decoding circuit for this construction obeys the noncatastrophicity
property because the augmented generator $\mathbf{u}^{\prime}\left(  D\right)
$ contains 1 as one of its entries.%
\begin{figure}
[ptb]
\begin{center}
\includegraphics[
natheight=7.639800in,
natwidth=5.626500in,
height=3.7403in,
width=2.2788in
]%
{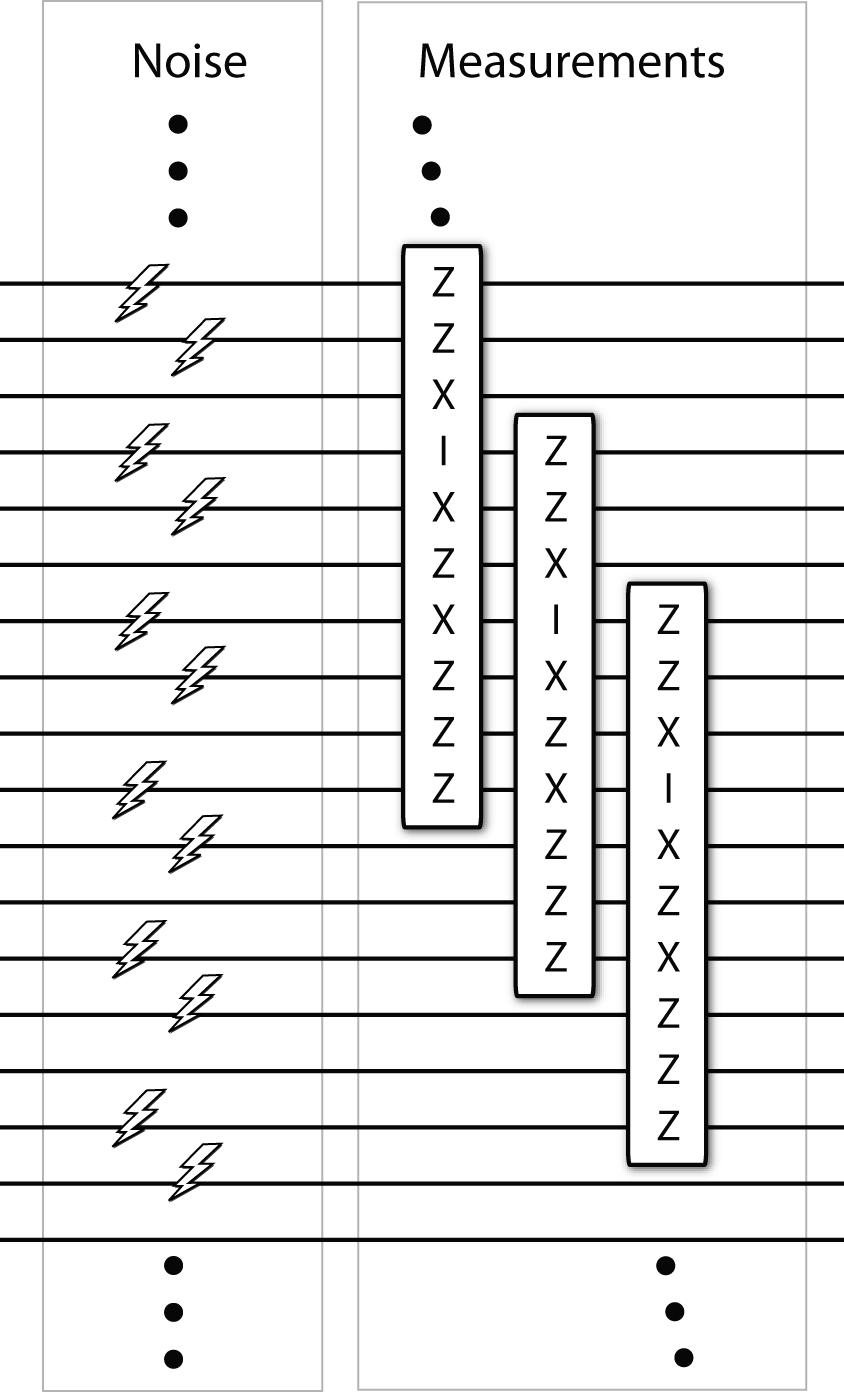}%
\caption{The above figure illustrates Bob's side of the convolutional
entanglement distillation protocol that uses entanglement assistance. The
noise affects the first and second of every three ebits that Bob shares with
Alice. Every third ebit that Alice and Bob share are noiseless. The
measurements correspond to those in Example~\ref{ex:conv-ed-example-one-gen}.}%
\label{fig:conv-dist-example}%
\end{center}
\end{figure}

\begin{example}
\label{ex:conv-ed-example-one-gen}Suppose we have the following generator%
\[
\mathbf{N}\left(  \mathbf{u}\left(  D\right)  \right)  =\left(  \cdots
|II|ZZ|IX|XZ|ZI|II|\cdots\right)  ,
\]
where%
\[
\mathbf{u}\left(  D\right)  =\left[  \left.
\begin{array}
[c]{cc}%
1+D^{3} & 1+D^{2}%
\end{array}
\right\vert
\begin{array}
[c]{cc}%
D^{2} & D
\end{array}
\right]  .
\]
The above generator corrects for an arbitrary single-qubit error in every
fourth frame. Table~\ref{tbl:syndromes} lists the unique syndromes for errors
in a single frame. The generator anticommutes with a shift of itself by one or
two to the left or right. The shifted symplectic product confirms these
commutation relations:%
\[
\left(  \mathbf{u\odot u}\right)  \left(  D\right)  =D+D^{2}+D^{-1}+D^{-2}.
\]
Let us follow the prescription in (\ref{eq:augment-conv}) for augmenting
generator $\mathbf{N}\left(  \mathbf{u}\left(  D\right)  \right)  $. The
following polynomial%
\begin{equation}
\mathbf{a}\left(  D\right)  =\left[  \left.
\begin{array}
[c]{c}%
\left(  \mathbf{u\odot u}\right)  \left(  D\right)  ^{+}%
\end{array}
\right\vert
\begin{array}
[c]{c}%
1
\end{array}
\right]  =\left[  \left.
\begin{array}
[c]{c}%
D+D^{2}%
\end{array}
\right\vert
\begin{array}
[c]{c}%
1
\end{array}
\right]  ,\nonumber
\end{equation}
has the same commutation relations as $\mathbf{u}\left(  D\right)  $:%
\begin{equation}
\left(  \mathbf{a\odot a}\right)  \left(  D\right)  =\left(  \mathbf{u\odot
u}\right)  \left(  D\right)  .
\end{equation}
We augment $\mathbf{u}\left(  D\right)  $\ as follows:%
\[
\mathbf{u}^{\prime}\left(  D\right)  =\left[  \left.
\begin{array}
[c]{cc}%
1+D^{3} & 1+D^{2}%
\end{array}%
\begin{array}
[c]{c}%
D+D^{2}%
\end{array}
\right\vert
\begin{array}
[c]{cc}%
D^{2} & D
\end{array}%
\begin{array}
[c]{c}%
1
\end{array}
\right]  .
\]
The overall generator now looks as follows in the Pauli representation:%
\[
\mathbf{N}\left(  \mathbf{u}^{\prime}\left(  D\right)  \right)  =(\cdots
|III|ZZX|IXZ|XZZ|ZII|III|\cdots).
\]
The yield of a protocol using the above construction is 1/2.
Figure~\ref{fig:conv-dist-example} illustrates Bob's side of the protocol. It
shows which of Bob's half of the ebits are noisy and noiseless, and it gives
the measurements that Bob performs.
\end{example}

\subsection{Yield (n-m)/n Convolutional Entanglement Distillation}

The construction in the above section uses only one generator for
distillation. We generalize the above construction to a code with an arbitrary
number of generators. We give an example that illustrates how to convert an
arbitrary classical quaternary convolutional code into a convolutional
entanglement distillation protocol.

Suppose we have the following $m$ generators%
\[
\left\{  \mathbf{N}\left(  \mathbf{u}_{i}\left(  D\right)  \right)  :1\leq
i\leq m\right\}  ,
\]
where%
\begin{equation}%
\begin{bmatrix}
\mathbf{u}_{1}\left(  D\right) \\
\mathbf{u}_{2}\left(  D\right) \\
\vdots\\
\mathbf{u}_{m}\left(  D\right)
\end{bmatrix}
=\left[  \left.
\begin{array}
[c]{c}%
\mathbf{z}_{1}\left(  D\right) \\
\mathbf{z}_{2}\left(  D\right) \\
\vdots\\
\mathbf{z}_{m}\left(  D\right)
\end{array}
\right\vert
\begin{array}
[c]{c}%
\mathbf{x}_{1}\left(  D\right) \\
\mathbf{x}_{2}\left(  D\right) \\
\vdots\\
\mathbf{x}_{m}\left(  D\right)
\end{array}
\right]  . \label{eq:n-k-construction-distillation}%
\end{equation}
We make no assumption about the commutation relations of the above generators.
We choose them solely for their error-correcting properties.

We again utilize the shifted symplectic product to design a convolutional
entanglement distillation protocol with multiple generators. Let us adopt the
following shorthand for the auto and cross shifted symplectic products of
generators $\mathbf{u}_{1}\left(  D\right)  ,\ldots,\mathbf{u}_{m}\left(
D\right)  $:%
\begin{align}
\mathbf{u}_{i}^{+} &  \equiv\left(  \mathbf{u}_{i}\odot\mathbf{u}_{i}\right)
\left(  D\right)  ^{+},\\
\mathbf{u}_{i,j} &  \equiv\left(  \mathbf{u}_{i}\odot\mathbf{u}_{j}\right)
\left(  D\right)  .
\end{align}
Consider the following matrix:%
\begin{equation}%
\begin{bmatrix}
\mathbf{a}_{1}\left(  D\right)  \\
\mathbf{a}_{2}\left(  D\right)  \\
\vdots\\
\mathbf{a}_{m}\left(  D\right)
\end{bmatrix}
=\left[  \left.
\begin{array}
[c]{cccc}%
\mathbf{u}_{1}^{+} & \mathbf{u}_{1,2} & \cdots & \mathbf{u}_{1,m}\\
0 & \mathbf{u}_{2}^{+} & \cdots & \mathbf{u}_{2,m}\\
\vdots &  & \ddots & \vdots\\
0 & \cdots & 0 & \mathbf{u}_{m}^{+}%
\end{array}
\right\vert \mathbf{I}_{m\times m}\right]  .\label{eq:augmented-conv-general}%
\end{equation}
The symplectic relations of the entries $\mathbf{a}_{i}\left(  D\right)  $ are
the same as the original $\mathbf{u}_{i}\left(  D\right)  $:%
\[
\left(  \mathbf{a}_{i}\odot\mathbf{a}_{j}\right)  \left(  D\right)  =\left(
\mathbf{u}_{i}\odot\mathbf{u}_{j}\right)  \left(  D\right)  \ \ \ \ \ \forall
i,j\in\left\{  1,\ldots,m\right\}  ,
\]
or equivalently, if $\Omega_{\mathbf{a}}\left(  D\right)  $ is the shifted
symplectic product matrix for the generators in
(\ref{eq:augmented-conv-general}) and $\Omega_{\mathbf{u}}\left(  D\right)  $
is the shifted symplectic product matrix for the \textquotedblleft$\mathbf{u}%
$\textquotedblright\ generators, then%
\[
\Omega_{\mathbf{a}}\left(  D\right)  =\Omega_{\mathbf{u}}\left(  D\right)  .
\]
We mention that the following matrix also has the same symplectic relations:%
\begin{equation}
\left[  \left.
\begin{array}
[c]{cccc}%
\mathbf{u}_{1}^{+} & 0 & \cdots & 0\\
\mathbf{u}_{2,1} & \mathbf{u}_{2}^{+} & \cdots & \vdots\\
\vdots &  & \ddots & 0\\
\mathbf{u}_{m,1} & \mathbf{u}_{m,2} & \cdots & \mathbf{u}_{m}^{+}%
\end{array}
\right\vert \mathbf{I}_{m\times m}\right]  .\label{eq:alternate-constr-matrix}%
\end{equation}
Let us rewrite (\ref{eq:augmented-conv-general}) as follows:%
\begin{equation}%
\begin{bmatrix}
\mathbf{a}_{1}\left(  D\right)  \\
\mathbf{a}_{2}\left(  D\right)  \\
\vdots\\
\mathbf{a}_{m}\left(  D\right)
\end{bmatrix}
=\left[  \left.
\begin{array}
[c]{c}%
\mathbf{z}_{1}^{\prime}\left(  D\right)  \\
\mathbf{z}_{2}^{\prime}\left(  D\right)  \\
\vdots\\
\mathbf{z}_{m}^{\prime}\left(  D\right)
\end{array}
\right\vert
\begin{array}
[c]{c}%
\mathbf{x}_{1}^{\prime}\left(  D\right)  \\
\mathbf{x}_{2}^{\prime}\left(  D\right)  \\
\vdots\\
\mathbf{x}_{m}^{\prime}\left(  D\right)
\end{array}
\right]  .
\end{equation}
The above matrix provides a straightforward way to make the original
generators commute with all of their shifts. We augment the generators in
(\ref{eq:n-k-construction-distillation}) by the generators $\mathbf{a}%
_{i}\left(  D\right)  $\ to get the following $m\times2\left(  n+m\right)
$\ matrix:%
\begin{equation}
\mathbf{U}^{\prime}\left(  D\right)  =\left[  \left.
\begin{array}
[c]{c}%
\mathbf{Z}\left(  D\right)
\end{array}
\right\vert
\begin{array}
[c]{c}%
\mathbf{X}\left(  D\right)
\end{array}
\right]  =\left[  \left.
\begin{array}
[c]{cc}%
\mathbf{z}_{1}\left(  D\right)   & \mathbf{z}_{1}^{\prime}\left(  D\right)  \\
\mathbf{z}_{2}\left(  D\right)   & \mathbf{z}_{2}^{\prime}\left(  D\right)  \\
\vdots & \vdots\\
\mathbf{z}_{m}\left(  D\right)   & \mathbf{z}_{m}^{\prime}\left(  D\right)
\end{array}
\right\vert
\begin{array}
[c]{cc}%
\mathbf{x}_{1}\left(  D\right)   & \mathbf{x}_{1}^{\prime}\left(  D\right)  \\
\mathbf{x}_{2}\left(  D\right)   & \mathbf{x}_{2}^{\prime}\left(  D\right)  \\
\vdots & \vdots\\
\mathbf{x}_{m}\left(  D\right)   & \mathbf{x}_{m}^{\prime}\left(  D\right)
\end{array}
\right]  .\nonumber
\end{equation}
Every row of the augmented matrix $\mathbf{U}^{\prime}\left(  D\right)  $ has
vanishing symplectic product with itself and any other row. This condition is
equivalent to the following matrix condition for shifted symplectic
orthogonality \cite{arxiv2004olliv}:%
\begin{equation}
\mathbf{Z}\left(  D\right)  \mathbf{X}^{T}\left(  D^{-1}\right)
-\mathbf{X}\left(  D\right)  \mathbf{Z}^{T}\left(  D^{-1}\right)  =0.
\end{equation}
The construction gives a commuting set of generators for arbitrary shifts and
thus forms a valid stabilizer.

We can readily develop a convolutional entanglement distillation protocol
using the above formulation. The generators in the augmented matrix
$\mathbf{U}^{\prime}\left(  D\right)  $\ correct for errors on the first $n$
ebits. The last $m$ ebits are noiseless ebits that help to obtain a commuting
stabilizer. It is necessary to catalyze the distillation protocol with
$\left(  n+m\right)  \nu$ noiseless ebits. Later frames can use the noiseless
ebits generated from previous frames. These initial noiseless ebits are
negligible when calculating the yield.

We comment more on the yield of the protocol. The protocol requires a set of
$m$ generators with $n+m$ Pauli entries. It generates $n$ ebits for every
frame. But it consumes $m$ noiseless ebits per frame. The net yield of a
protocol using the above construction is thus $\left(  n-m\right)  /n$.

The key benefit of the above construction is that we can use an arbitrary set
of Paulis for distilling noiseless ebits. This arbitrariness in the Paulis
implies that we can import an arbitrary classical convolutional binary or
quaternary code for use in a convolutional entanglement distillation protocol.

It is again straightforward to develop a noncatastrophic decoding circuit
using previous techniques \cite{isit2006grassl}. Every augmented generator in
$\mathbf{U}^{\prime}\left(  D\right)  $ has \textquotedblleft%
1\textquotedblright\ as an entry so that it satisfies the property required
for noncatastrophicity.

\begin{example}
We begin with a classical quaternary convolutional code with entries from
$GF\left(  4\right)  $:%
\begin{equation}
\left(  \cdots|0000|1\bar{\omega}10|1101|0000|\cdots\right)  .
\end{equation}
The above code is a convolutional version of the classical quaternary block
code from Ref.~\cite{science2006brun}. We multiply the above generator by
$\bar{\omega}$ and $\omega$ as prescribed in Refs.
\cite{ieee1998calderbank,ieee2007forney} and use the map in
(\ref{eq:gf4-pauli}) to obtain the following Pauli generators%
\begin{align}
\mathbf{N}\left(  \mathbf{u}_{1}\left(  D\right)  \right)   &  =\left(
\cdots|IIII|ZXZI|ZZIZ|IIII|\cdots\right)  ,\nonumber\\
\mathbf{N}\left(  \mathbf{u}_{2}\left(  D\right)  \right)   &  =\left(
\cdots|IIII|XYXI|XXIX|IIII|\cdots\right)  .
\end{align}
We determine binary polynomials corresponding to the above Pauli generators:%
\begin{equation}
\left[
\begin{array}
[c]{c}%
\mathbf{u}_{1}\left(  D\right)  \\
\mathbf{u}_{2}\left(  D\right)
\end{array}
\right]  =\left[  \left.
\begin{array}
[c]{cccc}%
1+D & D & 1 & D\\
0 & 1 & 0 & 0
\end{array}
\right\vert
\begin{array}
[c]{cccc}%
0 & 1 & 0 & 0\\
1+D & 1+D & 1 & D
\end{array}
\right]  .
\end{equation}
The first generator anticommutes with itself shifted by one to the left or
right, the second generator anticommutes with itself shifted by one to the
left or right, and the first generator anticommutes with the second shifted by
one to the left. The following shifted symplectic products confirm the above
commutation relations:%
\begin{equation}
\left(  \mathbf{u}_{1}\odot\mathbf{u}_{1}\right)  \left(  D\right)
=D^{-1}+D,\ \ \ \ \left(  \mathbf{u}_{2}\odot\mathbf{u}_{2}\right)  \left(
D\right)  =D^{-1}+D,\ \ \ \ \left(  \mathbf{u}_{1}\odot\mathbf{u}_{2}\right)
\left(  D\right)  =D^{-1}.\label{eq:symp-prod-relations1}%
\end{equation}
Consider the following two generators:%
\begin{equation}
\left[
\begin{array}
[c]{c}%
\mathbf{a}_{1}\left(  D\right)  \\
\mathbf{a}_{2}\left(  D\right)
\end{array}
\right]  =\left[  \left.
\begin{array}
[c]{cc}%
D & 0\\
D & D
\end{array}
\right\vert
\begin{array}
[c]{cc}%
1 & 0\\
0 & 1
\end{array}
\right]  .
\end{equation}
Their relations under the shifted symplectic product are the same as those in
(\ref{eq:symp-prod-relations1}).%
\begin{equation}
\left(  \mathbf{a}_{1}\odot\mathbf{a}_{1}\right)  \left(  D\right)  =\left(
\mathbf{u}_{1}\odot\mathbf{u}_{1}\right)  \left(  D\right)  ,\ \ \ \left(
\mathbf{a}_{2}\odot\mathbf{a}_{2}\right)  \left(  D\right)  =\left(
\mathbf{u}_{2}\odot\mathbf{u}_{2}\right)  \left(  D\right)  ,\ \ \ \left(
\mathbf{a}_{1}\odot\mathbf{a}_{2}\right)  \left(  D\right)  =\left(
\mathbf{u}_{1}\odot\mathbf{u}_{2}\right)  \left(  D\right)  .\nonumber
\end{equation}
We augment the generators $\mathbf{u}_{1}\left(  D\right)  $\ and
$\mathbf{u}_{2}\left(  D\right)  $\ to generators $\mathbf{u}_{1}^{\prime
}\left(  D\right)  $\ and $\mathbf{u}_{2}^{\prime}\left(  D\right)
$\ respectively as follows. The augmented matrix $\mathbf{U}^{\prime}\left(
D\right)  $ is%
\begin{equation}
\mathbf{U}^{\prime}\left(  D\right)  =\left[  \left.
\begin{array}
[c]{cccc}%
1+D & D & 1 & D\\
0 & 1 & 0 & 0
\end{array}%
\begin{array}
[c]{cc}%
D & 0\\
D & D
\end{array}
\right\vert
\begin{array}
[c]{cccc}%
0 & 1 & 0 & 0\\
1+D & 1+D & 1 & D
\end{array}%
\begin{array}
[c]{cc}%
1 & 0\\
0 & 1
\end{array}
\right]  .
\end{equation}
The first row of $\mathbf{U}^{\prime}\left(  D\right)  $ is generator
$\mathbf{u}_{1}^{\prime}\left(  D\right)  $\ and the second row is
$\mathbf{u}_{2}^{\prime}\left(  D\right)  $. The augmented generators have the
following Pauli representation.%
\begin{align}
\mathbf{N}\left(  \mathbf{u}_{1}^{\prime}\left(  D\right)  \right)   &
=\left(  \cdots|IIIIII|ZXZIXI|ZZIZZI|IIIIII|\cdots\right)  ,\nonumber\\
\mathbf{N}\left(  \mathbf{u}_{2}^{\prime}\left(  D\right)  \right)   &
=\left(  \cdots|IIIIII|XYXIIX|XXIXZZ|IIIIII|\cdots\right)  .
\end{align}
The original block code from Ref.~\cite{science2006brun} corrects for an
arbitrary single-qubit error. The above entanglement distillation protocol
corrects for a single-qubit error in eight qubits---two frames. This
error-correcting capability follows from the capability of the block code. The
yield of a protocol using the above stabilizer is again 1/2.
\end{example}

\subsection{CSS-Like Construction for Convolutional Entanglement Distillation}

We finally present a construction that allows us to import two arbitrary
binary classical codes for use in a convolutional entanglement distillation
protocol. The construction is similar to a CSS\ code because one code corrects
for bit flips and the other corrects for phase flips.

We could simply use the technique from the previous section to construct a
convolutional entanglement-distillation protocol. We could represent both
classical codes as codes over $GF\left(  4\right)  $. We could multiply the
bit-flip code by $\omega$ and the phase-flip code by $\bar{\omega}$ and use
the map in (\ref{eq:gf4-pauli})\ from $GF\left(  4\right)  $ to the Paulis. We
could then use the above method for augmentation and obtain a valid quantum
code for entanglement distillation. But there is a better method that exploits
the structure of a CSS\ code to minimize the number of initial catalytic
noiseless ebits.

Our algorithm below uses a Gram-Schmidt like orthogonalization procedure to
minimize the number of initial noiseless ebits. The procedure is similar to
the algorithm in \cite{arx2006brun}\ with some key differences.

Suppose we have $m$ generators $\left\{  \mathbf{N}\left(  \mathbf{w}%
_{i}\left(  D\right)  \right)  :1\leq i\leq m\right\}  $ where%
\begin{equation}%
\begin{bmatrix}
\mathbf{w}_{1}\left(  D\right) \\
\vdots\\
\mathbf{w}_{p}\left(  D\right) \\
\mathbf{w}_{p+1}\left(  D\right) \\
\vdots\\
\mathbf{w}_{m}\left(  D\right)
\end{bmatrix}
=\left[  \left.
\begin{array}
[c]{c}%
\mathbf{z}_{1}\left(  D\right) \\
\vdots\\
\mathbf{z}_{p}\left(  D\right) \\
\mathbf{0}\\
\vdots\\
\mathbf{0}%
\end{array}
\right\vert
\begin{array}
[c]{c}%
\mathbf{0}\\
\vdots\\
\mathbf{0}\\
\mathbf{x}_{1}\left(  D\right) \\
\vdots\\
\mathbf{x}_{m-p}\left(  D\right)
\end{array}
\right]  .
\end{equation}
and each vector $\mathbf{w}_{i}\left(  D\right)  $ has length $2n$. The above
matrix could come from two binary classical codes. The vectors $\mathbf{z}%
_{1}\left(  D\right)  $,\ldots,$\mathbf{z}_{p}\left(  D\right)  $ could come
from one code, and the vectors $\mathbf{x}_{1}\left(  D\right)  $%
,\ldots,$\mathbf{x}_{m-p}\left(  D\right)  $ could come from another code. The
following orthogonality relations hold for the above vectors:%
\begin{align}
\forall\ \ 1\leq i,j\leq p  &  :\left(  \mathbf{w}_{i}\odot\mathbf{w}%
_{j}\right)  \left(  D\right)  =0,\\
\forall\ \ p+1\leq i^{\prime},j^{\prime}\leq m  &  :\left(  \mathbf{w}%
_{i^{\prime}}\odot\mathbf{w}_{j^{\prime}}\right)  \left(  D\right)  =0.
\end{align}
We exploit the above orthogonality relations in the algorithm below.

We can perform a Gram-Schmidt process on the above set of vectors. This
process orthogonalizes the vectors with respect to the shifted symplectic
product. The procedure does not change the error-correcting properties of the
original codes because all operations are linear.

The algorithm breaks the set of vectors above into pairs. Each pair consists
of two vectors which are symplectically nonorthogonal to each other, but which
are symplectically orthogonal to all other pairs. Any remaining vectors that
are symplectically orthogonal to all other vectors are collected into a
separate set, which we call the set of isotropic vectors. This idea is similar
to the decomposition of a vector space into an isotropic and symplectic part.
We cannot label the decomposition as such because the shifted symplectic
product is not a true symplectic product.

We detail the initialization of the algorithm. Set parameters $i=0$, $c=0$,
$l=0$. The index $i$ labels the total number of vectors processed, $c$ gives
the number of pairs, and $l$ labels the number of vectors with no partner.
Initialize sets $\mathcal{U}$ and $\mathcal{V}$ to be null: $\mathcal{U}%
=\mathcal{V}=\emptyset$. $\mathcal{U}$ keeps track of the pairs and
$\mathcal{V}$ keeps track of the vectors with no partner.

The algorithm proceeds as follows. While $i\leq m$, let $j\geq2c+l+2$ be the
smallest index for a $\mathbf{w}_{j}\left(  D\right)  $ for which $\left(
\mathbf{w}_{2c+l+1}\odot\mathbf{w}_{j}\right)  \left(  D\right)  \neq0$.
Increment $l$ and $i$ by one, add $i$ to $\mathcal{V}$, and proceed to the
next round if no such pair exists. Otherwise, swap $\mathbf{w}_{j}\left(
D\right)  $ with $\mathbf{w}_{2c+l+2}\left(  D\right)  $. For $r\in\left\{
2c+l+3,\ldots,m\right\}  $, perform%
\begin{equation}
\mathbf{w}_{r}\left(  D\right)  =\left(  \mathbf{w}_{2c+l+2}\odot
\mathbf{w}_{2c+l+1}\right)  \left(  D\right)  \mathbf{w}_{r}\left(  D\right)
+\left(  \mathbf{w}_{r}\odot\mathbf{w}_{2c+l+2}\right)  \left(  D^{-1}\right)
\mathbf{w}_{2c+l+1}\left(  D\right)  ,\nonumber
\end{equation}
if $\mathbf{w}_{r}\left(  D\right)  $ has a purely $z$ component. Perform%
\begin{equation}
\mathbf{w}_{r}\left(  D\right)  =\left(  \mathbf{w}_{2c+l+1}\odot
\mathbf{w}_{2c+l+2}\right)  \left(  D\right)  \mathbf{w}_{r}\left(  D\right)
+\left(  \mathbf{w}_{r}\odot\mathbf{w}_{2c+l+1}\right)  \left(  D^{-1}\right)
\mathbf{w}_{2c+l+2}\left(  D\right)  ,\nonumber
\end{equation}
if $\mathbf{w}_{r}\left(  D\right)  $ has a purely $x$ component. Divide every
element in $\mathbf{w}_{r}\left(  D\right)  $ by the greatest common factor if
the GCF is not equal to one. Then%
\begin{equation}
\left(  \mathbf{w}_{r}\odot\mathbf{w}_{2c+l+1}\right)  \left(  D\right)
=\left(  \mathbf{w}_{r}\odot\mathbf{w}_{2c+l+2}\right)  \left(  D\right)  =0.
\end{equation}
Increment $c$ by one, increment $i$ by one, add $i$ to $\mathcal{U}$, and
increment $i$ by one. Proceed to the next round.

We now give the method for augmenting the above generators so that they form a
commuting stabilizer. At the end of the algorithm, the sets $\mathcal{U}$ and
$\mathcal{V}$\ have the following sizes: $\left\vert \mathcal{U}\right\vert
=c$ and $\left\vert \mathcal{V}\right\vert =l$. Let us relabel the vectors
$\mathbf{w}_{i}\left(  D\right)  $ for all $1\leq i\leq2c+l$. We relabel all
pairs: call the first $\mathbf{u}_{i}\left(  D\right)  $ and call its partner
$\mathbf{v}_{i}\left(  D\right)  $ for all $1\leq i\leq c$. Call any vector
without a partner $\mathbf{u}_{c+i}\left(  D\right)  $ for all $1\leq i\leq
l$. The relabeled vectors have the following shifted symplectic product
relations after the Gram-Schmidt procedure:%
\begin{align}
\left(  \mathbf{u}_{i}\odot\mathbf{v}_{j}\right)  \left(  D\right)   &
=f_{i}\left(  D\right)  \delta_{ij}\ \ \forall\ \ i,j\in\left\{
1,\ldots,c\right\}  ,\nonumber\\
\left(  \mathbf{u}_{i}\odot\mathbf{u}_{j}\right)  \left(  D\right)   &
=0\ \ \ \ \ \ \ \ \ \ \ \ \forall\ \ i,j\in\left\{  1,\ldots,l\right\}
,\nonumber\\
\left(  \mathbf{v}_{i}\odot\mathbf{v}_{j}\right)  \left(  D\right)   &
=0\ \ \ \ \ \ \ \ \ \ \ \ \forall\ \ i,j\in\left\{  1,\ldots,c\right\}  ,
\end{align}
where $f_{i}\left(  D\right)  $ is an arbitrary polynomial. Let us arrange the
above generators in a matrix as follows:%
\begin{equation}%
\begin{bmatrix}
\mathbf{u}_{1}^{T}\left(  D\right)  & \cdots & \mathbf{u}_{c}^{T}\left(
D\right)  & \mathbf{v}_{1}^{T}\left(  D\right)  & \cdots & \mathbf{v}_{c}%
^{T}\left(  D\right)  & \mathbf{u}_{c+1}^{T}\left(  D\right)  & \cdots &
\mathbf{u}_{c+l}^{T}\left(  D\right)
\end{bmatrix}
^{T}.
\end{equation}
We augment the above generators with the following matrix so that all vectors
are orthogonal to each other:%
\begin{equation}
\left[  \left.
\begin{array}
[c]{cccc}%
f_{1}\left(  D\right)  & 0 & \cdots & 0\\
0 & f_{2}\left(  D\right)  &  & \vdots\\
\vdots &  & \ddots & 0\\
0 & \cdots & 0 & f_{c}\left(  D\right) \\
\mathbf{0}_{c\times1} & \mathbf{0}_{c\times1} & \cdots & \mathbf{0}_{c\times
1}\\
\mathbf{0}_{l\times1} & \mathbf{0}_{l\times1} & \cdots & \mathbf{0}_{l\times1}%
\end{array}
\right\vert
\begin{array}
[c]{c}%
\mathbf{0}_{1\times c}\\
\mathbf{0}_{1\times c}\\
\vdots\\
\mathbf{0}_{1\times c}\\
\mathbf{I}_{c\times c}\\
\mathbf{0}_{l\times c}%
\end{array}
\right]  .
\end{equation}

The yield of a protocol using the above construction is $\left(  n-m\right)
/n$. Suppose we use an $\left[  n,k_{1}\right]  $ classical binary
convolutional code for the bit flips and an $\left[  n,k_{2}\right]  $
classical binary convolutional code for the phase flips. Then the
convolutional entanglement distillation protocol has yield $\left(
k_{1}+k_{2}-n\right)  /n$.

\begin{example}
Consider a binary classical convolutional code with the following parity check
matrix:%
\begin{equation}%
\begin{bmatrix}
1+D & D & 1
\end{bmatrix}
.
\end{equation}
We can use the above parity check matrix to correct both bit and phase flip
errors in an entanglement distillation protocol. Our initial quantum parity
check matrix is%
\begin{equation}
\left[  \left.
\begin{array}
[c]{ccc}%
1+D & D & 1\\
0 & 0 & 0
\end{array}
\right\vert
\begin{array}
[c]{ccc}%
0 & 0 & 0\\
1+D & D & 1
\end{array}
\right]  .
\end{equation}
The shifted symplectic product for the first and second row is $D^{-1}+D$. We
therefore augment the above matrix as follows:%
\begin{equation}
\left[  \left.
\begin{array}
[c]{cccc}%
1+D & D & 1 & D^{-1}+D\\
0 & 0 & 0 & 0
\end{array}
\right\vert
\begin{array}
[c]{cccc}%
0 & 0 & 0 & 0\\
1+D & D & 1 & 1
\end{array}
\right]  .
\end{equation}
The above matrix gives a valid stabilizer for use in an entanglement
distillation protocol. The yield of a protocol using the above stabilizer is 1/3.
\end{example}

\section{Closing Remarks}

We constructed a theory of convolutional entanglement distillation. The
entanglement-assisted protocol assumes that the sender and receiver have some
noiseless ebits to use as a catalyst for distilling more ebits. These
protocols have the benefit of lifting the self-orthogonality constraint. Thus
we are able to import an arbitrary classical convolutional code for use in a
convolutional entanglement distillation protocol. The error-correcting
properties and rate of the classical code translate to the quantum case. Brun,
Devetak, and Hsieh first constructed the method for importing an arbitrary
classical block code in their work on entanglement-assisted codes
\cite{arx2006brun,science2006brun}. Our theory of convolutional entanglement
distillation paves the way for exploring protocols that approach the optimal
distillable entanglement by using the well-established theory of classical
convolutional coding.

Convolutional entanglement distillation protocols also hold some key
advantages over block entanglement distillation protocols. They have a higher
yield of ebits, lower decoding complexity, and are an online protocol that a
sender and receiver can employ as they acquire more noisy ebits.

We suggest that convolutional entanglement distillation protocols may bear
some advantages for distillation of a secret key because of the strong
connection between distillation and privacy \cite{PhysRevLett.85.441}. We are
currently investigating whether convolutional entanglement distillation
protocols can improve the secret key rate for quantum key distribution.

\chapter{Conclusion}
\begin{saying}
Fly Photon! Trap Ion! Stay Spin!\\
Which of you we'll implement in?\\
For quantum coherence\\
Demands perseverance,\\
We'll try and God only knows when.
\end{saying}

Quantum error correction theory plays a fundamental role in quantum computing and communication. Without
error-correcting protocols, quantum computing and communication devices will fall prey to
the hands of decoherence. It is crucial
for theorists to continue developing techniques to protect quantum information because a
fundamental discovery in the theory might bring quantum computing closer to reality.

We have augmented the theory of quantum error correction by
contributing a theory of entanglement-assisted quantum convolutional coding.
With the ability to import arbitrary classical convolutional codes, we can now construct quantum convolutional codes
that inherit the desirable characteristics of their ancestral classical convolutional codes.
Our entanglement-assisted quantum convolutional coding theory should be useful for future quantum
communication engineers if they would like to have codes with a good performance/complexity trade-off.

We have said and again stress
that the next important line of investigation is to combine the quantum turbo coding theory \cite{arx2007poulin} with
this theory. Convolutional codes form the constituent codes of a quantum turbo code and it would
be interesting to investigate if we can enhance performance with entanglement-assisted codes
as the constituent codes.
This investigation might produce quantum codes that come close to achieving the entanglement-assisted or ``father'' capacity.

There are also other avenues to pursue---these avenues include any scenario
where entanglement assistance might help. It is useful to
inspect the results of quantum Shannon theory to determine whether we can construct a coding scenario that fits with
the constructions in the asymptotic scenario. We have conducted little analysis
of the performance of our codes beyond stating that they inherit the properties of the imported classical codes.
It would be interesting to observe the performance of these codes in a realistic noisy
quantum channel when using syndrome-based Viterbi processing for correction of quantum error. 

We have seen much creativity in the field of quantum
error correction in the past decade because of the many strange resources available in quantum theory.
Experimentalists are increasingly using an array of quantum error correction techniques with the goal of bringing us closer
to having qubits with good quantum coherence. One can only
imagine what resources future quantum coding theorists will exploit
to protect their valuable quantum information.

\bibliographystyle{unsrt}
\bibliography{wilde-thesis}

\end{document}